%% file: Diplomathesis.tex
\documentclass[a4paper,12pt,titlepage,twoside,dvips]{report}
\usepackage{amsmath} %mathematical symbols etc.
\usepackage{amssymb} %mathematical symbols etc.
\usepackage{amsthm}  %for theorems, propositions, proofs etc.
\usepackage[dvips,xdvi]{graphics} %graphics
\usepackage[dvips,xdvi]{graphicx} %graphics
\usepackage{fancyhdr}
\usepackage{array}
\usepackage{longtable}
\usepackage[figuresleft,clockwise]{rotating}
\usepackage{bbm}

\usepackage[a4paper,twoside, hmarginratio={1:1}]{geometry}
\usepackage[Lenny]{fncychap}
\usepackage{palatino}
\swapnumbers %"1.1 Satz" anstatt "Satz 1.1"
\theoremstyle{definition} %Roman text instead of italic(plain)
\newtheorem{define}{Definition}[section]

\theoremstyle{definition} %Roman text instead of italic(plain)
\newtheorem{prop}[define]{Proposition} %[define] bewirkt, dass prop nicht extra nummeriert wird, sondern zusammen mit define.

\theoremstyle{definition} %Roman text instead of italic(plain)
\newtheorem{defprop}[define]{Definition and Proposition} %[define] bewirkt, dass defprop nicht extra nummeriert wird, sondern zusammen mit define.

\theoremstyle{definition} %Roman text instead of italic(plain)
\newtheorem{lemma}[define]{Lemma} %[define] bewirkt, dass lemma nicht extra nummeriert wird, sondern zusammen mit define.

\theoremstyle{definition} %Roman text instead of italic(plain)
\newtheorem{sublemma}[define]{Sublemma} %[define] bewirkt, dass sublemma nicht extra nummeriert wird, sondern zusammen mit define.

\theoremstyle{definition} %Roman text instead of italic(plain)
\newtheorem{cor}[define]{Corollary} %[define] bewirkt, dass cor nicht extra nummeriert wird, sondern zusammen mit define.

\theoremstyle{definition} %Roman text instead of italic(plain)
\newtheorem{crit}[define]{Criterion} %[define] bewirkt, dass crit nicht extra nummeriert wird, sondern zusammen mit define.

\theoremstyle{definition} %Roman text instead of italic(plain)
\newtheorem{theorem}[define]{Theorem}

\newtheoremstyle{example}% name
{12pt}% space above
{3pt}% space below
{}% body font
{}% indent amount
{\bfseries}% theorem head font
{:}% punctuation after theorem head
{.5em}% space after theorem head
{}%

\theoremstyle{example}
\newtheorem*{example}{Example}

\newtheoremstyle{convention}% name
{12pt}% space above
{3pt}% space below
{}% body font
{}% indent amount
{\bfseries}% theorem head font
{:}% punctuation after theorem head
{.5em}% space after theorem head
{}%

\theoremstyle{convention}
\newtheorem*{conv}{Convention}

\begin{document}

%---------------------------------------------%
%                                             %
%                Titlepage                    %
%                                             %
%---------------------------------------------%

\begin{titlepage}

\hspace{0mm}\\*[2cm]

\begin{center}
\Large Systematic analysis of finite family\\ symmetry groups and their application\\ to the lepton sector\end{center}

\vspace*{3mm}

\begin{center}
\large Patrick Otto Ludl$^{\ast}$
\end{center}

\begin{center}
\small University of Vienna, Faculty of Physics 
\medskip
\\
\small Boltzmanngasse 5, A--1090 Vienna, Austria
\end{center}

\begin{center}
27 July 2010
\end{center}

\vspace*{10mm}

\begin{center}
\textbf{Abstract}
\end{center}

In this work we will investigate Lagrangians of the standard model extended by three right-handed neutrinos, and the consequences of invariance under finite groups $G$ for lepton masses and mixing matrices are studied. The main part of this work is the systematic analysis of finite subgroups of $SU(3)$. The analysis of these groups may act as a toolkit for future model building.
\\
The original version of this work has been published as a diploma thesis at the University of Vienna in 2009. Compared to the original version some minor errors have been corrected.
\\
\vspace*{4cm}
\\
$^{\ast}$E-mail: patrick.ludl@univie.ac.at

\end{titlepage}

\newpage
\pagestyle{empty}
\hspace{0mm}

\newpage
\pagestyle{empty}
\hspace{0mm}

	\begin{flushright}
	\includegraphics[scale=0.25, keepaspectratio=true]{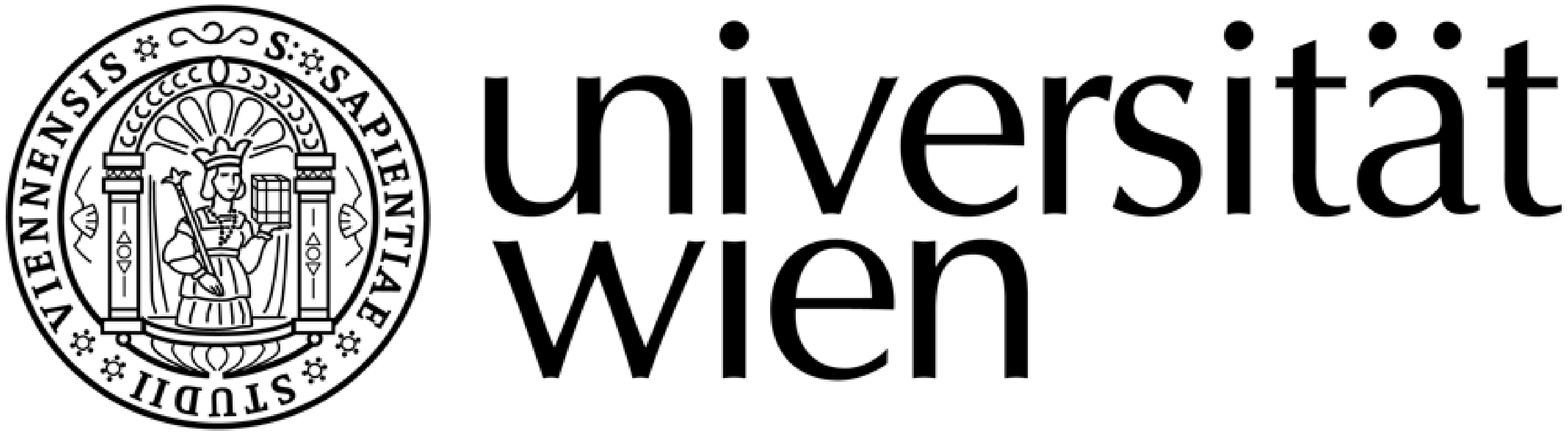}
	\end{flushright}
	\vspace*{2cm}
	\begin{center}
		\huge{Diplomarbeit}
		\bigskip
		\bigskip
		\bigskip
		\\
		\normalsize{Titel der Diplomarbeit}
		\bigskip
		\bigskip
		\begin{quote}
		\LARGE{\textit{Systematic analysis of finite family symmetry groups and their application to the lepton sector}}
		\end{quote}
		\bigskip
		\bigskip
		\normalsize{angestrebter akademischer Grad}
		\bigskip
		\bigskip
		\\
		\large{Magister der Naturwissenschaften (Mag. rer. nat.)}
	\end{center}
\vspace*{3cm}
Verfasser: Patrick Otto Ludl
\smallskip
\\
Matrikelnummer: 0301971
\smallskip
\\
Studienrichtung: A 411 Diplomstudium Physik
\smallskip
\\
Betreuer: Ao. Univ.-Prof. Dr. Walter Grimus
\bigskip
\bigskip
\bigskip
\bigskip
\\
Wien, am 6. Juni 2009

\newpage

\pagestyle{empty}
\hspace{0mm}

\newpage

\vspace*{7cm}
\begin{center}
\begin{large}\textit{Meinem Vater in Liebe gewidmet.}\end{large}
\end{center}

\newpage
\hspace{0mm}

\newpage
\vspace*{3cm}
\begin{center}
\begin{Large}Danksagung\end{Large}
\end{center}
\bigskip
\bigskip
Ich m\"ochte mich herzlich bei jenen Personen bedanken, ohne die diese Arbeit nicht zustande kommen h\"atte k\"onnen.
\medskip
\\
Besonders herzlichen Dank m\"ochte ich meinem Betreuer Walter Grimus f\"ur die ausgezeichnete Betreuung und vor allem f\"ur die vielen Stunden fruchtbarer Dis\-kus\-si\-on\-en aussprechen. W\"ahrend der Entstehung dieser Arbeit hat er sich stets Zeit genommen Probleme zu er\"ortern, Neuigkeiten zu diskutieren und Hilfestellungen zu bieten.
\\
Gro\ss er Dank geb\"uhrt auch meinem guten Freund und Studienkollegen Andreas Singraber, der mir bei vielen kleinen Problemen, insbesondere auch bei Computerproblemen, stets mit Rat und Tat zur Seite gestanden hat, und der sich immer Zeit genommen hat diese Probleme mit mir zu diskutieren.
\\
Dank und Anerkennung sei auch allen meinen Freunden, Bekannten und Studienkollegen ausgesprochen. Sie stellten die Atmosph\"are her, die erfolgreiches Arbeiten erst m\"oglich macht.
\\
Dankbare Anerkennung geb\"uhrt auch allen meinen Lehrern und Professoren. Durch die Begeisterung f\"ur ihre F\"acher inspirieren sie  Sch\"uler und Studenten und legen so den Grundstein f\"ur den Fortschritt von Wissenschaft und Kultur.
\\
Von ganzem Herzen m\"ochte ich auch meinen Eltern danken, ohne denen mir mein Studium nicht m\"oglich gewesen w\"are, und die mir in allen Lebenslagen stets hilfreich zur Seite standen.
\bigskip
\\
Unserem Sch\"opfer danke ich, dass er mir die F\"ahigkeit verliehen hat das Stu\-di\-um der Naturwissenschaften mit Neugier und Eifer zu betreiben. Ich hoffe, dass meine Arbeit einen kleinen Anteil zum Verst\"andnis seiner Sch\"opfung beitragen kann.

\newpage

\pagestyle{empty}
\hspace{0mm}

\newpage

\vspace*{7cm}
\begin{center}
\begin{large}\textit{I want to know God's thoughts, the rest are details.}\end{large}
\end{center}
\hspace{9.5cm}Albert Einstein

\newpage
\hspace{0mm}

\newpage

\pagestyle{fancy}

\fancyhf{}
\fancyhead[RO,LE]{ \thepage}
\fancyhead[RE]{\small \rmfamily \nouppercase{\leftmark}}
\fancyhead[LO]{\small \rmfamily \nouppercase{\rightmark}}
\fancyfoot{} %no footline, especially no page numbering on the lower side of the page

%---------------------------------------------%
%                                             %
%             Table of contents               %
%                                             %
%---------------------------------------------%

\tableofcontents

\newpage

\fancyhf{}
\fancyhead[RO,LE]{ \thepage}
\fancyhead[RE]{\small \rmfamily \nouppercase{Foreword}}
\fancyhead[LO]{\small \rmfamily \nouppercase{Foreword}}
\fancyfoot{} %no footline, especially no page numbering on the lower side of the page

\input{foreword/foreword.tex}

\newpage

\fancyhf{}
\fancyhead[RO,LE]{ \thepage}
\fancyhead[RE]{\small \rmfamily \nouppercase{\leftmark}}
\fancyhead[LO]{\small \rmfamily \nouppercase{\rightmark}}
\fancyfoot{} %no footline, especially no page numbering on the lower side of the page

\input{sm/sm}

\input{neutrinooscillations/neutrinooscillations.tex}

\input{family/family.tex}
\input{yukawa/yukawa.tex}

\input{clebschgordan/clebschgordan.tex}

\input{su3/su3.tex}

\input{leptonsector/leptonsector}

\newpage

\fancyhf{}
\fancyhead[RO,LE]{ \thepage}
\fancyhead[RE]{\small \rmfamily \nouppercase{Conclusions}}
\fancyhead[LO]{\small \rmfamily \nouppercase{Conclusions}}
\fancyfoot{} %no footline, especially no page numbering on the lower side of the page

\input{conclusions/conclusions.tex}

\newpage

\fancyhf{}
\fancyhead[RO,LE]{ \thepage}
\fancyhead[RE]{\small \rmfamily \nouppercase{\leftmark}}
\fancyhead[LO]{\small \rmfamily \nouppercase{\rightmark}}
\fancyfoot{} %no footline, especially no page numbering on the lower side of the page

\begin{appendix}
\input{group_theory/group_theory.tex}

\input{lie/lie.tex}

\input{majorana/majorana.tex}

\input{programexamples/examples.tex}

\end{appendix}

\input{symbols/symbols.tex}
\input{conventions/conventions.tex}

\newpage

\fancyhf{}
\fancyhead[RO,LE]{ \thepage}
\fancyhead[RE]{\small \rmfamily \nouppercase{List of tables}}
\fancyhead[LO]{\small \rmfamily \nouppercase{List of tables}}
\fancyfoot{} %no footline, especially no page numbering on the lower side of the page

\addcontentsline{toc}{chapter}{List of Tables}
\listoftables

\newpage

\fancyhf{}
\fancyhead[RO,LE]{ \thepage}
\fancyhead[RE]{\small \rmfamily \nouppercase{Bibliography}}
\fancyhead[LO]{\small \rmfamily \nouppercase{Bibliography}}
\fancyfoot{} %no footline, especially no page numbering on the lower side of the page

\begin{small}

\input{references/references.tex}
\end{small}

\newpage

\fancyhf{}
\fancyhead[RO,LE]{ \thepage}
\fancyhead[RE]{\small \rmfamily \nouppercase{Zusammenfassung}}
\fancyhead[LO]{\small \rmfamily \nouppercase{Zusammenfassung}}
\fancyfoot{} %no footline, especially no page numbering on the lower side of the page

\input{zusammenfassung/zusammenfassung.tex}

\newpage

\fancyhf{}
\fancyhead[RO,LE]{ \thepage}
\fancyhead[RE]{\small \rmfamily \nouppercase{Curriculum vitae}}
\fancyhead[LO]{\small \rmfamily \nouppercase{Curriculum vitae}}
\fancyfoot{} %no footline, especially no page numbering on the lower side of the page

\input{cv/cv.tex}

\end{document}

%% file: foreword/foreword.tex
\chapter*{Foreword}
When Joseph John Thomson discovered the electron in 1897 he could not have imagined that he had become a member of the \textquotedblleft great ancestors\textquotedblright\hspace{1mm} who led the foundations for the unimaginably wide and beautiful field of particle physics, neither could he have imagined that the electron is only one member of the lepton sector involving the electron's \textquotedblleft heavier relatives\textquotedblright, the muon and the tauon, as well as the corresponding neutrinos whose masses are so small that they withstood direct measurement up to now.
\medskip
\\
While the gauge bosons $\gamma, W^{\pm}$ and $Z$ are manifestations of the mathematical formalism of gauge field theories there is no successful explanation for the fermion spectrum of the standard model. The search for a mathematical theory describing the (currently known) 3 generations of fermions remained unsuccessful within the 20\textsuperscript{th} century, but there is hope that at least a little glance of light is shed onto this problem at the beginning of the 21\textsuperscript{st} century.
\\
The most striking experimental fact is that while the mass ratios within one of the fermion families (charged leptons, neutrinos, up-quarks, down-quarks) can take enormous values like
	\begin{displaymath}
	\frac{m_{\tau}}{m_{e}}\sim 3500
	\end{displaymath}
the couplings to the gauge fields are equal within a fermion family. This is why one usually interprets the heavier members of a fermion family as the \textquotedblleft heavy relatives\textquotedblright\hspace{1mm} of the lightest member. By now there is no satisfying explanation for the existence of the three generations, which is outlined by Abraham Pais in a very pointed way, when he writes about the muon:
	\begin{quote}
	\textit{The new punchline, 'There is a muon', may have caused laughter in the heavens, but man was, and still is, ignorant of the joke. What else was the muon good for other than being the pion's favorite decay product? To be sure, the discovery of the electron had also been unexpected, but its universal use as an ingredient of the atomic periphery was recognized rapidly. The neutron, less of a surprise, made it possible almost at once to develop theories of nuclear structure and $\beta$-decay. But the muon? Now, forty years later, the divine laughter continues unabated.}
	\smallskip
	\\
	\cite{pais}, p.454f.
	\end{quote}
The only point of the standard model where a difference between the members of a family is implemented in the Lagrangian is the Yukawa-coupling that generates the fermion mass terms via spontaneous symmetry breaking. At this point there occur connections to other measurable quantities - the fermion mixing matrices $U_{CKM}$ and $U_{PMNS}$.
\\
The lepton mixing matrix $U_{PMNS}$ seems to be quite close to the so-called Harrison-Perkins-Scott mixing matrix
	\begin{displaymath}
	U_{\mathrm{HPS}}=\left(\begin{matrix}
		 \sqrt{\frac{2}{3}} & \frac{1}{\sqrt{3}} & 0 \\
		 -\frac{1}{\sqrt{6}} & \frac{1}{\sqrt{3}} & \frac{1}{\sqrt{2}} \\
		 \frac{1}{\sqrt{6}} & -\frac{1}{\sqrt{3}} & \frac{1}{\sqrt{2}}
	\end{matrix}\right).
	\end{displaymath}
The nice appearance of this matrix induced the idea of an underlying discrete symmetry in the Lagrangian. In this work we will investigate Lagrangians of the standard model extended by three right-handed neutrinos, and the consequences of invariance under finite groups $G$ for lepton masses and mixing matrices are studied. The main part of this thesis will be the systematic analysis of finite subgroups of $SU(3)$. I hope that my analysis of these groups can act as a toolkit for future model building.
\medskip
\\
A future discovery of an appropriate finite group that can describe lepton masses and mixing will not solve the mystery of the three generations of fermions, but maybe an appropriate symmetry group can give a hint on an underlying more general theory that includes the present day standard model as a suitable limit.

%% file: sm/sm.tex
\chapter{Introduction}

\section{The GWS-theory of electroweak interactions}
The GWS-theory\footnote{GWS = S.L. Glashow, S. Weinberg, A. Salam. \cite{Glashow,Weinberg,Salam}} of electroweak interactions \cite{Glashow,Weinberg,Salam} is currently the most established theory in particle physics. Experimentally it stands on firm ground, and together with quantum chromodynamics it forms the standard model of particle physics - the currently best description of the microscopic world. However the standard model has some missing features, for example it cannot describe neutrino masses without being extended.
\medskip
\\
At first we will summarize the basic features of the GWS-model of electroweak interactions.

\subsection{The $SU(2)_{I}\times U(1)_{Y}$-theory of electroweak interactions}

The GWS-theory is a gauge theory based on the gauge group $SU(2)_{I}\times U(1)_{Y}$. $I$ is called \textit{weak isospin}, and $Y$ is called \textit{weak hypercharge}. The Lagrangian of the $SU(2)_{I}\times U(1)_{Y}$-theory is
	\begin{equation}
		\mathcal{L}_{SU(2)_{I}\times U(1)_{Y}}=-\frac{1}{2}\mathrm{Tr}(W_{\lambda\rho}W^{\lambda\rho})-\frac{1}{4}B_{\lambda\rho}B^{\lambda\rho}+\sum_{\alpha=e,\mu,\tau}\bar{\psi}_{\alpha}i\gamma^{\mu}D_{\mu}\psi_{\alpha}
	\end{equation}
with
	\begin{displaymath}
		W_{\lambda\rho}=\partial_{\lambda}W_{\rho}-\partial_{\rho}W_{\lambda}+ig[W_{\lambda},W_{\rho}],\quad\quad B_{\lambda\rho}=\partial_{\lambda}B_{\rho}-\partial_{\rho}B_{\lambda}.
	\end{displaymath}
$W_{\lambda}(x)=W_{\lambda}^{a}(x)\frac{\tau^{a}}{2}$ is the $SU(2)_{I}$-gauge field, where $\lbrace -i\frac{\tau^{a}}{2}\rbrace_{a=1,2,3}$ forms a basis of the Lie algebra $su(2)$. $B_{\lambda}$ is the $U(1)_{Y}$-gauge field.

\begin{table}
\begin{center}
\renewcommand{\arraystretch}{1.4}
\begin{tabular}{|c@{\vrule width1pt}cc|}
\firsthline
	1\textsuperscript{st} Generation \rule{3mm}{0mm} & \rule{3mm}{0mm} $e^{-}$ & electron \rule{3mm}{0mm}\\ 
	× & \rule{3mm}{0mm} $\nu_{e}$ & electron neutrino \rule{3mm}{0mm}\\
\hline
	2\textsuperscript{nd} Generation \rule{3mm}{0mm} & \rule{3mm}{0mm} $\mu^{-}$ & muon \rule{3mm}{0mm}\\ 
	× & \rule{3mm}{0mm} $\nu_{\mu}$ & muon neutrino \rule{3mm}{0mm}\\ 
\hline
	3\textsuperscript{rd} Generation \rule{3mm}{0mm} & \rule{3mm}{0mm} $\tau^{-}$ & tauon \rule{3mm}{0mm}\\ 
	× & \rule{3mm}{0mm} $\nu_{\tau}$ & tauon neutrino \rule{3mm}{0mm}\\
\lasthline
\end{tabular}
\caption{The spin-$\frac{1}{2}$ fermions of the GWS-theory}
\label{GWSfermions}%label immer nach caption setzen, sonst funktioniert die Referenzierung nicht richtig!
\end{center}
\end{table}
\hspace{0mm}\\
Table \ref{GWSfermions} lists the fermions contained in the GWS-theory. The interaction between gauge bosons and fermions is described by the covariant derivative
	\begin{equation}
		D_{\mu}\psi_{\alpha}=(\partial_{\mu}+igW_{\mu}^{a}T^{a}+ig'B_{\mu}\frac{Y}{2})\psi_{\alpha},
	\end{equation}
where
	\begin{displaymath}
		T^{1}=\frac{1}{2}\left(
			\begin{matrix}
			 0 & 1 & 0 \\
			 1 & 0 & 0 \\
			 0 & 0 & 0
			\end{matrix}
			\right),
			T^{2}=\frac{1}{2}\left(
			\begin{matrix}
			 0 & -i & 0 \\
			 i & 0 & 0 \\
			 0 & 0 & 0
			\end{matrix}
			\right),
			T^{3}=\frac{1}{2}\left(
			\begin{matrix}
			 1 & 0 & 0 \\
			 0 & -1 & 0 \\
			 0 & 0 & 0
			\end{matrix}
			\right),
	\end{displaymath}
and
	\begin{displaymath}
		Y=\left(
			\begin{matrix}
			 Y_{L} & 0 & 0 \\
			 0 & Y_{L} & 0 \\
			 0 & 0 & Y_{R}
			\end{matrix}
			\right), \psi_{\alpha}=\left(\begin{matrix}
					\nu_{\alpha L} \\
					\alpha_{L} \\
					\alpha_{R}
			              \end{matrix}
					\right), D_{\alpha L}=\left(\begin{matrix} \nu_{\alpha L} \\ \alpha_{L} \end{matrix}\right).
	\end{displaymath}
$T^{a}$ correspond to $\frac{\tau^{a}}{2}$ acting on the $SU(2)_{I}$-doublet $D_{\alpha L}$. $Y_{L}$ is the hypercharge of $\nu_{\alpha L}$ and $\alpha_{L}$, $Y_{R}$ is the hypercharge of $\alpha_{R}$. The indices $R$ and $L$ label the left- and right-handed chiral fermion fields. (See section \ref{Chiralfields}.)
\bigskip
\\
\textbf{$SU(2)_{I}$-multiplets of the GWS-theory}
\medskip
\\
Table \ref{GWSmultiplets} lists the $SU(2)_{I}$-multiplets of the GWS-theory. $Q=I_{3}+\frac{1}{2}Y$ is the electrical charge.\footnote{We use the convention of \cite{grimus1} for the isospins and hypercharges of the multiplets.}
\begin{table}
\begin{center}
\renewcommand{\arraystretch}{1.4}
\begin{tabular}{|cccccc|}
\firsthline
	Fermion multiplets \rule{3mm}{0mm} & \rule{3mm}{0mm} $I$ & $I_{3}$ & $Y$ & $Q$ & \\
\hline
\rule{0mm}{11mm}
	$\left(\begin{matrix} \nu_{eL}\\e_{L} \end{matrix}\right)$, $\left(\begin{matrix} \nu_{\mu L}\\ \mu_{L} \end{matrix}\right)$, $\left(\begin{matrix} \nu_{\tau L}\\\tau_{L} \end{matrix}\right)$ & $\frac{1}{2}$ & $\frac{\tau^3}{2}$ & $-\mathbbm{1}_2$ & $\left(\begin{matrix} 0 & 0\\ 0 & -1 \end{matrix}\right)$ & left-handed doublets \\
	$e_{R}$, $\mu_{R}$, $\tau_{R}$ & $0$ & $0$ & $-2$ & $-1$ & right-handed singlets\\
 
\lasthline
\end{tabular}
\caption{$SU(2)_{I}$-multiplets of the GWS-theory}
\label{GWSmultiplets}
\end{center}
\end{table}
\bigskip
\\
\textbf{Gauge bosons in the $SU(2)_{I}\times U(1)_{Y}$-theory}
\medskip
\\
Table \ref{gaugebosons} lists the gauge bosons of the $SU(2)_{I}\times U(1)_{Y}$-theory. The definitions of the gauge bosons contain the \textit{weak mixing angle} $\vartheta_{W}$.
	\begin{displaymath}
		\mathrm{cos}\vartheta_{W}=\frac{g}{\sqrt{g^{2}+g'^{2}}}
	\end{displaymath}
Comparing the parts of the Lagrangian containing $A_{\mu}$ with the Lagrangian for QED one finds
	\begin{displaymath}
		e=g\mathrm{sin}\vartheta_{W}=g'\mathrm{cos}\vartheta_{W}.
	\end{displaymath}

\begin{table}
\begin{center}
\renewcommand{\arraystretch}{1.4}
\begin{tabular}{|ccc|}
\firsthline
	$W_{\mu}^{+}=\frac{1}{\sqrt{2}}(W_{\mu}^{1}-iW_{\mu}^{2})$ \rule{3mm}{0mm} & \rule{3mm}{0mm} $W^{+}$ & charged W-boson \rule{3mm}{0mm}\\ 

	$W_{\mu}^{-}=\frac{1}{\sqrt{2}}(W_{\mu}^{1}+iW_{\mu}^{2})$ \rule{3mm}{0mm} & \rule{3mm}{0mm} $W^{-}$ & charged W-boson \rule{3mm}{0mm}\\ 

	$Z_{\mu}=W_{\mu}^{3}\mathrm{cos}\vartheta_{W}-B_{\mu}\mathrm{sin}\vartheta_{W}$ \rule{3mm}{0mm} & \rule{3mm}{0mm} $Z$ & Z-boson \rule{3mm}{0mm}\\

	$A_{\mu}=W_{\mu}^{3}\mathrm{sin}\vartheta_{W}+B_{\mu}\mathrm{cos}\vartheta_{W}$ \rule{3mm}{0mm} & \rule{3mm}{0mm} $\gamma$ & photon \rule{3mm}{0mm}\\  
\lasthline
\end{tabular}
\caption{The gauge bosons of the $SU(2)_{I}\times U(1)_{Y}$-theory}
\label{gaugebosons}
\end{center}
\end{table}

\subsection{The Higgs-mechanism for the $SU(2)_{I}\times U(1)_{Y}$-theory}
The $SU(2)_{I}\times U(1)_{Y}$-theory does not allow mass terms for the described particles (because these terms would break gauge invariance), which is of course a severe problem. The well known solution to this problem is the so called \textit{Higgs-mechanism} \cite{Higgs1,Higgs2,Higgs3,Englert,Guralnik}.
\bigskip
\\
In the standard model masses are generated through spontaneous symmetry breaking of the $SU(2)_{I}\times U(1)_{Y}$-gauge group by a \textit{Higgs doublet}
	\begin{displaymath}
		\phi=\left(\begin{matrix}
			  \phi_{1}\\\phi_{2} 
		     \end{matrix}\right),
	\end{displaymath}
which is an $SU(2)_{I}$-doublet consisting of two complex valued scalar fields $\phi_{1}$ and $\phi_{2}$. The Lagrangian reads
	\begin{equation}\label{HiggsLagrangianEqu}
		\mathcal{L}_{\mathrm{Higgs}}=(D_{\mu}\phi)^{\dagger}(D^{\mu}\phi)-V(\phi),
	\end{equation}
where $V(\phi)$ is the well known \textit{Higgs potential}
	\begin{equation}
		V(\phi)=\mu^{2}\rho^{2}+\lambda\rho^{4}
	\end{equation}
with
	\begin{displaymath}
		\rho^{2}=\phi^{\dagger}\phi \quad\mbox{and}\quad \lambda>0, \enspace \mu^{2}<0.
	\end{displaymath}
$V(\phi)$ has a continuous set of minima $\phi_{0}$ which is characterized by
	\begin{displaymath}
		\rho_{0}^{2}=-\frac{\mu^{2}}{2\lambda}.
	\end{displaymath}
For all possible $\phi_{0}$ are minimizers of $V(\phi)$ they are called the \textquotedblleft vacua\textquotedblright\hspace{1mm} of the system. In fact when the system evolves to the ground state it will go to one of the possible states $\phi_{0}$, which all have the same probability. This phenomenon is what is meant by the term \textquotedblleft spontaneous symmetry breaking\textquotedblright, because when the system goes to a special ground state $\phi_{0}$ the symmetry is hidden.
\medskip
\\
The covariant derivative in (\ref{HiggsLagrangianEqu}) is given by
	\begin{equation}
		D_{\mu}\phi=(\partial_{\mu}+igW_{\mu}^{a}\frac{\tau^{a}}{2}+ig'B_{\mu}\frac{Y_{H}}{2}\mathbbm{1}_{2})\phi,
	\end{equation}
where $Y_{H}$ is the hypercharge of the Higgs doublet $\phi$.
\bigskip
\\
\textbf{Spontaneous symmetry breaking}
\medskip
\\
The procedure of spontaneous symmetry breaking of gauge theories is the following:
\medskip
\\
The Higgs doublet $\phi$, which has a non vanishing vacuum expectation value, is replaced by a new field $H$ via
	\begin{displaymath}
		\phi=\langle0\vert\phi\vert0\rangle+H.
	\end{displaymath}
Therefore the vacuum expectation value of $H$ is zero. Because of $SU(2)_{I}$ gauge freedom $\phi$ can always be chosen to have the form
	\begin{equation}\label{unitarygauge}
		\phi=\left(\begin{matrix}
			  0\\ \frac{1}{\sqrt{2}}(v+h)
		     \end{matrix}\right),
	\end{equation}
with
	\begin{displaymath}
		     \langle0\vert\phi\vert0\rangle=\left(\begin{matrix}
			  0\\ \frac{1}{\sqrt{2}}v
		     \end{matrix}\right), \quad H=\left(\begin{matrix}
			  0\\ \frac{1}{\sqrt{2}}h
		     \end{matrix}\right),
	\end{displaymath}
where $h$ is a real valued scalar field with vacuum expectation value zero. This gauge is called \textit{unitary gauge}.
\bigskip
\\
\textbf{Gauge boson masses}
\medskip
\\
Fixing $Y_{H}=1$ and using the unitary gauge (which we indicate by an index $u$) one gets
	\begin{equation}
		\begin{split}
		\mathcal{L}_{\mathrm{Higgs}}^{(u)} &= \frac{1}{2}(\partial_{\mu}h)(\partial^{\mu}h)+\\
		& +\frac{g^{2}}{4}W_{\mu}^{+}W^{\mu-}(v+h)^{2}+\frac{g^{2}+g'^{2}}{8}Z_{\mu}Z^{\mu}(v+h)^{2}+\\
		& -\frac{1}{2}(2\lambda v^{2})h^{2}\left[1+\frac{h}{v}+\frac{1}{4}\left(\frac{h}{v}\right)^{2}\right]+\frac{\lambda v^{4}}{4},
		\end{split}
	\end{equation}
which leads to
	\begin{displaymath}
		m_{W}^{2}=\frac{g^{2}v^{2}}{4}, \quad\quad m_{Z}^{2}=\frac{(g^{2}+g'^{2})v^{2}}{4}, \quad\quad m_{h}^{2}=2\lambda v^{2}
	\end{displaymath}
and the famous relation
	\begin{displaymath}
		\frac{m_{W}}{m_{Z}}=\mathrm{cos}\vartheta_{W}.
	\end{displaymath}
\bigskip
\\
\textbf{Fermion masses}
\medskip
\\
The Higgs mechanism offers the possibility to describe massive fermions too. For this purpose one introduces \textit{Yukawa-couplings} between the fermions and the Higgs doublets.
	\begin{equation}
		\mathcal{L}_{\mathrm{Yukawa}}=-\sum_{\alpha=e,\mu,\tau}c_{\alpha}\bar{\alpha}_{R}\phi^{\dagger}D_{\alpha L} + \mathrm{H.c.},
	\end{equation}
where $c_{\alpha}$ are coupling constants.
\medskip
\\
Breaking the symmetry by $\phi\mapsto \langle0\vert\phi\vert0\rangle+H$ and using the unitary gauge (\ref{unitarygauge}) one gets mass terms for the fermions.
	\begin{displaymath}
		\mathcal{L}_{\mathrm{Yukawa}}^{(u)}=-\sum_{\alpha=e,\mu,\tau}(\frac{c_{\alpha}\rho_{0}}{\sqrt{2}}\bar{\alpha}\alpha+\frac{c_{\alpha}h}{\sqrt{2}}\bar{\alpha}\alpha)
	\end{displaymath}
Thus the fermion masses become
	\begin{displaymath}
		m_{\alpha}=\frac{c_{\alpha}\rho_{0}}{\sqrt{2}} \quad\quad m_{\nu_{\alpha L}}=0.
	\end{displaymath}
The absence of right-handed neutrino fields $\nu_{\alpha R}$ and the simple structure of the scalar sector, consisting of only one Higgs doublet, are the reasons for massless neutrinos in the GWS-model.

\subsection{The Lagrangian of the GWS-model}
The Lagrangians considered so far build the GWS-theory of electroweak interactions.
	\begin{equation}
		\mathcal{L}_{\mathrm{GWS}}= \mathcal{L}_{SU(2)_{I}\times U(1)_{Y}}+\mathcal{L}_{\mathrm{Higgs}}+\mathcal{L}_{\mathrm{Yukawa}}
	\end{equation}
\bigskip
	\begin{displaymath}
		\begin{split}
		\mathcal{L}_{\mathrm{GWS}} & = -\frac{1}{2}\mathrm{Tr}(W_{\lambda\rho}W^{\lambda\rho})-\frac{1}{4}B_{\lambda\rho}B^{\lambda\rho}+\\
		&
		+(D_{\mu}\phi)^{\dagger}(D^{\mu}\phi)-V(\phi)+\\		
 		&+\sum_{\alpha=e,\mu,\tau}(\bar{\psi}_{\alpha}i\gamma^{\mu}D_{\mu}\psi_{\alpha}-c_{\alpha}\bar{\alpha}_{R}\phi^{\dagger}D_{\alpha L}-c_{\alpha}\bar{D}_{\alpha L}\phi\alpha_{R})
		\end{split}
	\end{displaymath}

\section{Neutrino masses and mixing}
From the phenomenon of \textit{neutrino oscillations} we know that at least two neutrinos are massive, and that the mass eigenstates $\vert\nu_{i}^m\rangle$ do not coincide with the flavour eigenstates $\vert\nu_{\alpha}\rangle$. (For further details we refer the reader to chapter \ref{neutrinoosc} on neutrino oscillations.) We are therefore looking for extensions of the standard model.

\subsection{The GIM-model}\label{GIMsubsection}
Both neutrino masses and lepton mixing can be described by a method similar to the GIM-construction\footnote{GIM = S.L. Glashow, J. Iliopoulos, L. Maiani. An excellent description of the GIM-construction can be found in \cite{horejsi}, the original work of GIM can be found in \cite{gim}.} in the quark sector. For this purpose one adds three right-handed neutrinos
	\begin{displaymath}
		\nu_{\alpha R},
	\end{displaymath}
which would lead to the following model for lepton mixing, which we will refer to as the \textit{GIM-model}. Note that in this model neutrinos will be \textit{Dirac-particles}.
	\begin{displaymath}
		\nu_{\alpha}=\nu_{\alpha L}+\nu_{\alpha R}
	\end{displaymath}
Adding right-handed neutrinos the fermion content of the model is given in table \ref{GIMmultiplets}.
\begin{table}
\begin{center}
\renewcommand{\arraystretch}{1.4}
\begin{tabular}{|cccccc|}
\firsthline
	Fermion multiplets \rule{3mm}{0mm} & \rule{3mm}{0mm} $I$ & $I_{3}$ & $Y$ & $Q$ & \\
\hline
\rule{0mm}{11mm}
	$\left(\begin{matrix} \nu_{eL}\\e_{L}' \end{matrix}\right)$, $\left(\begin{matrix} \nu_{\mu L}\\ \mu_{L}' \end{matrix}\right)$, $\left(\begin{matrix} \nu_{\tau L}\\\tau_{L}' \end{matrix}\right)$ & $\frac{1}{2}$ & $\frac{1}{2}$ & $-1$ & $0$ & left-handed doublets \\
	$e_{R}'$, $\mu_{R}'$, $\tau_{R}'$ & $0$ & $0$ & $-2$ & $-1$ & right-handed singlets\\
	$\nu_{eR}$, $\nu_{\mu R}$, $\nu_{\tau R}$ & $0$ & $0$ & $0$ & $0$ & right-handed singlets\\
\lasthline
\end{tabular}
\caption{$SU(2)_{I}$-multiplets of the GIM-model}
\label{GIMmultiplets}
\end{center}
\end{table}
The Yukawa-couplings are
	\begin{equation}
		\mathcal{L}_{\mathrm{Yukawa}}^{\mathrm{GIM} (cl)}=-\sum_{\alpha,\beta=e,\mu,\tau}M_{\alpha\beta}^{cl}\bar{\alpha}'_{R}\phi^{\dagger}D_{\beta L}' + \mathrm{H.c.},
	\end{equation}
	\begin{equation}
		\mathcal{L}_{\mathrm{Yukawa}}^{\mathrm{GIM} (\nu)}=-\sum_{\alpha,\beta=e,\mu,\tau}M^{\nu}_{\alpha\beta}\bar{\nu}_{\alpha R}\tilde{\phi}^{\dagger}D_{\beta L}' + \mathrm{H.c.},
	\end{equation}
with $\tilde{\phi}=i\tau^{2}\phi^{\ast}$ and non-singular complex $3\times 3$-matrices $M^{cl}$ and $M^{\nu}$. The indices $\nu$ and $cl$ stand for neutrinos and charged leptons. $D_{\alpha L}'=\left(\begin{array}{cc} \nu_{\alpha L} & \alpha_{L}'\end{array}\right)^{T}$.
\\
In the following we will frequently use the vectors
	\begin{displaymath}
	\nu_{L,R}=\left(
\begin{matrix}
 \nu_{e L,R} \\ \nu_{\mu L,R} \\ \nu_{\tau L,R}
\end{matrix}
\right)
,\quad\quad \mathnormal{l}_{L,R}'=\left(
\begin{matrix}
 e_{L,R}' \\ \mu_{L,R}' \\ \tau_{L,R}'
\end{matrix}
\right).
	\end{displaymath}
The primed charged lepton fields $\alpha_{L,R}'$ do not denote the standard model charged fermion fields. It will later turn out that the fields $\alpha_{L,R}'$ are linear combinations of the standard model charged fermions $\alpha_{L,R}$. As a consequence the fields $\alpha_{L,R}'$ are not mass eigenfields.
\medskip
\\
Using the unitary gauge (\ref{unitarygauge}) one gets
	\begin{equation}
	\begin{split}
		& \mathcal{L}_{\mathrm{Yukawa}}^{\mathrm{GIM} (cl)(u)}=-\frac{1}{\sqrt{2}}(v+h)
		\bar{l}_{R}'M^{cl}l_{L}'
		 + \mathrm{H.c.},\\
		& \mathcal{L}_{\mathrm{Yukawa}}^{\mathrm{GIM} (\nu)(u)}=
		-\frac{1}{\sqrt{2}}(v+h)
		\bar{\nu}_{R}M^{\nu}
		\nu_{L}
		 + \mathrm{H.c.},
	\end{split}
	\end{equation}
where $\mathcal{M}^{\nu}=\frac{v}{\sqrt{2}}M^{\nu}$ and $\mathcal{M}^{cl}=\frac{v}{\sqrt{2}}M^{cl}$ are the \textit{Dirac mass matrices} for the neutrinos and the charged leptons, respectively.
To find the fermion masses one has to diagonalize the mass matrices $\mathcal{M}^{\nu}$ and $\mathcal{M}^{cl}$ such that all diagonal entries are positive (because the masses must of course be positive). This can be managed by the following theorem:

\begin{theorem}\label{TY1}
	Let $\mathcal{M}$ be a non-singular complex square matrix. Then there exist unitary matrices $U$ and $V$ such that
		\begin{displaymath}
			\mathcal{M}=V\hat{\mathcal{M}}U^{\dagger},
		\end{displaymath}
where $\hat{\mathcal{M}}$ is diagonal, real and positive.
\end{theorem}

\begin{proof}\footnote{This proof has been adapted from \cite{horejsi} (p.230f.).}
	Let $\mathcal{M}$ be a non-singular complex $n\times n$-matrix. $\Rightarrow$ $\mathcal{M}\mathcal{M}^{\dagger}$ is Hermitian and positive. $\Rightarrow$ $\mathcal{M}\mathcal{M}^{\dagger}$ is diagonalizeable and has only positive eigenvalues.
	\begin{displaymath}
		\mathcal{M}\mathcal{M}^{\dagger}=V\hat{\mathcal{M}}^{2}V^{\dagger},
	\end{displaymath}
with $\hat{\mathcal{M}}=\mathrm{diag}(m_{1},...,m_{n}), m_{i}\in \mathbb{R}^{+}$.
\medskip
\\
Now we define
	\begin{displaymath}
		U^{\dagger}=\hat{\mathcal{M}}^{-1}V^{\dagger}\mathcal{M}.
	\end{displaymath}
$U$ is unitary, because 
	\begin{displaymath}
		U^{\dagger}U=\hat{\mathcal{M}}^{-1}\underbrace{V^{\dagger}\mathcal{M}\mathcal{M}^{\dagger}V}_{\hat{\mathcal{M}}^{2}}(\hat{\mathcal{M}}^{-1})^{\dagger}=\hat{\mathcal{M}}^{-1}\hat{\mathcal{M}}\hat{\mathcal{M}}\hat{\mathcal{M}}^{-1}=\mathbbm{1}_{n}.
	\end{displaymath}
	\begin{displaymath}
		\Rightarrow V\hat{\mathcal{M}}U^{\dagger}=V\hat{\mathcal{M}}\hat{\mathcal{M}}^{-1}V^{\dagger}\mathcal{M}=\mathcal{M}.
	\end{displaymath}
\end{proof}
\hspace{0mm}\\
The procedure described in theorem \ref{TY1} is called \textit{bidiagonalization} of a non-singular matrix.
\bigskip
\\
Bidiagonalizing $\mathcal{M}^{\nu}$ and $\mathcal{M}^{cl}$ one gets
	\begin{displaymath}
		\mathcal{M}^{\nu}=V_{\nu}\left(\begin{matrix}
		m_{\nu_{1}}  & 0 & 0 \\
		0  & m_{\nu_{2}} & 0 \\
		0  & 0 & m_{\nu_{3}}
		  \end{matrix}\right)U_{\nu}^{\dagger},\quad
		\mathcal{M}^{cl}=V_{cl}\left(\begin{matrix}
		m_{e}  & 0 & 0 \\
		0  & m_{\mu} & 0 \\
		0  & 0 & m_{\tau}
		  \end{matrix}\right)U_{cl}^{\dagger}
	\end{displaymath}
with unitary matrices $V_{\nu}$, $U_{\nu}$, $V_{cl}$ and $U_{cl}$. After bidiagonalization one finds the following mass terms:
	\begin{displaymath}
	\begin{split}
	& \mathcal{L}^{\nu}_{\mathrm{mass}}=-\bar{\nu}_{R}V_{\nu}\hat{\mathcal{M}}^{\nu}U_{\nu}^{\dagger}\nu_{L}+\mathrm{H.c.},\\
	&
	\mathcal{L}^{cl}_{\mathrm{mass}}=-\bar{l}_{R}'V_{cl}\hat{\mathcal{M}}^{cl}U_{cl}^{\dagger}l_{L}'+\mathrm{H.c.}
	\end{split}
	\end{displaymath}
Defining the \textit{mass eigenfields}
	\begin{displaymath}
	\nu_{R}^{m}=V_{\nu}^{\dagger}\nu_{R},\quad \nu_{L}^{m}=U_{\nu}^{\dagger}\nu_{L},\quad
	l_{R}=V_{cl}^{\dagger}l_{R}',\quad
	l_{L}=U_{cl}^{\dagger}l_{L}'
	\end{displaymath}
one finds
	\begin{displaymath}
	\begin{split}
		& \mathcal{L}_{\mathrm{mass}}^{(\nu)}=-\bar{\nu}^{m}_{R}\hat{\mathcal{M}}^{\nu}\nu_{L}^{m}+\mathrm{H.c.}=-m_{\nu_{1}}\bar{\nu}_{1}^{m}\nu_{1}^{m}-m_{\nu_{2}}\bar{\nu}_{2}^{m}\nu_{2}^{m}-m_{\nu_{3}}\bar{\nu}_{3}^{m}\nu_{3}^{m},\\
		& 
		\mathcal{L}_{\mathrm{mass}}^{(cl)}=-\bar{l}_{R}\hat{\mathcal{M}}^{cl}l_{L}+\mathrm{H.c.}=-m_{e}\bar{e}e-m_{\mu}\bar{\mu}\mu-m_{\tau}\bar{\tau}\tau,
	\end{split}
	\end{displaymath}
with $\nu^{m}=\nu^{m}_{L}+\nu^{m}_{R}$ and $l=l_{L}+l_{R}$.
\bigskip
\bigskip
\\
\textbf{The lepton mixing matrix}
\medskip
\\
The neutrino flavour $\alpha$ is defined by the charged lepton $\alpha$ involved in production/detection processes. The responsible part of the GWS-Lagrangian for detection processes is the \textit{charged current} Lagrangian
	\begin{displaymath}
		\mathcal{L}_{\mathrm{CC}}=-\frac{g}{\sqrt{2}}W_{\lambda}^{-}\enspace\bar{l}_{L}'\gamma^{\lambda}\nu_{L} + \mathrm{H.c.}
	\end{displaymath}
Rewriting this in terms of mass eigenfields one gets
	\begin{equation}
	\mathcal{L}_{\mathrm{CC}}=-\frac{g}{\sqrt{2}}W_{\lambda}^{-}\enspace\bar{l}_{L}\gamma^{\lambda}\underbrace{U_{cl}^{\dagger}U_{\nu}}_{U}\nu_{L}^{m} + \mathrm{H.c.}
	\end{equation}
Thus all $\nu_{iL}^{m}$ contribute to the Lagrangian of the detection of a particular $\nu_{\alpha L}$ (which is done by detecting the corresponding charged lepton $\alpha$).
\medskip
\\
This is why $U=U_{cl}^{\dagger}U_{\nu}$ is called \textit{lepton mixing matrix}. Often $U$ is also called PMNS-matrix\footnote{PMNS = Pontecorvo, Maki, Nakagawa, Sakata.}. It is the analogue to the CKM-matrix\footnote{Cabibbo-Kobayashi-Maskawa-matrix.} in the quark sector.

\subsection{Mass hierarchies in fermion families}
The origin of mass ratios of the different members of fermion families is one of the most unexplored regions of particle physics. Currently there is no explanation for the mass ratios of the standard-model fermions. Every successful model will extend the standard-model of particle physics in a far-reaching way.
\medskip
\\
Let us now take a look on the masses and mass ratios (\textquotedblleft hierarchies\textquotedblright) in the different fermion families. They are listed in table \ref{smhierarchies}.\footnote{Here $u, d$ and $s$ masses are \textquotedblleft current quark masses\textquotedblright, $c$ and $b$ masses are the \textquotedblleft running\textquotedblright\hspace{0.7mm} masses. The $t$ mass was determined by direct observation of top events.  There is much controversy in defining quark masses. For further details see \cite{pdg}.} Here we have defined the mass ratio as the mass of the particle over the mass of the lightest family member. One can see that the mass ratios of the families are quite different, but the most striking fact is that
	\begin{displaymath}
		\frac{m_{\nu}}{m_{e}}<4\cdot10^{-6},
	\end{displaymath}
while
	\begin{displaymath}
		\frac{m_{u}}{m_{e}}\sim 3-6, \quad\quad \frac{m_{d}}{m_{e}}\sim 7-12,
	\end{displaymath}
which points towards a completely different mechanism of mass generation for neutrinos than the mechanism for mass generation in the quark sector. This is one reason why \textit{alternatives to the GIM-construction} are intensively studied today. Two important proposals are
	\begin{itemize}
	 \item radiative neutrino masses \cite{grimus1,grimus2}
	 \item the \textit{seesaw mechanism}. \cite{grimus1,grimus2,phenneutinoosc}
	\end{itemize}

\begin{table}
\begin{center}
% use packages: array
\begin{tabular}{|lcll|}
\firsthline
Fermion family & members & $mc^{2}$[MeV]\textsuperscript{\cite{pdg}} & mass ratio \\
\hline 
up quarks & $u$ & 1.5-3.3 & 1 \\
× & $c$ & $(1.27_{-0.11}^{+0.07})\cdot 10^{3}$ & 350-890 \\
× & $t$ & $(171.2\pm2.1)\cdot10^{3}$ & 51000-116000 \\
\hline
down quarks & $d$ & 3.5-6.0 & 1\\
× & $s$ & $104_{-34}^{+26}$ & 12-37 \\
× & $b$ & $(4.20_{-0.07}^{+0.17})\cdot10^{3}$ & 690-1200 \\
\hline
charged leptons & $e$ & 0.510998910(13) & 1\\
× & $\mu$ & 105.658367(04) & 206.7683 \\
× & $\tau$ & 1776.84(17) & 3477 \\
\hline
neutrinos & $\nu_{1}$ & $<2\cdot 10^{-6}$ & unknown\\
× & $\nu_{2}$ & $<2\cdot 10^{-6}$ & unknown \\
× & $\nu_{3}$ & $<2\cdot 10^{-6}$ & unknown \\
\lasthline
\end{tabular}
\end{center}
\caption{Masses and hierarchies of the standard-model fermions.}
\label{smhierarchies}
\end{table}
\hspace{0mm}\\
Another very important question is, whether neutrinos are Dirac or Majorana particles. (For a discussion of Dirac and Majorana fields see appendix \ref{appendixdiracmajorana}.)  For a detailed discussion of Dirac and Majorana neutrinos see \cite{grimus2,phenneutinoosc,neutrinoantip}.
\medskip
\\
Let us mention that there exists a promising road to experimentally distinguish between Dirac and Majorana-neutrinos, namely the so called \textit{neutrinoless double beta decay} \cite{0nubetabeta1,0nubetabeta2}:
	\begin{equation}\label{doublebeta}
	(A,Z)\rightarrow (A,Z+2)+2e^{-},
	\end{equation}
which is only possible if neutrinos have Majorana nature. Process (\ref{doublebeta}) clearly violates lepton number conservation. A possible candidate for a nucleus that could decay via neutrinoless double beta decay is \textsuperscript{76}Ge:
	\begin{displaymath}
	^{76}\mathrm{Ge}\rightarrow\mathrm{\textsuperscript{76}Se}+2e^{-}.
	\end{displaymath}
The GERDA experiment \cite{Gerda1,Gerda2} searches for this process. Summaries describing other experiments on neutrinoless double beta decay can be found in \cite{0nubetabeta2,0nubetabetaexp1,0nubetabetaexp2}.
\bigskip
\\
Though the GIM-model itself may not be as promising as models involving Majorana neutrinos it provides the starting point to many other important models. Allowing total lepton number violation one can add a \textit{Majorana mass term} (see appendix \ref{appendixdiracmajorana}) to the Lagrangian of the GIM-model, which automatically implies Majorana nature for the neutrinos. In such models one has so-called Dirac-Majorana mass terms (see \cite{phenneutinoosc}) which are ideal starting points for the seesaw mechanism.
\bigskip
\\
So far we have only taken a look onto the experimental results for the masses of the standard model fermions. In the next chapter we will take a glance at neutrino oscillations and the results of recent neutrino oscillation experiments.

%% file: neutrinooscillations/neutrinooscillations.tex
\chapter{Neutrino oscillations and the lepton mixing matrix}\label{neutrinoosc}
In this chapter we will consider neutrino oscillations and their connection to neutrino masses and lepton mixing. An excellent review of this field can be found in \cite{neutrinoreview}. A summary of the current status of this field of research can be found in \cite{altarelli3}.
\bigskip
\\
\underline{Remark:} In all the following considerations the rows of the mixing matrix $U$ are numbered by indices $e,\mu,\tau$ instead of the numbers $1,2,3$. This leads to illustrative formulae like
	\begin{displaymath}
	\nu_{\mu L}(x)=U_{\mu j}\nu_{jL}^{m}(x), \quad\mathrm{etc.}
	\end{displaymath}

\section{Neutrino oscillations}
Let us consider the Lagrangian for the charged current interactions again:
	\begin{equation}\label{CCinteractions}
	\mathcal{L}_{\mathrm{CC}}=-\frac{g}{\sqrt{2}}W_{\lambda}^{-}\enspace\bar{l}_{L}'\gamma^{\lambda} \nu_{L}+\mathrm{H.c.}=-\frac{g}{\sqrt{2}}W_{\lambda}^{-}\enspace\bar{l}_{L}\gamma^{\lambda} U\nu_{L}^{m}+\mathrm{H.c.}
	\end{equation}
The only way the neutrino flavour can be defined is via the corresponding charged lepton in production and detection processes. This means:
	\begin{quote}
	\textit{A neutrino $\nu$ has flavour $\alpha$ if it is produced/detected in a CC-process involving the charged lepton of flavour $\alpha$.}
	\end{quote}
Therefore from the Lagrangian for the CC-interactions we find
	\begin{displaymath}
	\nu_{L}=U\nu^{m}_{L}
	\end{displaymath}
for the left-handed neutrino flavour eigenfields. Let us now consider the field operator for a massive Dirac neutrino:
	\begin{displaymath}
	\nu_{j}^{m}(x)=\sum_{s=1,2}\int \frac{d^{3}p}{(2\pi)^{3}\sqrt{2E}}[u_{j}(p,s)b_{j}(p,s)e^{-ipx}+v_{j}(p,s)d_{j}^{\dagger}(p,s)e^{ipx}].
	\end{displaymath}
(No summation over j.) The 1-neutrino mass eigenstate $\vert\nu_{i}^{m}\rangle$ with 4-momen\-tum $p$ and spin orientation $s$ is given by
	\begin{displaymath}
	\vert\nu_{i}^{m}\rangle=b_{i}^{\dagger}(p,s)\vert0\rangle,
	\end{displaymath}
where $\vert0\rangle$ denotes the vacuum of the quantum field theory. Therefore we have to consider the field operator $\bar{\nu}_{j}^{m}(x)$ for the creation of the 1-neutrino mass eigenstate. Since $\bar{\nu}_{L}=U^{\ast}\bar{\nu}_{L}^{m}$ we finally find
	\begin{equation}
	\vert\nu_{\alpha}\rangle=U^{\ast}_{\alpha j}\vert\nu_{j}^{m}\rangle.
	\end{equation}
This is the key formula that describes \textit{neutrino oscillations}. The central idea of neutrino oscillations is now the following:
	\begin{quote}
	\textit{Neutrinos are produced and detected as flavour eigenstates, but they propagate as a coherent superposition of mass eigenstates.}
	\end{quote}
The easiest description of neutrino oscillations involves the following assumptions:
	\begin{itemize}
	 \item The states propagate in the vacuum\footnote{The important effect of surrounding matter on neutrino oscillations was first studied by Mikheyev, Smirnov and Wolfenstein \cite{Wolfenstein1,Wolfenstein2,Mikheyev1,Mikheyev2,Mikheyev3}. For a detailed description of neutrino oscillations in matter see \cite{grimus2}. Important remark: Matter effects are \underline{essential} for the explanation of the solar neutrino deficit \cite{solarpuzzle} through neutrino oscillations.} in $x$-direction.
	 \item All mass eigenstates contained in the coherent superposition have the same energy $E$, but different momenta $p_{j}=\sqrt{E^{2}-m_{j}^{2}}$.
	 \item Since all neutrinos that usually occur in nature are ultrarelativistic ($E\gg m$), we use the ultrarelativistic limit
		\begin{displaymath}
		p_{j}=\sqrt{E^{2}-m_{j}^{2}}\simeq E-\frac{m_{j}^{2}}{2E}.
		\end{displaymath}
	\end{itemize}
Let $\vert\nu_{\alpha}\rangle$ be produced at $(t,x)=(0,0)$.
	\begin{displaymath}
	\begin{split}
	\Rightarrow \vert \nu_{\alpha};t,x\rangle=e^{-i(Ht-Px)}\vert\nu_{\alpha}\rangle &=\sum_{j=1}^{3}e^{-i(Et-p_{j}x)}U_{\alpha j}^{\ast}\vert\nu_{j}^{m}\rangle\simeq\\
	&\simeq e^{-iE(t-x)}\sum_{j=1}^{3}e^{-i\frac{m_{j}^{2}x}{2E}}U_{\alpha j}^{\ast}\vert\nu_{j}^{m}\rangle.
	\end{split}
	\end{displaymath}
The probability to find $\vert\nu_{\beta}\rangle$ at $(t,x)=(t,L)$ is given by
	\begin{equation}\label{neutrinooscequ}
	P_{\nu_{\alpha}\rightarrow\nu_{\beta}}(L)=\vert \langle\nu_{\beta}\vert\nu_{\alpha};t,L\rangle\vert^{2}\simeq \left| \sum_{j=1}^{3}e^{-i\frac{m_{j}^{2}L}{2E}}U_{\beta j}U_{\alpha j}^{\ast}\right|^{2},
	\end{equation}
which does not depend on $t$.
\bigskip
\\
$P_{\nu_{\alpha}\rightarrow\nu_{\beta}}$ has the following interesting properties:
	\begin{itemize}
	 \item It depends only on the mass square differences $\Delta m_{ij}^{2}:=m_{i}^{2}-m_{j}^{2}$.
	 \item It shows oscillatory behaviour in $\frac{L}{E}$ $\Rightarrow$ \textbf{neutrino oscillations}.
	 \item It is invariant under $U\mapsto DUD'$ with diagonal phase matrices $D$ and $D'$.
	 \item It indicates a violation of the family lepton numbers.
	\end{itemize}

\section{The lepton mixing matrix}
In the last section we found that the elements of the lepton mixing matrix (PMNS-matrix) $U$ are related to the observable effect of neutrino oscillations via equation (\ref{neutrinooscequ}). Therefore one can find the elements of $U$ by performing neutrino oscillation experiments. Before we list the results of previous oscillation experiments we choose some standard conventions for the mixing matrix $U$.
\bigskip
\\
According to \cite{grimus1} the unitary matrix $U$ can be decomposed in the following way:
	\begin{displaymath}
	U=DU' \mathrm{diag}(e^{i\rho},e^{i\sigma},1),
	\end{displaymath}
where $\rho, \sigma$ are real numbers, $D$ is a diagonal phase matrix and
	\begin{displaymath}
	U'=U_{23}U_{13}U_{12}
	\end{displaymath}
with
	\begin{displaymath}
	\begin{split}
	& U_{23}=\left(\begin{matrix}
		 1 & 0 & 0 \\
		 0 & c_{23} & s_{23} \\
		 0 & -s_{23} & c_{23}
	\end{matrix}\right),\quad
	U_{13}=\left(\begin{matrix}
		 c_{13} & 0 & s_{13}e^{-i\delta} \\
		 0 & 1 & 0 \\
		 -s_{13}e^{i\delta} & 0 & c_{13}
	\end{matrix}\right),\\
	& U_{12}=\left(\begin{matrix}
		 c_{12} & s_{12} & 0 \\
		 -s_{12} & c_{12} & 0 \\
		 0 & 0 & 1
	\end{matrix}\right), \quad c_{ij}=\mathrm{cos}\hspace{0.5mm}\theta_{ij},\enspace s_{ij}=\mathrm{sin}\hspace{0.5mm}\theta_{ij}.
	\end{split}
	\end{displaymath}
	\begin{equation}\label{Uprime}
	\Rightarrow U'=\left(
	\begin{matrix}
	c_{12}c_{13} & s_{12}c_{13} & s_{13}e^{-i\delta}\\
	-s_{12}c_{23}-c_{12}s_{13}s_{23}e^{i\delta} & c_{12}c_{23}-s_{12}s_{13}s_{23}e^{i\delta} & c_{13}s_{23}\\
	s_{12}s_{23}-c_{12}s_{13}c_{23}e^{i\delta} & -c_{12}s_{23}-s_{12}s_{13}c_{23}e^{i\delta} & c_{13}c_{23}
	\end{matrix}
	\right).
	\end{equation}
For the \textit{mixing angles} $\theta_{ij}$ we use the convention $\theta_{ij}\in [0,\frac{\pi}{2}]$. $\delta\in [0,2\pi)$ is a phase. Note that $\delta$ is a physical phase, because in general it occurs in the oscillation probability (\ref{neutrinooscequ}). $\delta$ is the analogue to the CKM-phase in quark mixing.
\medskip
\\
Let us now look at the matrices $D$ and $\mathrm{diag}(e^{i\rho},e^{i\sigma},1)$. In the last section we already found that the oscillation probability $P_{\nu_{\alpha}\rightarrow \nu_{\beta}}$ is invariant under left and right multiplication of $U$ with diagonal phase matrices, therefore $D$ and $\mathrm{diag}(e^{i\rho},e^{i\sigma},1)$ can be discarded when one considers neutrino oscillations only.
\\
In fact it turns out that $D$ is not physical in the charged current interactions, because it can be absorbed into the charged lepton fields (consider the Lagrangian for the CC-interactions (\ref{CCinteractions})). The case is more complicated for the phase matrix $\mathrm{diag}(e^{i\rho},e^{i\sigma},1)$. Here we have to consider two cases.
	\begin{enumerate}
	 \item If neutrinos are \textit{Dirac particles}, the phase matrix $\mathrm{diag}(e^{i\rho},e^{i\sigma},1)$ can be absorbed into the neutrino mass eigenfields, because all other parts of the Lagrangian are invariant under $\nu_{L}^{m}\mapsto \mathrm{diag}(e^{-i\rho},e^{-i\sigma},1)\nu_{L}^{m}$.
	 \item If neutrinos are \textit{Majorana particles}, the phase matrix $\mathrm{diag}(e^{i\rho},e^{i\sigma},1)$ cannot be absorbed into the neutrino mass eigenfields, because the Majorana mass term (see subsection \ref{majoranapsubsection})
		\begin{displaymath}
		\mathcal{L}^{\mathrm{(mass)}}_{\mathrm{Maj.}}=\frac{1}{2}(\nu_{L}^{m})^{T}C^{-1}\hat{\mathcal{M}}^{\nu}_{\mathrm{Maj.}}\nu_{L}^{m}
		\end{displaymath}
	is \underline{not} invariant under $\nu^{m}_{L}\mapsto \mathrm{diag}(e^{-i\rho},e^{-i\sigma},1)\nu_{L}^{m}$, because this would correspond to a transformation of the neutrino masses
		\begin{displaymath}
		m_{1}\mapsto e^{-2i\rho}m_{1},\quad m_{2}\mapsto e^{-2i\sigma}m_{2},\quad m_{3}\mapsto m_{3},
		\end{displaymath}
	and the new neutrino fields would not be mass eigenfields any longer. Therefore $\rho$ and $\sigma$ are physical for Majorana neutrinos, and they are called \textit{Majorana phases}.
	\end{enumerate}

\section{Results of the recent neutrino oscillation experiments}
Before we list the results we have to explain some conventions again. Since the numbering of the neutrino mass eigenfields is arbitrary, we have to choose a convention. We use the convention of \cite{grimus2}:
	\begin{displaymath}
	m_{1}<m_{2} \enspace(\Delta m_{21}^{2}>0),\quad \Delta m_{21}^{2}<\vert \Delta m_{31}^{2}\vert.
	\end{displaymath}
Using this convention there are two possibilities for the neutrino mass spectrum, depending on the sign of $\Delta m_{31}^{2}$.
	\begin{itemize}
	 \item[(a)] \textit{Normal spectrum}: $m_{1}<m_{2}<m_{3}$.
	 \item[(b)] \textit{Inverted spectrum}: $m_{3}<m_{1}<m_{2}$.
	\end{itemize}

	\begin{center}
	\includegraphics[scale=1.2, keepaspectratio=true]{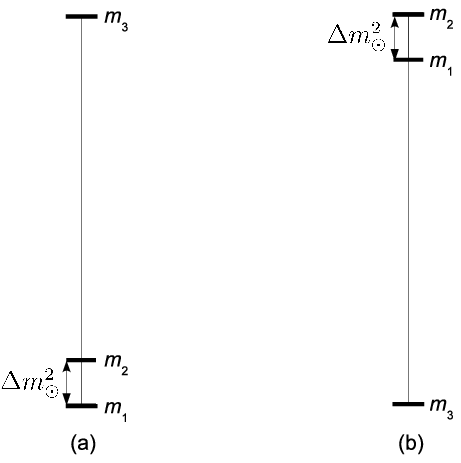}
	\\
	Normal (a) and inverted (b) neutrino mass spectrum
	\end{center}
The smallest mass square difference $\Delta m_{21}^{2}$ is important for \textit{solar} neutrino oscillations \cite{gonzalez}, and is therefore also denoted as $\Delta m_{\odot}^{2}$. The neutrino mass eigenstates $\nu_{1}^{m}$ and $\nu_{2}^{m}$ are therefore sometimes referred to as the \textquotedblleft solar pair\textquotedblright.
\\
The largest mass square difference is important for \textit{atmospheric} neutrino oscillations \cite{gonzalez}, and is therefore denoted as $\Delta m_{\mathrm{atm}}^{2}$.
\\
$\Delta m_{\mathrm{atm}}^{2}=\Delta m_{31}^{2}$ for the normal spectrum, and $\Delta m_{\mathrm{atm}}^{2}=\Delta m_{23}^{2}$ for the inverted spectrum.
\bigskip
\\
Schwetz et al. give a summary of best fit results and errors from global neutrino oscillation experiments in \cite{schwetz}. The results provided by \cite{schwetz} are shown in table \ref{Oscillationdata}.
	\begin{table}[h]
	\begin{center}
	\renewcommand{\arraystretch}{1.4}
	\begin{tabular}{|l@{\vrule width1pt}ccc|}
	\firsthline
	parameter \hspace{15mm} & \hspace{0mm} best fit & $2\sigma$ & $3\sigma$\\
	\hline
	$\Delta m_{21}^{2}\enspace [10^{-5}\mathrm{eV}^{2}]$ & \hspace{0mm} $7.65_{-0.20}^{+0.23}$ & $7.25-8.11$ & $7.05-8.34$\\
	$\vert\Delta m_{31}^{2}\vert\enspace [10^{-3}\mathrm{eV}^{2}]$ & \hspace{0mm} $2.40_{-0.11}^{+0.12}$ & $2.18-2.64$ & $2.07-2.75$\\
	$\mathrm{sin}^{2}\theta_{12}$ & \hspace{0mm} $0.304_{-0.016}^{+0.022}$ & $0.27-0.35$ & $0.25-0.37$\\
	$\mathrm{sin}^{2}\theta_{23}$ & \hspace{0mm} $0.50_{-0.06}^{+0.07}$ & $0.39-0.63$ & $0.36-0.67$\\
	$\mathrm{sin}^{2}\theta_{13}$ & \hspace{0mm} $0.01_{-0.011}^{+0.016}$ & $\leq 0.04$ & $\leq 0.056$\\
	\lasthline
	\end{tabular}
	\caption[Three-flavour neutrino oscillation parameters from global data]{Best fit values with 1$\sigma$-errors, $2\sigma$- and $3\sigma$-intervals for three-flavour neutrino oscillation parameters from global data as given in \cite{schwetz}. The analysis of Schwetz et al. is based on \cite{schwetz2,schwetz3,adamson,aharmim,abe,borexino}.}
	\label{Oscillationdata}
	\end{center}
	\end{table}
\\
From the quite high values for $\mathrm{sin}^{2}\theta_{ij}$ for two of the mixing angles we see that in contrast to the CKM-matrix $U_{\mathrm{CKM}}$ the lepton mixing matrix $U$ is very far from the unit matrix \cite{ckmreview}.
\medskip
\\
Using the best fit values for $\mathrm{sin}^{2}\theta_{ij}$ given in table \ref{Oscillationdata} one can try to find analytic expressions for the mixing matrix which are in agreement with the presented experimental data. We will give two popular examples here.

\subsection{The tribimaximal mixing matrix}
So-called \textit{tribimaximal mixing}, which was first suggested by Harrison, Perkins and Scott in \cite{HPS}, is a very popular suggestion for lepton mixing, which is described by the following assumptions:
	\begin{itemize}
	 \item $\mathrm{sin}^{2}\theta_{12}=\frac{1}{3}$, which is within the currently $2\sigma$-bound (see table \ref{Oscillationdata}).
	 \item $\mathrm{sin}^{2}\theta_{23}=\frac{1}{2}$, which is within the currently $1\sigma$-bound.
	 \item $\mathrm{sin}^{2}\theta_{13}=0$, which is within the currently $1\sigma$-bound.
	\end{itemize}
For the mixing angles we get:
	\begin{displaymath}
	\theta_{12}=\mathrm{arcsin}\frac{1}{\sqrt{3}},\quad \theta_{23}=\frac{\pi}{4},\quad \theta_{13}=0.
	\end{displaymath}
Inserting these values into expression (\ref{Uprime}) for the mixing matrix $U'$ we find:
	\begin{displaymath}
	U'=U_{\mathrm{HPS}}=\left(\begin{matrix}
		 \sqrt{\frac{2}{3}} & \frac{1}{\sqrt{3}} & 0 \\
		 -\frac{1}{\sqrt{6}} & \frac{1}{\sqrt{3}} & \frac{1}{\sqrt{2}} \\
		 \frac{1}{\sqrt{6}} & -\frac{1}{\sqrt{3}} & \frac{1}{\sqrt{2}}
	\end{matrix}\right),
	\end{displaymath}
which is also called Harrison-Perkins-Scott mixing matrix. (Remark: In the literature there occur different phase conventions for the rows and columns of $U_{\mathrm{HPS}}$.)

\subsection{Trimaximal mixing}
Another interesting suggestion for lepton mixing is the so-called \textit{trimaximal mixing} pattern, which was suggested by Grimus and Lavoura in \cite{grimus3}. It is described by the assumptions
	\begin{displaymath}
	\vert U_{e2}\vert^{2}=\vert U_{\mu2}\vert^{2}=\vert U_{\tau2}\vert^{2}=\frac{1}{3}.
	\end{displaymath}
Because the trimaximal mixing pattern is less restrictive than the tribimaximal mixing assumption it offers interesting possibilities like $U_{e3}\neq 0$ and $\mathrm{sin}^{2}\theta_{23}\neq \frac{1}{2}$. For more information on trimaximal mixing we refer the reader to \cite{grimus3}.

\subsection{A selection of important neutrino experiments}
At last we want to list a selection of important neutrino experiments.
	\begin{itemize}
	 \item SAGE \cite{sage} (radiochemical detector; measurement of solar neutrino flux)
	 \item Superkamiokande \cite{superkamiokande2} (\v{C}erenkov detector; detection of solar and atmospheric neutrinos)
	 \item SNO \cite{sno1,sno2} (\v{C}erenkov detector; detection of neutrinos from the sun and other astrophysical neutrino sources)
	 \item Amanda/IceCube \cite{icecube} (\v{C}erenkov detector; neutrino telescope, detection of cosmic ray air showers)
	 \item KamLAND \cite{kamland} (scintillation detector; observation of reactor neutrinos)
	 \item K2K \cite{k2k,k2k2} (accelerator neutrino experiment)
	 \item T2K \cite{k2k,t2k} (accelerator neutrino experiment)
	 \item MINOS \cite{minos,minos2} (accelerator neutrino experiment)
	 \item OPERA \cite{opera} (accelerator neutrino experiment)
	 \item MiniBooNE \cite{miniboone} (accelerator neutrino experiment)
	 \item Borexino \cite{borexino,borexino2} (scintillation detector; solar neutrinos)
	 \item KATRIN \cite{katrin} (determination of the absolute mass of electron neutrinos through study of beta decay of \textsuperscript{3}H)
	\end{itemize}

%% file: family/family.tex
\chapter{Finite family symmetry groups}
In the last chapter we found that the current experimental data on neutrino oscillations allow interesting mixing patterns like tribimaximal mixing. If one allows all matrices as lepton mass matrices it is no problem to construct mass matrices that lead to a specific mixing matrix $U$. However, just finding a mass matrix that allows a specific form of $U$ does not tell us much about the physical (and mathematical) principles that govern the phenomenon of lepton mixing.
\medskip
\\
The important question is therefore: Can we find a mathematical framework that restricts the mass matrices and therefore the mixing matrix in a way such that the experimental results can be reproduced within the current error bounds?
\medskip
\\
A promising road to answer this question is the idea of so-called \textit{finite family symmetries}, which we will explain now.
\medskip
\\
Consider a Yukawa-coupling of the form
	\begin{equation}\label{FFSG-GinvYuk}
	\mathcal{L}_{\mathrm{Yukawa}}=-\bar{a}(\Gamma_{j}\phi_{j})b+\mathrm{H.c.},
	\end{equation}
where $\Gamma_{j}$ are complex $3\times 3$-matrices, $\phi_{j}$ are standard model Higgs doublets and $a,b$ are vectors of the form
	\begin{displaymath}
	a=\left(\begin{matrix}
		D_{e}' \\ D_{\mu}' \\ D_{\tau}'
	\end{matrix}\right),\quad
	b=\left(\begin{matrix}
		f_{e} \\ f_{\mu} \\ f_{\tau}
	\end{matrix}\right).
	\end{displaymath}
$f$ denotes a specific type of chiral fermion field, and $D$ are $SU(2)_{I}$-doublets built of two chiral fermion fields. In our analysis we concentrate on the lepton sector (extended by 3 right-handed neutrinos), thus
	\begin{displaymath}
	\begin{split}
	&f=l_{R}',\nu_{R},\\
	&D'=\left(\begin{matrix} \nu_{L}\\l_{L}'  \end{matrix}\right).
	\end{split}
	\end{displaymath}
(In general one would have to include quarks too.) The basic idea is now to consider a Yukawa-coupling that is invariant under a finite group $G$:

\begin{quote}
The Yukawa-coupling
	\begin{displaymath}
	\mathcal{L}_{\mathrm{Yukawa}}=-\bar{a}\phi_{j}\Gamma_{j}b+\mathrm{H.c.}
	\end{displaymath}
is called invariant under a finite group $G$ (or $G$-invariant), if there are matrix representations $D_{a}, D_{b}, D_{\phi}$ of the group $G$ such that $\mathcal{L}_{\mathrm{Yukawa}}$ is invariant under
	\begin{displaymath}
	a\mapsto D_{a}a,\quad b\mapsto D_{b}b,\quad \phi\mapsto D_{\phi}\phi.
	\end{displaymath}
\end{quote}
Since the representations $D_{a}, D_{b}$ of the finite group $G$ act on vectors containing all members of the family $\{f_{e}, f_{\mu}, f_{\tau}\}$ of a lepton type $f$, we call $G$ a \textit{finite family symmetry group}.
\medskip
\\
Note that by introducing the Yukawa-coupling (\ref{FFSG-GinvYuk}) we have automatically extended the scalar sector by an arbitrary number of standard model Higgs doublets.
\medskip
\\
As we will see in chapter \ref{yukawa} in a $G$-invariant Yukawa-coupling the number of Higgs doublets $\phi_{i}$ is not arbitrary, but restricted by the finite group $G$. Furthermore a finite group $G$ restricts the possible matrices $\Gamma_{j}$, so $G$ restricts the mass matrix
	\begin{displaymath}
	\mathcal{M}=\left(\frac{v_{j}}{\sqrt{2}} \Gamma_{j}\right)^{\dagger},
	\end{displaymath}
where $v_{j}$ is defined via
	\begin{displaymath}
	\langle 0\vert\phi_{j}(x)\vert0\rangle=\frac{1}{\sqrt{2}}\left(
	\begin{matrix}
	0
	\\
	v_{j}
	\end{matrix}
	\right).
	\end{displaymath}
Here we have made the following (important) assumption:
\begin{quote}
Let $\langle 0\vert\phi_{j}\vert0\rangle$ be the vacuum expectation value of the Higgs doublet $\phi_{j}$, then all vacuum expectation values $\langle 0\vert\phi_{j}\vert0\rangle$ have the form
	\begin{displaymath}
	\langle0\vert\phi_{j}(x)\vert0\rangle=\frac{1}{\sqrt{2}}\left(
	\begin{matrix}
	0\\ v_{j}
	\end{matrix}
	\right),
	\end{displaymath}
where $v_{j}$ is a constant complex number. This assumption leads to the following geometrical interpretation: Let $\langle0\vert\phi_{j}(x)\vert0\rangle$ be the vacuum expectation values of $\phi_{j}$, then all $\langle0\vert\phi_{j}(x)\vert0\rangle$ are proportional to each other, which means that, interpreted as vectors in $\mathbb{C}^2$, all $\langle0\vert\phi_{j}(x)\vert0\rangle$ are parallel.
\end{quote}
Note that this is an assumption. We have not analysed under which conditions (on the generalized Higgs potential $V(\phi_{1},...,\phi_{n})$) this assumption is justified. If the assumption was not fulfilled, a nonvanishing first component $u_{j}$ of $\langle0\vert\phi_{j}(x)\vert0\rangle$ would lead to charge breaking terms of the form
	\begin{displaymath}
	u_{j}\bar{\nu}_{\alpha L}\beta_{R}+\mathrm{H.c.},
	\end{displaymath}
which are clearly not invariant under $U(1)_{\mathrm{EM}}$ (EM=electromagnetic). This problem has been studied for the 2-Higgs doublet model by Barroso, Ferreira and Santos in \cite{ferreira} and by Maniatis et al. in \cite{maniatis}.
\medskip
\\
Introducing $G$-invariant Yukawa-couplings in the lepton sector we find that a chosen finite group $G$ restricts
	\begin{itemize}
	 \item the number of Higgs doublets,
	 \item the structure of the mass matrices,
	 \item the structure of the mixing matrix.
	\end{itemize}
In the following chapter we will investigate the mathematical properties of $G$-invariant Yukawa-couplings in detail.
 

%% file: yukawa/yukawa.tex
\chapter{$G$-invariant Yukawa-couplings}\label{yukawa}
Our aim is to develop a method for obtaining Yukawa-couplings which are invariant under a finite group $G$. To reach this goal we will first define the basic objects that are needed. Then we will use the theory of finite groups to construct and classify the possible $G$-invariant Yukawa-couplings for a given group $G$.
\bigskip
\\
All needed group theoretical definitions and theorems can be found in appendix \ref{groupappendix}.

\section{$G$-invariance of functions and Yukawa-cou\-plings}

\begin{define}\label{DY1}
	Let $G$ be a finite group and $f: V_{1}\times...\times V_{p}\rightarrow \mathbb{C}(\mathbb{R})$ a map from a Cartesian product of $p$ vectorspaces over $\mathbb{C}$ to $\mathbb{C}$($\mathbb{R}$). $f$ is called \textit{invariant under} $G$ or $G$-\textit{invariant}, if there exist representations $D_{1},...,D_{p}$ of $G$ on $V_{1},...,V_{p}$ such that
	\begin{equation}
		f(D_{1}(a)v_{1},...,D_{p}(a)v_{p})=f(v_{1},...,v_{p})
	\end{equation}
	$\forall v_{k}\in V_{k}$ $(k=1,...,p)$ and $\forall a\in G$.
\end{define}
\hspace{0mm}
\\
We will now define the prototype of a Yukawa-coupling in an abstract mathematical form, for it is easier to investigate the mathematical properties of an object, if one discards its physical origin for a while. We will breach the gap to the Yukawa-couplings of particle physics in chapter \ref{LeptonSector}.

\begin{define}\label{DY2}
	Let $V_{a}\simeq \mathbb{C}^{n_{a}} ,V_{b}\simeq \mathbb{C}^{n_{b}}$ and $V_{\phi}\simeq \mathbb{C}^{n_{\phi}}$ be finite-dimensional vectorspaces. Let $\Gamma_{j}$ ($j=1,...,n_{\phi}$) be complex $n_{a}\times n_{b}$-matrices. (We will use the isomorphism of the vectorspaces $V_{i}\simeq\mathbb{C}^{n_{i}}$ to define matrix-multiplication, transposition and complex conjugation. $[A]$ denotes the matrix representation of a linear operator $A$. In the case of vectors we will write $a$ instead of $[a]$ for the matrix representation of a vector $a$.)
	\\
	We define
		\begin{gather}
		Y_{\Gamma}: V_{a}\times V_{b}\times V_{\phi}\rightarrow \mathbb{R}
		\\
		Y_{\Gamma}(a,b,\phi)=a^{\dagger}(\Gamma_{j}\phi_{j})b+b^{\dagger}(\Gamma_{j}\phi_{j})^{\dagger}a.
		\end{gather}
$Y_{\Gamma}$ is called a \textit{Yukawa-coupling}.
\end{define}

\begin{prop}\label{PY3}
	Let $Y_{\Gamma}:V_{a}\times V_{b}\times V_{\phi}\rightarrow \mathbb{R}$ be a $G$-invariant Yukawa-coupling, and let $D_{a}, D_{b}, D_{\phi}$ be the corresponding representations of $G$ on $V_{a},V_{b},V_{\phi}$.
	\\
	If any representation of $D_{a},D_{b},D_{\phi}$ is reducible\footnote{For the definitions of reducible and irreducible representations see definitions \ref{DA44} and \ref{DA46}.}, then $Y_{\Gamma}$ can be decomposed into a sum of $G$-invariant Yukawa-couplings acting on subsets of $V_{a}\times V_{b}\times V_{\phi}$ such that the corresponding representations of the new Yukawa-couplings are irreducible.
\end{prop}

\begin{proof}
	W.l.o.g. let $D_{a}$ be reducible s.t. $D_{a}=D_{a}^{1}\oplus D_{a}^{2}$ on $V_{a}=V_{a}^{1}\oplus V_{a}^{2}$. (Let $D_{a}^{1}, D_{a}^{2}$ be irreducible.)
	\smallskip
	\\
	$\Rightarrow$ There exists a basis of $V_{a}$ s.t.
		\begin{displaymath}
			[D_{a}]=
				\left(
				\begin{matrix}
				[D_{a}^{1}]  & \textbf{0} \\
				 \textbf{0} & [D_{a}^{2}]
			        \end{matrix}
				\right),
		\end{displaymath}
where \textbf{0} are appropriate null matrices.
\\
In this basis we have
	\begin{displaymath}
		\begin{split}
			Y_{\Gamma}(D_{a}a,D_{b}b,D_{\phi}\phi) & = a^{\dagger}[D_{a}]^{\dagger}(\Gamma_{j}[D_{\phi}]_{jk}\phi_{k})[D_{b}]b + \mathrm{H.c.} = \\
			& = (a_{1}^{\dagger}[D_{a}^{1}]^{\dagger}, a_{2}^{\dagger}[D_{a}^{2}]^{\dagger})(\Gamma_{j}[D_{\phi}]_{jk}\phi_{k})[D_{b}]b + \mathrm{H.c.} = \\
			& = \lbrace(a_{1}^{\dagger}[D_{a}^{1}]^{\dagger}, 0)+(0,a_{2}^{\dagger}[D_{a}^{2}]^{\dagger})\rbrace(\Gamma_{j}[D_{\phi}]_{jk}\phi_{k})[D_{b}]b + \mathrm{H.c.} = \\
			& = Y_{\Gamma}((D_{a}^{1}a_{1},0),D_{b}b,D_{\phi}\phi)+Y_{\Gamma}((0, D_{a}^{2}a_{2}),D_{b}b,D_{\phi}\phi)=\\
			& = Y_{\Gamma}(a,b,\phi) = Y^{1}_{\Gamma}(a_{1},b,\phi)+Y^{2}_{\Gamma}(a_{2},b,\phi),
		\end{split}
	\end{displaymath}
where $Y^{1}_{\Gamma}(a_{1},b,\phi)$ and $Y^{2}_{\Gamma}(a_{2},b,\phi)$ are independent $G$-invariant Yukawa-cou\-plings acting on the subsets $V_{a}^{1}\times V_{b}\times V_{\phi}$ and $V_{a}^{2}\times V_{b}\times V_{\phi}$ of $V_{a}\times V_{b}\times V_{\phi}$.
\bigskip
\\
The proof is easily extendable to all other cases of reducible representations occurring in the transformation law of a general $Y_{\Gamma}$.
\end{proof}
\hspace{0mm}\\
Having obtained proposition \ref{PY3} we will only consider irreducible representations $D_{a}, D_{b}, D_{\phi}$ from now on.

\section{Properties of $G$-invariant Yukawa-cou\-plings}\label{Yproperties}
The aim of the following considerations is to solve the following purely mathematical problem:

\begin{quote}
	Given $G, V_{a}, V_{b}, D_{a}, D_{b}$, where $D_{a}, D_{b}$ are irreducible, we want to find $\Gamma$, $V_{\phi}$ and $D_{\phi}$ (irreducible) such that $Y_{\Gamma}$ is $G$-invariant.
\end{quote}

\begin{prop}\label{PY4}
	Suppose $G, V_{a}, V_{b}, V_{\phi}, D_{a}, D_{b}, D_{\phi}$ are given ($D_{a}, D_{b}, D_{\phi}$ irreducible, $V_{a}, V_{b}, V_{\phi}$ vectorspaces over $\mathbb{C}$), then $Y_{\Gamma}$ is $G$-invariant if and only if
		\begin{equation}\label{geninv}
			[D_{a}(f)]^{\dagger}\Gamma_{j}[D_{\phi}(f)]_{jk}[D_{b}(f)]=\Gamma_{k} \quad \forall f \in G.
		\end{equation}
\end{prop}

\begin{proof}
	\begin{displaymath}
		\begin{split}
			Y_{\Gamma}(D_{a}a, D_{b}b, D_{\phi}\phi) & = a^{\dagger}[D_{a}]^{\dagger}(\Gamma_{j}[D_{\phi}]_{jk}\phi_{k})[D_{b}]b + \mathrm{H.c.} =\\
			& = y_{\Gamma}(D_{a}a, D_{b}b, D_{\phi}\phi) + y^{\ast}_{\Gamma}(D_{a}a, D_{b}b, D_{\phi}\phi),
		\end{split}
	\end{displaymath}
where we have defined $y_{\Gamma}(a,b,\phi):=a^{\dagger}(\Gamma_{j}\phi_{j})b$.
\medskip
\\
$y_{\Gamma}$ $G$-invariant $\Rightarrow$ $y_{\Gamma}^{\ast}$ $G$-invariant $\Rightarrow$ $Y_{\Gamma}$ $G$-invariant.
\medskip
\\
$y_{\Gamma}$ is $G$-invariant if and only if
	\begin{displaymath}
		[D_{a}(f)]^{\dagger}\Gamma_{j}[D_{\phi}(f)]_{jk}[D_{b}(f)]=\Gamma_{k} \quad \forall f \in G.
	\end{displaymath}
Is it possible for $Y_{\Gamma}$ to be $G$-invariant without fulfilling (\ref{geninv})? Well, on complex vectorspaces the answer is no, for the following reason:
\medskip
\\
Define 
		\begin{displaymath}
			y_{\Gamma}':=a^{\dagger}([D_{a}]^{\dagger}\Gamma_{k}[D_{b}][D_{\phi}]_{kj})\phi_{j}b.
		\end{displaymath}
and
		\begin{displaymath}
		\alpha:=y_{\Gamma}-y_{\Gamma}'.
		\end{displaymath}
$\Rightarrow \Gamma_{j}$ is a solution of equation (\ref{geninv})  $\Leftrightarrow \alpha(a,b,\phi)=0 \enspace\forall a,b,\phi$. $G$-invariance imposes
	\begin{displaymath}
		\begin{split}
			0 &= Y_{\Gamma}-Y_{\Gamma}'=\\ & = (y_{\Gamma}+y_{\Gamma}^{\ast})-(y_{\Gamma}'+y_{\Gamma}'\hspace{0mm}^{\ast})=\\
			& = (y_{\Gamma}-y_{\Gamma}')+(y_{\Gamma}-y_{\Gamma}')^{\ast}=\alpha+\alpha^{\ast}.
		\end{split}
	\end{displaymath}
	\begin{displaymath}
		\alpha=y_{\Gamma}-y_{\Gamma}'=a^{\dagger}(\underbrace{\Gamma_{j}-[D_{a}]^{\dagger}\Gamma_{k}[D_{b}][D_{\phi}]_{kj}}_{=:\tilde{\Gamma}_{j}})\phi_{j}b.
	\end{displaymath}
It follows:
	\begin{equation}\label{PY4cequ}
		\alpha+\alpha^{\ast}=a^{\dagger}(\tilde{\Gamma}_{j}\phi_{j})b+b^{\dagger}(\tilde{\Gamma}_{j}\phi_{j})^{\dagger}a=0 \quad\forall a,b,\phi.
	\end{equation}
On complex vectorspaces $V_{a}, V_{b}, V_{\phi}$ the only solution of (\ref{PY4cequ}) is $\tilde{\Gamma}_{j}=0\Rightarrow \Gamma_{j}$ fulfil (\ref{geninv}).
\end{proof}
\hspace{0mm}\\
Because Yukawa-couplings on real vectorspaces $V_{a}, V_{b}, V_{\phi}$ are not appropriate for particle physics, we will not consider $G$-invariant Yukawa-couplings with $\Gamma_{j}$ not fulfilling (\ref{geninv}).

\begin{define}\label{DY5}
	Let $G$ be a finite group. A set of \textit{generators of} $G$, i.s. $\mathrm{gen}(G)$, is a subset of $G$ s.t. every $a\in G$ can be written as a product of elements in $\mathrm{gen}(G)$. $\mathrm{gen}(G)$ always exists, but it is of course not unique. Clearly it is desirable to find a set $\mathrm{gen}(G)$ that is as small as possible.
\end{define}

\begin{example}
	$G=Z_{n}=\lbrace e^{\frac{2\pi ik}{n}}\rbrace_{k=1,...,n}$, $ \mathrm{gen}(G)=\lbrace e^{\frac{2\pi i}{n}}\rbrace$ $\Rightarrow$ $Z_{n}$ has only one generator.
\end{example}

\begin{prop}\label{PY6}
	\begin{displaymath}
		[D_{a}(f)]^{\dagger}\Gamma_{j}[D_{\phi}(f)]_{jk}[D_{b}(f)]=\Gamma_{k} \quad \forall f \in \mathrm{gen}(G)
	\end{displaymath}
	\begin{displaymath} \Rightarrow [D_{a}(f)]^{\dagger}\Gamma_{j}[D_{\phi}(f)]_{jk}[D_{b}(f)]=\Gamma_{k} \quad \forall f \in G. 
	\end{displaymath}
\end{prop}

\begin{proof}
	To show: If the equation holds for $r,s\in G$, then it holds for $rs\in G$.
	\begin{gather*}
		[D_{a}(r)]^{\dagger}\Gamma_{j}[D_{\phi}(r)]_{jk}[D_{b}(r)]=\Gamma_{k},\\
		[D_{a}(s)]^{\dagger}\Gamma_{j}[D_{\phi}(s)]_{jk}[D_{b}(s)]=\Gamma_{k}
	\end{gather*}
$\Rightarrow$
	\begin{displaymath}
		\begin{split}
			& \hspace{5.5mm} [D_{a}(rs)]^{\dagger}\Gamma_{j}[D_{\phi}(rs)]_{jk}[D_{b}(rs)]=\\
			& =[D_{a}(r)D_{a}(s)]^{\dagger}\Gamma_{j}[D_{\phi}(r)D_{\phi}(s)]_{jk}[D_{b}(r)D_{b}(s)]=\\
			& = [D_{a}(s)]^{\dagger}[D_{a}(r)]^{\dagger}\Gamma_{j}[D_{\phi}(r)]_{jl}[D_{\phi}(s)]_{lk}[D_{b}(r)][D_{b}(s)]=\\
			& = [D_{a}(s)]^{\dagger}\underbrace{[D_{a}(r)]^{\dagger}\Gamma_{j}[D_{\phi}(r)]_{jl}[D_{b}(r)]}_{\Gamma_{l}}[D_{\phi}(s)]_{lk}[D_{b}(s)]=\\
			& = [D_{a}(s)]^{\dagger}\Gamma_{l}[D_{\phi}(s)]_{lk}[D_{b}(s)]=\Gamma_{k}.
		\end{split}
	\end{displaymath}
\end{proof}
\hspace{0mm}\\
For we will use them very often, we will name
	\begin{equation}\label{inveq}
		[D_{a}(f)]^{\dagger}\Gamma_{j}[D_{\phi}(f)]_{jk}[D_{b}(f)]=\Gamma_{k} \quad f\in \mathrm{gen}(G), k=1,...,n_{\phi}
	\end{equation} 
the \textit{invariance equations} for the Yukawa-coupling $Y_{\Gamma}$.

\section{Solving the invariance equations}\label{solve}
We will now bring the invariance equations (\ref{inveq}) into a form such that they can be solved easily.
\medskip
\\
At first we notice that the invariance equations are linear in the unknown matrices $\Gamma_{j}$. This is the key property that makes solving the equations surprisingly simple.

\begin{prop}\label{PY7}
	The invariance equations (\ref{inveq}) can be interpreted as an eigenvalue problem of an $n_{a}n_{b}n_{\phi}\times n_{a}n_{b}n_{\phi}$-matrix $N$. The numbers $\Gamma_{jkl}$ can be interpreted as an $n_{a}n_{b}n_{\phi}$-dimensional columnvector $\Gamma$, which is an eigenvector of $N$ to the eigenvalue $1$.
		\begin{displaymath}
			N\Gamma=\Gamma.
		\end{displaymath}
\end{prop}

\begin{proof}
	This becomes clear, if we write the invariance equations in components:
		\begin{displaymath}
			[D_{\phi}]_{ir}[D_{a}]_{js}^{\ast}[D_{b}]_{kt}\Gamma_{ijk}=\Gamma_{rst}.
		\end{displaymath}
The lefthand side of the above equation is a linear function of all $\Gamma_{ijk}$, thus there exists a matrix $N$ s.t. the equation can be written as
		\begin{displaymath}
			N\Gamma=\Gamma
		\end{displaymath}
with
		\begin{displaymath}
			\Gamma=\left(
				\begin{matrix}
					\Gamma_{111}  \\
					\Gamma_{112}  \\
					\vdots  \\
					\Gamma_{11n_{b}}\\
					\Gamma_{121}\\
					\Gamma_{122}\\
					\vdots\\
 					\Gamma_{n_{\phi}n_{a}n_{b}}
			       \end{matrix}
				\right).
		\end{displaymath}
We will also give the explicit form of $N$. For the sake of clarity we will use the abbreviations
	\begin{displaymath}
		\begin{split}
			& [D_{a}]^{\ast}=:A,\\
			& [D_{b}]=:B,\\
			& [D_{\phi}]=:\Phi.
		\end{split}
	\end{displaymath}
	\begin{equation}\label{eqY7}
		\Rightarrow \Phi_{ir}A_{js}B_{kt}\Gamma_{ijk}=\Gamma_{rst}.
	\end{equation}
We read off $N$:
	\begin{displaymath}
		N=
		\left(
		\begin{matrix}
		 \Phi_{11}A_{11}B_{11} & \Phi_{11}A_{11}B_{21} & ... & \Phi_{n_{\phi}1}A_{n_{a}1}B_{n_{b}1} \\
		 \Phi_{11}A_{11}B_{12} & \Phi_{11}A_{11}B_{22} & ... & \Phi_{n_{\phi}1}A_{n_{a}1}B_{n_{b}2} \\
		 \vdots &  &  & \vdots \\
		 \Phi_{1n_{\phi}}A_{1n_{a}}B_{1n_{b}} & \Phi_{1n_{\phi}}A_{1n_{a}}B_{2n_{b}} & ... & \Phi_{n_{\phi}n_{\phi}}A_{n_{a}n_{a}}B_{n_{b}n_{b}}
		\end{matrix}
		\right).
	\end{displaymath}
\end{proof}

\begin{cor}\label{CY8}
	\textbf{Construction of solutions of the invariance equations}
	\begin{enumerate}
		\item Construct the matrix $N(f)$ for all generators $f\in \mathrm{gen}(G)$.
		\item Determine the eigenspace of $N(f)$ to the eigenvalue $1$.
		\item Determine the intersection of the eigenspaces to find the common eigenspace to the eigenvalue $1$ of all $N(f)$.
		\item Determine a basis of the common eigenspace. Each basisvector is a solution of the invariance equations.
	\end{enumerate}
\end{cor}
\begin{proof}
	This follows directly from proposition \ref{PY7}.
\end{proof}

\section{Mathematical interpretation of $\Gamma_{j}$}
In section \ref{solve} we gave a procedure that enables calculating the unknown matrices $\Gamma_{j}$ if one knows $D_{\phi}$ and $V_{\phi}$. As stated earlier we want to find $\Gamma_{j}$, $V_{\phi}$ and $D_{\phi}$ (irreducible) for given $G, V_{a}, V_{b}, D_{a}, D_{b}$.
\\
Investigating the solution to this more general problem one can find the mathematical meaning of $\Gamma_{j}$ - it are matrices of Clebsch-Gordan coefficients.

\begin{prop}\label{PY9}
	Let $D_{a}, D_{b}$ and $D^{\lambda}$ be irreducible representations of a finite group $G$, then $C^{\lambda}_{ijk}$ (not all equal zero) are a solution of the invariance equations
	\begin{equation}\label{inveqPY9a}
	[D^{\lambda}]_{ir}[(D_{a}^{-1})^{T}]_{js}[(D_{b}^{-1})^{T}]_{kt}C^{\lambda}_{ijk}=C^{\lambda}_{rst}
	\end{equation}
	if and only if they are Clebsch-Gordan coefficients for the Clebsch-Gordan decomposition\footnote{See definitions \ref{DA89} and \ref{DA90}.}
	\begin{displaymath}
		D_{a}\otimes D_{b}=\bigoplus_{\lambda}D^{\lambda},\quad V_{a}\otimes V_{b}=\bigoplus_{\lambda}V^{\lambda}
	\end{displaymath}
	(see section \ref{ACGC}). 
\end{prop}

\begin{proof}
We first prove that the Clebsch-Gordan coefficients are solutions of the invariance equations (\ref{inveqPY9a}).
\medskip
\\
Let $\{e^{a}_{i}\}_{i}$ and $\{e^{b}_{j}\}_{j}$ be bases of $V_{a}$ and $V_{b}$, then the Clebsch-Gordan coefficients define a basis transformation of a basis of $V_{a}\otimes V_{b}$ to a basis $\{u^{\lambda}_{k}\}_{k}$ of $V^{\lambda}$:
	\begin{displaymath}
	u^{\lambda}_{j}=C^{\lambda}_{jkl}e^{a}_{k}\otimes e^{b}_{l},
	\end{displaymath}
such that
	\begin{displaymath}
	(D_{a}\otimes D_{b})u^{\lambda}_{j}=D^{\lambda}u^{\lambda}_{j}.
	\end{displaymath}
(no summation over $\lambda$!) It follows
	\begin{displaymath}
	\begin{split}
	(D_{a}\otimes D_{b})u^{\lambda}_{j} & = (D_{a}\otimes D_{b})(C^{\lambda}_{jkl}e^{a}_{k}\otimes e^{b}_{l})=C^{\lambda}_{jkl}[D_{a}]_{mk}[D_{b}]_{nl}e^{a}_{m}\otimes e^{b}_{n}=\\
	& = D^{\lambda}u^{\lambda}_{j}=[D^{\lambda}]_{sj}u^{\lambda}_{s}=[D^{\lambda}]_{sj}C^{\lambda}_{smn}e^{a}_{m}\otimes e^{b}_{n}.
	\end{split}
	\end{displaymath}
	\begin{displaymath}
	\Rightarrow C^{\lambda}_{jkl}[D_{a}]_{mk}[D_{b}]_{nl}=[D^{\lambda}]_{sj}C^{\lambda}_{smn} /\cdot [D_{a}^{-1}]_{am}[D_{b}^{-1}]_{bn}
	\end{displaymath}
	\begin{displaymath}
	\Rightarrow [D^{\lambda}]_{sj}[(D_{a}^{-1})^{T}]_{ma}[(D_{b}^{-1})^{T}]_{nb}C^{\lambda}_{smn}=C^{\lambda}_{jab}.
	\end{displaymath}
The proof of the converse works similar. We start from equation (\ref{inveqPY9a}):
	\begin{displaymath}
	[D^{\lambda}]_{ir}[(D_{a}^{-1})^{T}]_{js}[(D_{b}^{-1})^{T}]_{kt}C^{\lambda}_{ijk}=C^{\lambda}_{rst}
	\end{displaymath}
	\begin{displaymath}
	\Rightarrow [D^{\lambda}]_{ir}C^{\lambda}_{ijk}=[D_{a}]_{js}[D_{b}]_{kt}C^{\lambda}_{rst}/\cdot e^{a}_{j}\otimes e^{b}_{k}
	\end{displaymath}
	\begin{displaymath}
	[D^{\lambda}]_{ir}\underbrace{C^{\lambda}_{ijk}e^{a}_{j}\otimes e^{b}_{k}}_{=:u^{\lambda}_{i}}=\underbrace{[D_{a}]_{js}[D_{b}]_{kt}C^{\lambda}_{rst}e^{a}_{j}\otimes e^{b}_{k}}_{C^{\lambda}_{rst}(D_{a}\otimes D_{b})(e^{a}_{s}\otimes e^{b}_{t})}
	\end{displaymath}
	\begin{displaymath}
	\Rightarrow D^{\lambda}u^{\lambda}_{r}=(D_{a}\otimes D_{b})u^{\lambda}_{r}
	\end{displaymath}
$\Rightarrow$ $u^{\lambda}_{r}\in V^{\lambda}$ $\Rightarrow$ $C^{\lambda}_{rst}$ are Clebsch-Gordan coefficients for
	\begin{displaymath}
	D_{a}\otimes D_{b}=\bigoplus_{\lambda}D^{\lambda},\quad V_{a}\otimes V_{b}=\bigoplus_{\lambda}V^{\lambda}.
	\end{displaymath}
\end{proof}

\begin{cor}\label{CY10}
Suppose $G, V_{a}, V_{b}, D_{a}, D_{b}$ are given ($D_{a}, D_{b}$ irreducible, $V_{a}, V_{b}$ vectorspaces over $\mathbb{C}$). Let
	\begin{displaymath}
		(D_{a}^{-1})^{\dagger}\otimes (D_{b}^{-1})^{T}=\bigoplus_{\lambda}D^{\lambda},\quad V_{a}\otimes V_{b}=\bigoplus_{\lambda}V^{\lambda}
	\end{displaymath}
	be the Clebsch-Gordan decomposition of $(D_{a}^{-1})^{\dagger}\otimes (D_{b}^{-1})^{T}$, and let $\Gamma^{\lambda}_{jkl}$ be the Clebsch-Gordan coefficients.
	\medskip
	\\
	Then $Y_{\Gamma^{\lambda}}(a,b,\phi)$ (and their linear combinations) with $D_{\phi}=D^{\lambda}$ are the only Yukawa-couplings that are invariant under
		\begin{displaymath}
		\begin{split}
		& a\mapsto D_{a}a,\\
		& b\mapsto D_{b}b,\\
		& \phi\mapsto D_{\phi}\phi.
		\end{split}
		\end{displaymath}
\end{cor}

\begin{proof}
This follows from propositions \ref{PY4} and \ref{PY9}.
\end{proof}
\hspace{0mm}\\
We will later systematically analyse tensor products of irreducible representations. For the sake of efficiency it is helpful to think about the relation between $D_{a}\otimes D_{b}$ and $D_{b}\otimes D_{a}$.

\begin{prop}\label{PY11}
Let $D_{a}$ and $D_{b}$ be irreducible representations of a finite group $G$, let $D_{a}\otimes D_{b}=\bigoplus_{\lambda}D^{\lambda}$ be the Clebsch-Gordan decomposition, and let $C^{\lambda}_{ijk}$ be the Clebsch-Gordan coefficients. Then the Clebsch-Gordan decomposition of $D_{b}\otimes D_{a}$ is the same as for $D_{a}\otimes D_{b}$, and the Clebsch-Gordan coefficients $\widehat{C}^{\lambda}_{ijk}$ for $D_{b}\otimes D_{a}$ are given by
	\begin{displaymath}
	\widehat{C}^{\lambda}_{ijk}=C^{\lambda}_{ikj}.
	\end{displaymath}
\end{prop}

\begin{proof}
The Clebsch-Gordan decomposition can be calculated using a Hermitian scalar product on the set of characters of the representations ($\rightarrow$ proposition \ref{PA59a}). Since the characters of $D_{a}\otimes D_{b}$ and $D_{b}\otimes D_{a}$ are equal, also the Clebsch-Gordan decomposition must be the same.
\\
Let $u$ be an element of an invariant subspace of $V_{a}\otimes V_{b}$:
	\begin{displaymath}
	u=C^{\lambda}_{ijk}e^{a}_{j}\otimes e^{b}_{k}.
	\end{displaymath}
It follows that
	\begin{displaymath}
	\underbrace{C^{\lambda}_{ijk}}_{\widehat{C}^{\lambda}_{ikj}}e^{b}_{k}\otimes e^{a}_{j}
	\end{displaymath}
is an element of the corresponding invariant subspace of $V_{b}\otimes V_{a}$. $\Rightarrow \widehat{C}^{\lambda}_{ikj}=C^{\lambda}_{ijk}$.
\end{proof}

%% file: clebschgordan/clebschgordan.tex
\chapter{The Clebsch-Gordan coefficients and their calculation}\label{chapterclebsch}
In chapter \ref{yukawa} we saw that we will need Clebsch-Gordan coefficients to construct $G$-invariant Yukawa-couplings. The aim of this chapter is to investigate some of their properties and an efficient method for their calculation.

\section{Important properties of Clebsch-Gordan coefficients}

At first we notice that the Clebsch-Gordan coefficients are basis dependent, thus it will be useful to know their transformation properties under basis change.

\begin{prop}\label{Pclebsch1}
Let $C_{ijk}$ be the Clebsch-Gordan coefficients for
	\begin{displaymath}
	D_{a}\otimes D_{b}=D^{\lambda}\oplus...
	\end{displaymath}
in a given basis. (The basis is determined by the matrix representations of the operators $D_{a},D_{b}$ and $D^{\lambda}$.) Let the matrix $C_{i}$ be defined by $(C_{i})_{jk}=C_{ijk}$. Let furthermore $A, B, \Phi$ be matrix representations of $D_{a},D_{b}$ and $D^{\lambda}$. Then under a basis transformation
	\begin{displaymath}
	\begin{split}
	& A\mapsto A'=S_{A}^{-1}AS_{A},\\
	& B\mapsto B'=S_{B}^{-1}BS_{B},\\
	& \Phi\mapsto \Phi'=S_{\lambda}^{-1}\Phi S_{\lambda}
	\end{split}
	\end{displaymath}
the Clebsch-Gordan coefficients transform as
	\begin{displaymath}
	C_{i}\mapsto S_{A}^{-1}C_{m}(S_{B}^{-1})^{T}(S_{\lambda})_{mi}.
	\end{displaymath}
\end{prop}

\begin{proof}
The Clebsch-Gordan coefficients are the solutions of the invariance equations
	\begin{displaymath}
	\Phi_{ir}(A^{-1})^{T}_{js}(B^{-1})^{T}_{kt}C_{ijk}=C_{rst},
	\end{displaymath}
which are in matrix form
	\begin{displaymath}
	(A^{-1}C_{i}(B^{-1})^{T})\Phi_{ir}=C_{r}.
	\end{displaymath}
After a basis transformation the new Clebsch-Gordan coefficients $C_{r}'$ must fulfil the invariance equations
	\begin{displaymath}
	(A'\hspace{0mm}^{-1}C'_{i}(B'\hspace{0mm}^{-1})^{T})\Phi'_{ir}=C'_{r}.
	\end{displaymath}
Inserting the expressions for the new matrix representations we find
	\begin{displaymath}
	(S_{A}^{-1}A^{-1}S_{A}C_{i}'S_{B}^{T}(B^{-1})^{T}(S_{B}^{-1})^{T})(S_{\lambda}^{-1})_{is}\Phi_{sm}(S_{\lambda})_{mr}=C_{r}',
	\end{displaymath}
from which follows
	\begin{displaymath}
	A^{-1}(\underbrace{S_{A}C_{i}'S_{B}^{T}(S_{\lambda}^{-1})_{is}}_{C_{s}})(B^{-1})^{T}\Phi_{sz}=\underbrace{S_{A}C_{r}'S_{B}^{T}(S_{\lambda}^{-1})_{rz}}_{C_{z}}.
	\end{displaymath}
Therefore
	\begin{displaymath}
	C_{s}'=S_{A}^{-1}C_{i}(S_{B}^{-1})^{T}(S_{\lambda})_{is}.
	\end{displaymath}
\end{proof}

\begin{prop}\label{Pclebsch2}
Let $[D_{a}\otimes D_{b}]$ and $\bigoplus_{\lambda}[D^{\lambda}]$ be unitary matrix representations of a finite group $G$ with respect to a basis $\{e_{i}\}_{i}$ of $V_{a}\otimes V_{b}=\bigoplus_{\lambda}V^{\lambda}$, and let
	\begin{displaymath}
	D_{a}\otimes D_{b}=\bigoplus_{\lambda}D^{\lambda}.
	\end{displaymath}
Let $M$ be the matrix of basis change to the bases $\{u^{\lambda}_{j}\}_{j}$ of the invariant subspaces $V^{\lambda}$, i.s.
	\begin{displaymath}
	Me_{i}=u_{i}.
	\end{displaymath}
(The components of $M$ are the Clebsch-Gordan coefficients. To be exact we should rather have written $M[e_{i}]=[u_{i}]$, because we want $M$ to be a matrix, not an operator. $[.]$ denotes matrix representation with respect to the basis $\{e_{i}\}_{i}$ of $V_{a}\otimes V_{b}$.) Let $R=\bigoplus_{\lambda}[D^{\lambda}]$ be the reduced matrix representation of $D_{a}\otimes D_{b}$ in the basis $\{e_{i}\}_{i}$, then
	\begin{equation}\label{clebsch1}
	M^{-1}[D_{a}\otimes D_{b}]M=R,
	\end{equation}
and $M$ is unitary.
\end{prop}

\begin{proof}
	\begin{displaymath}
	\begin{split}
	M^{-1}[D_{a}\otimes D_{b}]Me_{i} & = M^{-1}[D_{a}\otimes D_{b}]u_{i}=\\
	& = M^{-1}R_{ji}u_{j}=R_{ji}e_{j}=Re_{i}.
	\end{split}
	\end{displaymath}
We will now show that $M$ is unitary.
	\begin{displaymath}
	(M^{-1}[D_{a}\otimes D_{b}](g)M)^{\dagger}=R^{\dagger}(g)=R(g)^{-1}=R(g^{-1})
	\end{displaymath}
	\begin{displaymath}
	\Rightarrow M^{\dagger}[D_{a}\otimes D_{b}](g^{-1})(M^{-1})^{\dagger}=R(g^{-1})
	\end{displaymath}
Since equation (\ref{clebsch1}) must be valid for every $g\in G$ (so especially for $g^{-1}$) it follows $M^{-1}=M^{\dagger}$.
\end{proof}
\hspace{0mm}\\
When we will later calculate Clebsch-Gordan coefficients we would like to choose some kind of \textquotedblleft standard form\textquotedblright, which is not possible, because of their basis dependence. From proposition \ref{Pclebsch2} we know that the matrix $M$ of Clebsch-Gordan coefficients can always be transformed to a unitary form. This does of course not fix $M$ uniquely, but it is better than choosing an arbitrary basis.
\bigskip
\bigskip
\\
We will now proof some lemmata that will save us a large amount of work later.

\bigskip
\begin{lemma}\label{LSU313}
Let $D_{1}, D_{2}, D_{1}', D_{2}', d_{1}, d_{2}$ be irreducible representations of a finite group $G$. Let $\mathrm{dim}d_{1}=\mathrm{dim}d_{2}=1$, and let furthermore
	\begin{displaymath}
	D_{1}'=d_{1}\otimes D_{1},\quad\mbox{and}\quad D_{2}'=d_{2}\otimes D_{2}.
	\end{displaymath}
Then the Clebsch-Gordan coefficients for $D_{1}\otimes D_{2}$ and $D_{1}'\otimes D_{2}'$ are equal.
\end{lemma}

\begin{proof}
Let $M$ be the matrix of Clebsch-Gordan coefficients that reduces $D_{1}\otimes D_{2}$.
	\begin{displaymath}
	M^{-1}[D_{1}\otimes D_{2}]M=R,
	\end{displaymath}
where $R$ is a matrix representation of the reduced tensor product.
	\begin{displaymath}
	\begin{split}
	\Rightarrow M^{-1}[D_{1}'\otimes D_{2}']M & =M^{-1}[(d_{1}\otimes D_{1})\otimes (d_{2}\otimes D_{2})]M\\
	& = [d_{1}]\cdot[d_{2}]M^{-1}[D_{1}\otimes D_{2}]M=[d_{1}]\cdot[d_{2}]R.
	\end{split}
	\end{displaymath}
$[d_{1}]\cdot[d_{2}]R$ is the reduced form of the tensor product $[D_{1}'\otimes D_{2}']$.
\end{proof}

\bigskip
\begin{lemma}\label{LSU314}
Let $M$ be the matrix of Clebsch-Gordan coefficients that reduces the tensor product $D_{1}\otimes D_{2}$ ($D_{1}, D_{2}$ irreducible), then $M^{\ast}$ reduces $D_{1}^{\ast}\otimes D_{2}^{\ast}$.
\end{lemma}

\begin{proof}
	\begin{displaymath}
	M^{-1}[D_{1}\otimes D_{2}]M=R \Rightarrow (M^{\ast})^{-1}[D_{1}^{\ast}\otimes D_{2}^{\ast}]M^{\ast}=R^{\ast}.
	\end{displaymath}
Since $R$ is block diagonal (reduced form), $R^{\ast}$ is block diagonal too.
\end{proof}

\section{An algorithm for the calculation of Clebsch-Gordan coefficients}

\begin{theorem}\label{Tclebsch3}
\textbf{Diagonalizeability of normal matrices.} Let $A$ be a complex $n\times n$-matrix with $AA^{\dagger}=A^{\dagger}A$ (such a matrix is called \textit{normal}), then there exists a unitary matrix $S$ such that
	\begin{displaymath}
	S^{-1}AS=\mathrm{diag}(\lambda_{i}),\quad \lambda_{i}\in \mathbb{C}.
	\end{displaymath}
A is called \textit{unitarily diagonalizeable}. As special cases we find that unitary and Hermitian matrices are unitarily diagonalizeable.
\medskip
\\ 
The proof of this theorem can be found in \cite{zieschang} (p.271f).
\end{theorem}

\begin{cor}\label{Cclebsch4}
Let $A$ be equivalent to a unitary matrix $U$, that means $\exists T$: $A=TUT^{-1}$, then $A$ is diagonalizeable and has the same eigenvalues as $U$.
\end{cor}

\begin{proof}
$U=T^{-1}AT$. From theorem \ref{Tclebsch3} we know that $U$ is diagonalizeable, thus
	\begin{displaymath}
	\mathrm{diag}(\lambda_{i})=S^{-1}US=S^{-1}T^{-1}ATS=(TS)^{-1}A(TS).
	\end{displaymath}
Remark: In general $TS$ is not unitary, thus $A$ is not unitarily diagonalizeable.
\end{proof}
\hspace{0mm}\\
We now have the tools for the development of a method for the construction of Clebsch-Gordan coefficients.

\begin{prop}\label{Pclebsch5}
Let $D_{a}\otimes D_{b}=D^{\lambda}\oplus...$ . The Clebsch-Gordan coefficients for this decomposition can be obtained by the following procedure:
	\begin{enumerate}
	 \item $A:=[D_{a}^{-1}]^{T}, B:=[D_{b}^{-1}]^{T}, \Phi:=[D^{\lambda}]$. Construct $N(f)$ $\forall f\in \mathrm{gen}(G)$.
	\begin{displaymath}
	N:=
		\left(
		\begin{matrix}
		 \Phi_{11}A_{11}B_{11} & \Phi_{11}A_{11}B_{21} & ... & \Phi_{n_{\phi}1}A_{n_{a}1}B_{n_{b}1} \\
		 \Phi_{11}A_{11}B_{12} & \Phi_{11}A_{11}B_{22} & ... & \Phi_{n_{\phi}1}A_{n_{a}1}B_{n_{b}2} \\
		 \vdots &  &  & \vdots \\
		 \Phi_{1n_{\phi}}A_{1n_{a}}B_{1n_{b}} & \Phi_{1n_{\phi}}A_{1n_{a}}B_{2n_{b}} & ... & \Phi_{n_{\phi}n_{\phi}}A_{n_{a}n_{a}}B_{n_{b}n_{b}}
		\end{matrix}
		\right)
	\end{displaymath}
	\item Determine the eigenspace of $N(f)$ to the eigenvalue $1$.
		\item Determine the intersection of the eigenspaces to find the common eigenspace to the eigenvalue $1$ of all $N(f)$.
		\item Determine a basis of the common eigenspace. Each basisvector corresponds to a set of Clebsch-Gordan coefficients via
	\begin{displaymath}
	\Gamma=\left(
	\begin{matrix}
	\Gamma_{111}  \\
	\Gamma_{112}  \\
	\vdots  \\
	\Gamma_{11n_{b}}\\
	\Gamma_{121}\\
	\Gamma_{122}\\
	\vdots\\
 	\Gamma_{n_{\phi}n_{a}n_{b}}
	\end{matrix}
	\right).
	\end{displaymath}
	\end{enumerate}
\end{prop}

\begin{proof}
This follows from propositions \ref{PY7} and \ref{PY9}.
\end{proof}
\hspace{0mm}\\
Applying this method in a calculation by hand is nearly impossible, because of the large dimensions of the involved matrices and vectors. In fact it turns out that for \textquotedblleft complicated\textquotedblright\hspace{1mm} generators $f$ (containing elements as $e^{\frac{2\pi i}{m}}$ and square roots like $\sqrt{m}$) this method cannot directly be used with \textit{Mathematica 6} \cite{mathematica}, because the calculation of the eigenvectors to the eigenvalue 1, and also the intersection of the eigenspaces would take extremely long. (In some tests \textit{Mathematica 6} did not give any result, even after 8 hours of computation.) So we have to improve the algorithm. After many little improvements an efficient implementation of the procedure described in proposition \ref{Pclebsch5} was found.

\newpage
\begin{center}
\textbf{An efficient algorithm for the calculation of Clebsch-Gordan coefficients.}\end{center}
\begin{enumerate}
	\item Let $f_{1}$ be the first generator of $G$. We search for the eigenvectors of $N(f_{1})$ to the eigenvalue 1. This is done by calculating the kernel of $N(f_{1})-\mathbbm{1}_{n_{a}n_{b}n_{\lambda}}$. In an arbitrary basis this will be an enormous amount of work, but performing a special basis transformation one can make this problem easily solveable.
		\begin{itemize}
		 \item[a)] Diagonalize $A(f_{1}), B(f_{1})$ and $\Phi(f_{1})$. This is possible, because representations of finite groups are equivalent to unitary representations ($\rightarrow$ theorem \ref{TA42} and corollary \ref{Cclebsch4}).
		 \item[b)] Calculate the kernel of $N(f_{1})-\mathbbm{1}_{n_{a}n_{b}n_{\lambda}}$. This is easy now, because in the new basis most elements of $N(f_{1})-\mathbbm{1}_{n_{a}n_{b}n_{\lambda}}$ are zero.
		 \item[c)] Using the transformation properties of Clebsch-Gordan coefficients (proposition \ref{Pclebsch1}) transform the kernel back to the previous basis. So one obtains a set of linearly independent eigenvectors $v_{i}$ of $N(f_{1})$ to the eigenvalue 1.
		\end{itemize}
	\item Solve the equation
		\begin{displaymath}
		\alpha_{i}(N(f_{2})-\mathbbm{1}_{n_{a}n_{b}n_{\lambda}})v_{i}=0,
		\end{displaymath}
		and define $v_{j}'=\alpha_{i}v_{i}$ for the different solutions $\alpha_{i}$. Then $v_{j}'$ are common eigenvectors of $N(f_{1})$ and $N(f_{2})$ to the eigenvalue 1.  Repeat this step for all other generators of $G$. At the end one obtains the common eigenvectors of $N(f_{i})$ to the eigenvalue 1.
\end{enumerate}
This algorithm works well in \textit{Mathematica 6} and gives results within a computation time of about 10 seconds (order of magnitude) for most groups we will study in this work.
\medskip
\\
\textbf{Remark:} If one wants to find the Clebsch-Gordan coefficients in the \textquotedblleft unitary form\textquotedblright, one has to start with unitary representations $A(f_{i}), B(f_{i}), \Phi(f_{i})$. At the end, when one constructs the matrix $M$ (defined in proposition \ref{Pclebsch2}) one has to normalize each column of $M$ to make it unitary.
\bigskip
\\
We will later list Clebsch-Gordan coefficients in the following form:
	\begin{enumerate}
	 \item $M$ has \textquotedblleft unitary form\textquotedblright.
	 \item Let $e_{ij}:=e_{i}^{a}\otimes e_{i}^{b}$ be basis vectors of $V_{a}\otimes V_{b}$, and let $D_{a}\otimes D_{b}=D^{\lambda}\oplus...$ We will give the Clebsch-Gordan coefficients for this decomposition by giving the basis vectors of $V^{\lambda}$:
	\begin{displaymath}
	u_{D^{\lambda}}^{D_{a}\otimes D_{b}}(i)=C^{\lambda}_{ijk}e_{jk},\quad i=1,...,n_{\lambda}.
	\end{displaymath}
	\end{enumerate}

%% file: su3/su3.tex
\chapter{The finite subgroups of $SU(3)$}\label{SU3chapter}

The aim of this work is to systematically analyse finite family symmetry groups. The idea of a systematic analysis of finite groups for application to particle physics is not new. Consider for example the work of Fairbairn, Fulton and Klink \cite{fairbairn} on the application of finite subgroups of $SU(3)$ in particle physics and the analysis of Frampton and Kephart \cite{frampton} who studied all finite groups of order smaller than 32.
\medskip
\\
From proposition \ref{PY3} we know that Yukawa-couplings
	\begin{displaymath}
	Y_{\Gamma}=a^{\dagger}(\Gamma_{j}\phi_{j})b+\mathrm{H.c.}
	\end{displaymath}
that are invariant under
	\begin{displaymath}
	a\mapsto D_{a}a, \quad b\mapsto D_{b}b,\quad \phi\mapsto D_{\phi}\phi
	\end{displaymath}
split up into independent $G$-invariant Yukawa-couplings if one of the representations $D_{a},D_{b},D_{\phi}$ is reducible. Therefore Yukawa-couplings that are invariant under the action of irreducible representations of a finite group $G$ form the basic building blocks of general $G$-invariant Yukawa-couplings.
\medskip
\\
Since every fermion family consists (to current knowledge) of 3 members we will investigate $G$-invariant Yukawa-couplings where $D_{a}$ and $D_{b}$ are 3-dimen\-sion\-al. Furthermore we will restrict ourselves to the special case of 3-dimensional irreducible representations $D_{a}$ and $D_{b}$. The central question is therefore:
	\begin{quote}
	\textit{Can we find all possible Yukawa-couplings that are invariant under the action of 3-dimensional irreducible representations $D_{a},D_{b}$ of a finite group $G$?}
	\end{quote}
This problem is equivalent to the question for all possible Clebsch-Gordan coefficients for $\textbf{\underline{3}}\otimes \textbf{\underline{3}}$-tensor products of $3$-dimensional irreducible representations of all finite groups.
\medskip
\\
This problem can be simplified using theorem \ref{TA42}:
\medskip
\\
\textit{Every representation of a finite group is equivalent to a unitary representation.}
\medskip
\\
Let now $D$ be a 3-dimensional irreducible representation of a finite group $G$. $\Rightarrow$ $D(G)$ is isomorphic to a finite subgroup of $U(3)$. It is therefore enough for our purpose to find all 3-dimensional (faithful) irreducible representations of finite subgroups of $U(3)$. Knowing all these irreducible representations one could calculate the corresponding Clebsch-Gordan coefficients. This would lead to
	\begin{itemize}
	 \item knowledge of all $G$-invariant Yukawa-couplings (invariant under irre\-duc\-i\-ble $D_{a},D_{b}$)
	 \item the possibility to create a list of \underline{all} possible Dirac mass matrices of models involving a $G$-invariant Yukawa-coupling (invariant under irreducible $D_{a},D_{b}$).
	\end{itemize}
Unfortunately the finite subgroups of $U(3)$ have (to our knowledge) not been classified yet. The finite subgroups of $SU(3)$ on the other hand were already classified at the beginning of the 20\textsuperscript{th} century by Miller, Dickson and Blichfeldt in \cite{miller}. This work is therefore restricted to a systematic analysis of the finite subgroups of $SU(3)$.
\bigskip
\bigskip
\\
In this chapter we will systematically list all known finite subgroups of $SU(3)$. Especially we will list all known finite subgroups of $SU(3)$ that have three-dimensional irreducible representations. For these subgroups we will either derive the irreducible representations and calculate the Clebsch-Gordan coefficients, or we will list these properties giving the references where the desired data can be found.

\section{The non-Abelian finite subgroups of $SU(3)$}\label{nonabeliansu3section}

The finite subgroups of $SU(3)$ have already been classified at the beginning of the 20\textsuperscript{th} century by Miller, Dickson and Blichfeldt in their textbook \cite{miller}. However, Miller et al. only classified the groups and listed their generators, but they did not calculate character tables, tensor products or Clebsch-Gordan coefficients.
\medskip
\\
Fairbairn, Fulton and Klink were the first who studied the applications of finite subgroups of $SU(3)$ to particle physics \cite{fairbairn}. In their paper \cite{fairbairn} they listed character tables and generators for most of the finite subgroups of $SU(3)$. From this paper there emerged a collaboration - let us call it FFK/BLW collaboration\footnote{FFK/BLW = Fairbairn, Fulton, Klink/ Bovier, L\"uling, Wyler.} - between two groups of physicists who studied the finite subgroups of $SU(3)$ in more detail.
\medskip
\\
This collaboration consisted of the participants Fairbairn, Fulton, Klink (FFK) as well as Bovier, L\"uling, Wyler (BLW). The main papers of this collaboration are \cite{fairbairn,BLW1,BLW2,fairbairn2}.
\bigskip
\\
Following \cite{fairbairn} the finite subgroups of $SU(3)$ can be divided into three different types:
	
\begin{center}
	\begin{itemize}
	 \item Abelian finite subgroups of $SU(3)$ (including all subgroups of $U(1)$)
	 \item Non-Abelian finite subgroups of $SU(2)$ and $SO(3)$
	 \item Non-Abelian finite subgroups of $SU(3)$ that are not subgroups of $SU(2)$ or $SO(3)$.
	\end{itemize}
\end{center}

\begin{lemma}\label{LSU31}
	All irreducible representations of Abelian groups are one-di\-men\-sion\-al.
\end{lemma}

\begin{proof}
	Let $D$ be an irreducible representation of an Abelian finite group $G$ on a vectorspace $V_{D}$. It follows
	\begin{displaymath}
		D(a)D(b)=D(b)D(a) \quad\forall a,b\in G.
	\end{displaymath}
	Using Schur's lemma (\ref{LA50})
	we see that $D(b)=\lambda(b)id$, $\lambda:G\rightarrow \mathbb{C}$. Since $D$ is irreducible $[D(b)]=\lambda(b)$ (if $[D(b)]=\lambda(b) \mathbbm{1}_{n}, n>1,$ it would be reducible), thus $V_{D}$ is one-dimensional.
\end{proof}
\hspace{0mm}\\
We are only interested in groups that have three-dimensional irreducible representations, thus we can rule out all Abelian groups.
\medskip
\\
Now for the finite non-Abelian subgroups of $SU(2)$. As it is commonly known $SU(2)$ is homomorphic to $SO(3)$, the group of proper rotations in $\mathbb{R}^{3}$. Using this fact one can construct all finite subgroups of $SU(2)$ from the finite subgroups of $SO(3)$, which can be found in nearly every textbook on group theory. They are a subset of the set of so-called \textit{point groups}. The finite non-Abelian subgroups of $SO(3)$ are listed in table \ref{SO3subgroups}.

\begin{table}[h]
\begin{center}
\renewcommand{\arraystretch}{1.4}
\begin{tabular}{|llcl|}
\firsthline
	Type of subgroups &  Subgroup & Symbol & Order \\
\hline
	Polyhedral groups  & tetrahedral group & $T$ & 12\\
	 & octahedral group & $O$ & 24\\
	 & icosahedral group & $I$ & 60\\
\hline
	Dihedral groups $D_{n}, n\in\mathbb{N}\backslash\{0\}$  & dihedral group & $D_{n}$ & $2n$\\
\lasthline
\end{tabular}
\caption[Non-Abelian finite subgroups of $SO(3)$]{Non-Abelian finite subgroups of $SO(3)$ as given in \cite{hamermesh}.}
\label{SO3subgroups}
\end{center}
\end{table}
\hspace{0mm}\\
It turns out that for every finite subgroup $G$ of $SO(3)$, there exists a \textquotedblleft double cover\textquotedblright\hspace{1mm} $\tilde{G}\subset SU(2)$ with $\mathrm{ord}(\tilde{G})=2\mathrm{ord}(G)$.
\bigskip
\\
The third kind of finite subgroups of $SU(3)$ are the non-Abelian finite subgroups of $SU(3)$ that cannot be interpreted as finite subgroups of $SU(2)$ or $SO(3)$. These are the groups investigated by the FFK/BLW collaboration. Fairbairn et al. list the following groups in \cite{fairbairn}:
	\begin{itemize}
	 \item $\Sigma(n\phi)$, $n=36,72,216,360$
	 \item $\Sigma(m)$, $m=60,168$
	 \item $\Delta(3n^{2})$, $n\in\mathbb{N}\backslash \{0,1\}$
	 \item $\Delta(6n^{2})$, $n\in\mathbb{N}\backslash \{0,1\}$
	\end{itemize}
It will turn out later that the group $\Sigma(60)$ is identical with the icosahedral group $I$ which is a finite subgroup of $SO(3)$, thus it is already included in the list of finite subgroups of $SO(3)$ and $SU(2)$. Furthermore $\Delta(3\cdot 2^2)\simeq A_4\simeq T$, $\Delta(6\cdot 2^2)\simeq S_4\simeq O$ \cite{escobar}.
\medskip
\\
At this point the FFK/BLW collaboration started when Bovier, L\"uling and Wyler published their paper \cite{BLW1}. In this paper they claimed that some of the groups they had considered in their earlier paper \cite{BLW2} are finite subgroups of $SU(3)$ that are not included in the list of FFK \cite{fairbairn}.
\\
In the last step of the investigations of the collaboration Fairbairn and Fulton found out that not all (but some) groups found by BLW in \cite{BLW1} are indeed new subgroups of $SU(3)$, which was published in \cite{fairbairn2}. BLW accepted the corrections of Fairbairn and Fulton.
\\
The analysis of Miller et al. \cite{miller} mentions two series (C) and (D) of groups that were not further investigated by them. We will later see that the series $\Delta(3n^{2})$, $\Delta(6n^{2})$ and the new groups found by FFK/BLW form subsets of (C) and (D), thus it is likely that there exist further new groups of the types (C) and (D) that have not been investigated by FFK/BLW. We will analyse the new groups found by FFK/BLW and the series (C) and (D) at the end of the chapter.
\bigskip
\\
Before we will start analysing the finite subgroups of $SU(3)$ we will investigate two types of Clebsch-Gordan decompositions that will play an important role in the following.

\section{Two important types of Clebsch-Gordan decompositions}\label{CGCdecosubsect}
Consider a 3-dimensional irreducible representation $\textbf{\underline{3}}$ of a finite group $G$, then there are two basic forms of tensor products that can be formed using $\textbf{\underline{3}}$, namely
	\begin{displaymath}
	\textbf{\underline{3}}\otimes \textbf{\underline{3}}\quad\quad\mbox{and}\quad\quad \textbf{\underline{3}}\otimes \textbf{\underline{3}}^{\ast}.
	\end{displaymath}
Though the explicit type of irreducible representations contained in these tensor products depends on the group one can make general statements about the Clebsch-Gordan decompositions of $\textbf{\underline{3}}\otimes \textbf{\underline{3}}$ and $\textbf{\underline{3}}\otimes \textbf{\underline{3}}^{\ast}$.

\begin{lemma}\label{CGdecompL1}
Let $\textbf{\underline{3}}$ be a 3-dimensional irreducible representation of a finite group $G$, then the vector space $V_{\textbf{\underline{3}}}\otimes V_{\textbf{\underline{3}}}$ has a $3$-dimensional invariant subspace $V_{a}$ and a $6$-dimensional invariant subspace $V_{s}$. Bases of the invariant subspaces are given by
	\begin{displaymath}
	\begin{split}
	V_{a}:\enspace & u_{i}=\varepsilon_{ijk}e_{j}\otimes e_{k},\\
	V_{s}:\enspace & u_{i}=e_{(i)}\otimes e_{(i)},\\
	& 
	u_{i+3}=\vert\varepsilon_{ijk}\vert e_{j}\otimes e_{k},
	\end{split}
	\end{displaymath}
where $\{e_{j}\}_{j=1,2,3}$ is a basis of $V_{\textbf{\underline{3}}}$ and $i=1,2,3$. If the resulting 3- and 6-di\-men\-sion\-al representations corresponding to $V_{a}$ and $V_{s}$ are irreducible we find a Clebsch-Gordan decomposition of the form
	\begin{equation}\label{33CGdecomp}
	\textbf{\underline{3}}\otimes \textbf{\underline{3}}=\textbf{\underline{3}}_{a}\oplus \textbf{\underline{6}}_{s}.
	\end{equation}
\end{lemma}

\begin{proof}
For the proof of this lemma we use proposition \ref{PA86}, which states that the symmetric and antisymmetric subspaces
	\begin{displaymath}
	\begin{split}
	&V_{a}:=\mathrm{Span}(e_{[i}\otimes e_{j]})\quad \mathrm{dim}V_{a}=\frac{n(n-1)}{2},\\
	&V_{s}:=\mathrm{Span}(e_{(i}\otimes e_{j)})\quad \mathrm{dim}V_{s}=\frac{n(n+1)}{2}
	\end{split}
	\end{displaymath}
are invariant subspaces of $V\otimes V$ ($\mathrm{dim}V=n$). Here $\mathrm{dim}V_{\textbf{\underline{3}}}=3$, thus the symmetric subspace is invariant and 6-dimensional, so we have found $V_{\textbf{\underline{6}}}$. The remaining part of $V_{\textbf{\underline{3}}}\otimes V_{\textbf{\underline{3}}}$ must be $V_{\textbf{\underline{3}}}$, which turns out to be the antisymmetric subspace. We will name the corresponding representations $\textbf{\underline{3}}_{a}$ and $\textbf{\underline{6}}_{s}$.
\end{proof}
\hspace{0mm}\\
We can also read off the Clebsch-Gordan coefficients\footnote{Remark: The coefficients listed here are not normalized. To obtain normalized Clebsch-Gordan coefficients one has to construct the $9\times 9$-matrix $M$ that reduces the matrix representation of the tensor product and normalize each column of $M$. For $V_{a}$ one would get $\Gamma_{ijk}=\frac{1}{\sqrt{2}}\varepsilon_{ijk}$ for example.} for the Clebsch-Gordan decomposition (\ref{33CGdecomp}):
	\begin{displaymath}
	\begin{split}
	V_{a}:\enspace & u_{i}=\varepsilon_{ijk}e_{j}\otimes e_{k}=\Gamma_{ijk}e_{j}\otimes e_{k}\Rightarrow \Gamma_{ijk}=\varepsilon_{ijk},\\
	V_{s}:\enspace & u_{i}=e_{(i)}\otimes e_{(i)}=\Gamma_{ijk}e_{j}\otimes e_{k}\Rightarrow \Gamma_{ijk}=\delta_{(i)j}\delta_{(i)k},\\
	& 
	u_{i+3}=\vert\varepsilon_{ijk}\vert e_{j}\otimes e_{k}=\Gamma_{i+3\hspace{0.3mm} jk}e_{j}\otimes e_{k}\Rightarrow \Gamma_{i+3\hspace{0.3mm} jk}=\vert \varepsilon_{ijk}\vert.
	\end{split}
	\end{displaymath}
Let us now explicitly construct $\textbf{\underline{3}}_{a}$:

\begin{prop}\label{PSU32b}
Let $\textbf{\underline{3}}\otimes \textbf{\underline{3}}=\textbf{\underline{3}}_{a}\oplus \textbf{\underline{6}}_{s}$, then
	\begin{displaymath}
	\textbf{\underline{3}}_{a}=(\mathrm{det}\textbf{\underline{3}})\otimes \textbf{\underline{3}}^{\ast}.
	\end{displaymath}
Note that $\mathrm{det}\textbf{\underline{3}}$ is a 1-dimensional representation of the finite group.
\end{prop}

\begin{proof}
The (non normalized) Clebsch-Gordan coefficients for $\textbf{\underline{3}}\otimes \textbf{\underline{3}}\rightarrow \textbf{\underline{3}}_{a}$ are given by
	\begin{displaymath}
	\Gamma_{ijk}=\varepsilon_{ijk},
	\end{displaymath}
thus the corresponding invariance equations are ($\rightarrow$ proposition \ref{PY9})
	\begin{equation}\label{PSU32bequ1}
	U_{ir}[(D^{-1})^T]_{js}[(D^{-1})^T]_{kt}\varepsilon_{ijk}=\varepsilon_{rst},
	\end{equation}
where we have defined $[\textbf{\underline{3}}_{a}]=U$, $[\textbf{\underline{3}}]=D$. Multiplying equation (\ref{PSU32bequ1}) by $[U^{-1}]_{rp}[(D^{-1})^T]_{pz}$ we find
	\begin{displaymath}
	\underbrace{[(D^{-1})^T]_{iz}[(D^{-1})^T]_{js}[(D^{-1})^T]_{kt}\varepsilon_{ijk}}_{\mathrm{det}((D^{-1})^T)\varepsilon_{zst}}=\varepsilon_{rst}\underbrace{[U^{-1}]_{rp}[(D^{-1})^T]_{pz}}_{[(D^{T}U)^{-1}]_{rz}}
	\end{displaymath}
Multiplication with $\varepsilon_{ast}$ (using $\varepsilon_{zst}\varepsilon_{ast}=2\delta_{za}$) leads to
	\begin{displaymath}
	\frac{\delta_{za}}{\mathrm{det}D}=[(D^{T}U)^{-1}]_{az},
	\end{displaymath}
which is in matrix form
	\begin{displaymath}
	U=\mathrm{det}D\cdot (D^{-1})^{T},
	\end{displaymath}
thus
\begin{displaymath}
	\textbf{\underline{3}}_{a}=(\mathrm{det}\textbf{\underline{3}})\otimes (\textbf{\underline{3}}^{-1})^T.
	\end{displaymath}
Since $(\textbf{\underline{3}}^{-1})^T$ is equivalent to $\textbf{\underline{3}}^{\ast}$ for representations of finite groups we can also write
	\begin{displaymath}
	\textbf{\underline{3}}_{a}=(\mathrm{det}\textbf{\underline{3}})\otimes \textbf{\underline{3}}^{\ast}.
	\end{displaymath}
\end{proof}
\hspace{0mm}\\
For the case of $\textbf{\underline{3}}\otimes \textbf{\underline{3}}^{\ast}$ we use the following lemma:

\begin{sublemma}\label{CGdecompSL}
Let $D$ be an $n$-dimensional unitary irreducible representation of a finite group $G$. Then
	\begin{displaymath}
	\frac{1}{\sqrt{n}}e_{j}\otimes e_{j}
	\end{displaymath}
spans a subspace of $\mathbb{C}^{n}\otimes \mathbb{C}^{n}$ which is invariant under $D\otimes D^{\ast}$. ($D$ shall be unitary with respect to the orthonormal basis $\{e_{j}\}_{j}$ of $\mathbb{C}^{n}$.)
\end{sublemma}

\begin{proof}
Let $De_{k}=[D]_{mk}e_{m}$.
	\begin{displaymath}
	\begin{split}
	D\otimes D^{\ast}(\frac{1}{\sqrt{n}}e_{j}\otimes e_{j})&=\frac{1}{\sqrt{n}}[D]_{kj}[D^{\ast}]_{mj}e_{k}\otimes e_{m}=\\
	&=\frac{1}{\sqrt{n}}\underbrace{([D][D]^{\dagger})_{km}}_{\delta_{km}}e_{k}\otimes e_{m}=\frac{1}{\sqrt{n}}e_{k}\otimes e_{k}.
	\end{split}
	\end{displaymath}
$\Rightarrow$ $\frac{1}{\sqrt{n}}(e_{j}\otimes e_{j})$ is invariant under $D\otimes D^{\ast}$.
\end{proof}
\hspace{0mm}\\
Using sublemma \ref{CGdecompSL} on our problem here, we find that
	\begin{displaymath}
	u=\frac{1}{\sqrt{3}}(e_{11}+ e_{22}+ e_{33})
	\end{displaymath}
spans a 1-dimensional invariant subspace of $V_{\textbf{\underline{3}}}\otimes V_{\textbf{\underline{3}}^{\ast}}$. The basis vectors of the remaining subspace $V_{\textbf{\underline{8}}}$ can be obtained by extending $u$ to an orthonormal basis of $\mathbb{C}^{9}$. So we have found:

\begin{lemma}\label{CGdecompL2}
Let $\textbf{\underline{3}}$ be a $3$-dimensional irreducible representation of a finite group $G$, then $V_{\textbf{\underline{3}}}\otimes V_{\textbf{\underline{3}}^{\ast}}$ has a 1-dimensional invariant subspace $V_{\textbf{\underline{1}}}$ and an 8-dimensional invariant subspace $V_{\textbf{\underline{8}}}$. If the corresponding representations are irreducible we have the Clebsch-Gordan decomposition
	\begin{displaymath}
	 \textbf{\underline{3}}\otimes \textbf{\underline{3}}^{\ast}=\textbf{\underline{1}}\oplus \textbf{\underline{8}}.
	\end{displaymath}
\end{lemma}

\section{The finite subgroups of $SU(2)$ and $SO(3)$}
At the beginning of our considerations about the finite subgroups of $SU(2)$ we should list some properties of $SO(3)$ and its universal covering group $SU(2)$. These important relations and formulae may be well known to any reader familiar with quantum theory, therefore they are considered in appendix \ref{lieappendix}. So we will start with the definition of the so called \textquotedblleft double covers\textquotedblright\hspace{1mm} $\tilde{G}$ of $G\subset SO(3)$.

\begin{prop}\label{PSU33}
	Let
	\begin{gather*}
	\phi: SU(2)\rightarrow SO(3)\\
	e^{-i\alpha\vec{n}\cdot\frac{\vec{\sigma}}{2}}\mapsto e^{-i\alpha\vec{n}\cdot\vec{T}}
	\end{gather*}
	be a group homomorphism of $SU(2)$ on $SO(3)$ ($\alpha\in \mathbb{R}, \vec{n}\in\mathbb{R}^{3}, \|\vec{n}\|=1, (T^{i})^{jk}=\frac{1}{i}\varepsilon_{ijk}$). If $G$ is a finite subgroup of $SO(3)$, $\mathrm{ord}(G)=g$, then there exists a subgroup $\tilde{G}$ of $SU(2)$, $\mathrm{ord}(\tilde{G})=2g$ s.t.
	\begin{displaymath}
		\phi(\tilde{G})=G.
	\end{displaymath}
	We call $\tilde{G}$ the \textit{double cover} of $G$.
\end{prop}

\begin{proof}%verlinken mit appendix bzgl SU(2)-SO(3)
$\phi$ is a surjective group homomorphism that is \textquotedblleft2 to 1\textquotedblright\hspace{1mm} (see appendix \ref{lieappendix}), therefore the proposition is true.
\end{proof}

\begin{prop}\label{PSU34}
Let $\tilde{C}$ be a conjugate class of $\tilde{G}$, then $\phi(\tilde{C})$ is a conjugate class of $G$.
\end{prop}

\begin{proof}
Every conjugate class can be constructed from a particular element of the group. Let $\tilde{C}$ be the conjugate class containing the element $\tilde{a}$.
	\begin{displaymath}
	\tilde{C}=\{\tilde{u}\tilde{a}\tilde{u}^{-1}\}_{\tilde{u}\in \tilde{G}}\Rightarrow \phi(\tilde{C})=\{uau^{-1}\}_{u\in G},
	\end{displaymath}
therefore $\phi(\tilde{C})$ is the conjugate class of $G$ containing $a=\phi(\tilde{a})$.
\end{proof}

\begin{cor}\label{CSU35}
Let $\tilde{a}\in \tilde{G}$ s.t. $\phi(\tilde{a})=a\in G$, and let $C_{a}$ be the conjugate class of $G$ that contains $a$. Then to every conjugate class $C_{a}$ of $G$, there exist the conjugate classes $C_{\tilde{a}}$ and $C_{-\tilde{a}}$ of $\tilde{G}$.
\end{cor}

\begin{proof}
This follows from proposition \ref{PSU34} and the fact that $\phi(\pm \tilde{a})=a.$
\end{proof}

\begin{theorem}\label{TSU36}
Let $D$ be an irreducible representation of a finite subgroup $G$ of $SO(3)$. Then $\tilde{D}=D\circ \phi$ is an irreducible representation of $\tilde{G}$.
\end{theorem}

\begin{proof}
$\tilde{G}$ is homomorphic to $G$, and $G$ is homomorphic to $D(G)$, therefore $\tilde{G}$ is homomorphic to $D(G)=D\circ\phi(\tilde{G})$, thus $D\circ\phi$ is a representation of $\tilde{G}$.
\\
To show the irreducibility we use proposition \ref{PA58}
: $D$ is irreducible if and only if $(\chi_{D},\chi_{D})=\frac{1}{\mathrm{ord}(G)}\sum_{a\in G}\chi_{D}(a)^{\ast}\chi_{D}(a)=1.$
\medskip
\\
We will now show $(\chi_{D},\chi_{D})_{G}=(\chi_{\tilde{D}},\chi_{\tilde{D}})_{\tilde{G}}$, which proves the theorem.
\medskip
\\
Using $\chi_{\tilde{D}}(\tilde{a})=\chi_{\tilde{D}}(-\tilde{a})=\chi_{D}(\phi(\pm\tilde{a}))$ we find:
	\begin{displaymath}
	\begin{split}
	(\chi_{\tilde{D}},\chi_{\tilde{D}})_{\tilde{G}}
	&=
	\frac{1}{2\mathrm{ord}(G)}\sum_{\tilde{a}\in\tilde{G}}\chi_{\tilde{D}}^{\ast}(\tilde{a})\chi_{\tilde{D}}(\tilde{a})=
	\\&=
	\frac{1}{2\mathrm{ord}(G)}2\sum_{a\in G}\chi_{D}^{\ast}(a)\chi_{D}(a)=(\chi_{D},\chi_{D})_{G}.
	\end{split}
	\end{displaymath}
\end{proof}

\subsection{The tetrahedral group $T$}\label{subsectionA4}
The tetrahedral group $T$ \cite{hamermesh,rajasekaran,he,altarelli1,ma,frampton2,altarelli2} is the smallest non-Abelian finite subgroup of $SO(3)$. It belongs to the polyhedral groups, which are groups of rotations that map the Platonic solids onto themselves. Table \ref{Platonic} lists the polyhedral groups and the subgroups of $S_{n}$ isomorphic to them. (The table is taken from \cite{ma}.)

\begin{table}[h]
\begin{center}
\renewcommand{\arraystretch}{1.4}
\begin{tabular}{|lcc|}
\firsthline
	Platonic solids & Group & Isomorphic subgroup of $S_{n}$ \\
\hline
	Tetrahedron  & $T$ & $A_{4}$\\
	Octahedron, cube & $O$ & $S_{4}$\\
	Dodecahedron, icosahedron  & $I$ & $A_{5}$\\
\lasthline
\end{tabular}
\caption[The Platonic solids and their corresponding rotation symmetry groups.]{The Platonic solids and their corresponding rotation symmetry groups. \cite{ma}}
\label{Platonic}
\end{center}
\end{table}
\hspace{0mm}\\
$T$ is isomorphic to $A_{4}$, which is the set of even permutations (=product of an even number of transpositions) of four objects. $A_{4}$ has two generators, for example one can take \cite{altarelli1}
	\begin{displaymath}
	S=(1\enspace 4)(2\enspace 3)\quad T=(1\enspace 2\enspace 3).
	\end{displaymath}
(For explanations about definitions and notations of permutations see section \ref{A1}.)
Forming products of $S$ and $T$ one finds
	\begin{displaymath}
	A_{4}=\{E,T,S,T^{2},ST,TS,ST^{2},T^{2}S,T^{2}ST,TST^{2},TST,STS\}
	\end{displaymath}
and the conjugate classes
	\begin{displaymath}
	\begin{split}
	& C_{E}=\{E\},\\
	& C_{S}=\{S,TST^{2},T^{2}ST\},\\
	& C_{T}=\{T,ST,TS,STS\},\\
	& C_{T^{2}}=\{T^{2},TST,T^{2}S,ST^{2}\}.
	\end{split}
	\end{displaymath}
$E$ always denotes the identity element. Since there are four conjugate classes there are four non-equivalent irreducible representations (theorem \ref{TA74}).
The dimensions of these representations must fulfil (theorem \ref{TA65})
	\begin{displaymath}
	n_{1}^{2}+n_{2}^{2}+n_{3}^{2}+n_{4}^{2}=\mathrm{ord}(A_{4})=12.
	\end{displaymath}
The only solution (up to label exchange) is
	\begin{displaymath}
	n_{1}=n_{2}=n_{3}=1, n_{4}=3.
	\end{displaymath}
We will now construct a four-dimensional representation of $A_{4}$, which will be reducible to a one-dimensional and a three-dimensional irreducible representation.
	\begin{displaymath}
	S=(1\enspace 4)(2\enspace 3)=\left(\begin{matrix}
	 1 & 2 & 3 & 4 \\
	 4 & 3 & 2 & 1                    \end{matrix}\right)\Rightarrow
	\hat{S}:=\left(\begin{matrix}
		 0 & 0 & 0 & 1 \\
		 0 & 0 & 1 & 0 \\
		 0 & 1 & 0 & 0 \\
		 1 & 0 & 0 & 0
	        \end{matrix}\right)
	\end{displaymath}
Since   \begin{displaymath}
	\hat{S}\left(\begin{matrix}
		 x_{1} \\
		 x_{2} \\
		 x_{3} \\
		 x_{4}
	       \end{matrix}\right)
	=\left(\begin{matrix}
		 x_{4} \\
		 x_{3} \\
		 x_{2} \\
		 x_{1}
	       \end{matrix}\right)
	=\left(\begin{matrix}
		 x_{S(1)} \\
		 x_{S(2)} \\
		 x_{S(3)} \\
		 x_{S(4)}
	       \end{matrix}\right),
	\end{displaymath}
$\hat{S}$ is a four-dimensional matrix representation of $S\in A_{4}$. Similarly we get
	\begin{displaymath}
	\hat{T}:=\left(\begin{matrix}
		 0 & 1 & 0 & 0 \\
		 0 & 0 & 1 & 0 \\
		 1 & 0 & 0 & 0 \\
		 0 & 0 & 0 & 1
	        \end{matrix}\right).
	\end{displaymath}
$\hat{S}$ and $\hat{T}$ have a common eigenvector $u_{1}$ to the eigenvalue 1. We construct an orthonormal basis consisting of $u_{1}$ and other eigenvectors of $\hat{S}$:
	\begin{displaymath}
	u_{1}=\frac{1}{2}\left(\begin{matrix}
		 1 \\
		 1 \\
		 1 \\
		 1
	       \end{matrix}\right),\enspace
	u_{2}=\frac{1}{2}\left(\begin{matrix}
		 1 \\
		 -1 \\
		 -1 \\
		 1
	       \end{matrix}\right),
	u_{3}=\frac{1}{2}\left(\begin{matrix}
		 1 \\
		 1 \\
		 -1 \\
		 -1
	       \end{matrix}\right),\enspace
	u_{4}=\frac{1}{2}\left(\begin{matrix}
		 1 \\
		 -1 \\
		 1 \\
		 -1
	       \end{matrix}\right).
	\end{displaymath}
$U:=\left(\begin{matrix}
		 u_{1} & u_{2} & u_{3} & u_{4}
	        \end{matrix}\right),$
	\begin{displaymath}
	\hat{S}':=U^{T}\hat{S}U=
	\left(\begin{matrix}
		 1 & 0 & 0 & 0 \\
		 0 & 1 & 0 & 0 \\
		 0 & 0 & -1 & 0 \\
		 0 & 0 & 0 & -1
	        \end{matrix}\right)
	,\enspace \hat{T}':=U^{T}\hat{T}U=\left(\begin{matrix}
		 1 & 0 & 0 & 0 \\
		 0 & 0 & 0 & -1 \\
		 0 & -1 & 0 & 0 \\
		 0 & 0 & 1 & 0
	        \end{matrix}\right).
	\end{displaymath}
We have found two irreducible representations (irreducibility can be checked by calculating characters and testing whether $(\chi_{D},\chi_{D})=1.$):
	\begin{equation}\label{A4rep}
	\textbf{\underline{1}}: S\mapsto 1,\enspace T\mapsto 1; \quad\quad \textbf{\underline{3}}: S\mapsto \left(\begin{matrix}
	 1 & 0 & 0 \\
	 0 & -1 & 0 \\
	 0 & 0 & -1
	\end{matrix}\right),\enspace
	T\mapsto
	\left(\begin{matrix}
	 0 & 0 & -1 \\
	 -1 & 0 & 0 \\
	 0 & 1 & 0
	\end{matrix}\right)
	\end{equation}
Using the fact that $S^{2}=T^{3}=E$ we can easily guess the other one-dimensional non-equivalent representations:
	\begin{displaymath}
	\textbf{\underline{1}}': S\mapsto 1,\enspace T\mapsto \omega; \quad\quad \textbf{\underline{1}}'': S\mapsto 1,\enspace T\mapsto \omega^{2},
	\end{displaymath}
where $\omega=e^{\frac{2\pi i}{3}}$. One can now also obtain the character table by taking traces. (The numbers in brackets denote the numbers of elements in the conjugate classes. For example $C_{S}(3)$ means that the conjugate class containing $S$ has three elements.)

\begin{table}
\begin{center}
\renewcommand{\arraystretch}{1.4}
\begin{tabular}{|l|cccc|}
\firsthline
	$A_{4}$ & $C_{E}(1)$ & $C_{S}(3)$ & $C_{T}(4)$ & $C_{T^{2}}(4)$\\
\hline
	\textbf{\underline{1}} & $1$ & $1$ & $1$ & $1$\\
	\textbf{\underline{1}}' & $1$ & $1$ & $\omega$ & $\omega^{2}$\\
	\textbf{\underline{1}}'' & $1$ & $1$ & $\omega^{2}$ & $\omega$\\
	\textbf{\underline{3}} & $3$ & $-1$ & $0$ & $0$\\
\lasthline
\end{tabular}
\caption{The character table of $A_{4}$.}
\label{A4charactertable}
\end{center}
\end{table}
\hspace{0mm}\\
We will now use the character table to calculate the tensor product $\textbf{\underline{3}}\otimes\textbf{\underline{3}}$ for $A_{4}$.

\begin{lemma}\label{LSU37}
Let $D^{(\alpha)}$ be an irreducible representation of a finite group $G$. The number of times $b_{\alpha}$ the irreducible representation $D^{(\alpha)}$ is containend in a tensor product $D_{a}\otimes D_{b}$ is given by
	\begin{displaymath}
	b_{\alpha}=(\chi_{D^{(\alpha)}},\chi_{a}\cdot \chi_{b}).
	\end{displaymath}
\end{lemma}

\begin{proof}
From proposition \ref{PA59a}
we know $b_{\alpha}=(\chi_{D^{(\alpha)}},\chi_{D_{a}\otimes D_{b}})$, and from proposition \ref{PA85}
we know $\chi_{D_{a}\otimes D_{b}}=\chi_{D_{a}}\cdot\chi_{D_{b}}$.
\end{proof}
\hspace{0mm}\\
Using sublemma \ref{LSU37} and the character table \ref{A4charactertable} we find
	\begin{displaymath}
	(\chi_{\textbf{\underline{1}}},\chi_{\textbf{\underline{3}}}^{2})=(\chi_{\textbf{\underline{1}}'},\chi_{\textbf{\underline{3}}}^{2})=(\chi_{\textbf{\underline{1}}''},\chi_{\textbf{\underline{3}}}^{2})=1,\quad (\chi_{\textbf{\underline{3}}},\chi_{\textbf{\underline{3}}}^{2})=2,
	\end{displaymath}
thus
	\begin{displaymath}
	\textbf{\underline{3}}\otimes\textbf{\underline{3}}=\textbf{\underline{1}}\oplus \textbf{\underline{1}}'\oplus\textbf{\underline{1}}''\oplus\textbf{\underline{3}}\oplus\textbf{\underline{3}}.
	\end{displaymath}
From equation (\ref{A4rep}) we see that all representations of $A_{4}$ are real, especially $\textbf{\underline{3}}=\textbf{\underline{3}}^{\ast}$. We can now use lemmata \ref{CGdecompL1} and \ref{CGdecompL2} to find that
	\begin{displaymath}
	\textbf{\underline{3}}\otimes\textbf{\underline{3}}=\textbf{\underline{3}}\otimes \textbf{\underline{3}}^{\ast}=\textbf{\underline{1}}\oplus \textbf{\underline{8}}=\textbf{\underline{3}}_{a}\oplus\textbf{\underline{6}}_{s}
	\end{displaymath}
$\textbf{\underline{1}}$ and $\textbf{\underline{3}}_{a}=\textbf{\underline{3}}^{\ast}=\textbf{\underline{3}}$ are irreducible, and from section \ref{CGCdecosubsect} we already know basis vectors of the invariant subspaces of $V_{\textbf{\underline{3}}}\otimes V_{\textbf{\underline{3}}}$ corresponding to them. Normalizing them we find (using the abbreviation $e_{ij}=e_{i}\otimes e_{j}$):
	\begin{displaymath}
	\begin{split}
	\textbf{\underline{1}}:\enspace & u=\frac{1}{\sqrt{3}}(e_{11}+e_{22}+e_{33}),\\
	\textbf{\underline{3}}_{a}:\enspace & u_{1}=\frac{1}{\sqrt{2}}(e_{23}-e_{32}),\\
	& u_{2}=\frac{1}{\sqrt{2}}(-e_{13}+e_{31}),\\
	& u_{3}=\frac{1}{\sqrt{2}}(e_{12}-e_{21}).
	\end{split}
	\end{displaymath}
The basis vectors of the invariant subspaces $V_{\textbf{\underline{1}}'}$ and $V_{\textbf{\underline{1}}''}$ can be obtained by calculating the Clebsch-Gordan coefficients using the algorithm developed in chapter \ref{chapterclebsch}. An example for a \textit{Mathematica 6}-implementation of this algorithm (used on the group $\Sigma(36\phi)$) can be found in appendix \ref{appendixD}. One finds
	\begin{displaymath}
	\begin{split}
	\textbf{\underline{1}}':\enspace & u=\frac{1}{\sqrt{3}}(\omega^{2}e_{11}+\omega e_{22}+e_{33}),\\
	\textbf{\underline{1}}'':\enspace & u=\frac{1}{\sqrt{3}}(\omega e_{11}+\omega^{2} e_{22}+e_{33}).
	\end{split}
	\end{displaymath}
Extending the found basis of $V_{\textbf{\underline{1}}}\oplus V_{\textbf{\underline{1}}'}\oplus V_{\textbf{\underline{1}}''}$ to a basis of $V_{\textbf{\underline{6}}_{s}}=V_{s}$ ($\rightarrow$ lemma \ref{CGdecompL1}) one finds the remaining basis vectors of $V_{\textbf{\underline{3}}}\otimes V_{\textbf{\underline{3}}}$:
	\begin{displaymath}
	\begin{split}
	\textbf{\underline{3}}:\enspace & u_{1}=\frac{1}{\sqrt{2}}(e_{23}+e_{32}),\\
	& u_{2}=\frac{1}{\sqrt{2}}(e_{13}+e_{31}),\\
	& u_{3}=\frac{1}{\sqrt{2}}(e_{12}+e_{21}).
	\end{split}
	\end{displaymath}
We have now obtained bases of all invariant subspaces of $V_{\textbf{\underline{3}}}\otimes V_{\textbf{\underline{3}}}$. The results are shown in table \ref{A4CGC}. As already mentioned at the end of chapter \ref{chapterclebsch} we will always list Clebsch-Gordan coefficients by showing the corresponding basis vectors of the invariant subspaces:
\\
Let $e_{ij}:=e_{i}^{a}\otimes e_{i}^{b}$ be basis vectors of $V_{a}\otimes V_{b}$, and let $D_{a}\otimes D_{b}=D^{\lambda}\oplus...$ . We will give the Clebsch-Gordan coefficients for this decomposition by giving the basis vectors of $V^{\lambda}$:
	\begin{displaymath}
	u_{D^{\lambda}}^{D_{a}\otimes D_{b}}(i)=C^{\lambda}_{ijk}e_{jk},\quad i=1,...,n_{\lambda}.
	\end{displaymath}
(Though it is incorrect we will use the term \textquotedblleft Clebsch-Gordan coefficients\textquotedblright\hspace{1mm} when we label the lists of these basis vectors.) At the end of this chapter we will develop a list of different \textquotedblleft types of Clebsch-Gordan coefficients\textquotedblright, which are labeled by Roman numerals. We will also indicate these types in the tables of Clebsch-Gordan coefficients (in the column named \textquotedblleft CGC\textquotedblright). For the exact meaning of these types see section \ref{summaryCGCsection} (especially table \ref{SU3CGClist}).

\begin{table}[t]
\begin{center}
\renewcommand{\arraystretch}{1.4}
\begin{tabular}{|l|l|l|}
\firsthline
	$A_{4}$ & $\textbf{\underline{3}}\otimes\textbf{\underline{3}}$ & CGC\\
\hline
	$\textbf{\underline{1}}$
	& $u^{\textbf{\underline{3}}\otimes\textbf{\underline{3}}}_{\textbf{\underline{1}}}=\frac{1}{\sqrt{3}}e_{11}+\frac{1}{\sqrt{3}}e_{22}+\frac{1}{\sqrt{3}}e_{33}$ & \textbf{Ia}\\

	$\textbf{\underline{1}'}$
	&	
	$u^{\textbf{\underline{3}}\otimes\textbf{\underline{3}}}_{\textbf{\underline{1}'}}=\frac{\omega^{2}}{\sqrt{3}}e_{11}+\frac{\omega}{\sqrt{3}}e_{22}+\frac{1}{\sqrt{3}}e_{33}$ & \textbf{Ib}\\

	$\textbf{\underline{1}''}$
	&	
	$u^{\textbf{\underline{3}}\otimes\textbf{\underline{3}}}_{\textbf{\underline{1}''}}=\frac{\omega}{\sqrt{3}}e_{11}+\frac{\omega^{2}}{\sqrt{3}}e_{22}+\frac{1}{\sqrt{3}}e_{33}$ & \textbf{Ic}\\

	$\textbf{\underline{3}}_{a}$
	&	
	$u^{\textbf{\underline{3}}\otimes\textbf{\underline{3}}}_{\textbf{\underline{3}}}(1)=\frac{1}{\sqrt{2}}e_{23}-\frac{1}{\sqrt{2}}e_{32}$ & \textbf{IIIa$_{(1,-1)}$}\\

	&	
	$u^{\textbf{\underline{3}}\otimes\textbf{\underline{3}}}_{\textbf{\underline{3}}}(2)=-\frac{1}{\sqrt{2}}e_{13}+\frac{1}{\sqrt{2}}e_{31}$& \\

	&	
	$u^{\textbf{\underline{3}}\otimes\textbf{\underline{3}}}_{\textbf{\underline{3}}}(3)=\frac{1}{\sqrt{2}}e_{12}-\frac{1}{\sqrt{2}}e_{21}$& \\

	$\textbf{\underline{3}}$
	&	
	$u^{\textbf{\underline{3}}\otimes\textbf{\underline{3}}}_{\textbf{\underline{3}}}(1')=\frac{1}{\sqrt{2}}e_{23}+\frac{1}{\sqrt{2}}e_{32}$ & \textbf{IIIa$_{(1,1)}$}\\

	&	
	$u^{\textbf{\underline{3}}\otimes\textbf{\underline{3}}}_{\textbf{\underline{3}}}(2')=\frac{1}{\sqrt{2}}e_{13}+\frac{1}{\sqrt{2}}e_{31}$& \\

	&	
	$u^{\textbf{\underline{3}}\otimes\textbf{\underline{3}}}_{\textbf{\underline{3}}}(3')=\frac{1}{\sqrt{2}}e_{12}+\frac{1}{\sqrt{2}}e_{21}$& \\
\lasthline
\end{tabular}
\caption{Clebsch-Gordan coefficients for $\textbf{\underline{3}}\otimes\textbf{\underline{3}}$ of $A_{4}$.}
\label{A4CGC}
\end{center}
\end{table}

\subsection{The double cover of the tetrahedral group $\tilde{T}$}

The double cover of $T$ is isomorphic to $\tilde{A}_{4}$. To construct $\tilde{A}_{4}$ we have to find its generators. For this purpose we use theorem \ref{TSU2SU34}
and proposition \ref{PSU2SU32}
. In the three-dimensional representation of $A_{4}$ we had
	\begin{displaymath}
	S\mapsto \left(\begin{matrix}
	 1 & 0 & 0 \\
	 0 & -1 & 0 \\
	 0 & 0 & -1
	\end{matrix}\right),\enspace
	T\mapsto
	\left(\begin{matrix}
	 0 & 0 & -1 \\
	 -1 & 0 & 0 \\
	 0 & 1 & 0
	\end{matrix}\right).
	\end{displaymath}
We will now construct a corresponding representation of $\tilde{A}_{4}$. Let $R(\alpha,\vec{n}):=e^{-i\alpha\vec{n}\cdot\vec{T}}$ and $U(\alpha,\vec{n}):=e^{-i\alpha\vec{n}\cdot\frac{\vec{\sigma}}{2}}$.
	\begin{displaymath}
	\left(\begin{matrix}
	 1 & 0 & 0 \\
	 0 & -1 & 0 \\
	 0 & 0 & -1
	\end{matrix}\right)=R(\alpha_{S},\vec{n}_{S})\Rightarrow \alpha_{S}=\pi, \vec{n}_{S}=
	\left(\begin{matrix}
	 1 \\
	 0 \\
	 0
	\end{matrix}\right).
	\end{displaymath}
Thus a generator of $\tilde{A}_{4}$ is
	\begin{displaymath}
	S:=U(\alpha_{S},\vec{n}_{S})=
	\left(\begin{matrix}
	 0 & -i \\
	 -i & 0
	\end{matrix}\right).
	\end{displaymath}
For $T$ we find $\alpha_{T}=\frac{2\pi}{3}, \vec{n}_{T}=\frac{1}{\sqrt{3}}
	\left(\begin{matrix}
	 1 \\
	 -1 \\
	 -1
	\end{matrix}\right)$
	\begin{displaymath}
	\Rightarrow T:=U(\alpha_{T},\vec{n}_{T})=
	\frac{1}{2}
	\left(\begin{matrix}
	 1+i & 1-i \\
	 -(1+i) & 1-i
	\end{matrix}\right).
	\end{displaymath}
Knowing the generators, noticing $S^{2}=T^{3}=-\mathbbm{1}_{2}$, we can write down all 24 elements of $\tilde{A}_{4}$.
	\begin{displaymath}
	\begin{split}
	\tilde{A}_{4}=\{& \pm E,\pm T,\pm S,\pm T^{2},\pm ST,\pm TS,\pm ST^{2},\\
	&\pm T^{2}S,\pm T^{2}ST,\pm TST^{2},\pm TST,\pm STS\}
	\end{split}
	\end{displaymath}
The conjugate classes are
	\begin{displaymath}
	\begin{split}
		& C_{E}=\{E\},\\
		& C_{-E}=\{-E\},\\
		& C_{S}=C_{-S}=\{\pm S,\pm TST^{2},\pm T^{2}ST\},\\
		& C_{T}=\{T,-ST,-TS,-STS\},\\
		& C_{-T}=\{-T,ST,TS,STS\},\\
		& C_{T^{2}}=\{T^{2},TST,T^{2}S,ST^{2}\},\\
		& C_{-T^{2}}=\{-T^{2},-TST,-T^{2}S,-ST^{2}\}.
	\end{split}
	\end{displaymath}
There are 7 conjugate classes, thus there are 7 non-equivalent irreducible representations. The only solution (up to label exchange) of
	\begin{displaymath}
	\sum_{k=1}^{7}n_{k}^{2}=24
	\end{displaymath}
is
	\begin{displaymath}
	n_{1}=n_{2}=n_{3}=1, n_{4}=n_{5}=n_{6}=2, n_{7}=3,
	\end{displaymath}
which is not surprising since the irreducible representations of $A_{4}$ remain, and the new representation $\textbf{\underline{2}}$ is also irreducible. Thus the irreducible representations of $\tilde{A}_{4}$ are:
	\begin{displaymath}
	\begin{split}
	& \textbf{\underline{1}}: S\mapsto 1,\enspace T\mapsto1,\\
	& \textbf{\underline{1}}': S\mapsto 1,\enspace T\mapsto\omega,\\
	& \textbf{\underline{1}}'': S\mapsto 1,\enspace T\mapsto\omega^{2},\\
	& \textbf{\underline{2}}: S\mapsto S,\enspace T\mapsto T,\\
	& \textbf{\underline{2}}': S\mapsto S,\enspace T\mapsto\omega T,\\
	& \textbf{\underline{2}}'': S\mapsto S,\enspace T\mapsto\omega^{2} T,\\
	& \textbf{\underline{3}}: S\mapsto \left(\begin{matrix}
	 1 & 0 & 0 \\
	 0 & -1 & 0 \\
	 0 & 0 & -1
	\end{matrix}\right),\enspace T\mapsto
	\left(\begin{matrix}
	 0 & 0 & -1 \\
	 -1 & 0 & 0 \\
	 0 & 1 & 0
	\end{matrix}\right).
	\end{split}
	\end{displaymath}
\begin{table}
\begin{center}
\renewcommand{\arraystretch}{1.4}
\begin{tabular}{|l|ccccccc|}
\firsthline
	$\tilde{A}_{4}$ & $C_{E}(1)$ & $C_{-E}(1)$ & $C_{S}(6)$ & $C_{T}(4)$ & $C_{-T}(4)$ & $C_{T^{2}}(4)$ & $C_{-T^{2}}(4)$\\
\hline
	\textbf{\underline{1}} & $1$ & $1$ & $1$ & $1$ & $1$ & $1$ & $1$\\
	\textbf{\underline{1}}' & $1$ & $1$ & $1$ & $\omega$ & $\omega$ & $\omega^{2}$ & $\omega^{2}$\\
	\textbf{\underline{1}}'' & $1$ & $1$ & $1$ & $\omega^{2}$ & $\omega^{2}$ & $\omega$ & $\omega$\\
	\textbf{\underline{2}} & $2$ & $-2$ & $0$ & $1$ & $-1$ & $-1$ & $1$\\
	\textbf{\underline{2}}' & $2$ & $-2$ & $0$ & $\omega$ & $-\omega$ & $-\omega^{2}$ & $\omega^{2}$\\
	\textbf{\underline{2}}'' & $2$ & $-2$ & $0$ & $\omega^{2}$ & $-\omega^{2}$ & $-\omega$ & $\omega$\\
	\textbf{\underline{3}} & $3$ & $3$ & $-1$ & $0$ & $0$ & $0$ & $0$\\
\lasthline
\end{tabular}
\caption{The character table of $\tilde{A}_{4}$.}
\label{tildeA4charactertable}
\end{center}
\end{table}
\hspace{0mm}\\
Using the character table \ref{tildeA4charactertable} one can calculate the tensor products. $\textbf{\underline{3}}\otimes\textbf{\underline{3}}$ remains completely unchanged, and since all representations occurring in $\textbf{\underline{3}}\otimes\textbf{\underline{3}}$ are the same as in the case of $A_{4}$, also the Clebsch-Gordan coefficients are the same.

\subsection{The octahedral group $O$}

$O$ is isomorphic to $S_{4}$, which is treated in \cite{bazzocchi,hagedorn}. We repeat here the conjugate classes, generators and representations given in \cite{bazzocchi}, and calculate the character table and Clebsch-Gordan coefficients using this data.
\medskip
\\
The conjugate classes are
	\begin{displaymath}
	\begin{split}
	& C_{E}=\{E\},\\
	& C_{S^{2}}=\{S^{2},TS^{2}T^{2},S^{2}TS^{2}T^{2}\},\\
	& C_{T}=\{T,T^{2},S^{2}T,S^{2}T^{2},STST^{2},STS, S^{2}TS^{2},S^{3}TS\},\\
	& C_{ST^{2}}=\{ST^{2},T^{2}S,TST,TSTS^{2},STS^{2},S^{2}TS\},\\
	& C_{S}=\{S,TST^{2},ST,TS,S^{3},S^{3}T^{2}\},
	\end{split}
	\end{displaymath}
where $S$ and $T$ are generators of $S_{4}$ fulfilling $S^{4}=T^{3}=E, ST^{2}S=T$. The irreducible representations are
	\begin{displaymath}
	\begin{split}
	& \textbf{\underline{1}}: S\mapsto 1, T\mapsto 1,\\
	& \textbf{\underline{1}}': S\mapsto -1, T\mapsto 1,\\
	& \textbf{\underline{2}}: S\mapsto 
	\left(
	\begin{matrix}
	 -1 & 0 \\
	 0 & 1
	\end{matrix}
	\right),
	T\mapsto
	-\frac{1}{2}
	\left(
	\begin{matrix}
	 1 & \sqrt{3} \\
	 -\sqrt{3} & 1
	\end{matrix}
	\right),\\
	& \textbf{\underline{3}}: S\mapsto 
	\left(
	\begin{matrix}
	 -1 & 0 & 0\\
	 0 & 0 & -1\\
	 0 & 1 & 0
	\end{matrix}
	\right),
	T\mapsto
	\left(
	\begin{matrix}
	 0 & 0 & 1 \\
	 1 & 0 & 0 \\
	 0 & 1 & 0
	\end{matrix}
	\right),\\
	& \textbf{\underline{3}}': S\mapsto 
	\left(
	\begin{matrix}
	 1 & 0 & 0\\
	 0 & 0 & 1\\
	 0 & -1 & 0
	\end{matrix}
	\right),
	T\mapsto
	\left(
	\begin{matrix}
	 0 & 0 & 1 \\
	 1 & 0 & 0 \\
	 0 & 1 & 0
	\end{matrix}
	\right).
	\end{split}	
	\end{displaymath}
Knowing the non-equivalent irreducible representations and the conjugate classes we can construct the character table \ref{S4charactertable}.

\begin{table}
\begin{center}
\renewcommand{\arraystretch}{1.4}
\begin{tabular}{|l|ccccc|}
\firsthline
	$S_{4}$ & $C_{E}(1)$ & $C_{S}(6)$ & $C_{S^{2}}(3)$ & $C_{T}(8)$ & $C_{ST^{2}}(6)$\\
\hline
	\textbf{\underline{1}} & $1$ & $1$ & $1$ & $1$ & $1$\\
	\textbf{\underline{1}}' & $1$ & $-1$ & $1$ & $1$ & $-1$\\
	\textbf{\underline{2}} & $2$ & $0$ & $2$ & $-1$ & $0$\\
	\textbf{\underline{3}} & $3$ & $-1$ & $-1$ & $0$ & $1$\\
	\textbf{\underline{3}}' & $3$ & $1$ & $-1$ & $0$ & $-1$\\
\lasthline
\end{tabular}
\caption{The character table of $S_{4}$.}
\label{S4charactertable}
\end{center}
\end{table}
\hspace{0mm}\\
From the character table we can calculate:
	\begin{displaymath}
	\begin{split}
	& \textbf{\underline{3}}\otimes\textbf{\underline{3}}=\textbf{\underline{1}}\oplus\textbf{\underline{2}} \oplus \textbf{\underline{3}}\oplus \textbf{\underline{3}}',\\
	&
	\textbf{\underline{3}}\otimes\textbf{\underline{3}}'=\textbf{\underline{1}}'\oplus\textbf{\underline{2}} \oplus \textbf{\underline{3}}\oplus \textbf{\underline{3}}',\\
	&
	\textbf{\underline{3}}'\otimes\textbf{\underline{3}}'=\textbf{\underline{1}}\oplus\textbf{\underline{2}} \oplus \textbf{\underline{3}}\oplus \textbf{\underline{3}}'.
	\end{split}
	\end{displaymath}
Looking at the matrix representations of the generators $S$ and $T$ we find
	\begin{displaymath}
	\textbf{\underline{3}}=\textbf{\underline{3}}^{\ast},\quad\textbf{\underline{3}}'=\textbf{\underline{3}}'\hspace{0mm}^{\ast},\quad \textbf{\underline{3}}'=\textbf{\underline{1}}'\otimes\textbf{\underline{3}}.
	\end{displaymath}
Using lemma \ref{LSU313} we conclude that the Clebsch-Gordan coefficients for
	\begin{displaymath}
	\begin{split}
	&\textbf{\underline{3}}\otimes\textbf{\underline{3}}'=(\textbf{\underline{1}}\otimes\textbf{\underline{1}}')\otimes(\textbf{\underline{3}}\otimes\textbf{\underline{3}}),\\ &\textbf{\underline{3}}'\otimes\textbf{\underline{3}}=(\textbf{\underline{1}}'\otimes\textbf{\underline{1}})\otimes(\textbf{\underline{3}}\otimes\textbf{\underline{3}}),\\ &\textbf{\underline{3}}'\otimes\textbf{\underline{3}}'=(\textbf{\underline{1}}'\otimes\textbf{\underline{1}}')\otimes(\textbf{\underline{3}}\otimes\textbf{\underline{3}})
	\end{split}
	\end{displaymath}
are equal to those of $\textbf{\underline{3}}\otimes\textbf{\underline{3}}$. Therefore there is only one type of Clebsch-Gordan coefficients, which we can construct by investigating the Clebsch-Gordan decomposition of $\textbf{\underline{3}}\otimes\textbf{\underline{3}}$. As in the case of $A_{4}\simeq T$ we have
	\begin{displaymath}
	\textbf{\underline{3}}\otimes\textbf{\underline{3}}=\textbf{\underline{3}}\otimes \textbf{\underline{3}}^{\ast}=\textbf{\underline{1}}\oplus \textbf{\underline{8}}=\textbf{\underline{3}}_{a}\oplus\textbf{\underline{6}}_{s}.
	\end{displaymath}
$\textbf{\underline{1}}$ and $\textbf{\underline{3}}_{a}=(\mathrm{det}\textbf{\underline{3}})\otimes\textbf{\underline{3}}^{\ast}=\textbf{\underline{1}}'\otimes\textbf{\underline{3}}=\textbf{\underline{3}}'$ are irreducible, and we get the same basis vectors for the corresponding invariant subspaces as we already found for $\textbf{\underline{3}}\otimes \textbf{\underline{3}}$ of $A_{4}\simeq T$ ($\rightarrow$ table \ref{A4CGC}).
\\
Basis vectors of $V_{\textbf{\underline{2}}}$ can be obtained by using the algorithm described in chapter \ref{chapterclebsch}. One gets
	\begin{displaymath}
	\begin{split}
	\textbf{\underline{2}}:\enspace & u_{1}=\frac{1}{\sqrt{2}}(e_{22}-e_{33}),\\
	& u_{2}=\frac{1}{\sqrt{6}}(-2e_{11}+e_{22}+e_{33}).
	\end{split}
	\end{displaymath}
Basis vectors of $V_{\textbf{\underline{3}}}$ can be obtained by extending the found basis of $V_{\textbf{\underline{1}}}\oplus V_{\textbf{\underline{2}}}$ to a basis of $V_{\textbf{\underline{6}}_{s}}=V_{s}$:
	\begin{displaymath}
	\begin{split}
	\textbf{\underline{3}}:\enspace & u_{1}=\frac{1}{\sqrt{2}}(e_{23}+e_{32}),\\
	& u_{2}=\frac{1}{\sqrt{2}}(e_{13}+e_{31}),\\
	& u_{3}=\frac{1}{\sqrt{2}}(e_{12}+e_{21}).
	\end{split}
	\end{displaymath}
So we have found basis vectors of all invariant subspaces of $V_{\textbf{\underline{3}}}\otimes V_{\textbf{\underline{3}}}$ (and therefore the Clebsch-Gordan coefficients). We list them in table \ref{S4CGC}.

\begin{table}
\begin{center}
\renewcommand{\arraystretch}{1.4}
\begin{tabular}{|l|l|l|}
\firsthline
	$S_{4}$ & $\textbf{\underline{3}}\otimes\textbf{\underline{3}}$ & CGC\\
\hline
	$\textbf{\underline{1}}$
	& $u^{\textbf{\underline{3}}\otimes\textbf{\underline{3}}}_{\textbf{\underline{1}}}=\frac{1}{\sqrt{3}}e_{11}+\frac{1}{\sqrt{3}}e_{22}+\frac{1}{\sqrt{3}}e_{33}$ & \textbf{Ia}\\

	$\textbf{\underline{2}}$
	&	
	$u^{\textbf{\underline{3}}\otimes\textbf{\underline{3}}}_{\textbf{\underline{2}}}(1)=\frac{1}{\sqrt{2}}e_{22}-\frac{1}{\sqrt{2}}e_{33}$ & \textbf{IIa}\\

	&	
	$u^{\textbf{\underline{3}}\otimes\textbf{\underline{3}}}_{\textbf{\underline{2}}}(2)=-\frac{2}{\sqrt{6}}e_{11}+\frac{1}{\sqrt{6}}e_{22}+\frac{1}{\sqrt{6}}e_{33}$ & \\

	$\textbf{\underline{3}}_{a}$
	&	
	$u^{\textbf{\underline{3}}\otimes\textbf{\underline{3}}}_{\textbf{\underline{3}}_{a}}(1)=\frac{1}{\sqrt{2}}e_{23}-\frac{1}{\sqrt{2}}e_{32}$ & \textbf{IIIa$_{(1,-1)}$}\\

	&	
	$u^{\textbf{\underline{3}}\otimes\textbf{\underline{3}}}_{\textbf{\underline{3}}_{a}}(2)=-\frac{1}{\sqrt{2}}e_{13}+\frac{1}{\sqrt{2}}e_{31}$ & \\

	&	
	$u^{\textbf{\underline{3}}\otimes\textbf{\underline{3}}}_{\textbf{\underline{3}}_{a}}(3)=\frac{1}{\sqrt{2}}e_{12}-\frac{1}{\sqrt{2}}e_{21}$ & \\

	$\textbf{\underline{3}}$
	&	
	$u^{\textbf{\underline{3}}\otimes\textbf{\underline{3}}}_{\textbf{\underline{3}}}(1)=\frac{1}{\sqrt{2}}e_{23}+\frac{1}{\sqrt{2}}e_{32}$ & \textbf{IIIa$_{(1,1)}$}\\

	&	
	$u^{\textbf{\underline{3}}\otimes\textbf{\underline{3}}}_{\textbf{\underline{3}}}(2)=\frac{1}{\sqrt{2}}e_{13}+\frac{1}{\sqrt{2}}e_{31}$ & \\

	&	
	$u^{\textbf{\underline{3}}\otimes\textbf{\underline{3}}}_{\textbf{\underline{3}}}(3)=\frac{1}{\sqrt{2}}e_{12}+\frac{1}{\sqrt{2}}e_{21}$ & \\
\lasthline
\end{tabular}
\caption{Clebsch-Gordan coefficients for $\textbf{\underline{3}}\otimes\textbf{\underline{3}}$ of $S_{4}$.}
\label{S4CGC}
\end{center}
\end{table}

\subsection{The double cover of the octahedral group $\tilde{O}$}

We proceed as in the case of $\tilde{T}\simeq \tilde{A}_{4}$. $\tilde{O}$ is isomorphic to $\tilde{S}_{4}$. To find generators of $\tilde{S}_{4}$ we need the generators of a faithful representation $D(S_{4})\subset SO(3)$. $\textbf{\underline{3}}'$ is such a representation of $S_{4}$.
	\begin{displaymath}
	S_{4}:\quad \textbf{\underline{3}}': S\mapsto 
	\hat{S}=\left(
	\begin{matrix}
	 1 & 0 & 0\\
	 0 & 0 & 1\\
	 0 & -1 & 0
	\end{matrix}
	\right),
	T\mapsto
	\hat{T}=\left(
	\begin{matrix}
	 0 & 0 & 1 \\
	 1 & 0 & 0 \\
	 0 & 1 & 0
	\end{matrix}
	\right)
	\end{displaymath}
We express $\hat{S}$ and $\hat{T}$ as $R(\alpha,\vec{n})\in SO(3)$:
	\begin{displaymath}
	\hat{S}=R(\alpha_{S},\vec{n}_{S})\Rightarrow \alpha_{S}=\frac{\pi}{2},\enspace \vec{n}_{S}=\left(\begin{matrix}
	-1 \\
	 0 \\
	 0
	\end{matrix}\right),
	\end{displaymath}
	\begin{displaymath}
	\hat{T}=R(\alpha_{T},\vec{n}_{T})\Rightarrow \alpha_{T}=\frac{2\pi}{3},\enspace \vec{n}_{T}=\frac{1}{\sqrt{3}}\left(\begin{matrix}
	1 \\
	1 \\
	1
	\end{matrix}\right).
	\end{displaymath}
Thus we obtain generators of $\tilde{S}_{4}$ by
       \begin{displaymath}
       S=U(\alpha_{S},\vec{n}_{S})=
       \frac{1}{\sqrt{2}}\left(\begin{matrix}
			 1 & i \\
			 i & 1
	                  \end{matrix}\right),
	\quad
	T=U(\alpha_{T},\vec{n}_{T})=
       \frac{1}{2}\left(\begin{matrix}
			 1-i & -1-i \\
			 1-i & 1+i
	                  \end{matrix}\right).
	\end{displaymath}
The new generators fulfil $S^{4}=T^{3}=-\mathbbm{1}_{2}$. Using these generators we construct the conjugate classes of $\tilde{S}_{4}$:
	\begin{displaymath}
	\begin{split}
	& C_{E}=\{E\},\\
	& C_{-E}=\{-E\},\\
	& C_{S^{2}}=\{\pm S^{2},\pm TS^{2}T^{2},\pm S^{2}TS^{2}T^{2}\},\\
	& C_{T}=\{T,-T^{2},S^{2}T,S^{2}T^{2},-STST^{2},STS,-S^{2}TS^{2},-S^{3}TS\},\\
	&
	C_{-T}=\{-T,T^{2},-S^{2}T,-S^{2}T^{2},STST^{2},-STS,S^{2}TS^{2},S^{3}TS\},\\
	& C_{ST^{2}}=\{\pm ST^{2},\pm T^{2}S,\pm TST,\pm TSTS^{2},\pm STS^{2},\pm S^{2}TS\},\\
	& C_{S}=\{S,-TST^{2},ST,TS,-S^{3},S^{3}T^{2}\},\\
	& C_{-S}=\{-S,TST^{2},-ST,-TS,S^{3},-S^{3}T^{2}\}.
	\end{split}
	\end{displaymath}	
Compared to $S_{4}$ we now have three conjugate classes more, thus we also get three new irreducible representations. Their dimensions must fulfil
	\begin{displaymath}
	\sum_{k=1}^{3}n_{k}^{2}=24.
	\end{displaymath}
The only solution (up to label exchange) is
	\begin{displaymath}
	n_{1}=n_{2}=2,\enspace n_{3}=4,
	\end{displaymath}
thus there are no additional three-dimensional representations.

\subsection{The icosahedral group $I$}\label{icosahedral}

The icosahedral group $I$ is isomorphic to the permutation group $A_{5}$ \cite{everett,luhn1}, which is well presented and analysed in \cite{everett}. We will take the character table (table \ref{A5charactertable}), the generators for the irreducible representations and the tensor products from \cite{everett}.

\begin{table}
\begin{center}
\renewcommand{\arraystretch}{1.4}
\begin{tabular}{|l|ccccc|}
\firsthline
	$A_{5}$ & $C_{1}(1)$ & $C_{2}(12)$ & $C_{3}(12)$ & $C_{4}(20)$ & $C_{5}(15)$\\
\hline
	\textbf{\underline{1}} & $1$ & $1$ & $1$ & $1$ & $1$\\
	\textbf{\underline{3}} & $3$ & $\phi$ & $1-\phi$ & $0$ & $-1$\\
	\textbf{\underline{3}}' & $3$ & $1-\phi$ & $\phi$ & $0$ & $-1$\\
	\textbf{\underline{4}} & $4$ & $-1$ & $-1$ & $1$ & $0$\\
	\textbf{\underline{5}} & $5$ & $0$ & $0$ & $-1$ & $1$\\
\lasthline
\end{tabular}
\caption[The character table of $A_{5}$.]{The character table of $A_{5}$ \cite{everett}. $\phi=\frac{1}{2}(1+\sqrt{5})$}
\label{A5charactertable}
\end{center}
\end{table}
\hspace{0mm}\\
The generators $S,T$ for the non-equivalent irreducible representations given in \cite{everett} are ($\phi=\frac{1}{2}(1+\sqrt{5})$):
	\begin{displaymath}
	\begin{split}
	& \textbf{\underline{1}}: S\mapsto 1, T\mapsto 1,\\
	& \textbf{\underline{3}}: S\mapsto 
	\frac{1}{2}\left(
	\begin{matrix}
	 -1 & \phi & \frac{1}{\phi}\\
	 \phi & \frac{1}{\phi} & 1\\
	 \frac{1}{\phi} & 1 & -\phi
	\end{matrix}
	\right),
	T\mapsto
	\frac{1}{2}\left(
	\begin{matrix}
	 1 & \phi & \frac{1}{\phi} \\
	 -\phi & \frac{1}{\phi} & 1 \\
	 \frac{1}{\phi} & -1 & \phi
	\end{matrix}
	\right),\\
	& \textbf{\underline{3}}': S\mapsto 
	\frac{1}{2}\left(
	\begin{matrix}
	 -\phi & \frac{1}{\phi} & 1\\
	 \frac{1}{\phi} & -1 & \phi\\
	 1 & \phi & \frac{1}{\phi}
	\end{matrix}
	\right),
	T\mapsto
	\frac{1}{2}\left(
	\begin{matrix}
	 -\phi & -\frac{1}{\phi} & 1 \\
	 \frac{1}{\phi} & 1 & \phi \\
	 -1 & \phi & -\frac{1}{\phi}
	\end{matrix}
	\right),\\
	& \textbf{\underline{4}}: S\mapsto 
	\frac{1}{4}\left(
	\begin{matrix}
	 -1 & -1 & -3 & -\sqrt{5}\\
	 -1 & 3 & 1 & -\sqrt{5}\\
	 -3 & 1 & -1 & \sqrt{5}\\
	-\sqrt{5} & -\sqrt{5} & \sqrt{5} & -1
	\end{matrix}
	\right),
	T\mapsto
	\frac{1}{4}\left(
	\begin{matrix}
	 -1 & 1 & -3 & \sqrt{5}\\
	 -1 & -3 & 1 & \sqrt{5}\\
	 3 & 1 & 1 & \sqrt{5}\\
	\sqrt{5} & -\sqrt{5} & -\sqrt{5} & -1
	\end{matrix}
	\right),
	\end{split}
	\end{displaymath}

	\begin{displaymath}
	\begin{split}
	\textbf{\underline{5}}:\enspace & S\mapsto 
	\frac{1}{2}\left(
	\begin{matrix}
	 \frac{1-3\phi}{4} & \frac{\phi^{2}}{2} & -\frac{1}{2\phi^{2}} & \frac{\sqrt{5}}{2} & \frac{\sqrt{3}}{4\phi}\\
	 \frac{\phi^{2}}{2} & 1 & 1 & 0 & \frac{\sqrt{3}}{2\phi}\\
	 -\frac{1}{2\phi^{2}} & 1 & 0 & -1 & -\frac{\sqrt{3}\phi}{2}\\
	\frac{\sqrt{5}}{2} & 0 & -1 & 1 & -\frac{\sqrt{3}}{2}\\
	\frac{\sqrt{3}}{4\phi} & \frac{\sqrt{3}}{2\phi} & -\frac{\sqrt{3}\phi}{2} & -\frac{\sqrt{3}}{2} & \frac{3\phi-1}{4}
	\end{matrix}
	\right),\\
	& T\mapsto
	\frac{1}{2}\left(
	\begin{matrix}
	 \frac{1-3\phi}{4} & -\frac{\phi^{2}}{2} & -\frac{1}{2\phi^{2}} & -\frac{\sqrt{5}}{2} & \frac{\sqrt{3}}{4\phi}\\
	 \frac{\phi^{2}}{2} & -1 & 1 & 0 & \frac{\sqrt{3}}{2\phi}\\
	 \frac{1}{2\phi^{2}} & 1 & 0 & -1 & \frac{\sqrt{3}\phi}{2}\\
	-\frac{\sqrt{5}}{2} & 0 & 1 & 1 & \frac{\sqrt{3}}{2}\\
	\frac{\sqrt{3}}{4\phi} & -\frac{\sqrt{3}}{2\phi} & -\frac{\sqrt{3}\phi}{2} & \frac{\sqrt{3}}{2} & \frac{3\phi-1}{4}
	\end{matrix}
	\right).
	\end{split}	
	\end{displaymath}
The tensor products are given by \cite{everett}
	\begin{displaymath}
	\begin{split}
	& \textbf{\underline{3}}\otimes \textbf{\underline{3}}=\textbf{\underline{1}}\oplus \textbf{\underline{3}} \oplus \textbf{\underline{5}},\\
	& \textbf{\underline{3}}\otimes \textbf{\underline{3}}'=\textbf{\underline{4}}\oplus \textbf{\underline{5}},\\
	& \textbf{\underline{3}}'\otimes \textbf{\underline{3}}'=\textbf{\underline{1}}\oplus \textbf{\underline{3}}' \oplus \textbf{\underline{5}}.
	\end{split}
	\end{displaymath}
\cite{everett} provides a set of only two generators of $A_{5}$, which can be very useful. However for our purpose it is comfortable to use the following set of 3 generators
	\begin{displaymath}
	\begin{split}
	& A =\left(\begin{matrix}
		 0 & 1 & 0 \\
		 0 & 0 & 1 \\
		 1 & 0 & 0
	     \end{matrix}\right), \quad
	  B =\left(\begin{matrix}
		 1 & 0 & 0 \\
		 0 & -1 & 0 \\
		 0 & 0 & -1
	     \end{matrix}\right),\\
	& C =\frac{1}{2}\left(\begin{matrix}
		 -1 & \mu_{2} & \mu_{1} \\
		 \mu_{2} & \mu_{1} & -1 \\
		 \mu_{1} & -1 & \mu_{2}
	     \end{matrix}\right),
	\end{split}
	\end{displaymath}
with $\mu_{1}=\frac{1}{2}(-1+\sqrt{5})$, $\mu_{2}=\frac{1}{2}(-1-\sqrt{5})$ and $\omega=e^{\frac{2\pi i}{3}}$. This set of generators is given by \cite{miller} as a set of generators for the group $\Sigma(60)$, which is listed there as a finite subgroup of $SU(3)$. Because all generators are real, it must be a finite subgroup of $SO(3)$, and it turns out that $\Sigma(60)=I$ ($\Rightarrow \Sigma(60)\simeq A_{5}$), because $S$ and $T$ can be expressed through $A,B$ and $C$:
	\begin{displaymath}
	S=A^{2}BA^{2}CBA,\quad T=BA^{2}CBA^{2}.
	\end{displaymath}
This relation was found using the computer algebra system \textit{GAP} \cite{gap}. There are two reasons why it is more comfortable for us to use the generators $A,B,C$ instead of $S,T$:
	\begin{enumerate}
	 \item It will soon turn out that we cannot use the lemmata we had successfully used in previous cases to find the Clebsch-Gordan coefficients for the tensor product $\textbf{\underline{3}}\otimes \textbf{\underline{3}}'$ of $A_{5}$. So we have to use the algorithm described in chapter \ref{chapterclebsch}. Using the generators $S$ and $T$ the following problems occurred:
		\begin{itemize}
	 	\item Because of the high dimensions of the matrices involved (up to 45$\times$23 for a basis of eigenvectors to the eigenvalue 1 of $N(f_{1})$) memory became too small while performing part 2 of the algorithm (see chapter \ref{chapterclebsch}), which caused \textit{Mathematica 6} to shut down the kernel.\footnote{Used hardware: Intel(R) Core(TM)2 CPU 6420 @ 2.13GHz, 2048 MB RAM.}
	 	\item By choosing a different $f_{1}$ the dimension of the matrix mentioned in 1. could be reduced to 45$\times$9. Unfortunately there occurred another problem. Solving the equations was now possible, but for useful results one has to use commands like \textquotedblleft Simplify\textquotedblright\hspace{1mm} or even \textquotedblleft FullSimplify\textquotedblright, which turn out to be extremely time intensive calculations here. In fact \textit{Mathematica 6} stops simplification, when an internal time constraint is reached.
		\end{itemize}
		A numeric calculation with \textit{Mathematica 6} is extremely fast and gives results that are in agreement with the analytic expressions given in \cite{everett}. For the one- and three-dimensional representations the coefficients could be calculated analytically without problems. The results equal those of \cite{everett} (up to irrelevant phase factors).
		\\
		For the four- and five-dimensional representations we used the Clebsch-Gordan coefficients from \cite{everett} and checked analytically (using \textit{Mathematica 6}) that they are common eigenvectors of $N(f)$ ($f=S,T$) to the eigenvalue 1. This was doubly useful, because on the one hand we have checked the Clebsch-Gordan coefficients given in \cite{everett}, and on the other hand we checked the validity of the algorithm provided in chapter \ref{chapterclebsch}.
	\medskip
	\\
	Our hope is now that using the three generators $A,B,C$ (of which $A$ and $B$ have a quite simple form) the algorithm can be successfully carried out by \textit{Mathematica 6}.
	\item The generators $A,B,C$ also occur as a subset of the generators of $\Sigma(360\phi)$ ($\rightarrow$ subsection \ref{subsectionSigma360phi}), so one can directly see that $I=\Sigma(60)$ is a subgroup of $\Sigma(360\phi)$.
	\end{enumerate}
In contrast to the groups we had considered earlier
	\begin{displaymath}
	\textbf{\underline{3}}'\neq \textbf{\underline{3}}^{\ast},\quad\nexists\hspace{1mm} \textbf{\underline{1}}': \textbf{\underline{3}}'=\textbf{\underline{1}}'\otimes \textbf{\underline{3}},
	\end{displaymath}
so we have to consider all possible tensor products that can be built from $\textbf{\underline{3}}$ and $\textbf{\underline{3}}'$. We start with the tensor product
	\begin{displaymath}
	\textbf{\underline{3}}\otimes \textbf{\underline{3}}=\textbf{\underline{1}}\oplus\textbf{\underline{3}}\oplus\textbf{\underline{5}}. 
	\end{displaymath}
Because $\textbf{\underline{3}}=\textbf{\underline{3}}^{\ast}$ the irreducible representation $\textbf{\underline{1}}$ comes from $\textbf{\underline{3}}\otimes \textbf{\underline{3}}=\textbf{\underline{1}}\oplus \textbf{\underline{8}}$. $\textbf{\underline{3}}$ is again $\textbf{\underline{3}}_{a}$ of $\textbf{\underline{3}}\otimes \textbf{\underline{3}}=\textbf{\underline{3}}_{a}\oplus \textbf{\underline{6}}_{s}$. Thus we already know basis vectors of $V_{\textbf{\underline{1}}}$ and $V_{a}=V_{\textbf{\underline{3}}}$ ($\rightarrow$ table \ref{A4CGC}). The basis vectors of $V_{\textbf{\underline{5}}}$ can be obtained by extending the basis of $V_{\textbf{\underline{1}}}\oplus V_{\textbf{\underline{3}}}$ to an orthonormal basis of $V_{\textbf{\underline{3}}}\otimes V_{\textbf{\underline{3}}}$:
	\begin{displaymath}
	\begin{split}
	\textbf{\underline{5}}:\enspace & u_{1}=\frac{1}{\sqrt{2}}(e_{22}-e_{33}),\\
	& u_{2}=\frac{1}{\sqrt{6}}(-2e_{11}+e_{22}+e_{33}),\\
	& u_{3}=\frac{1}{\sqrt{2}}(e_{23}+e_{32}),\\
	& u_{4}=\frac{1}{\sqrt{2}}(e_{13}+e_{31}),\\
	& u_{5}=\frac{1}{\sqrt{2}}(e_{12}+e_{21}).
	\end{split}
	\end{displaymath}
The resulting orthonormal basis of $V_{\textbf{\underline{3}}}\otimes V_{\textbf{\underline{3}}}$ is shown in table \ref{A5CGCa}.

\begin{table}
\begin{center}
\renewcommand{\arraystretch}{1.4}
\begin{tabular}{|l|l|l|}
\firsthline
	$A_{5}$ & $\textbf{\underline{3}}\otimes\textbf{\underline{3}}$ & CGC\\
\hline
	$\textbf{\underline{1}}$
	& $u^{\textbf{\underline{3}}\otimes\textbf{\underline{3}}}_{\textbf{\underline{1}}}=\frac{1}{\sqrt{3}}e_{11}+\frac{1}{\sqrt{3}}e_{22}+\frac{1}{\sqrt{3}}e_{33}$ & \textbf{Ia}\\

	$\textbf{\underline{3}}$
	&	
	$u^{\textbf{\underline{3}}\otimes\textbf{\underline{3}}}_{\textbf{\underline{3}}}(1)=\frac{1}{\sqrt{2}}e_{23}-\frac{1}{\sqrt{2}}e_{32}$ & \textbf{IIIa$_{(1,-1)}$}\\

	&	
	$u^{\textbf{\underline{3}}\otimes\textbf{\underline{3}}}_{\textbf{\underline{3}}}(2)=-\frac{1}{\sqrt{2}}e_{13}+\frac{1}{\sqrt{2}}e_{31}$ & \\

	&	
	$u^{\textbf{\underline{3}}\otimes\textbf{\underline{3}}}_{\textbf{\underline{3}}}(3)=\frac{1}{\sqrt{2}}e_{12}-\frac{1}{\sqrt{2}}e_{21}$ & \\

	$\textbf{\underline{5}}$
	&	
	$u^{\textbf{\underline{3}}\otimes\textbf{\underline{3}}}_{\textbf{\underline{5}}}(1)=\frac{1}{\sqrt{2}}e_{22}-\frac{1}{\sqrt{2}}e_{33}$ & \textbf{Va}\\

	&	
	$u^{\textbf{\underline{3}}\otimes\textbf{\underline{3}}}_{\textbf{\underline{5}}}(2)=-\frac{2}{\sqrt{6}}e_{11}+\frac{1}{\sqrt{6}}e_{22}+\frac{1}{\sqrt{6}}e_{33}$ & \\

	&
	$u^{\textbf{\underline{3}}\otimes\textbf{\underline{3}}}_{\textbf{\underline{5}}}(3)=\frac{1}{\sqrt{2}}e_{23}+\frac{1}{\sqrt{2}}e_{32}$ & \\

	&	
	$u^{\textbf{\underline{3}}\otimes\textbf{\underline{3}}}_{\textbf{\underline{5}}}(4)=\frac{1}{\sqrt{2}}e_{13}+\frac{1}{\sqrt{2}}e_{31}$ & \\

	&	
	$u^{\textbf{\underline{3}}\otimes\textbf{\underline{3}}}_{\textbf{\underline{5}}}(5)=\frac{1}{\sqrt{2}}e_{12}+\frac{1}{\sqrt{2}}e_{21}$ & \\
\lasthline
\end{tabular}
\caption{Clebsch-Gordan coefficients for $\textbf{\underline{3}}\otimes\textbf{\underline{3}}$ of $A_{5}$.}
\label{A5CGCa}
\end{center}
\end{table}
\hspace{0mm}\\
Using the Clebsch-Gordan coefficients given in table \ref{A5CGCa} we can reduce the tensor products using the matrix of Clebsch-Gordan coefficients $M$ to obtain the generators $A,B,C$ of the 5-dimensional representation.
\\
The matrix representation of a tensor product is the Kronecker product of the matrix representations of its components ($\rightarrow$ definition \ref{DA83}). $M$ reduces a tensor product $D_{a}\otimes D_{b}=D^{(1)}\oplus...\oplus D^{(p)}$ in the following way:
	\begin{equation}\label{reduceequ}
		M^{-1}[D_{a}\otimes D_{b}]M=\left(
		\begin{matrix}
		 [D^{(1)}] & \textbf{0} & \textbf{0} & \textbf{0} \\
		 \textbf{0} & [D^{(2)}] & \textbf{0} & \textbf{0} \\
		 \textbf{0} & \textbf{0} & \ddots & \textbf{0} \\
		 \textbf{0} & \textbf{0} & \textbf{0} & [D^{(p)}]
		\end{matrix}
		\right).
	\end{equation}
On the other hand, using proposition \ref{Pclebsch2} $M$ is the matrix of basis change from $\{e_{i}^{a}\otimes e_{j}^{b}\}_{ij}$ to the bases of the invariant subspaces. Thus
	\begin{displaymath}
	\exists q:\enspace M[e_{i}\otimes e_{j}]=[u_{q}],
	\end{displaymath}
where $[u_{p}]$ is the matrix representation of $u_{p}$ with respect to the basis $\{e^{a}_{i}\otimes e^{b}_{j}\}_{ij}$ of $V_{a}\otimes V_{b}$ and $q$ chooses a specific basis vector from the list of all basis vectors. The matrix representations of the basis vectors $e^{a}_{i}\otimes e^{b}_{j}$ with respect to themselves are
	\begin{displaymath}
	[e_{1}^{a}\otimes e^{b}_{1}]=\left(\begin{matrix}
			 1 \\
			 0 \\
			 \vdots \\
			 0 
	              \end{matrix}\right),\quad...\quad,
	[e^{a}_{n}\otimes e^{b}_{n}]=\left(\begin{matrix}
			  0 \\
			  \vdots \\
			  0 \\
			  1 
	              \end{matrix}\right),
	\end{displaymath}
where we have assumed that $D_{a}$ and $D_{b}$ have the same dimension $n$ (in our case here $n=3$). It follows
	\begin{displaymath}
	M=\left(\begin{matrix}
		[u_{1}] & [u_{2}] & ... & [u_{n^{2}}]
	  \end{matrix}\right).
	\end{displaymath}
To ensure the reducing property (\ref{reduceequ}) we have to order the basis vectors of the invariant subspaces in such a way that
	\begin{displaymath}
	M=\left(\begin{matrix}
		[u^{D_{a}\otimes D_{b}}_{D^{(1)}}(1)] & ... & [u^{D_{a}\otimes D_{b}}_{D^{(1)}}(n_{(1)})] & [u^{D_{a}\otimes D_{b}}_{D^{(2)}}(1)] & ... &
		[u^{D_{a}\otimes D_{b}}_{D^{(p)}}(n_{(p)})]
	  \end{matrix}\right).
	\end{displaymath}
Constructing $M$ for $\textbf{\underline{3}}\otimes \textbf{\underline{3}}$ of $A_{5}$ we find:
\begin{displaymath}
M=\left(
\begin{matrix}
 \frac{1}{\sqrt{3}} & 0 & 0 & 0 & 0 & -\frac{2}{\sqrt{6}} & 0 & 0 & 0 \\
 0 & 0 & 0 & \frac{1}{\sqrt{2}} & 0 & 0 & 0 & 0 & \frac{1}{\sqrt{2}} \\
 0 & 0 & -\frac{1}{\sqrt{2}} & 0 & 0 & 0 & 0 & \frac{1}{\sqrt{2}} & 0 \\
 0 & 0 & 0 & -\frac{1}{\sqrt{2}} & 0 & 0 & 0 & 0 & \frac{1}{\sqrt{2}} \\
 \frac{1}{\sqrt{3}} & 0 & 0 & 0 & \frac{1}{\sqrt{2}} & \frac{1}{\sqrt{6}} & 0 & 0 & 0 \\
 0 & \frac{1}{\sqrt{2}} & 0 & 0 & 0 & 0 & \frac{1}{\sqrt{2}} & 0 & 0 \\
 0 & 0 & \frac{1}{\sqrt{2}} & 0 & 0 & 0 & 0 & \frac{1}{\sqrt{2}} & 0 \\
 0 & -\frac{1}{\sqrt{2}} & 0 & 0 & 0 & 0 & \frac{1}{\sqrt{2}} & 0 & 0 \\
 \frac{1}{\sqrt{3}} & 0 & 0 & 0 & -\frac{1}{\sqrt{2}} & \frac{1}{\sqrt{6}} & 0 & 0 & 0
\end{matrix}
\right).
\end{displaymath}
We can define the representation $\textbf{\underline{3}}$ as the defining representation given by the generators $A,B,C$ listed before. Explicitly reducing the Kronecker product $[\textbf{\underline{3}}\otimes \textbf{\underline{3}}]$ using $M$ we obtain the follwing matrix representation of $\textbf{\underline{5}}$:
	\begin{small}
	\begin{displaymath}
	\begin{split}
	\textbf{\underline{5}}:\enspace & A\mapsto
	\left(
	\begin{matrix}
	 -\frac{1}{2} & \frac{\sqrt{3}}{2} & 0 & 0 & 0 \\
	 -\frac{\sqrt{3}}{2} & -\frac{1}{2} & 0 & 0 & 0 \\
	 0 & 0 & 0 & 1 & 0 \\
	 0 & 0 & 0 & 0 & 1 \\
	 0 & 0 & 1 & 0 & 0
	\end{matrix}
	\right),\quad
	B\mapsto
	\left(
	\begin{matrix}
	 1 & 0 & 0 & 0 & 0 \\
	 0 & 1 & 0 & 0 & 0 \\
	 0 & 0 & 1 & 0 & 0 \\
	 0 & 0 & 0 & -1 & 0 \\
	 0 & 0 & 0 & 0 & -1
	\end{matrix}
	\right),\\
	& C\mapsto
	\left(
	\begin{matrix}
	\frac{1}{8} & -\frac{\sqrt{15}}{8} & -\frac{\sqrt{5}}{4} & \frac{1}{8} \left(3+\sqrt{5}\right) & \frac{1}{8} \left(-3+\sqrt{5}\right) \\
	-\frac{\sqrt{15}}{8} & -\frac{1}{8} & \frac{\sqrt{3}}{4} & \frac{1}{8} \sqrt{3} \left(-1+\sqrt{5}\right) & -\frac{1}{8} \sqrt{3} \left(1+\sqrt{5}\right)
	\\
	-\frac{\sqrt{5}}{4} & \frac{\sqrt{3}}{4} & 0 & \frac{1}{2} & \frac{1}{2} \\
	\frac{1}{8} \left(3+\sqrt{5}\right) & \frac{1}{8} \sqrt{3} \left(-1+\sqrt{5}\right) & \frac{1}{2} & \frac{1}{2} & 0 \\
	\frac{1}{8} \left(-3+\sqrt{5}\right) & -\frac{1}{8} \sqrt{3} \left(1+\sqrt{5}\right) & \frac{1}{2} & 0 & \frac{1}{2}
	\end{matrix}
	\right).
	\end{split}
	\end{displaymath}
	\end{small}
\hspace{0mm}\\
Now we turn to the tensor products involving $\textbf{\underline{3}}'$. How can we find $\textbf{\underline{3}}'$? According to \cite{miller} the generators of $\Sigma(60)=I$ must fulfil the equations
	\begin{displaymath}
	A^{3}=B^{2}=C^{2}=E,\quad (AB)^{3}=(BC)^{3}=E,\quad (AC)^{2}=E,
	\end{displaymath}
where $E$ denotes the identity element of the group. The matrices $\textbf{\underline{3}}(A), \textbf{\underline{3}}(B), \textbf{\underline{3}}(C)$ given before fulfil this relations and the additional condition $\mathrm{det}\hspace{1mm}\textbf{\underline{3}}(A)=\mathrm{det}\hspace{1mm}\textbf{\underline{3}}(B)=\mathrm{det}\hspace{1mm}\textbf{\underline{3}}(C)=1$. There is another solution of these equations which can be found using \textit{Mathematica 6} or a calculation by hand\footnote{I thank Walter Grimus for doing the calculation by hand.}, and it must correspond to $\textbf{\underline{3}}'$:
	\begin{displaymath}
	\begin{split}
	\textbf{\underline{3}}':\enspace & A\mapsto\left(\begin{matrix}
		 0 & 1 & 0 \\
		 0 & 0 & 1 \\
		 1 & 0 & 0
	     \end{matrix}\right),\quad
	B\mapsto\left(\begin{matrix}
		 1 & 0 & 0 \\
		 0 & -1 & 0 \\
		 0 & 0 & -1
	     \end{matrix}\right),\\
	& C \mapsto\frac{1}{2}\left(\begin{matrix}
		 -1 & \mu_{1} & \mu_{2} \\
		 \mu_{1} & \mu_{2} & -1 \\
		 \mu_{2} & -1 & \mu_{1}
	     \end{matrix}\right).
	\end{split}
	\end{displaymath}
Note that from the character table we cannot distinguish which of the given representations is $\textbf{\underline{3}}$ and which is $\textbf{\underline{3}}'$, because the character table remains invariant under $\textbf{\underline{3}}\leftrightarrow\textbf{\underline{3}}'$ and $C_{2}(12)\leftrightarrow C_{3}(12)$, so choosing the defining representation to be $\textbf{\underline{3}}$ was allowed.
\bigskip
\\
We now turn to $\textbf{\underline{3}}\otimes \textbf{\underline{3}}'=\textbf{\underline{4}}\oplus \textbf{\underline{5}}$. At first we calculate the basis vectors of $V_{\textbf{\underline{5}}}$ using the algorithm developed in chapter \ref{chapterclebsch}. One finds
	\begin{displaymath}
	\begin{split}
	\textbf{\underline{5}}:\enspace & u_{1}=-\frac{\sqrt{\frac{5}{3}}}{2}e_{11} + \frac{-3+\sqrt{5}}{4 \sqrt{3}}e_{22} + \frac{3+\sqrt{5}}{4 \sqrt{3}}e_{33},\\
	& u_{2}=\frac{1}{2}e_{11} + \frac{1}{4} \left(-1-\sqrt{5}\right)e_{22} + \frac{1}{4} \left(-1+\sqrt{5}\right)e_{33},\\
	& u_{3}=\sqrt{\frac{1}{6} \left(3+\sqrt{5}\right)}e_{23} -\frac{-1+\sqrt{5}}{2 \sqrt{3}}e_{32},\\
	& u_{4}=-\frac{-1+\sqrt{5}}{2 \sqrt{3}}e_{13} + \sqrt{\frac{1}{6} \left(3+\sqrt{5}\right)}e_{31},\\
	& u_{5}=\sqrt{\frac{1}{6} \left(3+\sqrt{5}\right)}e_{12} -\frac{-1+\sqrt{5}}{2 \sqrt{3}}e_{21}.
	\end{split}
	\end{displaymath}
Extending the basis of $V_{\textbf{\underline{5}}}$ to an orthonormal basis of $V_{\textbf{\underline{3}}}\otimes V_{\textbf{\underline{3}}'}$ we find the basis vectors of $V_{\textbf{\underline{4}}}$. The resulting orthonormal basis of $V_{\textbf{\underline{3}}}\otimes V_{\textbf{\underline{3}}'}$ is shown in table \ref{A5CGCb}.

\begin{table}
\begin{center}
\renewcommand{\arraystretch}{1.4}
\begin{tabular}{|l|l|l|}
\firsthline
	$A_{5}$ & $\textbf{\underline{3}}\otimes\textbf{\underline{3}}'$ & CGC\\
\hline
	$\textbf{\underline{4}}$
	&	
	$u^{\textbf{\underline{3}}\otimes\textbf{\underline{3}}'}_{\textbf{\underline{4}}}(1)=\frac{1}{\sqrt{3}}e_{11}+\frac{1}{\sqrt{3}}e_{22}+\frac{1}{\sqrt{3}}e_{33}$ & \textbf{IVa}\\

	&	
	$u^{\textbf{\underline{3}}\otimes\textbf{\underline{3}}'}_{\textbf{\underline{4}}}(2)=\sqrt{\frac{1}{6} \left(3-\sqrt{5}\right)}e_{12}+\frac{1}{2}(-1+\sqrt{5})\sqrt{\frac{3+\sqrt{5}}{9-3\sqrt{5}}}e_{21}$ & \\

	&	
	$u^{\textbf{\underline{3}}\otimes\textbf{\underline{3}}'}_{\textbf{\underline{4}}}(3)=\sqrt{\frac{1}{6} \left(3+\sqrt{5}\right)}e_{13}+\frac{-1+\sqrt{5}}{2\sqrt{3}}e_{31}$ & \\

	&	
	$u^{\textbf{\underline{3}}\otimes\textbf{\underline{3}}'}_{\textbf{\underline{4}}}(4)=\sqrt{\frac{1}{6} \left(3-\sqrt{5}\right)}e_{23}+\frac{1}{2}(-1+\sqrt{5})\sqrt{\frac{3+\sqrt{5}}{9-3\sqrt{5}}}e_{32}$ & \\

	$\textbf{\underline{5}}$
	&	
	$u^{\textbf{\underline{3}}\otimes\textbf{\underline{3}}'}_{\textbf{\underline{5}}}(1)=-\frac{\sqrt{\frac{5}{3}}}{2}e_{11} + \frac{-3+\sqrt{5}}{4 \sqrt{3}}e_{22} + \frac{3+\sqrt{5}}{4 \sqrt{3}}e_{33}$ & \textbf{Vb}\\

	&	
	$u^{\textbf{\underline{3}}\otimes\textbf{\underline{3}}'}_{\textbf{\underline{5}}}(2)=\frac{1}{2}e_{11} + \frac{1}{4} \left(-1-\sqrt{5}\right)e_{22} + \frac{1}{4} \left(-1+\sqrt{5}\right)e_{33}$ & \\

	&
	$u^{\textbf{\underline{3}}\otimes\textbf{\underline{3}}'}_{\textbf{\underline{5}}}(3)=\sqrt{\frac{1}{6} \left(3+\sqrt{5}\right)}e_{23} -\frac{-1+\sqrt{5}}{2 \sqrt{3}}e_{32}$ & \\

	&	
	$u^{\textbf{\underline{3}}\otimes\textbf{\underline{3}}'}_{\textbf{\underline{5}}}(4)=-\frac{-1+\sqrt{5}}{2 \sqrt{3}}e_{13} + \sqrt{\frac{1}{6} \left(3+\sqrt{5}\right)}e_{31}$ & \\

	&	
	$u^{\textbf{\underline{3}}\otimes\textbf{\underline{3}}'}_{\textbf{\underline{5}}}(5)=\sqrt{\frac{1}{6} \left(3+\sqrt{5}\right)}e_{12} -\frac{-1+\sqrt{5}}{2 \sqrt{3}}e_{21}$ & \\
\lasthline
\end{tabular}
\caption{Clebsch-Gordan coefficients for $\textbf{\underline{3}}\otimes\textbf{\underline{3}}'$ of $A_{5}$.}
\label{A5CGCb}
\end{center}
\end{table}
\hspace{0mm}\\
We can now construct the matrix $M$ that reduces $[\textbf{\underline{3}}\otimes \textbf{\underline{3}}']$ to find $\textbf{\underline{4}}$:
	\begin{displaymath}
	\begin{split}
	\textbf{\underline{4}}:\enspace & A\mapsto
	\left(
	\begin{matrix}
	 1 & 0 & 0 & 0 \\
	 0 & 0 & 0 & 1 \\
	 0 & 1 & 0 & 0 \\
	 0 & 0 & 1 & 0
	\end{matrix}
	\right),\quad
	B\mapsto\left(\begin{matrix}
	1 & 0 & 0 & 0 \\
	0 & -1 & 0 & 0 \\
	0 & 0 & -1 & 0 \\
	0 & 0 & 0 & 1
	     \end{matrix}\right),\\
	& C \mapsto\frac{1}{4}\left(\begin{matrix}
	 -1 & \sqrt{5} & \sqrt{5} & \sqrt{5} \\
	 \sqrt{5} & -1 & 3 & -1 \\
	 \sqrt{5} & 3 & -1 & -1 \\
	 \sqrt{5} & -1 & -1 & 3
	     \end{matrix}\right).
	\end{split}
	\end{displaymath}
Using proposition \ref{PY11} one can easily construct the Clebsch-Gordan coefficients for $\textbf{\underline{3}}'\otimes \textbf{\underline{3}}$. In terms of basis vectors this would correspond to the replacement $e_{ij}\mapsto e_{ji}$ in table \ref{A5CGCb}.
\bigskip
\\
At last we have to consider the Clebsch-Gordan decomposition $\textbf{\underline{3}}'\otimes \textbf{\underline{3}}'=\textbf{\underline{1}}\oplus \textbf{\underline{3}}'\oplus \textbf{\underline{5}}$. Since all arguments used to find the Clebsch-Gordan coefficients for $\textbf{\underline{3}}\otimes \textbf{\underline{3}}$ also hold here, the coefficients can be chosen to be the same as those for $\textbf{\underline{3}}\otimes \textbf{\underline{3}}$ ($\rightarrow$ table \ref{A5CGCa}). However, there is still a little difference: Reducing $[\textbf{\underline{3}}'\otimes \textbf{\underline{3}}']$ with the matrix $M$ for $\textbf{\underline{3}}\otimes \textbf{\underline{3}}$ will lead to an irreducible matrix representation $\hat{\textbf{\underline{5}}}$, which is equivalent (but not equal) to the matrix representation $\textbf{\underline{5}}$ obtained before. We find:
	\begin{small}
	\begin{displaymath}
	\begin{split}
	\hat{\textbf{\underline{5}}}:\enspace & A\mapsto
	\left(
	\begin{matrix}
	 -\frac{1}{2} & \frac{\sqrt{3}}{2} & 0 & 0 & 0 \\
	 -\frac{\sqrt{3}}{2} & -\frac{1}{2} & 0 & 0 & 0 \\
	 0 & 0 & 0 & 1 & 0 \\
	 0 & 0 & 0 & 0 & 1 \\
	 0 & 0 & 1 & 0 & 0
	\end{matrix}
	\right),\quad
	B\mapsto
	\left(
	\begin{matrix}
	 1 & 0 & 0 & 0 & 0 \\
	 0 & 1 & 0 & 0 & 0 \\
	 0 & 0 & 1 & 0 & 0 \\
	 0 & 0 & 0 & -1 & 0 \\
	 0 & 0 & 0 & 0 & -1
	\end{matrix}
	\right),\\
	& C\mapsto
	\left(
	\begin{matrix}
	 \frac{1}{8} & \frac{\sqrt{15}}{8} & \frac{\sqrt{5}}{4} & \frac{1}{8} \left(3-\sqrt{5}\right) & \frac{1}{8} \left(-3-\sqrt{5}\right) \\
 \frac{\sqrt{15}}{8} & -\frac{1}{8} & \frac{\sqrt{3}}{4} & -\frac{1}{8} \sqrt{3} \left(1+\sqrt{5}\right) & \frac{1}{8} \sqrt{3} \left(-1+\sqrt{5}\right)
\\
 \frac{\sqrt{5}}{4} & \frac{\sqrt{3}}{4} & 0 & \frac{1}{2} & \frac{1}{2} \\
 \frac{1}{8} \left(3-\sqrt{5}\right) & -\frac{1}{8} \sqrt{3} \left(1+\sqrt{5}\right) & \frac{1}{2} & \frac{1}{2} & 0 \\
 \frac{1}{8} \left(-3-\sqrt{5}\right) & \frac{1}{8} \sqrt{3} \left(-1+\sqrt{5}\right) & \frac{1}{2} & 0 & \frac{1}{2}
	\end{matrix}
	\right).
	\end{split}
	\end{displaymath}
	\end{small}
\hspace{0mm}\\
Since $\hat{\textbf{\underline{5}}}\sim \textbf{\underline{5}}$, there must exist a matrix $S$ such that
	\begin{displaymath}
	S^{-1}\textbf{\underline{5}}(f)S=\hat{\textbf{\underline{5}}}(f)\quad\forall f\in \mathrm{gen}(A_{5}).
	\end{displaymath}
Rewriting this as
	\begin{displaymath}
	\textbf{\underline{5}}(f)S(\hat{\textbf{\underline{5}}}(f))^{-1}=S\quad\forall f\in \mathrm{gen}(A_{5})
	\end{displaymath}
we can interpret this as an eigenvalue problem
	\begin{displaymath}
	Ns=s
	\end{displaymath}
for
	\begin{displaymath}
	s=\left(\begin{matrix}
	 S_{11} & S_{12} & ... & S_{55}
	  \end{matrix}\right)^{T},
	\end{displaymath}
where $N$ is given by
	\begin{displaymath}
	N=\left(\begin{matrix}
	 [\textbf{\underline{5}}]_{11}[(\hat{\textbf{\underline{5}}})^{-1}]_{11} & ... & [\textbf{\underline{5}}]_{15}[(\hat{\textbf{\underline{5}}})^{-1}]_{51} \\
	 \vdots & & \vdots\\
	 [\textbf{\underline{5}}]_{51}[(\hat{\textbf{\underline{5}}})^{-1}]_{15} & ... & [\textbf{\underline{5}}]_{55}[(\hat{\textbf{\underline{5}}})^{-1}]_{55}
	  \end{matrix}\right).
	\end{displaymath}
We have to find the common eigenvectors to the eigenvalue 1 of $N(f)$ ($f=A,B,C$). Solving this problem (using a computer algebra system) leads to
	\begin{displaymath}
	S=\left(
	\begin{matrix}
	 -\frac{1}{4} & -\frac{\sqrt{15}}{4} & 0 & 0 & 0 \\
	 \frac{\sqrt{15}}{4} & -\frac{1}{4} & 0 & 0 & 0 \\
	 0 & 0 & 1 & 0 & 0 \\
	 0 & 0 & 0 & 1 & 0 \\
	 0 & 0 & 0 & 0 & 1
	\end{matrix}
	\right).
	\end{displaymath}

\subsection{The double cover of the icosahedral group $\tilde{I}$}
The double cover of $I$ is isomorphic to $\tilde{A}_{5}$, which is treated in \cite{luhn1}. According to \cite{luhn1} the new irreducible representations occurring (in addition to the irreducible representations of $A_{5}$) are $\textbf{\underline{2}}$, $\textbf{\underline{2}}'$, $\textbf{\underline{4}}$ and $\textbf{\underline{6}}$, thus there are no new three-dimensional irreducible representations and the Clebsch-Gordan coefficients (for tensor products of the three-dimensional representations) remain unchanged.

\subsection{The dihedral groups $D_{n}$ and their double covers $\tilde{D}_{n}$}
The dihedral group $D_{n}$ is the group of $3$-dimensional ($SO(3)$) rotations that leave the planar regular $n$-gon invariant. Therefore the operations of $D_{n}$ are \cite{hamermesh}:
	\begin{itemize}
	 \item Rotations about the angle $\frac{2\pi k}{n}$ in the plane of the $n$-gon ($k=0,...,n-1$).
	 \item Rotations about the angle $\pi$ around the $n$ symmetry axes of the regular $n$-gon.
	\end{itemize}
Thus we easily see that $D_{n}$ is of order $2n$. The rotations about the angle $\pi$ around the symmetry axes act on the polygon in the same way as reflections at the symmetry axes do. The difference is that using the rotations one gets a finite subgroup of $SO(3)$, and using the reflections one would get a finite subgroup of $O(3)$ instead. Example: Let the plane of the regular polygon be the $(x,y)-$plane, and let the $y$-axis be a symmetry axis of the polygon, then a reflection $P_{yz}$ at the $(y,z)-$plane (this corresponds to a 2-dimensional reflection at the $y$-axis in the $(x,y)-$plane) and a rotation $R_{y}(\pi)$ about the angle $\pi$ around the $y$-axis both leave the polygon invariant.
	\begin{displaymath}
	P_{yz}=\left(\begin{matrix}
-1 & 0 & 0\\
0 & 1 & 0 \\
0 & 0 & 1
	      \end{matrix}\right)\in O(3)
\quad\quad
R_{y}(\pi)=\left(\begin{matrix}
 -1 & 0 & 0 \\
 0 & 1 & 0 \\
 0 & 0 & -1
\end{matrix}\right)\in SO(3)
	\end{displaymath}
The dihedral groups and their double covers  are treated in \cite{blum}. According to this article $D_{n}$ and $\tilde{D}_{n}$ do not possess any three-dimensional irreducible representations, so we do not need to analyse them in our work.
\bigskip
\\
We have now treated all finite subgroups of $SO(3)$ and $SU(2)$ that have three-dimensional irreducible representations, and we will now go on working with the remaining finite subgroups of $SU(3)$.

\section{The finite subgroups of $SU(3)$}
The finite non-Abelian subgroups of $SU(3)$ that are not subgroups of $SU(2)$ or $SO(3)$ are listed in table \ref{SU3subgroups}. Before we will treat the groups separately, we will investigate some general properties of the finite subgroups of $SU(3)$.

\begin{define}\label{DSU32}
	The \textit{center of a group} $G$ is defined as the subgroup $C\subset G$ that fulfils
	\begin{displaymath}
		cg=gc \quad \forall c\in C, g\in G.
	\end{displaymath}
	For $SU(n)$ $C$ is given by
	\begin{displaymath}
		C=\{ \alpha\mathbbm{1}_{n}\vert\alpha\in Z_{n}\}.
	\end{displaymath}
\end{define}

\begin{prop}\label{PSU39}
Let $G$ be a finite subgroup of $SU(n)$, then the center of $G$ is
	\begin{displaymath}
	C=\{\alpha \mathbbm{1}_{n}\vert \alpha\in Y\subset Z_{n}\},
	\end{displaymath}
where $Y$ is a subgroup of $Z_{n}$.
\end{prop}

\begin{proof}
From Schur's lemma (\ref{LA50}) we see that all elements of the center must be proportional to $\mathbbm{1}_{n}$. Since $C$ is a subgroup of $SU(n)$: $\mathrm{det}(\alpha\mathbbm{1}_{n})=\alpha^{n}=1\Rightarrow \alpha\in Z_{n}$.
\end{proof}

\begin{cor}\label{CSU310}
Let $G$ be a finite subgroup of $SU(3)$, then the center of $G$ can be either
	\begin{displaymath}
	C=\{\alpha\mathbbm{1}_{3}\vert \alpha\in Z_{3}\} \quad\mbox{or}\quad C=\{\mathbbm{1}_{3}\}.
	\end{displaymath}
\end{cor}

\begin{proof}
This follows from proposition \ref{PSU39} and the fact that $\{1\}$ and $Z_{3}$ are the only subgroups of $Z_{3}$.
\end{proof}

\begin{define}\label{DSU311}
Let $G$ be a group, then the group $G/C$ is called \textit{collineation group} of $G$.\footnote{For the definition of the factor group $G/C$ we refer the reader to definition \ref{DA22}.}
\end{define}

\begin{define}\label{DSU312b}
The group $SU(n)/C$ is called \textit{projective special unitary group} in $n$ dimensions, i.s. $PSU(n)$.
\end{define}

\begin{table}
\begin{center}
\renewcommand{\arraystretch}{1.4}
\begin{tabular}{|lll|}
\firsthline
	Type of subgroups\rule{3mm}{0mm} & \rule{3mm}{0mm} Subgroup & Order of the subgroup \\
\hline
	$\Sigma(n\phi)$, $n=36, 72, 216, 360$ \rule{3mm}{0mm} & \rule{3mm}{0mm} $\Sigma(36\phi)$ & 108 \\
	\rule{3mm}{0mm} & \rule{3mm}{0mm} $\Sigma(72\phi)$ & 216\\
	\rule{3mm}{0mm} & \rule{3mm}{0mm} $\Sigma(216\phi)$ & 648\\
	\rule{3mm}{0mm} & \rule{3mm}{0mm} $\Sigma(360\phi)$ & 1080\\
\hline
	$\Sigma(m)$, $m=60, 168$\rule{3mm}{0mm} & \rule{3mm}{0mm} $\Sigma(60)$ & 60\\
	\rule{3mm}{0mm} & \rule{3mm}{0mm} $\Sigma(168)$ & 168\\
\hline
	$\Delta(3n^{2})$, $n\in\mathbb{N}\backslash \{0,1\}$\rule{3mm}{0mm} & \rule{3mm}{0mm} $\Delta(3n^{2})$ & $3n^{2}$\\
	$\Delta(6n^{2})$, $n\in\mathbb{N}\backslash \{0,1\}$\rule{3mm}{0mm} & \rule{3mm}{0mm} $\Delta(6n^{2})$ & $6n^{2}$\\
\lasthline
\end{tabular}
\caption[Non-Abelian finite subgroups of $SU(3)$ that are not subgroups of $SU(2)$ or $SO(3)$.]{Non-Abelian finite subgroups of $SU(3)$ that are not subgroups of $SU(2)$ or $SO(3)$, as given by Fairbairn et al. in \cite{fairbairn}. Note that as explained in section \ref{icosahedral} $\Sigma(60)=I$. Thus $\Sigma(60)$ is a finite subgroup of $SO(3)$ and should not be listed in this table.}
\label{SU3subgroups}
\end{center}
\end{table}
\hspace{0mm}\\
The notation $\Sigma(n\phi)$ in table \ref{SU3subgroups} means that these groups have collineation groups $\Sigma(n)$ of order $n$, which are finite subgroups of $PSU(3)$. According to \cite{miller} it could in principle be possible that the collineation group $\Sigma(n)$ has a three-dimensional irreducible representation, so we have to investigate the groups $\Sigma(n)$ too. Fortunately there exists a helpful theorem (theorem \ref{labour-saving-theorem}) which will tell us that the analysis of $\Sigma(n)$ is included in the analysis of $\Sigma(n\phi)$.

\begin{prop}
Let $G$ be a finite group and let $H\subset G$ be an invariant subgroup ($\rightarrow$ definition \ref{DA19}) of $G$. Then
	\begin{displaymath}
	\begin{split}
	\phi:\enspace & G\rightarrow G/H\\
	& a\mapsto aH
	\end{split}
	\end{displaymath}
is a group homomorphism.
\end{prop}

\begin{proof}
	\begin{displaymath}
	\begin{split}
	\phi(a_1 a_2) &= (a_1 a_2)H=\\
		&= \{a_1 a_2 h\vert h\in H\}=\\
		&= \{a_1 a_2 h_1 h_2\vert h_1, h_2\in H\}=\\
		&= \{a_1 a_2 h_1 a_2^{-1} a_2 h_2\vert h_1, h_2\in H\}.
	\end{split}
	\end{displaymath}
$H$ is an invariant subgroup of $G$ $\Rightarrow$ $a_2 h_1 a_2^{-1}=:h_3 \in H$ and $a_2 H a_2^{-1}=H$. Thus
	\begin{displaymath}
	\begin{split}
	\phi(a_1 a_2) &= \{a_1 a_2 h_1 a_2^{-1} a_2 h_2\vert h_1, h_2\in H\}=\\
	&= \{a_1 h_3 a_2 h_2\vert h_2, h_3\in H\}=\\
	&= (a_1 H)(a_2 H)=\phi(a_1)\phi(a_2).
	\end{split}
	\end{displaymath}
\end{proof}

\begin{theorem}\label{labour-saving-theorem}
Let $G$ be a finite group, let $H\subset G$ be an invariant subgroup of $G$, let $D$ be an irreducible representation of $G/H$ and let
	\begin{displaymath}
	\begin{split}
	\phi:\enspace & G\rightarrow G/H\\
	& a\mapsto aH
	\end{split}
	\end{displaymath}
be a group homomorphism. Then $D\circ\phi$ is an irreducible representation of $G$.
\end{theorem}

\begin{proof}
$\phi$ and $D$ are group homomorphisms. Thus
	\begin{displaymath}
	\begin{split}
	D\circ\phi:\enspace & G\rightarrow D(G/H)\\
	& a\mapsto D(aH)
	\end{split}
	\end{displaymath}
is a group homomorphism, and therefore a representation of $G$. $D\circ \phi$ is an irreducible representation of $G$ if and only if
	\begin{displaymath}
	(\chi_{D\circ\phi},\chi_{D\circ\phi})_{G}:=\frac{1}{\mathrm{ord}(G)}\sum_{a\in G}\vert\mathrm{Tr}(D(\phi(a)))\vert^2=1.
	\end{displaymath}
(See proposition \ref{PA58}.) Consider now the homomorphism $\phi$: There are $\mathrm{ord}(H)$ different $a\in G$ that are mapped onto the same $aH\in G/H$.
	\begin{displaymath}
	\Rightarrow \sum_{a\in G}\vert\mathrm{Tr}(D(\phi(a)))\vert^2=\mathrm{ord}(H)\sum_{k\in G/H}\vert\mathrm{Tr}(D(k))\vert^2.
	\end{displaymath}
	\begin{displaymath}
	\Rightarrow (\chi_{D\circ\phi},\chi_{D\circ\phi})_G=\underbrace{\frac{\mathrm{ord}(H)}{\mathrm{ord}(G)}}_{\frac{1}{\mathrm{ord}(G/H)}}\sum_{k\in G/H}\vert\mathrm{Tr}(D(k))\vert^2=(\chi_{D},\chi_{D})_{G/H}.
	\end{displaymath}
By assumption $D$ is an irreducible representation of $G/H$.
	\begin{displaymath}
	\Rightarrow(\chi_{D},\chi_{D})_{G/H}=1\Rightarrow(\chi_{D\circ\phi},\chi_{D\circ\phi})_G=1.
	\end{displaymath}
Thus $D\circ\phi$ is an irreducible representation of $G$.
\end{proof}

\begin{cor}\label{Sigmancorollary}
Let $G$ be a finite subgroup of $SU(3)$. Then all irreducible representations of the collineation group $G/C$ are irreducible representations of $G$ too.
\end{cor}

\begin{proof}
Since the elements of the center $C$ commute with all elements of $G$, $C$ is an invariant subgroup of $G$, and using theorem \ref{labour-saving-theorem} we find the corollary.
\end{proof}
\hspace{0mm}
\\
Since (following corollary \ref{Sigmancorollary}) all irreducible representations of $\Sigma(n)$ are irreducible representations of $\Sigma(n\phi)$ too, for the tensor products and Clebsch-Gordan coefficients it is enough to investigate the groups $\Sigma(n\phi)$.

\subsection{The group $\Sigma(60)$}
We already found out that $\Sigma(60)$ is identical with the icosahedral group $I\simeq A_{5}$, which was already treated in subsection \ref{icosahedral}, where one can also find the generators of $\Sigma(60)$. In subsection \ref{subsectionSigma360phi} we will find out that $\Sigma(60)$ is a subgroup of $\Sigma(360\phi)$.

\subsection{The group $\Sigma(168)$}\label{sigma168subsection}
The group $\Sigma(168)$ is treated in \cite{fairbairn,miller,luhn1,King}. The generators of $\Sigma(168)$ as given in \cite{miller} are
	\begin{displaymath}
	S=\left(
	\begin{matrix}
	 \eta & 0 & 0 \\
	 0 & \eta^{2} & 0 \\
	 0 & 0 & \eta^{4}
	\end{matrix}
	\right),\enspace
	T=\left(
	\begin{matrix}
	 0 & 1 & 0 \\
	 0 & 0 & 1 \\
	 1 & 0 & 0
	\end{matrix}
	\right),\enspace
	R=h\left(
	\begin{matrix}
	 \eta^{4}-\eta^{3} & \eta^{2}-\eta^{5} & \eta-\eta^{6} \\
	 \eta^{2}-\eta^{5} & \eta-\eta^{6} & \eta^{4}-\eta^{3} \\
	 \eta-\eta^{6} & \eta^{4}-\eta^{3} & \eta^{2}-\eta^{5}
	\end{matrix}
	\right),
	\end{displaymath}
where $h=\frac{1}{7}(\eta+\eta^{2}+\eta^{4}-\eta^{6}-\eta^{5}-\eta^{3})$ and $\eta=e^{\frac{2\pi i}{7}}$. According to \cite{luhn1} one needs only two generators to generate $\Sigma(168)$, namely
	\begin{displaymath}
	S'=\frac{i}{\sqrt{7}}\left(
	\begin{matrix}
	 \eta^{2}-\eta^{5} & \eta-\eta^{6} & \eta^{4}-\eta^{3} \\
	 \eta-\eta^{6} & \eta^{4}-\eta^{3} & \eta^{2}-\eta^{5} \\
	 \eta^{4}-\eta^{3} & \eta^{2}-\eta^{5} & \eta-\eta^{6}
	  \end{matrix}
	\right),\quad
	T'=\frac{i}{\sqrt{7}}\left(
	\begin{matrix}
	 \eta^{3}-\eta^{6} & \eta^{3}-\eta & \eta-1 \\
	 \eta^{2}-1 & \eta^{6}-\eta^{5} & \eta^{6}-\eta^{2} \\
	 \eta^{5}-\eta^{4} & \eta^{4}-1 & \eta^{5}-\eta^{3}
	  \end{matrix}
	\right).
	\end{displaymath}
Table \ref{Sigma168charactertable} shows the character table given in \cite{fairbairn}.
\begin{table}
\begin{center}
\renewcommand{\arraystretch}{1.4}
\begin{tabular}{|l|cccccc|}
\firsthline
	$\Sigma(168)$ & $C_{1}(1)$ & $C_{2}(21)$ & $C_{3}(42)$ & $C_{4}(56)$ & $C_{5}(24)$ & $C_{6}(24)$\\
\hline
	\textbf{\underline{1}} & $1$ & $1$ & $1$ & $1$ & $1$ & $1$\\
	\textbf{\underline{3}} & $3$ & $-1$ & $1$ & $0$ & $a$ & $a^{\ast}$\\
	\textbf{\underline{3}}$^{\ast}$ & $3$ & $-1$ & $1$ & $0$ & $a^{\ast}$ & $a$\\
	\textbf{\underline{6}} & $6$ & $2$ & $0$ & $0$ & $-1$ & $-1$\\
	\textbf{\underline{7}} & $7$ & $-1$ & $-1$ & $1$ & $0$ & $0$\\
	\textbf{\underline{8}} & $8$ & $0$ & $0$ & $-1$ & $1$ & $1$\\
\lasthline
\end{tabular}
\caption[The character table of $\Sigma(168)$.]{The character table of $\Sigma(168)$ \cite{fairbairn}. $a=\frac{1}{2}(-1+i\sqrt{7})$}
\label{Sigma168charactertable}
\end{center}
\end{table}
\hspace{0mm}\\
Using the character table we can compute the tensor products:
	\begin{displaymath}
	\begin{split}
	& \textbf{\underline{3}}\otimes\textbf{\underline{3}}=\textbf{\underline{3}}^{\ast}\oplus \textbf{\underline{6}},\\
	& \textbf{\underline{3}}\otimes\textbf{\underline{3}}^{\ast}=\textbf{\underline{1}}\oplus \textbf{\underline{8}},\\
	&
	\textbf{\underline{3}}^{\ast}\otimes\textbf{\underline{3}}^{\ast}=\textbf{\underline{3}}\oplus \textbf{\underline{6}}.\\
	\end{split}
	\end{displaymath}
Here we have the interesting case\footnote{Side remark: It is another interesting feature that (to current knowledge) $\Sigma(168)$ is the only finite subgroup of $SU(3)$ that possesses an irreducible 7-dimensional representation.} that $3\otimes 3$-tensor products of $\Sigma(168)$ only decompose in the \textquotedblleft standard ways\textquotedblright\hspace{1mm} covered by lemmata \ref{CGdecompL1} and \ref{CGdecompL2}. So we already know the basis vectors of the invariant subspaces, and normalizing them we find the orthonormal bases of $V_{\textbf{\underline{3}}}\otimes V_{\textbf{\underline{3}}}$ and $V_{\textbf{\underline{3}}}\otimes V_{\textbf{\underline{3}}^{\ast}}$, which are shown in tables \ref{Sigma168CGCa} and \ref{Sigma168CGCb}. (The basis vectors of $V_{\textbf{\underline{8}}}$ were obtained by extending $u^{\textbf{\underline{3}}\otimes \textbf{\underline{3}}^{\ast}}_{\textbf{\underline{1}}}$ to an orthonormal basis of $V_{\textbf{\underline{3}}}\otimes V_{\textbf{\underline{3}}^{\ast}}$.)

\begin{table}
\begin{center}
\renewcommand{\arraystretch}{1.4}
\begin{tabular}{|l|l|l|}
\firsthline
	$\Sigma(168)$ & $\textbf{\underline{3}}\otimes\textbf{\underline{3}}$/$\textbf{\underline{3}}^{\ast}\otimes\textbf{\underline{3}}^{\ast}$ & CGC\\
\hline
	$\textbf{\underline{3}}$/$\textbf{\underline{3}}^{\ast}$
	&	
	$u^{\textbf{\underline{3}}\otimes\textbf{\underline{3}}}_{\textbf{\underline{3}}}(1)=\frac{1}{\sqrt{2}}e_{23}-\frac{1}{\sqrt{2}}e_{32}$ & \textbf{IIIa$_{(1,-1)}$}\\

	&	
	$u^{\textbf{\underline{3}}\otimes\textbf{\underline{3}}}_{\textbf{\underline{3}}}(2)=-\frac{1}{\sqrt{2}}e_{13}+\frac{1}{\sqrt{2}}e_{31}$ & \\

	&	
	$u^{\textbf{\underline{3}}\otimes\textbf{\underline{3}}}_{\textbf{\underline{3}}}(3)=\frac{1}{\sqrt{2}}e_{12}-\frac{1}{\sqrt{2}}e_{21}$ & \\

	$\textbf{\underline{6}}$
	&	
	$u^{\textbf{\underline{3}}\otimes\textbf{\underline{3}}}_{\textbf{\underline{6}}}(1)=e_{11}$ & \textbf{VIa}\\

	&	
	$u^{\textbf{\underline{3}}\otimes\textbf{\underline{3}}}_{\textbf{\underline{6}}}(2)=e_{22}$ &\\

	&	
	$u^{\textbf{\underline{3}}\otimes\textbf{\underline{3}}}_{\textbf{\underline{6}}}(3)=e_{33}$ &\\

	&	
	$u^{\textbf{\underline{3}}\otimes\textbf{\underline{3}}}_{\textbf{\underline{6}}}(4)=\frac{1}{\sqrt{2}}e_{23}+\frac{1}{\sqrt{2}}e_{32}$ &\\

	&	
	$u^{\textbf{\underline{3}}\otimes\textbf{\underline{3}}}_{\textbf{\underline{6}}}(5)=\frac{1}{\sqrt{2}}e_{13}+\frac{1}{\sqrt{2}}e_{31}$ &\\

	&	
	$u^{\textbf{\underline{3}}\otimes\textbf{\underline{3}}}_{\textbf{\underline{6}}}(6)=\frac{1}{\sqrt{2}}e_{12}+\frac{1}{\sqrt{2}}e_{21}$ &\\
\lasthline
\end{tabular}
\caption{Clebsch-Gordan coefficients for $\textbf{\underline{3}}\otimes\textbf{\underline{3}}$ and $\textbf{\underline{3}}^{\ast}\otimes\textbf{\underline{3}}^{\ast}$ of $\Sigma(168)$.}
\label{Sigma168CGCa}
\end{center}
\end{table}

\begin{table}
\begin{center}
\renewcommand{\arraystretch}{1.4}
\begin{tabular}{|l|l|l|}
\firsthline
	$\Sigma(168)$ & $\textbf{\underline{3}}\otimes\textbf{\underline{3}}^{\ast}$ & CGC\\
\hline
	$\textbf{\underline{1}}$
	&	
	$u^{\textbf{\underline{3}}\otimes\textbf{\underline{3}}^{\ast}}_{\textbf{\underline{1}}}=\frac{1}{\sqrt{3}}e_{11}+\frac{1}{\sqrt{3}}e_{22}+\frac{1}{\sqrt{3}}e_{33}$ & \textbf{Ia}\\

	$\textbf{\underline{8}}$
	&
	$u^{\textbf{\underline{3}}\otimes\textbf{\underline{3}}^{\ast}}_{\textbf{\underline{8}}}(1)=e_{12}$ & \textbf{VIII}\\

	&
	$u^{\textbf{\underline{3}}\otimes\textbf{\underline{3}}^{\ast}}_{\textbf{\underline{8}}}(2)=e_{13}$& \\

	&
	$u^{\textbf{\underline{3}}\otimes\textbf{\underline{3}}^{\ast}}_{\textbf{\underline{8}}}(3)=e_{21}$& \\

	&
	$u^{\textbf{\underline{3}}\otimes\textbf{\underline{3}}^{\ast}}_{\textbf{\underline{8}}}(4)=\frac{1}{\sqrt{2}}e_{11}-\frac{1}{\sqrt{2}}e_{22}$ & \\

	&
	$u^{\textbf{\underline{3}}\otimes\textbf{\underline{3}}^{\ast}}_{\textbf{\underline{8}}}(5)=e_{23}$& \\

	&
	$u^{\textbf{\underline{3}}\otimes\textbf{\underline{3}}^{\ast}}_{\textbf{\underline{8}}}(6)=e_{31}$& \\

	&
	$u^{\textbf{\underline{3}}\otimes\textbf{\underline{3}}^{\ast}}_{\textbf{\underline{8}}}(7)=e_{32}$& \\

	&	
	$u^{\textbf{\underline{3}}\otimes\textbf{\underline{3}}^{\ast}}_{\textbf{\underline{8}}}(8)=\frac{1}{\sqrt{6}}e_{11}+\frac{1}{\sqrt{6}}e_{22}-\frac{2}{\sqrt{6}}e_{33}$ & \\
\lasthline
\end{tabular}
\caption{Clebsch-Gordan coefficients for $\textbf{\underline{3}}\otimes\textbf{\underline{3}}^{\ast}$ of $\Sigma(168)$.}
\label{Sigma168CGCb}
\end{center}
\end{table}

\subsection{The group $\Sigma(36\phi)$}\label{subsectionS36phi}
The generators of this group can be found in \cite{fairbairn} and \cite{miller}. Since Fairbairn et al. \cite{fairbairn} only work with the collineation groups of $\Sigma(n\phi)$, they do not give the character tables for $\Sigma(n\phi)$.
\medskip
\\
Fortunately there exists a computer algebra system, which is able to calculate character tables and irreducible representations of finite groups from their generators. This program is called \textit{GAP} \cite{gap}. (We already used \textit{GAP} in section \ref{icosahedral}.)
\medskip
\\
The generators of $\Sigma(36\phi)$ are \cite{miller}:
	\begin{displaymath}
		A=\left(\begin{matrix}
			 1 & 0 & 0 \\
			 0 & \omega & 0 \\
			 0 & 0 & \omega^2
		        \end{matrix}\right),\enspace
		B=\left(\begin{matrix}
			 0 & 1 & 0 \\
			 0 & 0 & 1 \\
			 1 & 0 & 0
		        \end{matrix}\right),\enspace
		C=\frac{1}{\omega-\omega^2}\left(\begin{matrix}
			 1 & 1 & 1 \\
			 1 & \omega & \omega^2 \\
			 1 & \omega^2 & \omega
		        \end{matrix}\right),
	\end{displaymath}
with $\omega=e^{\frac{2\pi i}{3}}$. The character table calculated with \textit{GAP} is shown in table \ref{sigma36charactertable}.

\begin{table}
\begin{footnotesize}
\begin{center}
\renewcommand{\arraystretch}{1.4}
\begin{tabular}{|l|ccccccc|}
\firsthline
	$\Sigma(36\phi)$ & $C_{1}(1)$ & $C_{2}(9)$ & $C_{3}(9)$ & $C_{4}(12)$ & $C_{5}(12)$ & $C_{6}(9)$ & $C_{7}(9)$\\
\hline
	$\textbf{\underline{1}}_{1}$ & $1$ & $1$ & $1$ & $1$ & $1$ & $1$ & $1$\\
	$\textbf{\underline{1}}_{2}$ & $1$ & $1$ & $1$ & $1$ & $1$ & $1$ & $-1$\\
	$\textbf{\underline{1}}_{3}$ & $1$ & $-1$ & $-1$ & $1$ & $1$ & $-1$ & $b$\\
	$\textbf{\underline{1}}_{4}$ & $1$ & $-1$ & $-1$ & $1$ & $1$ & $-1$ & $-b$\\
	$\textbf{\underline{3}}_{1}$ & $3$ & $a$ & $a^{\ast}$ & $0$ & $0$ & $-1$ & $a$\\
	$\textbf{\underline{3}}_{2}$ & $3$ & $a^{\ast}$ & $a$ & $0$ & $0$ & $-1$ & $a^{\ast}$\\
	$\textbf{\underline{3}}_{3}$ & $3$ & $a$ & $a^{\ast}$ & $0$ & $0$ & $-1$ & $-a$\\
	$\textbf{\underline{3}}_{4}$ & $3$ & $a^{\ast}$ & $a$ & $0$ & $0$ & $-1$ & $-a^{\ast}$\\
	$\textbf{\underline{3}}_{5}$ & $3$ & $-a^{\ast}$ & $-a$ & $0$ & $0$ & $1$ & $c$\\
	$\textbf{\underline{3}}_{6}$ & $3$ & $-a^{\ast}$ & $-a$ & $0$ & $0$ & $1$ & $-c$\\
	$\textbf{\underline{3}}_{7}$ & $3$ & $-a$ & $-a^{\ast}$ & $0$ & $0$ & $1$ & $-c^{\ast}$\\
	$\textbf{\underline{3}}_{8}$ & $3$ & $-a$ & $-a^{\ast}$ & $0$ & $0$ & $1$ & $c^{\ast}$\\
	$\textbf{\underline{4}}_{1}$ & $4$ & $0$ & $0$ & $-2$ & $1$ & $0$ & $0$\\
	$\textbf{\underline{4}}_{2}$ & $4$ & $0$ & $0$ & $1$ & $-2$ & $0$ & $0$\\
\hline
 & $C_{8}(9)$ & $C_{9}(9)$ & $C_{10}(9)$ & $C_{11}(9)$ & $C_{12}(1)$ & $C_{13}(9)$ & $C_{14}(1)$\\
\hline
	$\textbf{\underline{1}}_{1}$ & $1$ & $1$ & $1$ & $1$ & $1$ & $1$ & $1$\\
	$\textbf{\underline{1}}_{2}$ & $-1$ & $-1$ & $-1$ & $-1$ & $1$ & $-1$ & $1$\\
	$\textbf{\underline{1}}_{3}$ & $b$ & $-b$ & $-b$ & $-b$ & $1$ & $b$ & $1$\\
	$\textbf{\underline{1}}_{4}$ & $-b$ & $b$ & $b$ & $b$ & $1$ & $-b$ & $1$\\
	$\textbf{\underline{3}}_{1}$ & $-1$ & $a^{\ast}$ & $-1$ & $a$ & $d$ & $a^{\ast}$ & $d^{\ast}$\\
	$\textbf{\underline{3}}_{2}$ & $-1$ & $a$ & $-1$ & $a^{\ast}$ & $d^{\ast}$ & $a$ & $d$\\
	$\textbf{\underline{3}}_{3}$ & $1$ & $-a^{\ast}$ & $1$ & $-a$ & $d$ & $-a^{\ast}$ & $d^{\ast}$\\
	$\textbf{\underline{3}}_{4}$ & $1$ & $-a$ & $1$ & $-a^{\ast}$ & $d^{\ast}$ & $-a$ & $d$\\
	$\textbf{\underline{3}}_{5}$ & $b$ & $c^{\ast}$ & $-b$ & $-c$ & $d^{\ast}$ & $-c^{\ast}$ & $d$\\
	$\textbf{\underline{3}}_{6}$ & $-b$ & $-c^{\ast}$ & $b$ & $c$ & $d^{\ast}$ & $c^{\ast}$ & $d$\\
	$\textbf{\underline{3}}_{7}$ & $b$ & $-c$ & $-b$ & $c^{\ast}$ & $d$ & $c$ & $d^{\ast}$\\
	$\textbf{\underline{3}}_{8}$ & $-b$ & $c$ & $b$ & $-c^{\ast}$ & $d$ & $-c$ & $d^{\ast}$\\
	$\textbf{\underline{4}}_{1}$ & $0$ & $0$ & $0$ & $0$ & $4$ & $0$ & $4$\\
	$\textbf{\underline{4}}_{2}$ & $0$ & $0$ & $0$ & $0$ & $4$ & $0$ & $4$\\
\lasthline
\end{tabular}
\caption[The character table of $\Sigma(36\phi)$.]{The character table of $\Sigma(36\phi)$ as calculated with \textit{GAP}.
$a=-\omega^2$, $b=-i$, $c=-e^{\frac{7\pi i}{6}}$, $d=3\omega^2$.}
\label{sigma36charactertable}
\end{center}
\end{footnotesize}
\end{table}
\hspace{0mm}\\
Using the character table we can calculate the tensor products of the three-dimensional representations. The results are shown in table \ref{sigma36tensor}.

\begin{table}
\begin{small}
\begin{center}
\renewcommand{\arraystretch}{1.4}
\begin{tabular}{|l@{\hspace{10mm}}l|}
\firsthline
	$\textbf{\underline{3}}_{1} \otimes \textbf{\underline{3}}_{1}=\textbf{\underline{3}}_{4}\oplus\textbf{\underline{3}}_{5}\oplus\textbf{\underline{3}}_{6}$ & $\textbf{\underline{3}}_{3} \otimes \textbf{\underline{3}}_{6}=\textbf{\underline{1}}_{4}\oplus\textbf{\underline{4}}_{1}\oplus\textbf{\underline{4}}_{2}$\\
	$\textbf{\underline{3}}_{1} \otimes \textbf{\underline{3}}_{2}=\textbf{\underline{1}}_{1}\oplus\textbf{\underline{4}}_{1}\oplus\textbf{\underline{4}}_{2}$ & $\textbf{\underline{3}}_{3} \otimes \textbf{\underline{3}}_{7}=\textbf{\underline{3}}_{2}\oplus\textbf{\underline{3}}_{4}\oplus\textbf{\underline{3}}_{5}$\\
	$\textbf{\underline{3}}_{1} \otimes \textbf{\underline{3}}_{3}=\textbf{\underline{3}}_{2}\oplus\textbf{\underline{3}}_{5}\oplus\textbf{\underline{3}}_{6}$ & $\textbf{\underline{3}}_{3} \otimes \textbf{\underline{3}}_{8}=\textbf{\underline{3}}_{2}\oplus\textbf{\underline{3}}_{4}\oplus\textbf{\underline{3}}_{6}$\\
	$\textbf{\underline{3}}_{1} \otimes \textbf{\underline{3}}_{4}=\textbf{\underline{1}}_{2}\oplus\textbf{\underline{4}}_{1}\oplus\textbf{\underline{4}}_{2}$ & $\textbf{\underline{3}}_{4} \otimes \textbf{\underline{3}}_{4}=\textbf{\underline{3}}_{3}\oplus\textbf{\underline{3}}_{7}\oplus\textbf{\underline{3}}_{8}$\\
	$\textbf{\underline{3}}_{1} \otimes \textbf{\underline{3}}_{5}=\textbf{\underline{1}}_{4}\oplus\textbf{\underline{4}}_{1}\oplus\textbf{\underline{4}}_{2}$ & $\textbf{\underline{3}}_{4} \otimes \textbf{\underline{3}}_{5}=\textbf{\underline{3}}_{1}\oplus\textbf{\underline{3}}_{3}\oplus\textbf{\underline{3}}_{7}$\\
	$\textbf{\underline{3}}_{1} \otimes \textbf{\underline{3}}_{6}=\textbf{\underline{1}}_{3}\oplus\textbf{\underline{4}}_{1}\oplus\textbf{\underline{4}}_{2}$ & $\textbf{\underline{3}}_{4} \otimes \textbf{\underline{3}}_{6}=\textbf{\underline{3}}_{1}\oplus\textbf{\underline{3}}_{3}\oplus\textbf{\underline{3}}_{8}$\\
	$\textbf{\underline{3}}_{1} \otimes \textbf{\underline{3}}_{7}=\textbf{\underline{3}}_{2}\oplus\textbf{\underline{3}}_{4}\oplus\textbf{\underline{3}}_{6}$ & $\textbf{\underline{3}}_{4} \otimes \textbf{\underline{3}}_{7}=\textbf{\underline{1}}_{3}\oplus\textbf{\underline{4}}_{1}\oplus\textbf{\underline{4}}_{2}$\\
	$\textbf{\underline{3}}_{1} \otimes \textbf{\underline{3}}_{8}=\textbf{\underline{3}}_{2}\oplus\textbf{\underline{3}}_{4}\oplus\textbf{\underline{3}}_{5}$ & $\textbf{\underline{3}}_{4} \otimes \textbf{\underline{3}}_{8}=\textbf{\underline{1}}_{4}\oplus\textbf{\underline{4}}_{1}\oplus\textbf{\underline{4}}_{2}$\\
	$\textbf{\underline{3}}_{2} \otimes \textbf{\underline{3}}_{2}=\textbf{\underline{3}}_{3}\oplus\textbf{\underline{3}}_{7}\oplus\textbf{\underline{3}}_{8}$ & $\textbf{\underline{3}}_{5} \otimes \textbf{\underline{3}}_{5}=\textbf{\underline{3}}_{1}\oplus\textbf{\underline{3}}_{7}\oplus\textbf{\underline{3}}_{8}$\\
	$\textbf{\underline{3}}_{2} \otimes \textbf{\underline{3}}_{3}=\textbf{\underline{1}}_{2}\oplus\textbf{\underline{4}}_{1}\oplus\textbf{\underline{4}}_{2}$ & $\textbf{\underline{3}}_{5} \otimes \textbf{\underline{3}}_{6}=\textbf{\underline{3}}_{3}\oplus\textbf{\underline{3}}_{7}\oplus\textbf{\underline{3}}_{8}$\\
	$\textbf{\underline{3}}_{2} \otimes \textbf{\underline{3}}_{4}=\textbf{\underline{3}}_{1}\oplus\textbf{\underline{3}}_{7}\oplus\textbf{\underline{3}}_{8}$ & $\textbf{\underline{3}}_{5} \otimes \textbf{\underline{3}}_{7}=\textbf{\underline{1}}_{2}\oplus\textbf{\underline{4}}_{1}\oplus\textbf{\underline{4}}_{2}$\\
	$\textbf{\underline{3}}_{2} \otimes \textbf{\underline{3}}_{5}=\textbf{\underline{3}}_{1}\oplus\textbf{\underline{3}}_{3}\oplus\textbf{\underline{3}}_{8}$ & $\textbf{\underline{3}}_{5} \otimes \textbf{\underline{3}}_{8}=\textbf{\underline{1}}_{1}\oplus\textbf{\underline{4}}_{1}\oplus\textbf{\underline{4}}_{2}$\\
	$\textbf{\underline{3}}_{2} \otimes \textbf{\underline{3}}_{6}=\textbf{\underline{3}}_{1}\oplus\textbf{\underline{3}}_{3}\oplus\textbf{\underline{3}}_{7}$ & $\textbf{\underline{3}}_{6} \otimes \textbf{\underline{3}}_{6}=\textbf{\underline{3}}_{1}\oplus\textbf{\underline{3}}_{7}\oplus\textbf{\underline{3}}_{8}$\\
	$\textbf{\underline{3}}_{2} \otimes \textbf{\underline{3}}_{7}=\textbf{\underline{1}}_{4}\oplus\textbf{\underline{4}}_{1}\oplus\textbf{\underline{4}}_{2}$ & $\textbf{\underline{3}}_{6} \otimes \textbf{\underline{3}}_{7}=\textbf{\underline{1}}_{1}\oplus\textbf{\underline{4}}_{1}\oplus\textbf{\underline{4}}_{2}$\\
	$\textbf{\underline{3}}_{2} \otimes \textbf{\underline{3}}_{8}=\textbf{\underline{1}}_{3}\oplus\textbf{\underline{4}}_{1}\oplus\textbf{\underline{4}}_{2}$ & $\textbf{\underline{3}}_{6} \otimes \textbf{\underline{3}}_{8}=\textbf{\underline{1}}_{2}\oplus\textbf{\underline{4}}_{1}\oplus\textbf{\underline{4}}_{2}$\\
	$\textbf{\underline{3}}_{3} \otimes \textbf{\underline{3}}_{3}=\textbf{\underline{3}}_{4}\oplus\textbf{\underline{3}}_{5}\oplus\textbf{\underline{3}}_{6}$ & $\textbf{\underline{3}}_{7} \otimes \textbf{\underline{3}}_{7}=\textbf{\underline{3}}_{2}\oplus\textbf{\underline{3}}_{5}\oplus\textbf{\underline{3}}_{6}$\\
	$\textbf{\underline{3}}_{3} \otimes \textbf{\underline{3}}_{4}=\textbf{\underline{1}}_{1}\oplus\textbf{\underline{4}}_{1}\oplus\textbf{\underline{4}}_{2}$ & $\textbf{\underline{3}}_{7} \otimes \textbf{\underline{3}}_{8}=\textbf{\underline{3}}_{4}\oplus\textbf{\underline{3}}_{5}\oplus\textbf{\underline{3}}_{6}$\\
	$\textbf{\underline{3}}_{3} \otimes \textbf{\underline{3}}_{5}=\textbf{\underline{1}}_{3}\oplus\textbf{\underline{4}}_{1}\oplus\textbf{\underline{4}}_{2}$ & $\textbf{\underline{3}}_{8} \otimes \textbf{\underline{3}}_{8}=\textbf{\underline{3}}_{2}\oplus\textbf{\underline{3}}_{5}\oplus\textbf{\underline{3}}_{6}$\\
\lasthline
\end{tabular}
\caption{Tensor products of the three-dimensional irreducible representations of $\Sigma(36\phi)$.}
\label{sigma36tensor}
\end{center}
\end{small}
\end{table}
\hspace{0mm}\\
Although the number of tensor products of three-dimensional irreducible representations of $\Sigma(36\phi)$ is very high, there are only three types of Clebsch-Gordan coefficients, which we will show now.
\medskip
\\
From the character table \ref{sigma36charactertable} we can calculate the relations between the inequivalent three-dimensional irreducible representations of $\Sigma(36\phi)$, which are shown in table \ref{sigma36tensorb}.

\begin{table}
\begin{center}
\renewcommand{\arraystretch}{1.4}
\begin{tabular}{|l@{\hspace{10mm}}l|}
\firsthline
	$\textbf{\underline{3}}_{3}=\textbf{\underline{1}}_{2}\otimes \textbf{\underline{3}}_{1}$ & $\textbf{\underline{3}}_{4}=\textbf{\underline{1}}_{2}\otimes \textbf{\underline{3}}_{2}$\\
	$\textbf{\underline{3}}_{7}=\textbf{\underline{1}}_{4}\otimes \textbf{\underline{3}}_{1}$ & $\textbf{\underline{3}}_{5}=\textbf{\underline{1}}_{4}\otimes \textbf{\underline{3}}_{2}$\\
	$\textbf{\underline{3}}_{8}=\textbf{\underline{1}}_{3}\otimes \textbf{\underline{3}}_{1}$ & $\textbf{\underline{3}}_{6}=\textbf{\underline{1}}_{3}\otimes \textbf{\underline{3}}_{2}$\\
\lasthline
\end{tabular}
\caption{Relations of the inequivalent three-dimensional irreducible representations of $\Sigma(36\phi)$.}
\label{sigma36tensorb}
\end{center}
\end{table}
\hspace{0mm}\\
Using the results shown in table \ref{sigma36tensorb} and applying lemma \ref{LSU313} we find that we only need to consider the tensor products
	\begin{displaymath}
	\begin{split}
	& \textbf{\underline{3}}_{1}\otimes \textbf{\underline{3}}_{1}= \textbf{\underline{3}}_{4}\oplus \textbf{\underline{3}}_{5}\oplus \textbf{\underline{3}}_{6},\\
	& \textbf{\underline{3}}_{1}\otimes \textbf{\underline{3}}_{2}= \textbf{\underline{1}}_{1}\oplus \textbf{\underline{4}}_{1}\oplus \textbf{\underline{4}}_{2},\\
	& \textbf{\underline{3}}_{2}\otimes \textbf{\underline{3}}_{2}= \textbf{\underline{3}}_{3}\oplus \textbf{\underline{3}}_{7}\oplus \textbf{\underline{3}}_{8}.
	\end{split}
	\end{displaymath}
Before we can start to calculate bases of the invariant subspaces we have to find out which representation listed in the character table the defining representation (given at the beginning of this subsection) corresponds to. We proceed in the following way:
	\begin{itemize}
	 \item We calculate all group elements in the defining representation. This can be done by multiplication of each generator with each, in the next step multiplying each then known group element with each, and so on, until all group elements are found. An example for a \textit{Mathematica 6} program implementing such an algorithm (used on the group $\Sigma(216\phi)$ ($\rightarrow$ subsection \ref{sigma216phisubsection})) can be found in appendix \ref{appendixD}.
	 \item We calculate the traces of all group elements to obtain the characters.
	 \item We compare the characters with those listed in the character table \ref{sigma36charactertable}.
	\end{itemize}
We find characters with the following multiplicities
	\begin{displaymath}
	\begin{split}
	& 1\times 3,\enspace 9\times a,\enspace 9\times a^{\ast},\enspace 18\times -a,\enspace18\times -a^{\ast},\\
	& 24\times 0,\enspace 18\times 1,\enspace 9\times -1,\enspace 1\times d,\enspace 1\times d^{\ast}.
	\end{split}
	\end{displaymath}
The only representations that allow these multiplicities of characters are $\textbf{\underline{3}}_{3}$ and $\textbf{\underline{3}}_{4}=\textbf{\underline{3}}_{3}^{\ast}$. From the character table alone one cannot decide whether $\textbf{\underline{3}}_{3}$ or $\textbf{\underline{3}}_{4}$ corresponds to the defining representation for the following reason: Consider the generators $A,B,C$ of the defining representation. In this representation we have:
	\begin{displaymath}
	A^{\ast}=A^{-1},\enspace B^{\ast}=B,\enspace C^{\ast}=C^{-1}.
	\end{displaymath}
Thus the defining representation and its complex conjugate representation, interpreted as matrix groups, consist of the same elements. Nevertheless the defining representation and its complex conjugate are not equivalent. Since the defining representation and its complex conjugate, interpreted as matrix groups, consist of the same elements, these two inequivalent representations cannot be distinguished from the character table.
\\
This can also be verified from a different point of view. Let $a$ be an element of a matrix group $G$ that fulfils
	\begin{displaymath}
	g\in G\Rightarrow g^{\ast}\in G\quad\forall g\in G.
	\end{displaymath}
Then the conjugate classes $C_{a}$ and $C_{a^{\ast}}$ fulfil
	\begin{displaymath}
	(C_{a})^{\ast}=C_{a^{\ast}}.
	\end{displaymath}
This is true, because $bab^{-1}=a'$ $\Rightarrow$ $b^{\ast}a^{\ast}(b^{\ast})^{-1}=a'\hspace{0mm}^{\ast}$. Thus the character table remains invariant under interchange of the irreducible representations and their complex conjugates, if at the same time the conjugate classes and their complex conjugates are interchanged.\footnote{Clearly one has to interchange all irreducible representations with their complex conjugates, because the interchange of the conjugate classes affects all representations.}
\medskip
\\
Therefore we could choose the defining representation to be equivalent to $\textbf{\underline{3}}_{3}$ (and with the same right we could choose that it shall be equivalent to $\textbf{\underline{3}}_4$). Furthermore we need generators for the one-dimensional irreducible representations, and we use \textit{GAP} to obtain them:
	\begin{displaymath}
	\begin{split}
	& \textbf{\underline{1}}_1: A\mapsto 1,\enspace B\mapsto 1,\enspace C\mapsto 1,\\
	& \textbf{\underline{1}}_2: A\mapsto 1,\enspace B\mapsto 1,\enspace C\mapsto -1,\\
	& \textbf{\underline{1}}_3: A\mapsto 1,\enspace B\mapsto 1,\enspace C\mapsto i,\\
	& \textbf{\underline{1}}_4: A\mapsto 1,\enspace B\mapsto 1,\enspace C\mapsto -i.
	\end{split}
	\end{displaymath}
Conjecture: Due to the choice of generators for the $1$-dimensional representations, the freedom of choice for the defining representation being equivalent to $\textbf{\underline{3}}_3$ or $\textbf{\underline{3}}_4$ is no longer given. The reason for the disappearance of this freedom is the following:
\\
We have argued that the character table remains invariant under an interchange of all irreducible representations with their complex conjugates, if at the same time all conjugate classes are interchanged with their complex conjugates. Due to the choice of the $1$-dimensional irreducible representations, $\textbf{\underline{1}}_3$ and $\textbf{\underline{1}}_4=\textbf{\underline{1}}_3^{\ast}$ cannot be interchanged any longer, and thus also $\textbf{\underline{3}}_3$ and $\textbf{\underline{3}}_4$ should be fixed.
\bigskip
\\
To find out which representation the defining representation corresponds to, we investigate the following matrix representations constructed with \textit{GAP}\footnote{Since $GAP$ was used to construct the $1$-dimensional representations too, there is no danger that the representations we will work with could be inconsistent.}:
	\begin{displaymath}
	\begin{split}
	& \textbf{\underline{3}}_3: A\mapsto
	\left(
	\begin{matrix}
	 \omega^2 & \omega^2 & -1 \\
	 0 & -\omega^2 & -\omega \\
	 0 & 1 & 0
	\end{matrix}
	\right),\enspace
	B\mapsto
	\left(
	\begin{matrix}
	 -\omega & -\omega & 0 \\
	 \omega & 0 & 0 \\
	 \omega & -1 & \omega
	\end{matrix}
	\right),\enspace
	C\mapsto
	\left(
	\begin{matrix}
	 1 & 1 & -\omega \\
	 \omega^2 & 0 & 0 \\
	 -1 & -1 & 0
	\end{matrix}
	\right),\\
	& \textbf{\underline{3}}_4: A\mapsto
	\left(
	\begin{matrix}
	 \omega & -\omega & \omega \\
	 0 & 0 & \omega^2 \\
	 0 & -1 & -\omega
	\end{matrix}
	\right),\enspace
	B\mapsto
	\left(
	\begin{matrix}
	 -\omega^2 & -1 & 0 \\
	 \omega & 0 & 0 \\
	 \omega^2 & -\omega^2 & \omega^2
	\end{matrix}
	\right),\enspace
	C\mapsto
	\left(
	\begin{matrix}
	 1 & -1 & 1 \\
	 1 & 0 & 0 \\
	 0 & 1 & 0
	\end{matrix}
	\right).
	\end{split}
	\end{displaymath}
Let $[D]$ denote the defining matrix representation given at the beginning of this subsection, and let $[\textbf{\underline{3}}_3]$ and $[\textbf{\underline{3}}_4]$ denote the matrix representations calculated with $GAP$. In subsection \ref{icosahedral} we developed an algorithm to solve the equation
	\begin{displaymath}
	S^{-1}\textbf{\underline{5}}(f)S=\hat{\textbf{\underline{5}}}(f)\quad\forall f\in \mathrm{gen}(A_{5}).
	\end{displaymath}
Using the same algorithm here we find
	\begin{displaymath}
	S^{-1}[D(f)]S=[\textbf{\underline{3}}_4(f)]\quad\quad f=A,B,C
	\end{displaymath}
with
	\begin{displaymath}
	S=\left(\begin{matrix}
	 0 & -\frac{1}{3} & \frac{1}{3} \\
	 -\frac{i}{\sqrt{3}} & \frac{1+i\sqrt{3}}{6} & -\frac{1+i\sqrt{3}}{6} \\
	 0 & \frac{1+i\sqrt{3}}{6} & \frac{-1+i\sqrt{3}}{6}
	  \end{matrix}\right).
	\end{displaymath}
Thus the defining representation must be equivalent to $\textbf{\underline{3}}_4$ (the algorithm does not find solutions of $S^{-1}[D(f)]S=[\textbf{\underline{3}}_3(f)]$). We will therefore use the matrices $A,B,C$ as (unitary) generators of the representation $\textbf{\underline{3}}_4$.
\medskip
\\
We can now calculate all other unitary generators of three-dimensional representations via
	\begin{center}
	$
	\begin{array}{l}
	\textbf{\underline{3}}_{1}=\textbf{\underline{1}}_{2}\otimes \textbf{\underline{3}}_{3},\\
	\textbf{\underline{3}}_{7}=\textbf{\underline{1}}_{3}\otimes \textbf{\underline{3}}_{3},\\
	\textbf{\underline{3}}_{8}=\textbf{\underline{1}}_{4}\otimes \textbf{\underline{3}}_{3},\\
	\textbf{\underline{3}}_{4}=\textbf{\underline{3}}_{3}^{\ast},\\
	\textbf{\underline{3}}_{2}=\textbf{\underline{1}}_{2}\otimes \textbf{\underline{3}}_{4},\\
	\textbf{\underline{3}}_{5}=\textbf{\underline{1}}_{3}\otimes \textbf{\underline{3}}_{4},\\
	\textbf{\underline{3}}_{6}=\textbf{\underline{1}}_{4}\otimes \textbf{\underline{3}}_{4},
	\end{array}
	$
	\end{center}
using the $1$-dimensional irreducible representations constructed with \textit{GAP}. Furthermore we need unitary generators of $\textbf{\underline{4}}_{1}$ and $\textbf{\underline{4}}_{2}$. \textit{GAP} can also calculate generators for these irreducible representations. The generators for the 4-dimensional irreducible representations given by \textit{GAP} are:
	 \begin{displaymath}
	 \begin{split}
	 \textbf{\underline{4}}_{1}:\enspace & A\mapsto 
	 \left(\begin{matrix}
	 \omega & 0 & 0 & 0 \\
	 0 & \omega^{2} & 0 & 0 \\
	 0 & 0 & \omega & 0 \\
	 0 & 0 & 0 & \omega^{2}
	\end{matrix}\right),\quad B\mapsto
	\left(\begin{matrix}
	 \omega & 0 & 0 & 0 \\
	 0 & \omega & 0 & 0 \\
	 0 & 0 & \omega^{2} & 0 \\
	 0 & 0 & 0 & \omega^{2}
	\end{matrix}\right),\\
	& C\mapsto 
	\left(\begin{matrix}
	 0 & 0 & \omega & 0 \\
	 \omega^{2} & 0 & 0 & 0 \\
	 0 & 0 & 0 & \omega^{2} \\
	 0 & \omega & 0 & 0
	\end{matrix}\right),\\
	\textbf{\underline{4}}_{2}:\enspace & A\mapsto
	\left(\begin{matrix}
	 \omega^{2} & 0 & 0 & 0 \\
	 0 & 1 & 0 & 0 \\
	 0 & 0 & \omega & 0 \\
	 0 & 0 & 0 & 1
	\end{matrix}\right),\enspace
	B\mapsto 
	\left(\begin{matrix}
	 1 & 0 & 0 & 0 \\
	 0 & \omega^{2} & 0 & 0 \\
	 0 & 0 & 1 & 0 \\
	 0 & 0 & 0 & \omega
	\end{matrix}\right),\\
	& C\mapsto
	\left(\begin{matrix}
	 0 & 0 & 0 & 1 \\
	 1 & 0 & 0 & 0 \\
	 0 & \omega^{2} & 0 & 0 \\
	 0 & 0 & \omega & 0
	\end{matrix}\right).
	 \end{split}
	 \end{displaymath}
In tables \ref{Sigma36CGCa} and \ref{Sigma36CGCb} we list the Clebsch-Gordan coefficients of the tensor products of three-dimensional irreducible representations of $\Sigma(36\phi)$. They were calculated using a \textit{Mathematica 6} program realizing the algorithm described in chapter \ref{chapterclebsch}. An example for this program can be found in appendix \ref{appendixD}. Since the Clebsch-Gordan coefficients for $\textbf{\underline{3}}_{1}\otimes\textbf{\underline{3}}_{1}$ are real, lemma \ref{LSU314} tells us that the coefficients of $\textbf{\underline{3}}_{2}\otimes\textbf{\underline{3}}_{2}=\textbf{\underline{3}}_{1}^{\ast}\otimes\textbf{\underline{3}}_{1}^{\ast}$ and $\textbf{\underline{3}}_{1}\otimes\textbf{\underline{3}}_{1}$ are equal. The Clebsch-Gordan coefficients for $\textbf{\underline{3}}_{2}\otimes \textbf{\underline{3}}_{1}$ can again be obtained by replacing $e_{ij}\mapsto e_{ji}$ in table \ref{Sigma36CGCb} ($\rightarrow$ proposition \ref{PY11}).

\begin{table}
\begin{center}
\renewcommand{\arraystretch}{1.4}
\begin{tabular}{|l|l|l|}
\firsthline
$\Sigma(36\phi)$
  & $\textbf{\underline{3}}_1\otimes \textbf{\underline{3}}_1$ & CGC\\
\hline
 $\textbf{\underline{3}}_4$ & $u_{\textbf{\underline{3}}_{4}}^{\textbf{\underline{3}}_{1}\otimes \textbf{\underline{3}}_{1}}\text{(1)=}-\frac{1}{\sqrt{2}}e_{23} +\frac{1}{\sqrt{2}}e_{32}$ & \textbf{IIIa$_{(1,-1)}$}$^{(\ast)}$\\
  & $u_{\textbf{\underline{3}}_{4}}^{\textbf{\underline{3}}_{1}\otimes \textbf{\underline{3}}_{1}}\text{(2)=}\frac{1}{\sqrt{2}}e_{13} - \frac{1}{\sqrt{2}}e_{31}$ &\\
  & $u_{\textbf{\underline{3}}_{4}}^{\textbf{\underline{3}}_{1}\otimes \textbf{\underline{3}}_{1}}\text{(3)=}-\frac{1}{\sqrt{2}}e_{12} +\frac{1}{\sqrt{2}}e_{21}$ &\\
 $\textbf{\underline{3}}_5$ & $u_{\textbf{\underline{3}}_{5}}^{\textbf{\underline{3}}_{1}\otimes \textbf{\underline{3}}_{1}}\text{(1)=}-\frac{\tau_{-}}{\sqrt{12}}e_{11} + \frac{1}{\tau_{-}}e_{23} + \frac{1}{\tau_{-}}e_{32}$ & \textbf{IIIb}\\
  & $u_{\textbf{\underline{3}}_{5}}^{\textbf{\underline{3}}_{1}\otimes \textbf{\underline{3}}_{1}}\text{(2)=}\frac{1}{\tau_{-}}e_{13} -\frac{\tau_{-}}{\sqrt{12}}e_{22} + \frac{1}{\tau_{-}}e_{31}$ &\\
  & $u_{\textbf{\underline{3}}_{5}}^{\textbf{\underline{3}}_{1}\otimes \textbf{\underline{3}}_{1}}\text{(3)=}\frac{1}{\tau_{-}}e_{12} + \frac{1}{\tau_{-}}e_{21} -\frac{\tau_{-}}{\sqrt{12}}e_{33}$ &\\
 $\textbf{\underline{3}}_6$ & $u_{\textbf{\underline{3}}_{6}}^{\textbf{\underline{3}}_{1}\otimes \textbf{\underline{3}}_{1}}\text{(1)=}\frac{\tau_{+}}{\sqrt{12}}e_{11} + \frac{1}{\tau_{+}}e_{23} + \frac{1}{\tau_{+}}e_{32}$ & \textbf{IIIc}\\
  & $u_{\textbf{\underline{3}}_{6}}^{\textbf{\underline{3}}_{1}\otimes \textbf{\underline{3}}_{1}}\text{(2)=}\frac{1}{\tau_{+}}e_{13} + \frac{\tau_{+}}{\sqrt{12}}e_{22} + \frac{1}{\tau_{+}}e_{31}$ &\\
  & $u_{\textbf{\underline{3}}_{6}}^{\textbf{\underline{3}}_{1}\otimes \textbf{\underline{3}}_{1}}\text{(3)=}\frac{1}{\tau_{+}}e_{12} + \frac{1}{\tau_{+}}e_{21} + \frac{\tau_{+}}{\sqrt{12}}e_{33}$ &
\\
\lasthline
\end{tabular}
\caption[Clebsch-Gordan coefficients for $\textbf{\underline{3}}_{1}\otimes\textbf{\underline{3}}_{1}$ of $\Sigma(36\phi)$.]{Clebsch-Gordan coefficients for $\textbf{\underline{3}}_{1}\otimes\textbf{\underline{3}}_{1}$ of $\Sigma(36\phi)$. We use the abbreviation $\tau_{\pm}=\sqrt{2(3\pm \sqrt{3})}$. ($(\ast)$...up to irrelevant phase factors.)}
\label{Sigma36CGCa}
\end{center}
\end{table}

\begin{table}
\begin{center}
\renewcommand{\arraystretch}{1.4}
\begin{tabular}{|l|l|l|}
\firsthline
$\Sigma(36\phi)$ & $\textbf{\underline{3}}_1\otimes \textbf{\underline{3}}_2$ & CGC\\
\hline
$\textbf{\underline{1}}_1$ & $u_{\textbf{\underline{1}}_{1}}^{\textbf{\underline{3}}_{1}\otimes\textbf{\underline{3}}_{2}}=\frac{1}{\sqrt{3}}e_{11} + \frac{1}{\sqrt{3}}e_{22} +\frac{1}{\sqrt{3}}e_{33}$ & \textbf{Ia}\\

$\textbf{\underline{4}}_1$ & $u_{\textbf{\underline{4}}_1}^{\textbf{\underline{3}}_{1}\otimes\textbf{\underline{3}}_{2}}(1)=\frac{\omega^{2}}{\sqrt{3}}e_{12} + \frac{1}{\sqrt{3}}e_{23} + \frac{\omega}{\sqrt{3}}e_{31}$ & \textbf{IVb}
\\

& $u_{\textbf{\underline{4}}_1}^{\textbf{\underline{3}}_{1}\otimes\textbf{\underline{3}}_{2}}(2)=\frac{\omega^2}{\sqrt{3}}e_{13} + \frac{1}{\sqrt{3}}e_{21} + \frac{\omega}{\sqrt{3}}e_{32}$ &
\\

& $u_{\textbf{\underline{4}}_1}^{\textbf{\underline{3}}_{1}\otimes\textbf{\underline{3}}_{2}}(3)=\frac{\omega^2}{\sqrt{3}}e_{12} + \frac{\omega}{\sqrt{3}}e_{23} + \frac{1}{\sqrt{3}}e_{31}$ &
\\

& $u_{\textbf{\underline{4}}_1}^{\textbf{\underline{3}}_{1}\otimes\textbf{\underline{3}}_{2}}(4)=\frac{\omega^{2}}{\sqrt{3}}e_{13} + \frac{\omega}{\sqrt{3}}e_{21} + \frac{1}{\sqrt{3}}e_{32}$ &
\\

$\textbf{\underline{4}}_2$ & $u_{\textbf{\underline{4}}_2}^{\textbf{\underline{3}}_{1}\otimes\textbf{\underline{3}}_{2}}(1)=\frac{\omega}{\sqrt{3}}e_{13} + \frac{\omega}{\sqrt{3}}e_{21} + \frac{\omega}{\sqrt{3}}e_{32}$ & \textbf{IVc}\\

& $u_{\textbf{\underline{4}}_2}^{\textbf{\underline{3}}_{1}\otimes\textbf{\underline{3}}_{2}}(2)=\frac{\omega}{\sqrt{3}}e_{11} + \frac{1}{\sqrt{3}}e_{22} + \frac{\omega^2}{\sqrt{3}}e_{33}$ &\\

& $u_{\textbf{\underline{4}}_2}^{\textbf{\underline{3}}_{1}\otimes\textbf{\underline{3}}_{2}}(3)=\frac{\omega^2}{\sqrt{3}}e_{12} + \frac{\omega^2}{\sqrt{3}}e_{23} + \frac{\omega^2}{\sqrt{3}}e_{31}$ &
\\

& $u_{\textbf{\underline{4}}_2}^{\textbf{\underline{3}}_{1}\otimes\textbf{\underline{3}}_{2}}(4)=\frac{\omega}{\sqrt{3}}e_{11} + \frac{\omega^2}{\sqrt{3}}e_{22} + \frac{1}{\sqrt{3}}e_{33}$ & \\
\lasthline
\end{tabular}
\caption[Clebsch-Gordan coefficients for $\textbf{\underline{3}}_{1}\otimes\textbf{\underline{3}}_{2}$ of $\Sigma(36\phi)$.]{Clebsch-Gordan coefficients for $\textbf{\underline{3}}_{1}\otimes\textbf{\underline{3}}_{2}$ of $\Sigma(36\phi)$. ($\omega=e^{\frac{2\pi i}{3}}$)}
\label{Sigma36CGCb}
\end{center}
\end{table}

\subsection{The group $\Sigma(72\phi)$}
The generators of $\Sigma(72\phi)$ are \cite{miller}:
	\begin{displaymath}
		A=\left(\begin{matrix}
			 1 & 0 & 0 \\
			 0 & \omega & 0 \\
			 0 & 0 & \omega^2
		        \end{matrix}\right),\enspace
		B=\left(\begin{matrix}
			 0 & 1 & 0 \\
			 0 & 0 & 1 \\
			 1 & 0 & 0
		        \end{matrix}\right),\enspace
		C=\frac{1}{\omega-\omega^2}\left(\begin{matrix}
			 1 & 1 & 1 \\
			 1 & \omega & \omega^2 \\
			 1 & \omega^2 & \omega
		        \end{matrix}\right),\enspace
	\end{displaymath}
and
	\begin{displaymath}
		D=\frac{1}{\omega-\omega^2}\left(\begin{matrix}
			 1 & 1 & \omega^{2} \\
			 1 & \omega & \omega \\
			 \omega & 1 & \omega
		        \end{matrix}\right),
	\end{displaymath}
with $\omega=e^{\frac{2\pi i}{3}}$, so $\Sigma(36\phi)$ must be a subgroup of $\Sigma(72\phi)$. The character table calculated with \textit{GAP} is shown in table \ref{Sigma72charactertable}.

\begin{table}
\begin{scriptsize}
\begin{center}
\renewcommand{\arraystretch}{1.4}
\begin{tabular}{|l|cccccccc|}
\firsthline
	$\Sigma(72\phi)$ & $C_{1}(1)$ & $C_{2}(9)$ & $C_{3}(9)$ & $C_{4}(24)$ & $C_{5}(9)$ & $C_{6}(18)$ & $C_{7}(18)$ & $C_{8}(18)$\\
\hline
	$\textbf{\underline{1}}_{1}$ & $1$ & $1$ & $1$ & $1$ & $1$ & $1$ & $1$ & $1$\\
	$\textbf{\underline{1}}_{2}$ & $1$ & $1$ & $1$ & $1$ & $1$ & $-1$ & $-1$ & $-1$\\
	$\textbf{\underline{1}}_{3}$ & $1$ & $1$ & $1$ & $1$ & $1$ & $-1$ & $-1$ & $1$\\
	$\textbf{\underline{1}}_{4}$ & $1$ & $1$ & $1$ & $1$ & $1$ & $1$ & $1$ & $-1$\\
	$\textbf{\underline{2}}$ & $2$ & $-2$ & $-2$ & $2$ & $-2$ & $0$ & $0$ & $0$\\
	$\textbf{\underline{3}}_{1}$ & $3$ & $a$ & $a^{\ast}$ & $0$ & $-1$ & $a$ & $-1$ & $a$\\
	$\textbf{\underline{3}}_{2}$ & $3$ & $a^{\ast}$ & $a$ & $0$ & $-1$ & $a^{\ast}$ & $-1$ & $a^{\ast}$\\
	$\textbf{\underline{3}}_{3}$ & $3$ & $a$ & $a^{\ast}$ & $0$ & $-1$ & $a$ & $-1$ & $-a$\\
	$\textbf{\underline{3}}_{4}$ & $3$ & $a^{\ast}$ & $a$ & $0$ & $-1$ & $a^{\ast}$ & $-1$ & $-a^{\ast}$\\
	$\textbf{\underline{3}}_{5}$ & $3$ & $a$ & $a^{\ast}$ & $0$ & $-1$ & $-a$ & $1$ & $a$\\
	$\textbf{\underline{3}}_{6}$ & $3$ & $a^{\ast}$ & $a$ & $0$ & $-1$ & $-a^{\ast}$ & $1$ & $a^{\ast}$\\
	$\textbf{\underline{3}}_{7}$ & $3$ & $a$ & $a^{\ast}$ & $0$ & $-1$ & $-a$ & $1$ & $-a$\\
	$\textbf{\underline{3}}_{8}$ & $3$ & $a^{\ast}$ & $a$ & $0$ & $-1$ & $-a^{\ast}$ & $1$ & $-a^{\ast}$\\
	$\textbf{\underline{6}}_{1}$ & $6$ & $b$ & $b^{\ast}$ & $0$ & $2$ & $0$ & $0$ & $0$\\
	$\textbf{\underline{6}}_{2}$ & $6$ & $b^{\ast}$ & $b$ & $0$ & $2$ & $0$ & $0$ & $0$\\
	$\textbf{\underline{8}}$ & $8$ & $0$ & $0$ & $-1$ & $0$ & $0$ & $0$ & $0$\\
\hline
 & $C_{9}(18)$ & $C_{10}(18)$ & $C_{11}(18)$ & $C_{12}(18)$ & $C_{13}(18)$ & $C_{14}(18)$ & $C_{15}(1)$ & $C_{16}(1)$\\
\hline
	$\textbf{\underline{1}}_{1}$ & $1$ & $1$ & $1$ & $1$ & $1$ & $1$ & $1$ & $1$\\
	$\textbf{\underline{1}}_{2}$ & $-1$ & $1$ & $1$ & $-1$ & $1$ & $-1$ & $1$ & $1$\\
	$\textbf{\underline{1}}_{3}$ & $1$ & $-1$ & $-1$ & $-1$ & $-1$ & $1$ & $1$ & $1$\\
	$\textbf{\underline{1}}_{4}$ & $-1$ & $-1$ & $-1$ & $1$ & $-1$ & $-1$ & $1$ & $1$\\
	$\textbf{\underline{2}}$ & $0$ & $0$ & $0$ & $0$ & $0$ & $0$ & $2$ & $2$\\
	$\textbf{\underline{3}}_{1}$ & $-1$ & $-a$ & $1$ & $a^{\ast}$ & $-a^{\ast}$ & $a^{\ast}$ & $c$ & $c^{\ast}$\\
	$\textbf{\underline{3}}_{2}$ & $-1$ & $-a^{\ast}$ & $1$ & $a$ & $-a$ & $a$ & $c^{\ast}$ & $c$\\
	$\textbf{\underline{3}}_{3}$ & $1$ & $a$ & $-1$ & $a^{\ast}$ & $a^{\ast}$ & $-a^{\ast}$ & $c$ & $c^{\ast}$\\
	$\textbf{\underline{3}}_{4}$ & $1$ & $a^{\ast}$ & $-1$ & $a$ & $a$ & $-a$ & $c^{\ast}$ & $c$\\
	$\textbf{\underline{3}}_{5}$ & $-1$ & $a$ & $-1$ & $-a^{\ast}$ & $a^{\ast}$ & $a^{\ast}$ & $c$ & $c^{\ast}$\\
	$\textbf{\underline{3}}_{6}$ & $-1$ & $a^{\ast}$ & $-1$ & $-a$ & $a$ & $a$ & $c^{\ast}$ & $c$\\
	$\textbf{\underline{3}}_{7}$ & $1$ & $-a$ & $1$ & $-a^{\ast}$ & $-a^{\ast}$ & $-a^{\ast}$ & $c$ & $c^{\ast}$\\
	$\textbf{\underline{3}}_{8}$ & $1$ & $-a^{\ast}$ & $1$ & $-a$ & $-a$ & $-a$ & $c^{\ast}$ & $c$\\
	$\textbf{\underline{6}}_{1}$ & $0$ & $0$ & $0$ & $0$ & $0$ & $0$ & $d$ & $d^{\ast}$\\
	$\textbf{\underline{6}}_{2}$ & $0$ & $0$ & $0$ & $0$ & $0$ & $0$ & $d^{\ast}$ & $d$\\
	$\textbf{\underline{8}}$ & $0$ & $0$ & $0$ & $0$ & $0$ & $0$ & $8$ & $8$\\
\lasthline
\end{tabular}
\caption[The character table of $\Sigma(72\phi)$.]{The character table of $\Sigma(72\phi)$ as calculated with \textit{GAP}.
$a=-\omega^2$, $b=2\omega^{2}$, $c=3\omega^2$, $d=6\omega^2$.}
\label{Sigma72charactertable}
\end{center}
\end{scriptsize}
\end{table}
\hspace{0mm}\\
Using the character table we can calculate the tensor products of the three-dimensional representations. The results are shown in \ref{sigma72tensor}.

\begin{table}
\begin{small}
\begin{center}
\renewcommand{\arraystretch}{1.4}
\begin{tabular}{|l@{\hspace{10mm}}l|}
\firsthline
	$\textbf{\underline{3}}_{1} \otimes \textbf{\underline{3}}_{1}=\textbf{\underline{3}}_{8}\oplus\textbf{\underline{6}}_{2}$ & $\textbf{\underline{3}}_{3} \otimes \textbf{\underline{3}}_{6}=\textbf{\underline{1}}_{2}\oplus\textbf{\underline{8}}$ \\
	$\textbf{\underline{3}}_{1} \otimes \textbf{\underline{3}}_{2}=\textbf{\underline{1}}_{1}\oplus\textbf{\underline{8}}$ & $\textbf{\underline{3}}_{3} \otimes \textbf{\underline{3}}_{7}=\textbf{\underline{3}}_{4}\oplus\textbf{\underline{6}}_{2}$\\
	$\textbf{\underline{3}}_{1} \otimes \textbf{\underline{3}}_{3}=\textbf{\underline{3}}_{6}\oplus\textbf{\underline{6}}_{2}$ & $\textbf{\underline{3}}_{3} \otimes \textbf{\underline{3}}_{8}=\textbf{\underline{1}}_{3}\oplus\textbf{\underline{8}}$\\
	$\textbf{\underline{3}}_{1} \otimes \textbf{\underline{3}}_{4}=\textbf{\underline{1}}_{4}\oplus\textbf{\underline{8}}$ & $\textbf{\underline{3}}_{4} \otimes \textbf{\underline{3}}_{4}=\textbf{\underline{3}}_{7}\oplus\textbf{\underline{6}}_{1}$\\
	$\textbf{\underline{3}}_{1} \otimes \textbf{\underline{3}}_{5}=\textbf{\underline{3}}_{4}\oplus\textbf{\underline{6}}_{2}$ & $\textbf{\underline{3}}_{4} \otimes \textbf{\underline{3}}_{5}=\textbf{\underline{1}}_{2}\oplus\textbf{\underline{8}}$\\
	$\textbf{\underline{3}}_{1} \otimes \textbf{\underline{3}}_{6}=\textbf{\underline{1}}_{3}\oplus\textbf{\underline{8}}$ & $\textbf{\underline{3}}_{4} \otimes \textbf{\underline{3}}_{6}=\textbf{\underline{3}}_{1}\oplus\textbf{\underline{6}}_{1}$\\
	$\textbf{\underline{3}}_{1} \otimes \textbf{\underline{3}}_{7}=\textbf{\underline{3}}_{2}\oplus\textbf{\underline{6}}_{2}$ & $\textbf{\underline{3}}_{4} \otimes \textbf{\underline{3}}_{7}=\textbf{\underline{1}}_{3}\oplus\textbf{\underline{8}}$\\
	$\textbf{\underline{3}}_{1} \otimes \textbf{\underline{3}}_{8}=\textbf{\underline{1}}_{2}\oplus\textbf{\underline{8}}$ & $\textbf{\underline{3}}_{4} \otimes \textbf{\underline{3}}_{8}=\textbf{\underline{3}}_{3}\oplus\textbf{\underline{6}}_{1}$\\
	$\textbf{\underline{3}}_{2} \otimes \textbf{\underline{3}}_{2}=\textbf{\underline{3}}_{7}\oplus\textbf{\underline{6}}_{1}$ & $\textbf{\underline{3}}_{5} \otimes \textbf{\underline{3}}_{5}=\textbf{\underline{3}}_{8}\oplus\textbf{\underline{6}}_{2}$\\
	$\textbf{\underline{3}}_{2} \otimes \textbf{\underline{3}}_{3}=\textbf{\underline{1}}_{4}\oplus\textbf{\underline{8}}$ & $\textbf{\underline{3}}_{5} \otimes \textbf{\underline{3}}_{6}=\textbf{\underline{1}}_{1}\oplus\textbf{\underline{8}}$\\
	$\textbf{\underline{3}}_{2} \otimes \textbf{\underline{3}}_{4}=\textbf{\underline{3}}_{5}\oplus\textbf{\underline{6}}_{1}$ & $\textbf{\underline{3}}_{5} \otimes \textbf{\underline{3}}_{7}=\textbf{\underline{3}}_{6}\oplus\textbf{\underline{6}}_{2}$\\
	$\textbf{\underline{3}}_{2} \otimes \textbf{\underline{3}}_{5}=\textbf{\underline{1}}_{3}\oplus\textbf{\underline{8}}$ & $\textbf{\underline{3}}_{5} \otimes \textbf{\underline{3}}_{8}=\textbf{\underline{1}}_{4}\oplus\textbf{\underline{8}}$\\
	$\textbf{\underline{3}}_{2} \otimes \textbf{\underline{3}}_{6}=\textbf{\underline{3}}_{3}\oplus\textbf{\underline{6}}_{1}$ & $\textbf{\underline{3}}_{6} \otimes \textbf{\underline{3}}_{6}=\textbf{\underline{3}}_{7}\oplus\textbf{\underline{6}}_{1}$\\
	$\textbf{\underline{3}}_{2} \otimes \textbf{\underline{3}}_{7}=\textbf{\underline{1}}_{2}\oplus\textbf{\underline{8}}$ & $\textbf{\underline{3}}_{6} \otimes \textbf{\underline{3}}_{7}=\textbf{\underline{1}}_{4}\oplus\textbf{\underline{8}}$\\
	$\textbf{\underline{3}}_{2} \otimes \textbf{\underline{3}}_{8}=\textbf{\underline{3}}_{1}\oplus\textbf{\underline{6}}_{1}$ & $\textbf{\underline{3}}_{6} \otimes \textbf{\underline{3}}_{8}=\textbf{\underline{3}}_{5}\oplus\textbf{\underline{6}}_{1}$\\
	$\textbf{\underline{3}}_{3} \otimes \textbf{\underline{3}}_{3}=\textbf{\underline{3}}_{8}\oplus\textbf{\underline{6}}_{2}$ & $\textbf{\underline{3}}_{7} \otimes \textbf{\underline{3}}_{7}=\textbf{\underline{3}}_{8}\oplus\textbf{\underline{6}}_{2}$\\
	$\textbf{\underline{3}}_{3} \otimes \textbf{\underline{3}}_{4}=\textbf{\underline{1}}_{1}\oplus\textbf{\underline{8}}$ & $\textbf{\underline{3}}_{7} \otimes \textbf{\underline{3}}_{8}=\textbf{\underline{1}}_{1}\oplus\textbf{\underline{8}}$\\
	$\textbf{\underline{3}}_{3} \otimes \textbf{\underline{3}}_{5}=\textbf{\underline{3}}_{2}\oplus\textbf{\underline{6}}_{2}$ & $\textbf{\underline{3}}_{8} \otimes \textbf{\underline{3}}_{8}=\textbf{\underline{3}}_{7}\oplus\textbf{\underline{6}}_{1}$\\
\lasthline
\end{tabular}
\caption{Tensor products of the three-dimensional irreducible representations of $\Sigma(72\phi)$.}
\label{sigma72tensor}
\end{center}
\end{small}
\end{table}
\hspace{0mm}\\
As in the case of $\Sigma(36\phi)$ all irreducible three-dimensional representations can be constructed from two representations, e.g. $\textbf{\underline{3}}_{1}$ and $\textbf{\underline{3}}_{2}=\textbf{\underline{3}}_{1}^{\ast}$, as is shown in table \ref{sigma72tensorb}.

\begin{table}
\begin{center}
\renewcommand{\arraystretch}{1.4}
\begin{tabular}{|l@{\hspace{10mm}}l|}
\firsthline
	$\textbf{\underline{3}}_{3}=\textbf{\underline{1}}_{4}\otimes \textbf{\underline{3}}_{1}$ & $\textbf{\underline{3}}_{4}=\textbf{\underline{1}}_{4}\otimes \textbf{\underline{3}}_{2}$\\
	$\textbf{\underline{3}}_{5}=\textbf{\underline{1}}_{3}\otimes \textbf{\underline{3}}_{1}$ &
	$\textbf{\underline{3}}_{6}=\textbf{\underline{1}}_{3}\otimes \textbf{\underline{3}}_{2}$\\
	$\textbf{\underline{3}}_{7}=\textbf{\underline{1}}_{2}\otimes \textbf{\underline{3}}_{1}$ & $\textbf{\underline{3}}_{8}=\textbf{\underline{1}}_{2}\otimes \textbf{\underline{3}}_{2}$\\
\lasthline
\end{tabular}
\caption{Relations of the inequivalent three-dimensional irreducible representations of $\Sigma(72\phi)$.}
\label{sigma72tensorb}
\end{center}
\end{table}
\hspace{0mm}\\
Therefore for the calculation of the Clebsch-Gordan coefficients for $\Sigma(72\phi)$ we have to consider the tensor products
	\begin{displaymath}
	\begin{split}
	& \textbf{\underline{3}}_{1}\otimes \textbf{\underline{3}}_{1}= \textbf{\underline{3}}_{8}\oplus \textbf{\underline{6}}_{2},\\
	& \textbf{\underline{3}}_{1}\otimes \textbf{\underline{3}}_{2}= \textbf{\underline{1}}_{1}\oplus \textbf{\underline{8}}_{1},\\
	& \textbf{\underline{3}}_{2}\otimes \textbf{\underline{3}}_{2}= \textbf{\underline{3}}_{7}\oplus \textbf{\underline{6}}_{1}.
	\end{split}
	\end{displaymath}	
These tensor products are of the two types analysed when we investigated the group $\Sigma(168)$ in subsection \ref{sigma168subsection}, namely
	$\textbf{\underline{3}}\otimes\textbf{\underline{3}}=\textbf{\underline{3}}_{a}+\textbf{\underline{6}}_{s}$
and
$\textbf{\underline{3}}\otimes\textbf{\underline{3}}^{\ast}=\textbf{\underline{1}}+\textbf{\underline{8}}$.

\subsection{The group $\Sigma(216\phi)$}\label{sigma216phisubsection}
$\Sigma(216\phi)$ is also known as the \textit{Hessian group} \cite{miller}. This group is also mentioned in \cite{fairbairn}. Both references give the same generators, but Fairbairn et al. \cite{fairbairn} list them as generators of $\Sigma(216)=\Sigma(216\phi)/C$. To test which group the given generators belong to, one can for instance directly calculate all elements of the group by multiplication of the generators. A \textit{Mathematica 6} program constructed for this purpose can be found in appendix \ref{appendixD}. The result is that the generators given in \cite{fairbairn} and \cite{miller} generate $\Sigma(216\phi)$. As an additional result using the program one finds that the two generators $A$ and $B$ of the four generators $A,B,C,D$ given in \cite{fairbairn} and \cite{miller} generate $\Sigma(216\phi)$ alone:
	\begin{displaymath}
	\begin{split}
	& A=\frac{1}{\omega -\omega ^2}\left(
	\begin{array}{ccc}
 	1 & 1 & 1 \\
 	1 & \omega  & \omega ^2 \\
 	1 & \omega ^2 & \omega 
	\end{array}\right),\quad
	B=\epsilon \left(
	\begin{array}{ccc}
 	1 & 0 & 0 \\
 	0 & 1 & 0 \\
 	0 & 0 & \omega 
	\end{array}
	\right),\\
	& C=\left(
	\begin{array}{ccc}
 	1 & 0 & 0 \\
 	0 & \omega & 0 \\
 	0 & 0 & \omega^2 
	\end{array}
	\right),\quad
	D=\left(
	\begin{array}{ccc}
 	0 & 1 & 0 \\
 	0 & 0 & 1 \\
 	1 & 0 & 0 
	\end{array}
	\right),
	\end{split}
	\end{displaymath}
where $\omega=e^{\frac{2\pi i}{3}}, \epsilon=e^{\frac{4\pi i}{9}}$. Knowing these generators \textit{GAP} can construct the character table, which is shown in tables \ref{Sigma216charactertable1} and \ref{Sigma216charactertable2}.
\\
Remark: Since $\Sigma(36\phi)$ is generated by $A,C,D$ it is a subgroup of $\Sigma(216\phi)$.

\begin{table}
\begin{tiny}
\begin{center}
\renewcommand{\arraystretch}{1.4}
\begin{tabular}{|l|cccccccc|}
\firsthline
	$\Sigma(216\phi)$ & $C_{1}(1)$ & $C_{2}(9)$ & $C_{3}(9)$ & $C_{4}(24)$ & $C_{5}(72)$ & $C_{6}(36)$ & $C_{7}(36)$ & $C_{8}(36)$\\
\hline
	$\textbf{\underline{1}}_{1}$ & $1$ & $1$ & $1$ & $1$ & $1$ & $1$ & $1$ & $1$\\
	$\textbf{\underline{1}}_{2}$ & $1$ & $1$ & $1$ & $1$ & $-a$ & $-a$ & $-a$ & $-a$\\
	$\textbf{\underline{1}}_{3}$ & $1$ & $1$ & $1$ & $1$ & $-a^{\ast}$ & $-a^{\ast}$ & $-a^{\ast}$ & $-a^{\ast}$\\
	$\textbf{\underline{2}}_{1}$ & $2$ & $-2$ & $-2$ & $2$ & $-1$ & $1$ & $1$ & $1$\\
	$\textbf{\underline{2}}_{2}$ & $2$ & $-2$ & $-2$ & $2$ & $a^{\ast}$ & $-a^{\ast}$ & $-a^{\ast}$ & $-a^{\ast}$\\
	$\textbf{\underline{2}}_{3}$ & $2$ & $-2$ & $-2$ & $2$ & $a$ & $-a$ & $-a$ & $-a$\\
	$\textbf{\underline{3}}_{1}$ & $3$ & $3$ & $3$ & $3$ & $0$ & $0$ & $0$ & $0$\\
	$\textbf{\underline{3}}_{2}$ & $3$ & $a$ & $a^{\ast}$ & $0$ & $0$ & $d$ & $f$ & $E$\\
	$\textbf{\underline{3}}_{3}$ & $3$ & $a$ & $a^{\ast}$ & $0$ & $0$ & $E$ & $d$ & $f$\\
	$\textbf{\underline{3}}_{4}$ & $3$ & $a$ & $a^{\ast}$ & $0$ & $0$ & $f$ & $E$ & $d$\\
	$\textbf{\underline{3}}_{5}$ & $3$ & $a^{\ast}$ & $a$ & $0$ & $0$ & $E^{\ast}$ & $d^{\ast}$ & $f^{\ast}$\\
	$\textbf{\underline{3}}_{6}$ & $3$ & $a^{\ast}$ & $a$ & $0$ & $0$ & $d^{\ast}$ & $f^{\ast}$ & $E^{\ast}$\\
	$\textbf{\underline{3}}_{7}$ & $3$ & $a^{\ast}$ & $a$ & $0$ & $0$ & $f^{\ast}$ & $E^{\ast}$ & $d^{\ast}$\\
	$\textbf{\underline{6}}_{1}$ & $6$ & $b$ & $b^{\ast}$ & $0$ & $0$ & $f$ & $E$ & $d$\\
	$\textbf{\underline{6}}_{2}$ & $6$ & $b$ & $b^{\ast}$ & $0$ & $0$ & $d$ & $f$ & $E$\\
	$\textbf{\underline{6}}_{3}$ & $6$ & $b$ & $b^{\ast}$ & $0$ & $0$ & $E$ & $d$ & $f$\\
	$\textbf{\underline{6}}_{4}$ & $6$ & $b^{\ast}$ & $b$ & $0$ & $0$ & $f^{\ast}$ & $E^{\ast}$ & $d^{\ast}$\\
	$\textbf{\underline{6}}_{5}$ & $6$ & $b^{\ast}$ & $b$ & $0$ & $0$ & $E^{\ast}$ & $d^{\ast}$ & $f^{\ast}$\\
	$\textbf{\underline{6}}_{6}$ & $6$ & $b^{\ast}$ & $b$ & $0$ & $0$ & $d^{\ast}$ & $f^{\ast}$ & $E^{\ast}$\\
	$\textbf{\underline{8}}_{1}$ & $8$ & $0$ & $0$ & $-1$ & $-1$ & $0$ & $0$ & $0$\\
	$\textbf{\underline{8}}_{2}$ & $8$ & $0$ & $0$ & $-1$ & $a$ & $0$ & $0$ & $0$\\
	$\textbf{\underline{8}}_{3}$ & $8$ & $0$ & $0$ & $-1$ & $a^{\ast}$ & $0$ & $0$ & $0$\\
	$\textbf{\underline{9}}_{1}$ & $9$ & $c$ & $c^{\ast}$ & $0$ & $0$ & $0$ & $0$ & $0$\\
	$\textbf{\underline{9}}_{2}$ & $9$ & $c^{\ast}$ & $c$ & $0$ & $0$ & $0$ & $0$ & $0$\\
\hline
 & $C_{9}(72)$ & $C_{10}(36)$ & $C_{11}(36)$ & $C_{12}(36)$ & $C_{13}(9)$ & $C_{14}(54)$ & $C_{15}(54)$ & $C_{16}(54)$\\
\hline
	$\textbf{\underline{1}}_{1}$ & $1$ & $1$ & $1$ & $1$ & $1$ & $1$ & $1$ & $1$\\
	$\textbf{\underline{1}}_{2}$ & $-a^{\ast}$ & $-a^{\ast}$ & $-a^{\ast}$ & $-a^{\ast}$ & $1$ & $1$ & $1$ & $1$\\
	$\textbf{\underline{1}}_{3}$ & $-a$ & $-a$ & $-a$ & $-a$ & $1$ & $1$ & $1$ & $1$\\
	$\textbf{\underline{2}}_{1}$ & $-1$ & $1$ & $1$ & $1$ & $-2$ & $0$ & $0$ & $0$\\
	$\textbf{\underline{2}}_{2}$ & $a$ & $-a$ & $-a$ & $-a$ & $-2$ & $0$ & $0$ & $0$\\
	$\textbf{\underline{2}}_{3}$ & $a^{\ast}$ & $-a^{\ast}$ & $-a^{\ast}$ & $-a^{\ast}$ & $-2$ & $0$ & $0$ & $0$\\
	$\textbf{\underline{3}}_{1}$ & $0$ & $0$ & $0$ & $0$ & $3$ & $-1$ & $-1$ & $-1$\\
	$\textbf{\underline{3}}_{2}$ & $0$ & $f^{\ast}$ & $d^{\ast}$ & $E^{\ast}$ & $-1$ & $-a$ & $1$ & $-a^{\ast}$\\
	$\textbf{\underline{3}}_{3}$ & $0$ & $d^{\ast}$ & $E^{\ast}$ & $f^{\ast}$ & $-1$ & $-a$ & $1$ & $-a^{\ast}$\\
	$\textbf{\underline{3}}_{4}$ & $0$ & $E^{\ast}$ & $f^{\ast}$ & $d^{\ast}$ & $-1$ & $-a$ & $1$ & $-a^{\ast}$\\
	$\textbf{\underline{3}}_{5}$ & $0$ & $d$ & $E$ & $f$ & $-1$ & $-a^{\ast}$ & $1$ & $-a$\\
	$\textbf{\underline{3}}_{6}$ & $0$ & $f$ & $d$ & $E$ & $-1$ & $-a^{\ast}$ & $1$ & $-a$\\
	$\textbf{\underline{3}}_{7}$ & $0$ & $E$ & $f$ & $d$ & $-1$ & $-a^{\ast}$ & $1$ & $-a$\\
	$\textbf{\underline{6}}_{1}$ & $0$ & $E^{\ast}$ & $f^{\ast}$ & $d^{\ast}$ & $2$ & $0$ & $0$ & $0$\\
	$\textbf{\underline{6}}_{2}$ & $0$ & $f^{\ast}$ & $d^{\ast}$ & $E^{\ast}$ & $2$ & $0$ & $0$ & $0$\\
	$\textbf{\underline{6}}_{3}$ & $0$ & $d^{\ast}$ & $E^{\ast}$ & $f^{\ast}$ & $2$ & $0$ & $0$ & $0$\\
	$\textbf{\underline{6}}_{4}$ & $0$ & $E$ & $f$ & $d$ & $2$ & $0$ & $0$ & $0$\\
	$\textbf{\underline{6}}_{5}$ & $0$ & $d$ & $E$ & $f$ & $2$ & $0$ & $0$ & $0$\\
	$\textbf{\underline{6}}_{6}$ & $0$ & $f$ & $d$ & $E$ & $2$ & $0$ & $0$ & $0$\\
	$\textbf{\underline{8}}_{1}$ & $-1$ & $0$ & $0$ & $0$ & $0$ & $0$ & $0$ & $0$\\
	$\textbf{\underline{8}}_{2}$ & $a^{\ast}$ & $0$ & $0$ & $0$ & $0$ & $0$ & $0$ & $0$\\
	$\textbf{\underline{8}}_{3}$ & $a$ & $0$ & $0$ & $0$ & $0$ & $0$ & $0$ & $0$\\
	$\textbf{\underline{9}}_{1}$ & $0$ & $0$ & $0$ & $0$ & $-3$ & $a$ & $-1$ & $a^{\ast}$\\
	$\textbf{\underline{9}}_{2}$ & $0$ & $0$ & $0$ & $0$ & $-3$ & $a^{\ast}$ & $-1$ & $a$\\
\lasthline
\end{tabular}
\caption[The character table of $\Sigma(216\phi)$. (Part 1)]{The character table of $\Sigma(216\phi)$ as calculated with \textit{GAP}. (Part 1)
$a=-\omega^2$, $b=2\omega^{2}$, $c=-3\omega^2$, $d=-e^{\frac{10\pi i}{9}}$, $E=-e^{\frac{4\pi i}{9}}$, $f=e^{\frac{4\pi i}{9}}+e^{\frac{10 \pi i}{9}}$.}
\label{Sigma216charactertable1}
\end{center}
\end{tiny}
\end{table}

\begin{table}
\begin{tiny}
\begin{center}
\renewcommand{\arraystretch}{1.4}
\begin{tabular}{|l|cccccccc|}
\firsthline
	$\Sigma(216\phi)$ & $C_{17}(1)$ & $C_{18}(1)$ & $C_{19}(12)$ & $C_{20}(12)$ & $C_{21}(12)$ & $C_{22}(12)$ & $C_{23}(12)$ & $C_{24}(12)$\\
\hline
	$\textbf{\underline{1}}_{1}$ & $1$ & $1$ & $1$ & $1$ & $1$ & $1$ & $1$ & $1$\\
	$\textbf{\underline{1}}_{2}$ & $1$ & $1$ & $-a$ & $-a$ & $-a^{\ast}$ & $-a^{\ast}$ & $-a^{\ast}$ & $-a$\\
	$\textbf{\underline{1}}_{3}$ & $1$ & $1$ & $-a^{\ast}$ & $-a^{\ast}$ & $-a$ & $-a$ & $-a$ & $-a^{\ast}$\\
	$\textbf{\underline{2}}_{1}$ & $2$ & $2$ & $-1$ & $-1$ & $-1$ & $-1$ & $-1$ & $-1$\\
	$\textbf{\underline{2}}_{2}$ & $2$ & $2$ & $a^{\ast}$ & $a^{\ast}$ & $a$ & $a$ & $a$ & $a^{\ast}$\\
	$\textbf{\underline{2}}_{3}$ & $2$ & $2$ & $a$ & $a$ & $a^{\ast}$ & $a^{\ast}$ & $a^{\ast}$ & $a$\\
	$\textbf{\underline{3}}_{1}$ & $3$ & $3$ & $0$ & $0$ & $0$ & $0$ & $0$ & $0$\\
	$\textbf{\underline{3}}_{2}$ & $-c$ & $-c^{\ast}$ & $I$ & $j$ & $I^{\ast}$ & $j^{\ast}$ & $k^{\ast}$ & $k$\\
	$\textbf{\underline{3}}_{3}$ & $-c$ & $-c^{\ast}$ & $j$ & $k$ & $j^{\ast}$ & $k^{\ast}$ & $I^{\ast}$ & $I$\\
	$\textbf{\underline{3}}_{4}$ & $-c$ & $-c^{\ast}$ & $k$ & $I$ & $k^{\ast}$ & $I^{\ast}$ & $j^{\ast}$ & $j$\\
	$\textbf{\underline{3}}_{5}$ & $-c^{\ast}$ & $-c$ & $j^{\ast}$ & $k^{\ast}$ & $j$ & $k$ & $I$ & $I^{\ast}$\\
	$\textbf{\underline{3}}_{6}$ & $-c^{\ast}$ & $-c$ & $I^{\ast}$ & $j^{\ast}$ & $I$ & $j$ & $k$ & $k^{\ast}$\\
	$\textbf{\underline{3}}_{7}$ & $-c^{\ast}$ & $-c$ & $k^{\ast}$ & $I^{\ast}$ & $k$ & $I$ & $j$ & $j^{\ast}$\\
	$\textbf{\underline{6}}_{1}$ & $g$ & $g^{\ast}$ & $-k$ & $-I$ & $-k^{\ast}$ & $-I^{\ast}$ & $-j^{\ast}$ & $-j$\\
	$\textbf{\underline{6}}_{2}$ & $g$ & $g^{\ast}$ & $-I$ & $-j$ & $-I^{\ast}$ & $-j^{\ast}$ & $-k^{\ast}$ & $-k$\\
	$\textbf{\underline{6}}_{3}$ & $g$ & $g^{\ast}$ & $-j$ & $-k$ & $-j^{\ast}$ & $-k^{\ast}$ & $-I^{\ast}$ & $-I$\\
	$\textbf{\underline{6}}_{4}$ & $g^{\ast}$ & $g$ & $-k^{\ast}$ & $-I^{\ast}$ & $-k$ & $-I$ & $-j$ & $-j^{\ast}$\\
	$\textbf{\underline{6}}_{5}$ & $g^{\ast}$ & $g$ & $-j^{\ast}$ & $-k^{\ast}$ & $-j$ & $-k$ & $-I$ & $-I^{\ast}$\\
	$\textbf{\underline{6}}_{6}$ & $g^{\ast}$ & $g$ & $-I^{\ast}$ & $-j^{\ast}$ & $-I$ & $-j$ & $-k$ & $-k^{\ast}$\\
	$\textbf{\underline{8}}_{1}$ & $8$ & $8$ & $2$ & $2$ & $2$ & $2$ & $2$ & $2$\\
	$\textbf{\underline{8}}_{2}$ & $8$ & $8$ & $b$ & $b$ & $b^{\ast}$ & $b^{\ast}$ & $b^{\ast}$ & $b$\\
	$\textbf{\underline{8}}_{3}$ & $8$ & $8$ & $b^{\ast}$ & $b^{\ast}$ & $b$ & $b$ & $b$ & $b^{\ast}$\\
	$\textbf{\underline{9}}_{1}$ & $h$ & $h^{\ast}$ & $0$ & $0$ & $0$ & $0$ & $0$ & $0$\\
	$\textbf{\underline{9}}_{2}$ & $h^{\ast}$ & $h$ & $0$ & $0$ & $0$ & $0$ & $0$ & $0$\\
\lasthline
\end{tabular}
\caption[The character table of $\Sigma(216\phi)$. (Part 2)]{The character table of $\Sigma(216\phi)$ as calculated with \textit{GAP}. (Part 2)
$a=-\omega^2$, $b=2\omega^{2}$, $c=-3\omega^2$, $d=-e^{\frac{10\pi i}{9}}$, $E=-e^{\frac{4\pi i}{9}}$, $f=e^{\frac{4\pi i}{9}}+e^{\frac{10 \pi i}{9}}$, $g=6\omega^{2}$, $h=9\omega^{2}$, $I=2e^{\frac{4\pi i}{9}}+e^{\frac{10\pi i}{9}}$, $j=-e^{\frac{4\pi i}{9}}-2e^{\frac{10\pi i}{9}}$, $k=-e^{\frac{4\pi i}{9}}+e^{\frac{10 \pi i}{9}}$.}
\label{Sigma216charactertable2}
\end{center}
\end{tiny}
\end{table}
\hspace{0mm}\\
Using the character table we can calculate the tensor products of the three-dimensional representations. The results are shown in \ref{sigma216tensor}.

\begin{table}
\begin{small}
\begin{center}
\renewcommand{\arraystretch}{1.4}
\begin{tabular}{|l@{\hspace{10mm}}l|}
\firsthline
	$\textbf{\underline{3}}_{1} \otimes \textbf{\underline{3}}_{1}=\textbf{\underline{1}}_{1}\oplus\textbf{\underline{1}}_{2}\oplus\textbf{\underline{1}}_{3}\oplus\textbf{\underline{3}}_{1}\oplus\textbf{\underline{3}}_{1}$ & $\textbf{\underline{3}}_{3} \otimes \textbf{\underline{3}}_{4}=\textbf{\underline{3}}_{6}\oplus\textbf{\underline{6}}_{5}$ \\
	$\textbf{\underline{3}}_{1} \otimes \textbf{\underline{3}}_{2}=\textbf{\underline{9}}_{1}$ & $\textbf{\underline{3}}_{3} \otimes \textbf{\underline{3}}_{5}=\textbf{\underline{1}}_{1}\oplus\textbf{\underline{8}}_{1}$\\
	$\textbf{\underline{3}}_{1} \otimes \textbf{\underline{3}}_{3}=\textbf{\underline{9}}_{1}$ & $\textbf{\underline{3}}_{3} \otimes \textbf{\underline{3}}_{6}=\textbf{\underline{1}}_{2}\oplus\textbf{\underline{8}}_{2}$\\
	$\textbf{\underline{3}}_{1} \otimes \textbf{\underline{3}}_{4}=\textbf{\underline{9}}_{1}$ & $\textbf{\underline{3}}_{3} \otimes \textbf{\underline{3}}_{7}=\textbf{\underline{1}}_{3}\oplus\textbf{\underline{8}}_{3}$\\
	$\textbf{\underline{3}}_{1} \otimes \textbf{\underline{3}}_{5}=\textbf{\underline{9}}_{2}$ & $\textbf{\underline{3}}_{4} \otimes \textbf{\underline{3}}_{4}=\textbf{\underline{3}}_{7}\oplus\textbf{\underline{6}}_{6}$\\
	$\textbf{\underline{3}}_{1} \otimes \textbf{\underline{3}}_{6}=\textbf{\underline{9}}_{2}$ & $\textbf{\underline{3}}_{4} \otimes \textbf{\underline{3}}_{5}=\textbf{\underline{1}}_{2}\oplus\textbf{\underline{8}}_{2}$\\
	$\textbf{\underline{3}}_{1} \otimes \textbf{\underline{3}}_{7}=\textbf{\underline{9}}_{2}$ & $\textbf{\underline{3}}_{4} \otimes \textbf{\underline{3}}_{6}=\textbf{\underline{1}}_{3}\oplus\textbf{\underline{8}}_{3}$\\
	$\textbf{\underline{3}}_{2} \otimes \textbf{\underline{3}}_{2}=\textbf{\underline{3}}_{6}\oplus\textbf{\underline{6}}_{5}$ & $\textbf{\underline{3}}_{4} \otimes \textbf{\underline{3}}_{7}=\textbf{\underline{1}}_{1}\oplus\textbf{\underline{8}}_{1}$\\
	$\textbf{\underline{3}}_{2} \otimes \textbf{\underline{3}}_{3}=\textbf{\underline{3}}_{7}\oplus\textbf{\underline{6}}_{6}$ & $\textbf{\underline{3}}_{5} \otimes \textbf{\underline{3}}_{5}=\textbf{\underline{3}}_{3}\oplus\textbf{\underline{6}}_{1}$\\
	$\textbf{\underline{3}}_{2} \otimes \textbf{\underline{3}}_{4}=\textbf{\underline{3}}_{5}\oplus\textbf{\underline{6}}_{4}$ & $\textbf{\underline{3}}_{5} \otimes \textbf{\underline{3}}_{6}=\textbf{\underline{3}}_{4}\oplus\textbf{\underline{6}}_{2}$\\
	$\textbf{\underline{3}}_{2} \otimes \textbf{\underline{3}}_{5}=\textbf{\underline{1}}_{3}\oplus\textbf{\underline{8}}_{3}$ & $\textbf{\underline{3}}_{5} \otimes \textbf{\underline{3}}_{7}=\textbf{\underline{3}}_{2}\oplus\textbf{\underline{6}}_{3}$\\
	$\textbf{\underline{3}}_{2} \otimes \textbf{\underline{3}}_{6}=\textbf{\underline{1}}_{1}\oplus\textbf{\underline{8}}_{1}$ & $\textbf{\underline{3}}_{6} \otimes \textbf{\underline{3}}_{6}=\textbf{\underline{3}}_{2}\oplus\textbf{\underline{6}}_{3}$\\
	$\textbf{\underline{3}}_{2} \otimes \textbf{\underline{3}}_{7}=\textbf{\underline{1}}_{2}\oplus\textbf{\underline{8}}_{2}$ & $\textbf{\underline{3}}_{6} \otimes \textbf{\underline{3}}_{7}=\textbf{\underline{3}}_{3}\oplus\textbf{\underline{6}}_{1}$\\
	$\textbf{\underline{3}}_{3} \otimes \textbf{\underline{3}}_{3}=\textbf{\underline{3}}_{5}\oplus\textbf{\underline{6}}_{4}$ & $\textbf{\underline{3}}_{7} \otimes \textbf{\underline{3}}_{7}=\textbf{\underline{3}}_{4}\oplus\textbf{\underline{6}}_{2}$\\
\lasthline
\end{tabular}
\caption{Tensor products of the three-dimensional irreducible representations of $\Sigma(216\phi)$.}
\label{sigma216tensor}
\end{center}
\end{small}
\end{table}
\hspace{0mm}\\
All irreducible three-dimensional representations, except $\textbf{\underline{3}}_{1}$, can be constructed from two representations, e.g. $\textbf{\underline{3}}_{2}$ and $\textbf{\underline{3}}_{6}=\textbf{\underline{3}}_{2}^{\ast}$, as is shown in table \ref{sigma216tensorb}.

\begin{table}
\begin{center}
\renewcommand{\arraystretch}{1.4}
\begin{tabular}{|l@{\hspace{10mm}}l|}
\firsthline
	$\textbf{\underline{3}}_{3}=\textbf{\underline{1}}_{2}\otimes \textbf{\underline{3}}_{2}$ & $\textbf{\underline{3}}_{5}=\textbf{\underline{1}}_{3}\otimes \textbf{\underline{3}}_{6}$\\
	$\textbf{\underline{3}}_{4}=\textbf{\underline{1}}_{3}\otimes \textbf{\underline{3}}_{2}$ &
	$\textbf{\underline{3}}_{7}=\textbf{\underline{1}}_{2}\otimes \textbf{\underline{3}}_{6}$\\
\lasthline
\end{tabular}
\caption{Relations of the inequivalent three-dimensional irreducible representations of $\Sigma(216\phi)$.}
\label{sigma216tensorb}
\end{center}
\end{table}
\hspace{0mm}\\
The irreducible representation $\textbf{\underline{3}}_{1}$ is different of the others, because of the fact that it is not faithful. This can be easily seen from the character table. Obviously there is more than one element of the group that is mapped to $\mathbbm{1}_{3}$ under $\textbf{\underline{3}}_{1}$, because several conjugate classes have character $3$. (Let $U\in SU(3)$, then $\mathrm{Tr}(U)=3$ $\Longleftrightarrow$ $U=\mathbbm{1}_{3}$.) For further investigation of this representation we need generators, which are provided by \textit{GAP}:
	\begin{displaymath}
	 \textbf{\underline{3}}_{1}(A)=\left(\begin{array}{rrr}
	 -1 & -1 & -1 \\
	 0 & 0 & 1 \\
	 0 & 1 & 0
	\end{array}\right),\quad
	\textbf{\underline{3}}_{1}(B)=\left(\begin{array}{rrr}
	 1 & 0 & 0 \\
	 -1 & -1 & -1 \\
	 0 & 1 & 0
	\end{array}\right).
	\end{displaymath}
Constructing the group $\textbf{\underline{3}}_{1}(\Sigma(216\phi))$ using the \textit{Mathematica 6} program already used to calculate all elements of $\textbf{\underline{3}}_{4}(\Sigma(216\phi))$, we find that $\textbf{\underline{3}}_{1}(\Sigma(216\phi))$ is a non-Abelian group with 12 elements, and comparing the characters with those of $A_{4}$ we find $\textbf{\underline{3}}_{1}(\Sigma(216\phi))\simeq A_{4}$.
\bigskip
\\
Due to lemma \ref{LSU313}, for the Clebsch-Gordan coefficients we only need to consider the following tensor products:
	\begin{displaymath}
	\begin{split}
	& \textbf{\underline{3}}_{1}\otimes \textbf{\underline{3}}_{1}= \textbf{\underline{1}}_{1}\oplus \textbf{\underline{1}}_{2}\oplus \textbf{\underline{1}}_{3}\oplus \textbf{\underline{3}}_{1}\oplus \textbf{\underline{3}}_{1},\\
	& \textbf{\underline{3}}_{1}\otimes \textbf{\underline{3}}_{2}=\textbf{\underline{9}}_{1},\\
	& \textbf{\underline{3}}_{1}\otimes \textbf{\underline{3}}_{6}=\textbf{\underline{9}}_{2},\\
	& \textbf{\underline{3}}_{2}\otimes \textbf{\underline{3}}_{2}=\textbf{\underline{3}}_{6}\oplus \textbf{\underline{6}}_{5},\\
	& \textbf{\underline{3}}_{2}\otimes \textbf{\underline{3}}_{6}=\textbf{\underline{1}}_{1}\oplus \textbf{\underline{8}}_{1},\\
	& \textbf{\underline{3}}_{6}\otimes \textbf{\underline{3}}_{6}=\textbf{\underline{3}}_{2}\oplus \textbf{\underline{6}}_{3}.
	\end{split}
	\end{displaymath}
The Clebsch-Gordan coefficients for $\textbf{\underline{3}}_{1}\otimes \textbf{\underline{3}}_{1}$ can be found in subsection \ref{subsectionA4}, and the tensor products $\textbf{\underline{3}}_{1}\otimes \textbf{\underline{3}}_{2}$ and $\textbf{\underline{3}}_{1}\otimes \textbf{\underline{3}}_{6}$ are irreducible, thus the matrix $M$ of Clebsch-Gordan coefficients is an arbitrary unitary $9\times 9$-matrix and can be chosen to be the identity matrix $\mathbbm{1}_{9}$.
\medskip
\\
As in the case of $\Sigma(72\phi)$ the remaining tensor products are of the types $\textbf{\underline{3}}\otimes\textbf{\underline{3}}=\textbf{\underline{3}}_{a}+\textbf{\underline{6}}_{s}$
and
$\textbf{\underline{3}}\otimes\textbf{\underline{3}}^{\ast}=\textbf{\underline{1}}+\textbf{\underline{8}}$ ($\rightarrow$ tables \ref{Sigma168CGCa}, \ref{Sigma168CGCb}).

\subsection{The group $\Sigma(360\phi)$}\label{subsectionSigma360phi}
Generators of $\Sigma(360\phi)$ are provided by Miller et al. in \cite{miller}:
	\begin{displaymath}
	\begin{split}
	& A =\left(\begin{matrix}
		 0 & 1 & 0 \\
		 0 & 0 & 1 \\
		 1 & 0 & 0
	     \end{matrix}\right), \quad
	  B =\left(\begin{matrix}
		 1 & 0 & 0 \\
		 0 & -1 & 0 \\
		 0 & 0 & -1
	     \end{matrix}\right),\\
	& C =\frac{1}{2}\left(\begin{matrix}
		 -1 & \mu_{2} & \mu_{1} \\
		 \mu_{2} & \mu_{1} & -1 \\
		 \mu_{1} & -1 & \mu_{2}
	     \end{matrix}\right), \quad
	  D =\left(\begin{matrix}
		 -1 & 0 & 0 \\
		 0 & 0 & -\omega \\
		 0 & -\omega^{2} & 0
	     \end{matrix}\right),
	\end{split}
	\end{displaymath}
with $\mu_{1}=\frac{1}{2}(-1+\sqrt{5})$, $\mu_{2}=\frac{1}{2}(-1-\sqrt{5})$ and $\omega=e^{\frac{2\pi i}{3}}$. Using these generators \textit{GAP} can calculate the character table, which is shown in table \ref{Sigma360charactertable}.
\medskip
\\
According to \cite{miller} $\Sigma(60)=I\simeq A_{5}$ is generated by $A$, $B$ and $C$. Therefore $\Sigma(60)$ is a subgroup of $\Sigma(360\phi)$. In subsection \ref{icosahedral} we listed two matrices $S$ and $T$ (given in \cite{everett}) that generate the icosahedral group $I\simeq A_{5}$.  Therefore one could also use $S,T$ and $D$ as generators of $\Sigma(360\phi)$.

\begin{table}
\begin{tiny}
\begin{center}
\renewcommand{\arraystretch}{1.4}
\begin{tabular}{|l|ccccccccc|}
\firsthline
	$\Sigma(360\phi)$ & $C_{1}(1)$ & $C_{2}(45)$ & $C_{3}(90)$ & $C_{4}(90)$ & $C_{5}(45)$ & $C_{6}(45)$ & $C_{7}(1)$ & $C_{8}(1)$ & $C_{9}(90)$\\
\hline
	$\textbf{\underline{1}}_{1}$ & $1$ & $1$ & $1$ & $1$ & $1$ & $1$ & $1$ & $1$ & $1$\\
	$\textbf{\underline{3}}_{1}$ & $3$ & $-1$ & $a$ & $a^{\ast}$ & $-a^{\ast}$ & $-a$ & $c$ & $c^{\ast}$ & $1$\\
	$\textbf{\underline{3}}_{2}$ & $3$ & $-1$ & $a$ & $a^{\ast}$ & $-a^{\ast}$ & $-a$ & $c$ & $c^{\ast}$ & $1$\\
	$\textbf{\underline{3}}_{3}$ & $3$ & $-1$ & $a^{\ast}$ & $a$ & $-a$ & $-a^{\ast}$ & $c^{\ast}$ & $c$ & $1$\\
	$\textbf{\underline{3}}_{4}$ & $3$ & $-1$ & $a^{\ast}$ & $a$ & $-a$ & $-a^{\ast}$ & $c^{\ast}$ & $c$ & $1$\\
	$\textbf{\underline{5}}_{1}$ & $5$ & $1$ & $-1$ & $-1$ & $1$ & $1$ & $5$ & $5$ & $-1$\\
	$\textbf{\underline{5}}_{2}$ & $5$ & $1$ & $-1$ & $-1$ & $1$ & $1$ & $5$ & $5$ & $-1$\\
	$\textbf{\underline{6}}_{1}$ & $6$ & $2$ & $0$ & $0$ & $b$ & $b^{\ast}$ & $d$ & $d^{\ast}$ & $0$\\
	$\textbf{\underline{6}}_{2}$ & $6$ & $2$ & $0$ & $0$ & $b^{\ast}$ & $b$ & $d^{\ast}$ & $d$ & $0$\\
	$\textbf{\underline{8}}_{1}$ & $8$ & $0$ & $0$ & $0$ & $0$ & $0$ & $8$ & $8$ & $0$\\
	$\textbf{\underline{8}}_{2}$ & $8$ & $0$ & $0$ & $0$ & $0$ & $0$ & $8$ & $8$ & $0$\\
	$\textbf{\underline{9}}_{1}$ & $9$ & $1$ & $1$ & $1$ & $1$ & $1$ & $9$ & $9$ & $1$\\
	$\textbf{\underline{9}}_{2}$ & $9$ & $1$ & $a$ & $a^{\ast}$ & $a^{\ast}$ & $a$ & $E$ & $E^{\ast}$ & $1$\\
	$\textbf{\underline{9}}_{3}$ & $9$ & $1$ & $a^{\ast}$ & $a$ & $a$ & $a^{\ast}$ & $E^{\ast}$ & $E$ & $1$\\
	$\textbf{\underline{10}}$ & $10$ & $-2$ & $0$ & $0$ & $-2$ & $-2$ & $10$ & $10$ & $0$\\
	$\textbf{\underline{15}}_{1}$ & $15$ & $-1$ & $-a$ & $-a^{\ast}$ & $-a^{\ast}$ & $-a$ & $f$ & $f^{\ast}$ & $-1$\\
	$\textbf{\underline{15}}_{2}$ & $15$ & $-1$ & $-a^{\ast}$ & $-a$ & $-a$ & $-a^{\ast}$ & $f^{\ast}$ & $f$ & $-1$\\
\hline
	 & $C_{10}(120)$ & $C_{11}(72)$ & $C_{12}(72)$ & $C_{13}(72)$ & $C_{14}(72)$ & $C_{15}(72)$ & $C_{16}(72)$ & $C_{17}(120)$ &\\
\hline
	$\textbf{\underline{1}}_{1}$ & $1$ & $1$ & $1$ & $1$ & $1$ & $1$ & $1$ & $1$ &\\
	$\textbf{\underline{3}}_{1}$ & $0$ & $g$ & $j$ & $h$ & $I^{\ast}$ & $I$ & $h^{\ast}$ & $0$ &\\
	$\textbf{\underline{3}}_{2}$ & $0$ & $j$ & $g$ & $I$ & $h^{\ast}$ & $h$ & $I^{\ast}$ & $0$ &\\
	$\textbf{\underline{3}}_{3}$ & $0$ & $j$ & $g$ & $I^{\ast}$ & $h$ & $h^{\ast}$ & $I$ & $0$ &\\
	$\textbf{\underline{3}}_{4}$ & $0$ & $g$ & $j$ & $h^{\ast}$ & $I$ & $I^{\ast}$ & $h$ & $0$ &\\
	$\textbf{\underline{5}}_{1}$ & $-1$ & $0$ & $0$ & $0$ & $0$ & $0$ & $0$ & $2$ &\\
	$\textbf{\underline{5}}_{2}$ & $2$ & $0$ & $0$ & $0$ & $0$ & $0$ & $0$ & $-1$ &\\
	$\textbf{\underline{6}}_{1}$ & $0$ & $1$ & $1$ & $a$ & $a^{\ast}$ & $a$ & $a^{\ast}$ & $0$ &\\
	$\textbf{\underline{6}}_{2}$ & $0$ & $1$ & $1$ & $a^{\ast}$ & $a$ & $a^{\ast}$ & $a$ & $0$ &\\
	$\textbf{\underline{8}}_{1}$ & $-1$ & $g$ & $j$ & $g$ & $j$ & $j$ & $g$ & $-1$ &\\
	$\textbf{\underline{8}}_{2}$ & $-1$ & $j$ & $g$ & $j$ & $g$ & $g$ & $j$ & $-1$ &\\
	$\textbf{\underline{9}}_{1}$ & $0$ & $-1$ & $-1$ & $-1$ & $-1$ & $-1$ & $-1$ & $0$ &\\
	$\textbf{\underline{9}}_{2}$ & $0$ & $-1$ & $-1$ & $-a$ & $-a^{\ast}$ & $-a$ & $-a^{\ast}$ & $0$ &\\
	$\textbf{\underline{9}}_{3}$ & $0$ & $-1$ & $-1$ & $-a^{\ast}$ & $-a$ & $-a^{\ast}$ & $-a$ & $0$ &\\
	$\textbf{\underline{10}}$ & $1$ & $0$ & $0$ & $0$ & $0$ & $0$ & $0$ & $1$ &\\
	$\textbf{\underline{15}}_{1}$ & $0$ & $0$ & $0$ & $0$ & $0$ & $0$ & $0$ & $0$ &\\
	$\textbf{\underline{15}}_{2}$ & $0$ & $0$ & $0$ & $0$ & $0$ & $0$ & $0$ & $0$ &\\
\lasthline
\end{tabular}
\caption[The character table of $\Sigma(360\phi)$.]{The character table of $\Sigma(360\phi)$ as calculated by \textit{GAP}.
$a=\omega$, $b=2\omega^{2}$, $c=3\omega$, $d=6\omega$, $E=9\omega$, $f=15\omega$, $g=-e^{\frac{4\pi i}{5}}-e^{\frac{6\pi i}{5}}=\frac{1}{2}(1+\sqrt{5})$, $h=-e^{\frac{22\pi i}{15}}-e^{\frac{28\pi i}{15}}$, $I=-e^{\frac{4\pi i}{15}}-e^{\frac{16\pi i}{15}}$, $j=-e^{\frac{2\pi i}{5}}-e^{\frac{8\pi i}{5}}=\frac{1}{2}(1-\sqrt{5})$.}
\label{Sigma360charactertable}
\end{center}
\end{tiny}
\end{table}
\hspace{0mm}\\
Using the character table we can calculate the tensor products of the three-dimensional irreducible representations. The results are shown in \ref{sigma360tensor}.

\begin{table}
\begin{small}
\begin{center}
\renewcommand{\arraystretch}{1.4}
\begin{tabular}{|l@{\hspace{10mm}}l|}
\firsthline
	$\textbf{\underline{3}}_{1} \otimes \textbf{\underline{3}}_{1}=\textbf{\underline{3}}_{4}\oplus\textbf{\underline{6}}_{2}$ & $\textbf{\underline{3}}_{2} \otimes \textbf{\underline{3}}_{3}=\textbf{\underline{1}}\oplus\textbf{\underline{8}}_{2}$ \\
	$\textbf{\underline{3}}_{1} \otimes \textbf{\underline{3}}_{2}=\textbf{\underline{9}}_{3}$ & $\textbf{\underline{3}}_{2} \otimes \textbf{\underline{3}}_{4}=\textbf{\underline{9}}_{1}$\\
	$\textbf{\underline{3}}_{1} \otimes \textbf{\underline{3}}_{3}=\textbf{\underline{9}}_{1}$ & $\textbf{\underline{3}}_{3} \otimes \textbf{\underline{3}}_{3}=\textbf{\underline{3}}_{2}\oplus\textbf{\underline{6}}_{1}$\\
	$\textbf{\underline{3}}_{1} \otimes \textbf{\underline{3}}_{4}=\textbf{\underline{1}}\oplus\textbf{\underline{8}}_{1}$ & $\textbf{\underline{3}}_{3} \otimes \textbf{\underline{3}}_{4}=\textbf{\underline{9}}_{2}$\\
	$\textbf{\underline{3}}_{2} \otimes \textbf{\underline{3}}_{2}=\textbf{\underline{3}}_{3}\oplus\textbf{\underline{6}}_{2}$ & $\textbf{\underline{3}}_{4} \otimes \textbf{\underline{3}}_{4}=\textbf{\underline{3}}_{1}\oplus\textbf{\underline{6}}_{1}$\\
\lasthline
\end{tabular}
\caption{Tensor products of the three-dimensional irreducible representations of $\Sigma(360\phi)$.}
\label{sigma360tensor}
\end{center}
\end{small}
\end{table}
\hspace{0mm}\\
Looking at the character table we find $\textbf{\underline{3}}_{4}=\textbf{\underline{3}}_{1}^{\ast}$ and $\textbf{\underline{3}}_{3}=\textbf{\underline{3}}_{2}^{\ast}$. Therefore all tensor products listed in table \ref{sigma360tensor} are either irreducible, or of the types $\textbf{\underline{3}}\otimes\textbf{\underline{3}}=\textbf{\underline{3}}_{a}+\textbf{\underline{6}}_{s}$
and
$\textbf{\underline{3}}\otimes\textbf{\underline{3}}^{\ast}=\textbf{\underline{1}}+\textbf{\underline{8}}$. For these tensor products we had already constructed the Clebsch-Gordan coefficients ($\rightarrow$ subsection \ref{sigma168subsection}, tables \ref{Sigma168CGCa}, \ref{Sigma168CGCb}; $\rightarrow$ subsection \ref{sigma216phisubsection}).

\subsection{The group $\Delta(3n^{2})$}\label{d3nnsubsect}
$\Delta(3n^{2})$ is an infinite series of non-Abelian finite subgroups of $SU(3)$ of order $3n^{2}$ ($n\in \mathbb{N}\backslash\{0,1\}$). It has been intensively studied by Luhn, Nasri and Ramond in \cite{luhn2}, as well as by Bovier, L\"uling and Wyler in \cite{BLW1}. Luhn et al. construct conjugate classes, generators for all irreducible representations, character tables and Clebsch-Gordan coefficients for $\Delta(3n^{2})$. We will concentrate especially on the Clebsch-Gordan coefficients for tensor products of three-dimensional irreducible representations.
\medskip
\\
It turns out that the structure of $\Delta(3n^{2})$ is mainly dependent on whether $n$ is divisible by $3$ ($n\in 3\mathbb{N}\backslash\{0\}$) or not ($n\not\in 3\mathbb{N}\backslash\{0\}$).
\medskip
\\
According to \cite{luhn2} $\Delta(3n^2)$ is isomorphic to the group
	\begin{displaymath}
	(Z_{n}\times Z_{n})\rtimes Z_3,
	\end{displaymath}
where $\rtimes$ denotes the so-called \textit{semidirect product} of two groups. Semidirect products $G\rtimes H$ are characterized by the issue that $H$ acts on $G$ via a homomorphism
	\begin{displaymath}
	\phi: H\rightarrow\mathrm{Aut}(G),
	\end{displaymath}
where $\mathrm{Aut}(G)$ is the set of homomorphisms of $G$ onto itself. For a precise definition of the semidirect product we refer the reader to \cite{sexl} (p.338).
\medskip
\\
$\Delta(3n^{2})$ is generated by three generators $a,c,d$ which fulfil \cite{luhn2}
	\begin{equation}\label{d3nnpresentation1}
	\begin{split}
	& a^3=e \quad(\mbox{generator of }Z_3),\\
	& c^n=e,\enspace d^n=e,\enspace cd=dc \quad(\mbox{generators of }Z_n\times Z_n).\\
	\end{split}
	\end{equation}
Due to the semidirect product structure there is an action of $Z_3$ on $Z_n\times Z_n$, which is given by \cite{luhn2}
	\begin{equation}\label{d3nnpresentation2}
	\begin{split}
	& \phi(a)(c)=aca^{-1}=c^{-1}d^{-1},\\
	& \phi(a)(d)=ada^{-1}=c.
	\end{split}
	\end{equation}
$\phi$ is a homomorphism
	\begin{displaymath}
	\begin{split}
	\phi: Z_3\rightarrow \mathrm{Aut}(Z_n\times Z_n),
	\end{split}
	\end{displaymath}
which is given by $\phi(x)(y)=xyx^{-1}\enspace\forall x\in Z_3, y\in Z_n\times Z_n$.
\medskip
\\
The relations (\ref{d3nnpresentation1}) and (\ref{d3nnpresentation2}) form the \textit{presentation} of $\Delta(3n^2)$. (The abstract generators together with their relations are called the presentation of a group.)
\medskip
\\
The irreducible representations are listed in table \ref{Delta3nn-irreps}.

\begin{table}
\begin{footnotesize}
\begin{center}
\renewcommand{\arraystretch}{1.4}
\begin{tabular}{|l|ccl|}
\firsthline
	$\Delta(3n^{2})$ & Representations & number of representations & parameter values \\
\hline
	$n\in 3\mathbb{N}\backslash\{0\}$ & $\textbf{\underline{1}}_{r,s}$ & $9$ & $r,s=0,1,2$\\
	 & $\textbf{\underline{3}}_{\widetilde{(k,l)}}$ & $\frac{n^{2}-3}{3}$ & $k,l=1,...,n-1$\\
	& & & $(k,l)\neq(\frac{n}{3},\frac{n}{3}), (\frac{2n}{3},\frac{2n}{3})$\\
\hline
	$n\not\in 3\mathbb{N}\backslash\{0\}$ & $\textbf{\underline{1}}_{r}$ & $3$ & $r=0,1,2$\\
	 & $\textbf{\underline{3}}_{\widetilde{(k,l)}}$ & $\frac{n^{2}-1}{3}$ & $k,l=1,...,n-1$\\
\lasthline
\end{tabular}
\caption[Irreducible representations of $\Delta(3n^{2})$.]{Irreducible representations of $\Delta(3n^{2})$ as given in \cite{luhn2}.}
\label{Delta3nn-irreps}
\end{center}
\end{footnotesize}
\end{table}
\hspace{0mm}\\
The generators of the irreducible representations as given by \cite{luhn2} are ($\omega=e^{\frac{2\pi i}{3}}$, $\eta=e^{\frac{2\pi i}{n}}$):
	
	\begin{enumerate}
	 \item $n\not\in 3\mathbb{N}\backslash\{0\}$
		\begin{displaymath}
		\begin{split}
		& \textbf{\underline{1}}_{r} : a\mapsto \omega^{r}, c\mapsto 1, d\mapsto 1. \\
		& \textbf{\underline{3}}_{(k,l)} : a\mapsto
			\left(\begin{matrix}
			 0 & 1 & 0 \\
			 0 & 0 & 1 \\
			 1 & 0 & 0
			\end{matrix}\right),\enspace
			c\mapsto
			\left(\begin{matrix}
			 \eta^{l} & 0 & 0 \\
			 0 & \eta^{k} & 0 \\
			 0 & 0 & \eta^{-k-l}
			\end{matrix}\right),\enspace
			d\mapsto
			\left(\begin{matrix}
			 \eta^{-k-l} & 0 & 0 \\
			 0 & \eta^{l} & 0 \\
			 0 & 0 & \eta^{k}
			\end{matrix}\right).
		\end{split}
		\end{displaymath}
	\item $n\in 3\mathbb{N}\backslash\{0\}$
		\begin{displaymath}
		\begin{split}
		& \textbf{\underline{1}}_{r,s} : a\mapsto \omega^{r}, c\mapsto \omega^{s}, d\mapsto \omega^{s}. \\
		& \textbf{\underline{3}}_{(k,l)} : a\mapsto
			\left(\begin{matrix}
			 0 & 1 & 0 \\
			 0 & 0 & 1 \\
			 1 & 0 & 0
			\end{matrix}\right),\enspace
			c\mapsto
			\left(\begin{matrix}
			 \eta^{l} & 0 & 0 \\
			 0 & \eta^{k} & 0 \\
			 0 & 0 & \eta^{-k-l}
			\end{matrix}\right),\enspace
			d\mapsto
			\left(\begin{matrix}
			 \eta^{-k-l} & 0 & 0 \\
			 0 & \eta^{l} & 0 \\
			 0 & 0 & \eta^{k}
			\end{matrix}\right).
		\end{split}
		\end{displaymath}
	\end{enumerate}
Remark: Though not named $\Delta(3n^{2})$ and only shortly mentioned, the group $\Delta(3n^{2})$ is also treated in \cite{miller}. Miller et al. describe a group (C) generated by
	\begin{displaymath}
	H=\left(
	\begin{matrix}
 \alpha & 0 & 0 \\
 0 & \beta & 0 \\
 0 & 0 & \gamma
	\end{matrix}
	\right),\quad T=\left(\begin{matrix}
			 0 & 1 & 0 \\
			 0 & 0 & 1 \\
			 1 & 0 & 0
			\end{matrix}\right).
	\end{displaymath}
Setting $\alpha=\eta^{l}, \beta=\eta^{k}$ and $\gamma=\eta^{-k-l}$ one obtains generators of $\Delta(3n^{2})$. These generators are related to those given in \cite{luhn2} via
	\begin{displaymath}
	\textbf{\underline{3}}_{(k,l)} : a\mapsto
			T,\enspace
			c\mapsto
			H,\enspace
			d\mapsto
			T^{2}HT^{-2}.
	\end{displaymath}
From this we can also see that $a$ and $c$ alone generate $\Delta(3n^{2})$, and one could discard $d$.
\bigskip
\\
It turns out that the representations labeled by
	\begin{displaymath}
	(k,l),\enspace (-k-l,k),\enspace (l,-k-l)
	\end{displaymath}
are equivalent \cite{luhn2}. Considering the generators of $\textbf{\underline{3}}_{(k,l)}$, $k,l\in \mathbb{Z}$ would be allowed, but one would overcount the non-equivalent representations. Therefore Luhn et al. developed a \textquotedblleft standard form\textquotedblright\hspace{1mm} for the three-dimensional irreducible representations, which is indicated by an index $\widetilde{(k,l)}$. The details of this standard form are explained in \cite{luhn2}. For we are only interested in the \textquotedblleft types\textquotedblright\hspace{1mm} of Clebsch-Gordan coefficients that can occur, we can (and will) at many places drop the notation $\widetilde{(k,l)}$.
\bigskip
\\
The tensor products are \cite{luhn2}:
	\begin{enumerate}
	 \item $n\not\in 3\mathbb{N}\backslash\{0\}$
		\begin{displaymath}
		\begin{split}
		\textbf{\underline{3}}_{(k,l)}\otimes \textbf{\underline{3}}_{(k',l')}=& \delta_{(k',l'),(-k,-l)}(\textbf{\underline{1}}_{0} \oplus \textbf{\underline{1}}_{1} \oplus\textbf{\underline{1}}_{2})\oplus\\ & \quad\oplus\textbf{\underline{3}}_{(k'+k,l'+l)}\oplus\\ & \quad\oplus\textbf{\underline{3}}_{(k'-k-l,l'+k)}\oplus\\ & \quad\oplus\textbf{\underline{3}}_{(k'+l,l'-k-l)}.
		\end{split}
		\end{displaymath}
		If $(k',l')=(-k,-l)$, then the representation $\textbf{\underline{3}}_{(k'+k,l'+l)}=\textbf{\underline{3}}_{(0,0)}$ has to be cancelled from the right hand side of the tensor product, thus the dimensions are correct.
	\item $n\in 3\mathbb{N}\backslash\{0\}$
		\begin{displaymath}
		\begin{split}
		\textbf{\underline{3}}_{(k,l)}\otimes \textbf{\underline{3}}_{(k',l')}=& \sum_{s=0}^{2}\delta_{(k',l'),(-k+sn/3,-l+sn/3)}(\textbf{\underline{1}}_{0,s} \oplus \textbf{\underline{1}}_{1,s} \oplus\textbf{\underline{1}}_{2,s})\oplus\\ & \quad\oplus\textbf{\underline{3}}_{(k'+k,l'+l)}\oplus\\ & \quad\oplus\textbf{\underline{3}}_{(k'-k-l,l'+k)}\oplus\\ & \quad\oplus\textbf{\underline{3}}_{(k'+l,l'-k-l)}.
		\end{split}
		\end{displaymath}
		Depending on $(k,l)$ there are 0, 3 or 9 one-dimensional irreducible representations in the tensor product (see later). If $(k',l')=(-k+sn/3,-l+sn/3)$ for an $s\in\{0,1,2\}$, then one of the three-dimensional representations becomes reducible and has to be cancelled from the right hand side of the tensor product, thus the dimensions are correct.
		\end{enumerate}
Remark: To be exact, one should write $\widetilde{\quad}$ over all indices of three-dimensional representations on the right hand side of the tensor products, but for the investigation of the general structure of Clebsch-Gordan coefficients this is not necessary.
\bigskip
\\
The Clebsch-Gordan coefficients for all tensor products are calculated in \cite{luhn2}. We will use their data, extract the information we need, and test, if the matrix $M$ of Clebsch-Gordan coefficients reduces the tensor product.
\medskip
\\
We start with the case $n\not\in 3\mathbb{N}\backslash\{0\}$.

\paragraph{\underline{(i) $n\not\in 3\mathbb{N}\backslash\{0\}$:}}
In the notation of \cite{luhn2} the Clebsch-Gordan coefficients for
	\begin{displaymath}
	\textbf{\underline{3}}_{(k,l)}\otimes \textbf{\underline{3}}_{(k',l')}= \delta_{(k',l'),(-k,-l)}(\textbf{\underline{1}}_{0} \oplus \textbf{\underline{1}}_{1} \oplus\textbf{\underline{1}}_{2})\oplus...
	\end{displaymath}
are given by
	\begin{equation}\label{cgcluhn1}
	\langle \textbf{\underline{3}}_{(k',l')}^{i'},\textbf{\underline{3}}_{(k,l)}^{i}\vert \textbf{\underline{1}}_{r}\rangle=\omega^{r(1-i)}\delta_{i',i}\delta_{(k',l'),(-k,-l)},
	\end{equation}
which is in our notation:
	\begin{displaymath}
	u^{\textbf{\underline{3}}_{(k',l')}\otimes\textbf{\underline{3}}_{(k,l)}}_{\textbf{\underline{1}}_{r}}=\omega^{r(1-i)}\delta_{i',i}\enspace e_{i'}\otimes e_{i}.
	\end{displaymath}
Thus \cite{luhn2} give the Clebsch-Gordan coefficients for $\textbf{\underline{3}}'\otimes\textbf{\underline{3}}$.
($\delta_{(k',l'),(-k,-l)}$ just denotes that $\textbf{\underline{1}}_{r}$ is not contained in $\textbf{\underline{3}}_{(k',l')}\otimes\textbf{\underline{3}}_{(k,l)}$, if $(k',l')\neq(-k,-l)$.)
\medskip
\\
The coefficients for 
	\begin{displaymath}
	\textbf{\underline{3}}_{(k',l')}\otimes \textbf{\underline{3}}_{(k,l)}= 3_{(k'',l'')}\oplus...
	\end{displaymath}
are given by \cite{luhn2}
	\begin{equation}\label{cgcluhn2}
	\langle \textbf{\underline{3}}_{(k',l')}^{i'},\textbf{\underline{3}}_{(k,l)}^{i}\vert \textbf{\underline{3}}_{(k'',l'')}^{i''}\rangle=\delta^{(3)}_{i'',i'-p}\delta^{(3)}_{i',i-q}\delta_{(\ast)},
	\end{equation}
where $\delta^{(3)}_{ij}=\left\{\begin{array}{ll}
  				1 \enspace\mbox{if}\enspace i\hspace{0.5mm}\mathrm{mod}\hspace{0.5mm} 3=j\hspace{0.5mm} \mathrm{mod}\hspace{0.5mm} 3\\
				0 \enspace\mbox{else}
                                \end{array}\right.$,$\quad$ and
$\delta_{(\ast)}=\delta_{\left(\begin{array}{ll}
                                k''\\l''
                               \end{array}\right),
			M^{p}\left[
				\left(\begin{array}{ll}
                                k'\\l'
                               \end{array}\right)
				+M^{q}
				\left(\begin{array}{ll}
                                k\\l
                               \end{array}\right)
				\right]}$.
\medskip
\\
$q$ numbers the three-dimensional irreducible representations contained in the tensor product:
	\begin{displaymath}
	\begin{split}
	& q=0,1,2 \quad\mbox{if}\enspace (k',l')\neq(-k,-l),\\
	& q=1,2\hspace{4mm} \quad\mbox{if}\enspace (k',l')=(-k,-l).
	\end{split}
	\end{displaymath}
The missing of $q=0$ is due to the fact that the tensor product contains only two 3-dimensional irreducible representations if $(k',l')=(-k,-l)$. The number $p$ has to do with the choice of a \textquotedblleft standard form\textquotedblright\hspace{1mm} for the representations (see \cite{luhn2} for more information), it can take the values $0,1,2$. The choice of $p$ only influences the order of the Clebsch-Gordan coefficients and furthermore we don't need to distinguish between equivalent representations, so we can choose $p:=0$ for our purpose.\footnote{$M$ is the matrix 			$\left(\begin{matrix}
		 -1 & -1 \\
		 1 & 0
		\end{matrix}\right)$. $M^{p}$, $p\in\{0,1,2\}$ maps a pair of indices $(k,l)^{T}$ onto the equivalent pairs of indices $(k,l)^{T}, (-k-l,k)^{T}, (l,-k-l)^{T}$.}
Another simplification is the following: $\delta_{(\ast)}$ determines $\textbf{\underline{3}}_{(k'',l'')}$ contained in the tensor product as a function of $p$ and $q$, thus we can always choose $(k'',l'')$ in such a way that $\delta_{(\ast)}=1$. We can therefore proceed in the following way:
	\begin{enumerate}
	 \item We construct the Clebsch-Gordan coefficients setting $p=0$, $\delta_{(\ast)}=1$.
	 \item We explicitly reduce the tensor product using the obtained Clebsch-Gordan coefficients.
	 \item We read off $(k'',l'')$ from the reduced tensor product.
	\end{enumerate}

\paragraph{\underline{(ia) $(k',l')=(-k,-l)$:}}
In this case we have:
	\begin{equation}\label{caseiatensorprodequ}
	\textbf{\underline{3}}_{(-k,-l)}\otimes\textbf{\underline{3}}_{(k,l)}=\textbf{\underline{1}}_{0}\oplus\textbf{\underline{1}}_{1}\oplus\textbf{\underline{1}}_{2}\oplus\textbf{\underline{3}}_{(-2k-l,k-l)}\oplus\textbf{\underline{3}}_{(l-k,-k-2l)}.
	\end{equation}
Remark: $\textbf{\underline{3}}_{(k,l)}=\textbf{\underline{3}}_{(-k,-l)}^{\ast}$.
\medskip
\\
The allowed values for $q$ are $1$ and $2$. Using equations (\ref{cgcluhn1}) and (\ref{cgcluhn2}) we can construct basis vectors of the invariant subspaces. We find
	\begin{displaymath}
	\begin{split}
	& u^{\textbf{\underline{3}}_{(-k,-l)}\otimes \textbf{\underline{3}}_{(k,l)}}_{\textbf{\underline{1}}_{0}}=e_{11}+e_{22}+e_{33},\\
	& u^{\textbf{\underline{3}}_{(-k,-l)}\otimes \textbf{\underline{3}}_{(k,l)}}_{\textbf{\underline{1}}_{1}}=e_{11}+\omega^{2}e_{22}+\omega e_{33},\\
	& u^{\textbf{\underline{3}}_{(-k,-l)}\otimes \textbf{\underline{3}}_{(k,l)}}_{\textbf{\underline{1}}_{2}}=e_{11}+\omega e_{22}+\omega^{2} e_{33},\\
	& u^{\textbf{\underline{3}}_{(-k,-l)}\otimes \textbf{\underline{3}}_{(k,l)}}_{\textbf{\underline{3}}_{q=1}}(1)=e_{12},\\
	& u^{\textbf{\underline{3}}_{(-k,-l)}\otimes \textbf{\underline{3}}_{(k,l)}}_{\textbf{\underline{3}}_{q=1}}(2)=e_{23},\\
	& u^{\textbf{\underline{3}}_{(-k,-l)}\otimes \textbf{\underline{3}}_{(k,l)}}_{\textbf{\underline{3}}_{q=1}}(3)=e_{31},\\
	& u^{\textbf{\underline{3}}_{(-k,-l)}\otimes \textbf{\underline{3}}_{(k,l)}}_{\textbf{\underline{3}}_{q=2}}(1)=e_{13},\\
	& u^{\textbf{\underline{3}}_{(-k,-l)}\otimes \textbf{\underline{3}}_{(k,l)}}_{\textbf{\underline{3}}_{q=2}}(2)=e_{21},\\
	& u^{\textbf{\underline{3}}_{(-k,-l)}\otimes \textbf{\underline{3}}_{(k,l)}}_{\textbf{\underline{3}}_{q=2}}(3)=e_{32}.\\
	\end{split}
	\end{displaymath}
Normalizing the basis vectors we can form a unitary matrix $M$ of Clebsch-Gordan coefficients
	\begin{displaymath}
	M=\left(\begin{matrix}
	 \frac{1}{\sqrt{3}} & \frac{1}{\sqrt{3}} & \frac{1}{\sqrt{3}} & 0 & 0 & 0 & 0 & 0 & 0 \\
	 0 & 0 & 0 & 1 & 0 & 0 & 0 & 0 & 0 \\
	 0 & 0 & 0 & 0 & 0 & 0 & 1 & 0 & 0 \\
	 0 & 0 & 0 & 0 & 0 & 0 & 0 & 1 & 0 \\
	 \frac{1}{\sqrt{3}} & \frac{\omega^{2}}{\sqrt{3}} & \frac{\omega}{\sqrt{3}} & 0 & 0 & 0 & 0 & 0 & 0 \\
	 0 & 0 & 0 & 0 & 1 & 0 & 0 & 0 & 0 \\
	 0 & 0 & 0 & 0 & 0 & 1 & 0 & 0 & 0 \\
	 0 & 0 & 0 & 0 & 0 & 0 & 0 & 0 & 1 \\
	 \frac{1}{\sqrt{3}} & \frac{\omega}{\sqrt{3}} & \frac{\omega^{2}}{\sqrt{3}} & 0 & 0 & 0 & 0 & 0 & 0
	  \end{matrix}\right).
	\end{displaymath}
We will now test, if this matrix reduces the tensor product. We find:
	\begin{displaymath}
	\begin{split}
	& M^{-1}[\textbf{\underline{3}}_{(-k,-l)}(a)\otimes\textbf{\underline{3}}_{(k,l)}(a)]M=
	\left(\begin{matrix}
	 1 & 0 & 0 & 0 & 0 & 0 & 0 & 0 & 0 \\
	 0 & \omega^{2} & 0 & 0 & 0 & 0 & 0 & 0 & 0 \\
	 0 & 0 & \omega & 0 & 0 & 0 & 0 & 0 & 0 \\
	 0 & 0 & 0 & 0 & 1 & 0 & 0 & 0 & 0 \\
	 0 & 0 & 0 & 0 & 0 & 1 & 0 & 0 & 0 \\
	 0 & 0 & 0 & 1 & 0 & 0 & 0 & 0 & 0 \\
	 0 & 0 & 0 & 0 & 0 & 0 & 0 & 1 & 0 \\
	 0 & 0 & 0 & 0 & 0 & 0 & 0 & 0 & 1 \\
	 0 & 0 & 0 & 0 & 0 & 0 & 1 & 0 & 0
	  \end{matrix}\right)=\\
	& =\left(\begin{matrix}
			 (\textbf{\underline{1}}_{1}\oplus\textbf{\underline{1}}_{2}\oplus\textbf{\underline{1}}_{3})(a) & \textbf{0} & \textbf{0} \\
			 \textbf{0} & \textbf{\underline{3}}_{(-2k-l,k-l)}(a) & \textbf{0} \\
			 \textbf{0} & \textbf{0} & \textbf{\underline{3}}_{(-k+l,-k-2l)}(a)
			\end{matrix}\right),
	\end{split}
	\end{displaymath}
	\begin{displaymath}
	M^{-1}[\textbf{\underline{3}}_{(-k,-l)}(c)\otimes\textbf{\underline{3}}_{(k,l)}(c)]M=\left(\begin{matrix}
			 \mathbbm{1}_{3} & \textbf{0} & \textbf{0} \\
			 \textbf{0} & \textbf{\underline{3}}_{(-2k-l,k-l)}(c) & \textbf{0} \\
			 \textbf{0} & \textbf{0} & \textbf{\underline{3}}_{(-k+l,-k-2l)}(c)
			\end{matrix}\right),
	\end{displaymath}
	\begin{displaymath}
	M^{-1}[\textbf{\underline{3}}_{(-k,-l)}(d)\otimes\textbf{\underline{3}}_{(k,l)}(d)]M=\left(\begin{matrix}
			 \mathbbm{1}_{3} & \textbf{0} & \textbf{0} \\
			 \textbf{0} & \textbf{\underline{3}}_{(-2k-l,k-l)}(d) & \textbf{0} \\
			 \textbf{0} & \textbf{0} & \textbf{\underline{3}}_{(-k+l,-k-2l)}(d)
	\end{matrix}\right),
	\end{displaymath}
where $\textbf{0}$ is the $3\times3$-nullmatrix.
\\
Thus $M$ reduces the tensor product to exactly those representations given in equation (\ref{caseiatensorprodequ}). Thus we have confirmed that the Clebsch-Gordan coefficients are correct. We list them in table \ref{Delta3nnCGC1}. It turns out that these Clebsch-Gordan coefficients are equivalent to those listed in table \ref{A4CGC} ($A_{4}$), which can be easily verified by reducing the Kronecker product $[\textbf{\underline{3}}\otimes \textbf{\underline{3}}]$ of $A_{4}$ using the Clebsch-Gordan coefficients given here (One obtains 1- and 3-dimensional irreducible representations of $A_{4}$ that are equivalent to those given in subsection \ref{subsectionA4}). 

%Folgendes Argument falsch?

%Another argument is the following: Consider the basis vectors of the invariant subspaces given in table \ref{A4CGC}. Since $\textbf{\underline{3}}_{a}=\textbf{\underline{3}}$ occurs twice in the reduced tensor product (which only occurs at $\textbf{\underline{3}}\otimes \textbf{\underline{3}}$ of $A_{4}$) the set of solutions of the invariance equations for $\textbf{\underline{3}}\otimes \textbf{\underline{3}}=\textbf{\underline{3}}\oplus...$ is 2-parametric. Thus any two bases formed of linear combinations of the bases of the two invariant subspaces $V_{\textbf{\underline{3}}}$ will again span two invariant subspaces of $V_{\textbf{\underline{3}}}\otimes V_{\textbf{\underline{3}}}$. The bases of $V_{\textbf{\underline{3}}_{(-2k-l,k-l)}}$ and $V_{\textbf{\underline{3}}_{(-k+l,-k-2l)}}$ are such linear combinations.

\paragraph{\underline{(ib) $(k',l')\neq(-k,-l)$:}}
In this case the three 1-dimensional irreducible representations are replaced by a 3-dimensional irreducible representation.
	\begin{displaymath}
		\textbf{\underline{3}}_{(k',l')}\otimes \textbf{\underline{3}}_{(k,l)}=\textbf{\underline{3}}_{(k'+k,l'+l)}\oplus \textbf{\underline{3}}_{(k'-k-l,l'+k)}\oplus \textbf{\underline{3}}_{(k'+l,l'-k-l)}.
		\end{displaymath}
Proceeding exactly as before one finds that the Clebsch-Gordan coefficients given in \cite{luhn2} are correct. We list the normalized coefficients in table \ref{Delta3nnCGC2}.
%Unklar ob folgendes Argument wirklich stimmt.
% Since the special case $\textbf{\underline{3}}_{(k,l)}\otimes \textbf{\underline{3}}_{(k,l)}$ is included here, one of the sets of Clebsch-Gordan coefficients given in table \ref{Delta3nnCGC2} must be equivalent to the coefficients for $\textbf{\underline{3}}\otimes \textbf{\underline{3}}=\textbf{\underline{3}}_{a}\oplus...$ .
\begin{table}
\begin{center}
\renewcommand{\arraystretch}{1.5}
\begin{tabular}[p]{|l|l|l|}
\firsthline
	$\Delta(3n^{2})$ & $\textbf{\underline{3}}_{(-k,-l)}\otimes\textbf{\underline{3}}_{(k,l)}$ & CGC\\
\hline
	$\textbf{\underline{1}}_{0}$
	&	
	$u^{\textbf{\underline{3}}_{(-k,-l)}\otimes\textbf{\underline{3}}_{(k,l)}}_{\textbf{\underline{1}}_{0}}=\frac{1}{\sqrt{3}}e_{11}+\frac{1}{\sqrt{3}}e_{22}+\frac{1}{\sqrt{3}}e_{33}$ & \textbf{Ia}\\

	$\textbf{\underline{1}}_{1}$
	&
	$u^{\textbf{\underline{3}}_{(-k,-l)}\otimes\textbf{\underline{3}}_{(k,l)}}_{\textbf{\underline{1}}_{1}}=\frac{1}{\sqrt{3}}e_{11}+\frac{\omega^{2}}{\sqrt{3}}e_{22}+\frac{\omega}{\sqrt{3}}e_{33}$ & \textbf{Ib}$^{(\ast)}$\\

	$\textbf{\underline{1}}_{2}$
	&
	$u^{\textbf{\underline{3}}_{(-k,-l)}\otimes\textbf{\underline{3}}_{(k,l)}}_{\textbf{\underline{1}}_{2}}=\frac{1}{\sqrt{3}}e_{11}+\frac{\omega}{\sqrt{3}}e_{22}+\frac{\omega^{2}}{\sqrt{3}}e_{33}$ & \textbf{Ic}$^{(\ast)}$\\

	$\textbf{\underline{3}}_{(-2k-l,k-l)}$
	&
	$u^{\textbf{\underline{3}}_{(-k,-l)}\otimes\textbf{\underline{3}}_{(k,l)}}_{\textbf{\underline{3}}_{(-2k-l,k-l)}}(1)=e_{12}$ & \textbf{IIIa$_{(1,0)}$}\\

	&
	$u^{\textbf{\underline{3}}_{(-k,-l)}\otimes\textbf{\underline{3}}_{(k,l)}}_{\textbf{\underline{3}}_{(-2k-l,k-l)}}(2)=e_{23}$ & \\

	&
	$u^{\textbf{\underline{3}}_{(-k,-l)}\otimes\textbf{\underline{3}}_{(k,l)}}_{\textbf{\underline{3}}_{(-2k-l,k-l)}}(3)=e_{31}$ & \\

	$\textbf{\underline{3}}_{(-k+l,-k-2l)}$
	&
	$u^{\textbf{\underline{3}}_{(-k,-l)}\otimes\textbf{\underline{3}}_{(k,l)}}_{\textbf{\underline{3}}_{(-k+l,-k-2l)}}(1)=e_{13}$ & \textbf{IIIa$_{(0,1)}$}\\

	&
	$u^{\textbf{\underline{3}}_{(-k,-l)}\otimes\textbf{\underline{3}}_{(k,l)}}_{\textbf{\underline{3}}_{(-k+l,-k-2l)}}(2)=e_{21}$ & \\

	&
	$u^{\textbf{\underline{3}}_{(-k,-l)}\otimes\textbf{\underline{3}}_{(k,l)}}_{\textbf{\underline{3}}_{(-k+l,-k-2l)}}(3)=e_{32}$ & \\
\lasthline
\end{tabular}
\caption[Clebsch-Gordan coefficients for $\textbf{\underline{3}}_{(-k,-l)}\otimes\textbf{\underline{3}}_{(k,l)}$ of $\Delta(3n^{2})$. $n\not\in 3\mathbb{N}\backslash\{0\}$.]{Clebsch-Gordan coefficients for $\textbf{\underline{3}}_{(-k,-l)}\otimes\textbf{\underline{3}}_{(k,l)}$ of $\Delta(3n^{2})$. $n\not\in 3\mathbb{N}\backslash\{0\}$. (\textbf{Ib}$^{(\ast)}$, \textbf{Ic}$^{(\ast)}$ $\rightarrow$ \textbf{Ib} and \textbf{Ic} up to an (irrelevant) phase factor.)}
\label{Delta3nnCGC1}
\end{center}
\end{table}

\begin{table}
\begin{center}
\renewcommand{\arraystretch}{1.5}
\begin{tabular}[p]{|l|l|l|}
\firsthline
	$\Delta(3n^{2})$ & $\textbf{\underline{3}}_{(k',l')\neq (-k,-l)}\otimes\textbf{\underline{3}}_{(k,l)}$ & CGC\\
\hline
	$\textbf{\underline{3}}_{(k'+k,l'+l)}$
	&	
	$u^{\textbf{\underline{3}}_{(k',l')}\otimes\textbf{\underline{3}}_{(k,l)}}_{\textbf{\underline{3}}_{(k'+k,l'+l)}}(1)=e_{11}$ & \textbf{IIId}\\

	&	
	$u^{\textbf{\underline{3}}_{(k',l')}\otimes\textbf{\underline{3}}_{(k,l)}}_{\textbf{\underline{3}}_{(k'+k,l'+l)}}(2)=e_{22}$ & \\

	&	
	$u^{\textbf{\underline{3}}_{(k',l')}\otimes\textbf{\underline{3}}_{(k,l)}}_{\textbf{\underline{3}}_{(k'+k,l'+l)}}(3)=e_{33}$ & \\

	$\textbf{\underline{3}}_{(k'-k-l,l'+k)}$
	&
	$u^{\textbf{\underline{3}}_{(k',l')}\otimes\textbf{\underline{3}}_{(k,l)}}_{\textbf{\underline{3}}_{(k'-k-l,l'+k)}}(1)=e_{12}$ & \textbf{IIIa$_{(1,0)}$}\\

	&
	$u^{\textbf{\underline{3}}_{(k',l')}\otimes\textbf{\underline{3}}_{(k,l)}}_{\textbf{\underline{3}}_{(k'-k-l,l'+k)}}(2)=e_{23}$ & \\

	&
	$u^{\textbf{\underline{3}}_{(k',l')}\otimes\textbf{\underline{3}}_{(k,l)}}_{\textbf{\underline{3}}_{(k'-k-l,l'+k)}}(3)=e_{31}$ & \\

	$\textbf{\underline{3}}_{(k'+l,l'-k-l)}$
	&
	$u^{\textbf{\underline{3}}_{(k',l')}\otimes\textbf{\underline{3}}_{(k,l)}}_{\textbf{\underline{3}}_{(k'+l,l'-k-l)}}(1)=e_{13}$ & \textbf{IIIa$_{(0,1)}$}\\

	&
	$u^{\textbf{\underline{3}}_{(k',l')}\otimes\textbf{\underline{3}}_{(k,l)}}_{\textbf{\underline{3}}_{(k'+l,l'-k-l)}}(2)=e_{21}$ & \\

	&
	$u^{\textbf{\underline{3}}_{(k',l')}\otimes\textbf{\underline{3}}_{(k,l)}}_{\textbf{\underline{3}}_{(k'+l,l'-k-l)}}(3)=e_{32}$ & \\
\lasthline
\end{tabular}
\caption{Clebsch-Gordan coefficients for $\textbf{\underline{3}}_{(k',l')\neq(-k,-l)}\otimes\textbf{\underline{3}}_{(k,l)}$ of $\Delta(3n^{2})$. $n\not\in 3\mathbb{N}\backslash\{0\}$.}
\label{Delta3nnCGC2}
\end{center}
\end{table}
\paragraph{\underline{(ii) $n\in 3\mathbb{N}\backslash\{0\}$:}}
Luhn et al. \cite{luhn2} list the following Clebsch-Gordan coefficients for the case $n\in 3\mathbb{N}\backslash\{0\}$:
	\begin{equation}\label{cgcluhn3}
	\langle \textbf{\underline{3}}^{i'}_{(k',l')}, \textbf{\underline{3}}^{i}_{(k,l)}\vert \textbf{\underline{1}}_{r,s}\rangle=\omega^{r(1-i)}\delta^{(3)}_{i',i\pm s}\delta_{(k',l'),(-k,-l)} \quad\mbox{if}\quad (k,l)=(0,\pm\frac{n}{3}),
	\end{equation}
	\begin{equation}\label{cgcluhn4}
	\langle \textbf{\underline{3}}^{i'}_{(k',l')}, \textbf{\underline{3}}^{i}_{(k,l)}\vert \textbf{\underline{1}}_{r,s}\rangle=\omega^{r(1-i)}\delta^{(3)}_{i',i-p}\delta_{(\ast\ast)} \quad\mbox{if}\quad (k,l)\neq(0,\pm\frac{n}{3}),
	\end{equation}
where $\delta_{(\ast\ast)}=\delta_{\left(\begin{array}{ll} k'\\l'\end{array}\right),M^{p}\left(\begin{array}{ll} -k+sn/3\\-l+sn/3\end{array}\right)}$. Again as before, the determination of $p$ ensures that the representations on the right hand side of the tensor product have \textquotedblleft standard form\textquotedblright. Since we don't need to distinguish between equivalent representations, we can choose $p:=0$. The Clebsch-Gordan coefficients for the $3$-dimensional irreducible representations are the same as in (i).
\medskip
\\
As is shown in \cite{luhn2} there are the following cases:
\paragraph{\underline{(iia):}}
$k,l,k',l'$ are multiples of $\frac{n}{3}$. In this case $\textbf{\underline{3}}_{(k,l)}\otimes\textbf{\underline{3}}_{(k',l')}$ contains either $0$ or $9$ one-dimensional irreducible representations. The Clebsch-Gordan coefficients for $0$ one-dimensional representations ($3$ three-dimensional representations) are the same as for (i). The split-up into $9$ one-dimensional representations occurs if $(k,l)=(-k',-l')=(0,\pm \frac{n}{3})$. We only need to consider the case $l=+\frac{n}{3}$, since $l'=-l$ (therefore $l=-\frac{n}{3}$ just corresponds to reversing the order of the factors of the tensor product) . The reducing matrix $M$ derived from equation (\ref{cgcluhn3}) is given by
	\begin{displaymath}
	 M=\frac{1}{\sqrt{3}}\left(\begin{matrix}
 1 & 1 & 1 & 0 & 0 & 0 & 0 & 0 & 0 \\
 0 & 0 & 0 & 0 & 0 & 0 & 1 & 1 & 1 \\
 0 & 0 & 0 & 1 & 1 & 1 & 0 & 0 & 0 \\
 0 & 0 & 0 & 1 & \omega^{2} & \omega & 0 & 0 & 0 \\
 1 & \omega^{2} & \omega & 0 & 0 & 0 & 0 & 0 & 0 \\
 0 & 0 & 0 & 0 & 0 & 0 & 1 & \omega^{2} & \omega \\
 0 & 0 & 0 & 0 & 0 & 0 & 1 & \omega & \omega^{2} \\
 0 & 0 & 0 & 1 & \omega & \omega^{2} & 0 & 0 & 0 \\
 1 & \omega & \omega^{2} & 0 & 0 & 0 & 0 & 0 & 0
	   \end{matrix}\right).
	\end{displaymath}
Testing the coefficients as in case (ia) we find that they reduce the tensor product to nine 1-dimensional representations.\footnote{Remark: When one reduces the tensor product using the coefficients given by \cite{luhn2}, one finds that the order of the 1-dimensional representations obtained is different to the expected one. There may be an error in the construction of the reducing matrix $M$ from equation (\ref{cgcluhn3}), or the error is already in equation (\ref{cgcluhn3}). It seems as the order would become right, if $\omega^{r(1-i)}$ was replaced by $\omega^{r(i-1)}$ in equation (\ref{cgcluhn3}). Table \ref{Delta3nnCGC3} shows the representations in the order one gets when one uses $M$ for the reduction. However, the order of the representations is completely irrelevant for our studies.} We list the normalized coefficients in table \ref{Delta3nnCGC3}.
\\
Remark: $\textbf{\underline{3}}_{(0,\frac{n}{3})}=\textbf{\underline{3}}_{(0,-\frac{n}{3})}^{\ast}$.

\begin{table}
\begin{center}
\renewcommand{\arraystretch}{1.6}
\begin{tabular}{|l|l|l|}
\firsthline
	$\Delta(3n^{2})$ & $\textbf{\underline{3}}_{(0,-\frac{n}{3})}\otimes\textbf{\underline{3}}_{(0,\frac{n}{3})}$ & CGC\\
\hline
	$\textbf{\underline{1}}_{0,0}$
	&	
	$u^{\textbf{\underline{3}}_{(0,-\frac{n}{3})}\otimes\textbf{\underline{3}}_{(0,\frac{n}{3})}}_{\textbf{\underline{1}}_{0,0}}=\frac{1}{\sqrt{3}}e_{11}+\frac{1}{\sqrt{3}}e_{22}+\frac{1}{\sqrt{3}}e_{33}$ & \textbf{Ia}\\

	$\textbf{\underline{1}}_{2,0}$
	&
	$u^{\textbf{\underline{3}}_{(0,-\frac{n}{3})}\otimes\textbf{\underline{3}}_{(0,\frac{n}{3})}}_{\textbf{\underline{1}}_{2,0}}=\frac{1}{\sqrt{3}}e_{11}+\frac{\omega^{2}}{\sqrt{3}}e_{22}+\frac{\omega}{\sqrt{3}}e_{33}$ & \textbf{Ib}$^{(\ast)}$\\

	$\textbf{\underline{1}}_{1,0}$
	&
	$u^{\textbf{\underline{3}}_{(0,-\frac{n}{3})}\otimes\textbf{\underline{3}}_{(0,\frac{n}{3})}}_{\textbf{\underline{1}}_{2,0}}=\frac{1}{\sqrt{3}}e_{11}+\frac{\omega}{\sqrt{3}}e_{22}+\frac{\omega^{2}}{\sqrt{3}}e_{33}$ & \textbf{Ic}$^{(\ast)}$\\

	$\textbf{\underline{1}}_{0,1}$
	&	
	$u^{\textbf{\underline{3}}_{(0,-\frac{n}{3})}\otimes\textbf{\underline{3}}_{(0,\frac{n}{3})}}_{\textbf{\underline{1}}_{0,1}}=\frac{1}{\sqrt{3}}e_{13}+\frac{1}{\sqrt{3}}e_{21}+\frac{1}{\sqrt{3}}e_{32}$ & \textbf{Id}\\

	$\textbf{\underline{1}}_{2,1}$
	&	
	$u^{\textbf{\underline{3}}_{(0,-\frac{n}{3})}\otimes\textbf{\underline{3}}_{(0,\frac{n}{3})}}_{\textbf{\underline{1}}_{1,1}}=\frac{1}{\sqrt{3}}e_{13}+\frac{\omega^{2}}{\sqrt{3}}e_{21}+\frac{\omega}{\sqrt{3}}e_{32}$ & \textbf{Ie}\\

	$\textbf{\underline{1}}_{1,1}$
	&	
	$u^{\textbf{\underline{3}}_{(0,-\frac{n}{3})}\otimes\textbf{\underline{3}}_{(0,\frac{n}{3})}}_{\textbf{\underline{1}}_{2,1}}=\frac{1}{\sqrt{3}}e_{13}+\frac{\omega}{\sqrt{3}}e_{21}+\frac{\omega^{2}}{\sqrt{3}}e_{32}$ & \textbf{If}\\

	$\textbf{\underline{1}}_{0,2}$
	&	
	$u^{\textbf{\underline{3}}_{(0,-\frac{n}{3})}\otimes\textbf{\underline{3}}_{(0,\frac{n}{3})}}_{\textbf{\underline{1}}_{0,2}}=\frac{1}{\sqrt{3}}e_{12}+\frac{1}{\sqrt{3}}e_{23}+\frac{1}{\sqrt{3}}e_{31}$ & \textbf{Ig}\\

	$\textbf{\underline{1}}_{2,2}$
	&	
	$u^{\textbf{\underline{3}}_{(0,-\frac{n}{3})}\otimes\textbf{\underline{3}}_{(0,\frac{n}{3})}}_{\textbf{\underline{1}}_{1,2}}=\frac{1}{\sqrt{3}}e_{12}+\frac{\omega^{2}}{\sqrt{3}}e_{23}+\frac{\omega}{\sqrt{3}}e_{31}$ & \textbf{Ih}\\

	$\textbf{\underline{1}}_{1,2}$
	&	
	$u^{\textbf{\underline{3}}_{(0,-\frac{n}{3})}\otimes\textbf{\underline{3}}_{(0,\frac{n}{3})}}_{\textbf{\underline{1}}_{2,2}}=\frac{1}{\sqrt{3}}e_{12}+\frac{\omega}{\sqrt{3}}e_{23}+\frac{\omega^{2}}{\sqrt{3}}e_{31}$ & \textbf{Ii}\\
\lasthline
\end{tabular}
\caption[Clebsch-Gordan coefficients for $\textbf{\underline{3}}_{(0,-\frac{n}{3})}\otimes\textbf{\underline{3}}_{(0,\frac{n}{3})}$ of $\Delta(3n^{2})$. $n\in 3\mathbb{N}\backslash\{0\}$.]{Clebsch-Gordan coefficients for $\textbf{\underline{3}}_{(0,-\frac{n}{3})}\otimes\textbf{\underline{3}}_{(0,\frac{n}{3})}$ of $\Delta(3n^{2})$. $n\in 3\mathbb{N}\backslash\{0\}$. ($(\ast)$... up to an (irrelevant) phase factor)}
\label{Delta3nnCGC3}
\end{center}
\end{table}

\paragraph{\underline{(iib):}}
$k,l,k',l'$ are not all multiples of $\frac{n}{3}$. The coefficients for the 1-dimensional representations are determined by equation (\ref{cgcluhn4}). Again we set $p:=0$, and we see that $\delta_{(\ast\ast)}$ becomes
	\begin{displaymath}
	\delta_{\left(\begin{array}{ll} k'\\l'\end{array}\right), \left(\begin{array}{ll} -k+sn/3\\-l+sn/3\end{array}\right)}.
	\end{displaymath}
This corresponds to the analysis of the tensor product $\textbf{\underline{3}}_{(-k,-l)}\otimes \textbf{\underline{3}}_{(k,l)}$, thus the coefficients are the same as in case (ia). The coefficients for the three-dimensional representations are equal to case (ia) too. Trying to reduce the tensor product using these coefficients one finds that they are correct.

\subsection{The group $\Delta(6n^{2})$}\label{d6nnsubsection}

The infinite series $\Delta(6n^{2})$ ($n\in \mathbb{N}\backslash\{0,1\}$) of finite subgroups of $SU(3)$ has been intensively studied by Escobar and Luhn in \cite{escobar}, as well as by Bovier, L\"uling and Wyler in \cite{BLW1}. Escobar et al. give all important information on $\Delta(6n^{2})$, including generators, conjugate classes, character tables, tensor products and Clebsch-Gordan coefficients. We will now summarize those properties of the groups given in \cite{escobar} that we will need for our purpose.
\medskip
\\
According to \cite{escobar} $\Delta(6n^2)$ is isomorphic to the semidirect product group
	\begin{displaymath}
	(Z_n \times Z_n)\rtimes S_3.
	\end{displaymath}
$\Delta(6n^2)$ is generated by four generators $a,b,c,d$ fulfilling \cite{escobar}
	\begin{equation}\label{d6nnpresentation1}
	\begin{split}
	& a^3=b^2=(ab)^2=e \quad(\mbox{generators of }S_3),\\
	& c^n=d^n=e,\enspace cd=dc\quad(\mbox{generators of }Z_n\times Z_n).
	\end{split}
	\end{equation}
Due to the semidirect product structure there is an action of the group $S_3$ on the group $Z_n\times Z_n$, which is given by \cite{escobar}
	\begin{equation}\label{d6nnpresentation2}
	\begin{split}
	& \phi(a)(c)=aca^{-1}=c^{-1}d^{-1},\\
	& \phi(a)(d)=ada^{-1}=c,\\
	& \phi(b)(c)=bcb^{-1}=d^{-1},\\
	& \phi(b)(d)=bdb^{-1}=c^{-1}.
	\end{split}
	\end{equation}
Similar to the case of $\Delta(3n^2)$ $\phi$ is a homomorphism
	\begin{displaymath}
	\begin{split}
	\phi: S_3\rightarrow \mathrm{Aut}(Z_n\times Z_n),
	\end{split}
	\end{displaymath}
which is given by $\phi(x)(y)=xyx^{-1}\enspace\forall x\in S_3, y\in Z_n\times Z_n$.
\bigskip
\\
Table \ref{Delta6nn-irreps} shows the irreducible representations of $\Delta(6n^{2})$.

\begin{table}
\begin{footnotesize}
\begin{center}
\renewcommand{\arraystretch}{1.4}
\begin{tabular}{|l|ccl|}
\firsthline
	$\Delta(6n^{2})$ & Representations & number of representations & parameter values \\
\hline
	$n\in 3\mathbb{N}\backslash\{0\}$ & $\textbf{\underline{1}}_{r}$ & $2$ & $r=1,2$\\
	 & $\textbf{\underline{2}}_{i}$ & $4$ & $i=1,2,3,4$\\
	 & $\textbf{\underline{3}}_{j(l)}$ & $2(n-1)$ & $j=1,2$; $l=1,...,n-1$\\
	 & $\textbf{\underline{6}}_{\widetilde{(k,l)}}$ & $\frac{n(n-3)}{6}$ & $k,l=0,...,n-1$\\
	 & & & $(k,l)\neq (\frac{n}{3},\frac{n}{3}),(\frac{2n}{3},\frac{2n}{3})$\\
	 & & & $(k+l)\hspace{1mm} \mathrm{mod}\hspace{1mm} n\neq 0$\\
	 & & & $(k,l)\neq (0,l),(k,0)$\\
\hline
	$n\not\in 3\mathbb{N}\backslash\{0\}$ & $\textbf{\underline{1}}_{r}$ & $2$ & $r=1,2$\\
	 & $\textbf{\underline{2}}_{1}$ & $1$ & \\
	 & $\textbf{\underline{3}}_{j(l)}$ & $2(n-1)$ & $j=1,2$; $l=1,...,n-1$\\
	 & $\textbf{\underline{6}}_{\widetilde{(k,l)}}$ & $\frac{(n-1)(n-2)}{6}$ & $k,l=0,...,n-1$\\
	 & & & $(k+l)\hspace{1mm} \mathrm{mod}\hspace{1mm} n\neq 0$\\
	 & & & $(k,l)\neq (0,l),(k,0)$\\
\lasthline
\end{tabular}
\caption[Irreducible representations of $\Delta(6n^{2})$.]{Irreducible representations of $\Delta(6n^{2})$ as given in \cite{escobar}. The indices $k$ and $l$ have a precise mathematical meaning when refering to the generators of the representations. In fact they can take all integer values, but then some of the representations listed above would be reducible or equivalent. The restrictions for $k$,$l$ ensure that every irreducible representation occurs exactly once in the above table.}
\label{Delta6nn-irreps}
\end{center}
\end{footnotesize}
\end{table}
\hspace{0mm}\\
The generators of $\Delta(6n^{2})$ for the three-dimensional representations $\textbf{\underline{3}}_{1(l)}$ have the form
	\begin{displaymath}
	\begin{split}
	& A_{3_{1(l)}}=\left(\begin{matrix}
		0  & 1 & 0 \\
		0  & 0 & 1 \\
		1  & 0 & 0
	          \end{matrix}\right),\quad
	B_{3_{1(l)}}=\left(\begin{matrix}
		0  & 0 & 1 \\
		0  & 1 & 0 \\
		1  & 0 & 0
	          \end{matrix}\right),\\
	& C_{3_{1(l)}}=\left(\begin{matrix}
		\eta^{l}  & 0 & 0 \\
		0  & \eta^{-l} & 0 \\
		0  & 0 & 1
	          \end{matrix}\right),\quad
	D_{3_{1(l)}}=\left(\begin{matrix}
		1 & 0 & 0 \\
		0  & \eta^{l} & 0 \\
		0  & 0 & \eta^{-l}
	          \end{matrix}\right),
	\end{split}
	\end{displaymath}
where $\eta=e^{\frac{2\pi i}{n}}$ \cite{escobar}.
\\
The generators of $\textbf{\underline{3}}_{2(l)}$ are the same, except that $B_{3_{2(l)}}=-B_{3_{1(l)}}$. From the explicit form of the generators one can easily see that $\textbf{\underline{3}}_{i(l)}=\textbf{\underline{3}}_{i(l\hspace{0.2mm} \mathrm{mod}\hspace{0.2mm}n)}\enspace \forall l$. It is also clear what happens if $(l+l')\hspace{1mm} \mathrm{mod}\hspace{1mm}n=0$. Then $C_{3_{i(l)}}=D_{3_{i(l)}}=\mathbbm{1}_{3}$, and we end up with a 3-dimensional representation of the group $S_{3}$, which breaks up into \cite{escobar}
	\begin{displaymath}
	\textbf{\underline{3}}_{i(0)}=\textbf{\underline{2}}_{1}\oplus\textbf{\underline{1}}_{i}.
	\end{displaymath}
Remark: Though not named $\Delta(6n^{2})$ and only shortly mentioned, the group $\Delta(6n^{2})$ is also treated in \cite{miller}. Miller et al. describe a group (D) ($\rightarrow$ subsection \ref{Dsubsection}) generated by
	\begin{displaymath}
	H=\left(
	\begin{matrix}
 \alpha & 0 & 0 \\
 0 & \beta & 0 \\
 0 & 0 & \gamma
	\end{matrix}
	\right),\quad T=\left(\begin{matrix}
			 0 & 1 & 0 \\
			 0 & 0 & 1 \\
			 1 & 0 & 0
			\end{matrix}\right), \quad
	R=\left(\begin{matrix}
			 a & 0 & 0 \\
			 0 & 0 & c \\
			 0 & b & 0
			\end{matrix}\right).
	\end{displaymath}
Setting $\alpha=\eta^{l}, \beta=\eta^{-l}$, $\gamma=1$ and $a=b=c=1$ one obtains generators of $\Delta(6n^{2})$. These generators are related to those given in \cite{escobar} via
	\begin{displaymath}
	\textbf{\underline{3}}_{1(l)}: A_{3_{1(l)}}=T, \enspace B_{3_{1(l)}}=TR,\enspace C_{3_{1(l)}}=H,\enspace D_{3_{1(l)}}=T^{2}HT^{-2}.
	\end{displaymath}
From this we can also see that $A, B$ and $C$ alone generate $\Delta(6n^{2})$, and one could discard $D$.
\medskip
\\
The generators of the 6-dimensional irreducible representations are \cite{escobar}:
	\begin{displaymath}
	\begin{split}
	& A_{6_{(k,l)}}=\left(\begin{matrix}
		 A_{1} & \textbf{0} \\
		     \textbf{0} & A_{2}
	          \end{matrix}\right),\quad
	B_{6_{(k,l)}}=\left(\begin{matrix}
		\textbf{0} & \mathbbm{1}_{3} \\
		\mathbbm{1}_{3} & \textbf{0}
	          \end{matrix}\right),\\
	& C_{6_{(k,l)}}=\left(\begin{matrix}
		 C_{1} & \textbf{0} \\
		\textbf{0} & C_{2} 
	          \end{matrix}\right),\quad
	 D_{6_{(k,l)}}=\left(\begin{matrix}
		 D_{1} & \textbf{0} \\
		\textbf{0} & D_{2} 
	          \end{matrix}\right),
	\end{split}
	\end{displaymath}
with
	\begin{displaymath}
	\begin{split}
	& A_{1}=A_{2}^{T}=\left(\begin{matrix}
		0  & 1 & 0 \\
		0  & 0 & 1 \\
		1  & 0 & 0
	          \end{matrix}\right),\quad
	C_{1}=D_{2}^{-1}=\left(\begin{matrix}
		\eta^{l}  & 0 & 0 \\
		0  & \eta^{k} & 0 \\
		0  & 0 & \eta^{-l-k}
	          \end{matrix}\right),\\
	& C_{2}=D_{1}^{-1}=\left(\begin{matrix}
		\eta^{l+k} & 0 & 0 \\
		0  & \eta^{-l} & 0 \\
		0  & 0 & \eta^{-k}
	          \end{matrix}\right).
	\end{split}
	\end{displaymath}
\textbf{0} denotes the $3\times 3$-nullmatrix.
\medskip
\\
Remark: Care must be taken especially for the six-dimensional representations. They can be labeled by two integer numbers $k$ and $l$. It turns out that the representations labeled by
	\begin{displaymath}
	(k,l),\enspace (-k-l,k),\enspace (l,-k-l),\enspace (-l,-k),\enspace (k+l,-l),\enspace (-k,k+l)
	\end{displaymath}
are equivalent, which we denote by an index $\widetilde{(k,l)}$, as in the similar situation of the 3-dimensional representations of $\Delta(3n^{2})$.
\medskip
\\
The tensor products we will need are:
	\begin{equation}\label{D6nntensorequ}
	\begin{split}
	& \textbf{\underline{3}}_{1(l)}\otimes\textbf{\underline{3}}_{1(l')}=\textbf{\underline{3}}_{1(l+l')}\oplus\textbf{\underline{6}}_{\widetilde{(l,-l')}},\\
	& \textbf{\underline{3}}_{1(l)}\otimes\textbf{\underline{3}}_{2(l')}=\textbf{\underline{3}}_{2(l+l')}\oplus\textbf{\underline{6}}_{\widetilde{(l,-l')}},\\
	& \textbf{\underline{3}}_{2(l)}\otimes\textbf{\underline{3}}_{2(l')}=\textbf{\underline{3}}_{1(l+l')}\oplus\textbf{\underline{6}}_{\widetilde{(l,-l')}}.
	\end{split}
	\end{equation} 
Note that some of the 6-dimensional representations occurring in the tensor products (\ref{D6nntensorequ}) can be reducible \cite{escobar}, namely
	\begin{displaymath}
	(\textbf{\underline{6}}_{(-l,l)},\enspace\textbf{\underline{6}}_{(0,-l)},\enspace\textbf{\underline{6}}_{(l,0)})\sim \textbf{\underline{6}}_{(l,0)} \sim \textbf{\underline{3}}_{1(l)}\oplus\textbf{\underline{3}}_{2(l)}
	\end{displaymath}
and
	\begin{displaymath}
	(\textbf{\underline{6}}_{(\frac{n}{3},\frac{n}{3})},\enspace\textbf{\underline{6}}_{(\frac{2n}{3},\frac{2n}{3})})\sim \textbf{\underline{2}}_{2}\oplus\textbf{\underline{2}}_{3}\oplus\textbf{\underline{2}}_{4}.
	\end{displaymath}
\hspace{0mm}
\\
In tables \ref{Delta6nnCGCa} and \ref{Delta6nnCGCb} we give the Clebsch-Gordan coefficients for the tensor products (\ref{D6nntensorequ}), where we suppose that the 3- and 6-dimensional representations do not break up into lower dimensional representations. (These cases will be treated separately.)
\medskip
\\
Escobar et al. write $\textbf{\underline{6}}_{\widetilde{(l,-l')}}$ for all representations equivalent to $\textbf{\underline{6}}_{(l,-l')}$. We will write the representation one gets by explicitly reducing the given tensor product using the given Clebsch-Gordan coefficients. From explicit reduction one finds that $\widetilde{(l,-l')}$ takes the value $(-l,l-l')$ when one uses the Clebsch-Gordan coefficients given in tables \ref{Delta6nnCGCa} and \ref{Delta6nnCGCb}. Therefore we can express equations (\ref{D6nntensorequ}) in the following form:
\begin{displaymath}
	\textbf{\underline{3}}_{i(l)}\otimes \textbf{\underline{3}}_{j(l')}=\textbf{\underline{3}}_{\kappa(i,j)(l+l')}\oplus\textbf{\underline{6}}_{(-l,l-l')},
	\end{displaymath}
where $\kappa(i,j)=2-\delta_{ij}$, $\enspace i,j\in \{1,2\}$.

\begin{table}
\begin{center}
\renewcommand{\arraystretch}{1.4}
\begin{tabular}{|l|l|l|}
\firsthline
	$\Delta(6n^{2})$ & $\textbf{\underline{3}}_{1(l)}\otimes\textbf{\underline{3}}_{1(l')}$ / $\textbf{\underline{3}}_{2(l)}\otimes\textbf{\underline{3}}_{2(l')}$ & CGC\\
\hline
	$\textbf{\underline{3}}_{1(l+l')}$
	&	
	$u^{\textbf{\underline{3}}_{1(l)}\otimes\textbf{\underline{3}}_{1(l')} / \textbf{\underline{3}}_{2(l)}\otimes\textbf{\underline{3}}_{2(l')}}_{\textbf{\underline{3}}_{1(l+l')}}(1)=e_{11}$  & \textbf{IIId}\\

	&	
	$u^{\textbf{\underline{3}}_{1(l)}\otimes\textbf{\underline{3}}_{1(l')} / \textbf{\underline{3}}_{2(l)}\otimes\textbf{\underline{3}}_{2(l')}}_{\textbf{\underline{3}}_{1(l+l')}}(2)=e_{22}$ & \\

	&	
	$u^{\textbf{\underline{3}}_{1(l)}\otimes\textbf{\underline{3}}_{1(l')} / \textbf{\underline{3}}_{2(l)}\otimes\textbf{\underline{3}}_{2(l')}}_{\textbf{\underline{3}}_{1(l+l')}}(3)=e_{33}$ & \\

	$\textbf{\underline{6}}_{(-l,l-l')}$
	&	
	$u^{\textbf{\underline{3}}_{1(l)}\otimes\textbf{\underline{3}}_{1(l')} / \textbf{\underline{3}}_{2(l)}\otimes\textbf{\underline{3}}_{2(l')}}_{\textbf{\underline{6}}_{(-l,l-l')}}(1)=e_{12}$  & \textbf{VIb}\\

	&	
	$u^{\textbf{\underline{3}}_{1(l)}\otimes\textbf{\underline{3}}_{1(l')} / \textbf{\underline{3}}_{2(l)}\otimes\textbf{\underline{3}}_{2(l')}}_{\textbf{\underline{6}}_{(-l,l-l')}}(2)=e_{23}$ & \\

	&	
	$u^{\textbf{\underline{3}}_{1(l)}\otimes\textbf{\underline{3}}_{1(l')} / \textbf{\underline{3}}_{2(l)}\otimes\textbf{\underline{3}}_{2(l')}}_{\textbf{\underline{6}}_{(-l,l-l')}}(3)=e_{31}$ & \\

	&	
	$u^{\textbf{\underline{3}}_{1(l)}\otimes\textbf{\underline{3}}_{1(l')} / \textbf{\underline{3}}_{2(l)}\otimes\textbf{\underline{3}}_{2(l')}}_{\textbf{\underline{6}}_{(-l,l-l')}}(4)=e_{32}$ & \\

	&	
	$u^{\textbf{\underline{3}}_{1(l)}\otimes\textbf{\underline{3}}_{1(l')} / \textbf{\underline{3}}_{2(l)}\otimes\textbf{\underline{3}}_{2(l')}}_{\textbf{\underline{6}}_{(-l,l-l')}}(5)=e_{21}$ & \\

	&	
	$u^{\textbf{\underline{3}}_{1(l)}\otimes\textbf{\underline{3}}_{1(l')} / \textbf{\underline{3}}_{2(l)}\otimes\textbf{\underline{3}}_{2(l')}}_{\textbf{\underline{6}}_{(-l,l-l')}}(6)=e_{13}$ & \\
\lasthline
\end{tabular}
\caption[Clebsch-Gordan coefficients for $\textbf{\underline{3}}_{1(l)}\otimes\textbf{\underline{3}}_{1(l')}$ and $\textbf{\underline{3}}_{2(l)}\otimes\textbf{\underline{3}}_{2(l')}$ of $\Delta(6n^{2})$.]{Clebsch-Gordan coefficients for $\textbf{\underline{3}}_{1(l)}\otimes\textbf{\underline{3}}_{1(l')}$ and $\textbf{\underline{3}}_{2(l)}\otimes\textbf{\underline{3}}_{2(l')}$ of $\Delta(6n^{2})$ as given in \cite{escobar}.}
\label{Delta6nnCGCa}
\end{center}
\end{table}

\begin{table}
\begin{center}
\renewcommand{\arraystretch}{1.4}
\begin{tabular}{|l|l|l|}
\firsthline
	$\Delta(6n^{2})$ & $\textbf{\underline{3}}_{1(l)}\otimes\textbf{\underline{3}}_{2(l')}=\textbf{\underline{3}}_{2(l+l')} \oplus \textbf{\underline{6}}_{(-l,l-l')}$ & CGC\\
\hline
	$\textbf{\underline{3}}_{2(l+l')}$
	&	
	$u^{\textbf{\underline{3}}_{1(l)}\otimes\textbf{\underline{3}}_{2(l')}}_{\textbf{\underline{3}}_{2(l+l')}}(1)=e_{11}$  & \textbf{IIId}\\

	&	
	$u^{\textbf{\underline{3}}_{1(l)}\otimes\textbf{\underline{3}}_{2(l')}}_{\textbf{\underline{3}}_{2(l+l')}}(2)=e_{22}$ & \\

	&	
	$u^{\textbf{\underline{3}}_{1(l)}\otimes\textbf{\underline{3}}_{2(l')}}_{\textbf{\underline{3}}_{2(l+l')}}(3)=e_{33}$ & \\

	$\textbf{\underline{6}}_{(-l,l-l')}$
	&	
	$u^{\textbf{\underline{3}}_{1(l)}\otimes\textbf{\underline{3}}_{2(l')}}_{\textbf{\underline{6}}_{(-l,l-l')}}(1)=e_{12}$  & \textbf{VIb}\\

	&	
	$u^{\textbf{\underline{3}}_{1(l)}\otimes\textbf{\underline{3}}_{2(l')}}_{\textbf{\underline{6}}_{(-l,l-l')}}(2)=e_{23}$ & \\

	&	
	$u^{\textbf{\underline{3}}_{1(l)}\otimes\textbf{\underline{3}}_{2(l')}}_{\textbf{\underline{6}}_{(-l,l-l')}}(3)=e_{31}$ & \\

	&	
	$u^{\textbf{\underline{3}}_{1(l)}\otimes\textbf{\underline{3}}_{2(l')}}_{\textbf{\underline{6}}_{(-l,l-l')}}(4)=-e_{32}$ & \\

	&	
	$u^{\textbf{\underline{3}}_{1(l)}\otimes\textbf{\underline{3}}_{2(l')}}_{\textbf{\underline{6}}_{(-l,l-l')}}(5)=-e_{21}$ & \\

	&	
	$u^{\textbf{\underline{3}}_{1(l)}\otimes\textbf{\underline{3}}_{2(l')}}_{\textbf{\underline{6}}_{(-l,l-l')}}(6)=-e_{13}$ & \\
\lasthline
\end{tabular}
\caption[Clebsch-Gordan coefficients for $\textbf{\underline{3}}_{1(l)}\otimes\textbf{\underline{3}}_{2(l')}$ of $\Delta(6n^{2})$.]{Clebsch-Gordan coefficients for $\textbf{\underline{3}}_{1(l)}\otimes\textbf{\underline{3}}_{2(l')}$ of $\Delta(6n^{2})$ as given in \cite{escobar}.}
\label{Delta6nnCGCb}
\end{center}
\end{table}
\hspace{0mm}\\
We have already investigated how the Clebsch-Gordan coefficients behave under general basis transformations ($\rightarrow$ proposition \ref{Pclebsch1}). We will now concentrate on a useful special case.

\begin{prop}\label{PSU315}
Let $T$ be a matrix representation of a tensor product, let $M$ be the matrix of Clebsch-Gordan coefficients which reduces $T$, and let $U$ be an invertible block-diagonal matrix, where the dimensions of the blocks equal those of the blocks of the reduced tensor product $R=M^{-1}TM$. Then $MU$ reduces $T$ too, and if the blocks of $R$ commute with the blocks of $U$, then $(MU)^{-1}TMU=R$.
\end{prop}

\begin{proof}
	\begin{displaymath}
	(MU)^{-1}T(MU)=U^{-1}\underbrace{M^{-1}TM}_{R}U=U^{-1}RU.
	\end{displaymath}
If $U$ is block-diagonal, $U^{-1}RU$ is block-diagonal too, and if the blocks of $R$ commute with the corresponding blocks of $U$ we have $U^{-1}RU=R$.
\end{proof}

\begin{cor}\label{CSU316}
Let $M$ be a matrix of Clebsch-Gordan coefficients in a given basis, and let $U$ be an invertible block-diagonal matrix, then there exists a basis in which the Clebsch-Gordan coefficients are given by $MU$. Especially each column of $M$ can be multiplied with a phase factor.
\end{cor}

\begin{proof}
	This follows directly from proposition \ref{PSU315}. The action of $U$ on the Clebsch-Gordan coefficients can be interpreted as a basis transformation $S_{\lambda}$ in the sense of proposition \ref{Pclebsch1}. If $U$ is a diagonal phase matrix, $M\mapsto MU$ means multiplication of each column of $M$ with a phase factor.
\end{proof}

\begin{define}\label{DSU317}
Let $C^{\lambda}_{ijk}$ and $D^{\lambda}_{ijk}$ be Clebsch-Gordan coefficients for two tensor products. We call the Clebsch-Gordan coefficients \textit{equivalent up to basis transformations} or simply \textit{equivalent}, if there exists a basis transformation in the sense of proposition \ref{Pclebsch1} such that $C^{\lambda}_{ijk}\hspace{0mm}'=D^{\lambda}_{ijk}$.
\end{define}
\hspace{0mm}\\
From corollary \ref{CSU316} follows that, if we are only interested in Clebsch-Gordan coefficients up to basis transformations, we can multiply each basis vector of the invariant subspaces by arbitrary phase factors. Therefore the Clebsch-Gordan coefficients given in table \ref{Delta6nnCGCb} are equivalent to those shown in table \ref{Delta6nnCGCa}.
\medskip
\\
Now to the cases where the 3- or 6-dimensional representations contained in the tensor product are reducible.

\paragraph{\underline{(i) $\textbf{\underline{3}}_{r(0)}=\textbf{\underline{1}}_{r}\oplus \textbf{\underline{2}}_{1}$}}
This case occurs in $\textbf{\underline{3}}_{i(l)}\otimes \textbf{\underline{3}}_{j(l')}$ if $(l+l')$\hspace{1mm}mod\hspace{1mm}$n=0$. Since $l=l'=0$ is not allowed (then the factors of the product would be reducible), we have to consider $l+l'=n$ $\Rightarrow l'=n-l$.
	\begin{displaymath}
	\begin{split}
	& \textbf{\underline{3}}_{1(l)}\otimes \textbf{\underline{3}}_{1(n-l)}=\textbf{\underline{1}}_{1}\oplus \textbf{\underline{2}}_{1}\oplus \textbf{\underline{6}}_{(\widetilde{l,l-n})},\\
	& \textbf{\underline{3}}_{1(l)}\otimes \textbf{\underline{3}}_{2(n-l)}=\textbf{\underline{1}}_{2}\oplus \textbf{\underline{2}}_{1}\oplus \textbf{\underline{6}}_{(\widetilde{l,l-n})},\\
	& \textbf{\underline{3}}_{2(l)}\otimes \textbf{\underline{3}}_{2(n-l)}=\textbf{\underline{1}}_{1}\oplus \textbf{\underline{2}}_{1}\oplus \textbf{\underline{6}}_{(\widetilde{l,l-n})}.\\
	\end{split}
	\end{displaymath}
There is no $l\in \{1,...,n-1\}$ s.t. $\textbf{\underline{6}}_{(\widetilde{l,l-n})}$ is reducible, thus we only have to analyse the split-up $\textbf{\underline{3}}_{r(0)}=\textbf{\underline{1}}_{r}\oplus \textbf{\underline{2}}_{1}$.
\medskip
\\
Reduction of $\textbf{\underline{3}}_{1(0)}$:
	\begin{displaymath}
	\begin{split}
	& A_{3_{1(0)}}=\left(\begin{matrix}
		0  & 1 & 0 \\
		0  & 0 & 1 \\
		1  & 0 & 0
	          \end{matrix}\right),\quad
	B_{3_{1(0)}}=\left(\begin{matrix}
		0  & 0 & 1 \\
		0  & 1 & 0 \\
		1  & 0 & 0
	          \end{matrix}\right),\\
	& C_{3_{1(0)}}=\left(\begin{matrix}
		1  & 0 & 0 \\
		0  & 1 & 0 \\
		0  & 0 & 1
	          \end{matrix}\right),\quad
	D_{3_{1(0)}}=\left(\begin{matrix}
		1 & 0 & 0 \\
		0  & 1 & 0 \\
		0  & 0 & 1
	          \end{matrix}\right),
	\end{split}
	\end{displaymath}
For the reduction of $\textbf{\underline{3}}_{1(0)}$ we will use the following useful lemma:
\begin{lemma}\label{LSU39}
Let $D$ contain a one-dimensional irreducible representation $\textbf{\underline{1}}$, then all $D(a),\enspace a\in G$ have at least one common eigenvector $v$, and $D(a)v=(\textbf{\underline{1}}(a))v$.
\end{lemma}

\begin{proof}
By definition there exists a basis such that $[D]$ has the following block form:
	\begin{displaymath}
	\left(
	\begin{matrix}
	 [\textbf{\underline{1}}] & \textbf{0} \\
	 \textbf{0} & X
	\end{matrix}
	\right),
	\end{displaymath}
where \textbf{0} are appropriate null vectors. In this basis the vector $\left(
	\begin{matrix}
	 1 & 0 & ... & 0 
	\end{matrix}
	\right)^{T}$ is an eigenvector to all $[D(a)]$, and the eigenvalue is $\textbf{\underline{1}}(a)$.
\end{proof}
\hspace{0mm}\\
Using this lemma we will do the following:
	\begin{enumerate}
	\item We search for the common eigenvectors of $\textbf{\underline{3}}_{1(0)}(f), \enspace{f=A,B,C,D}$ to the eigenvalue $\textbf{\underline{1}}_{1}(f)=1$.
	\item We construct an orthonormal basis of $\mathbb{C}^{3}$ having the common eigenvector as first vector. A matrix $S$ having these basisvectors as columns will reduce the representation $\textbf{\underline{3}}_{1(0)}$ via
	\begin{displaymath}
	[\textbf{\underline{3}}_{1(0)}]\mapsto S^{-1}[\textbf{\underline{3}}_{1(0)}]S=[\textbf{\underline{1}}_{1}]\oplus[\textbf{\underline{2}}_{1}].
	\end{displaymath}
	\end{enumerate}
A common eigenvector of $A_{3_{1(0)}}, B_{3_{1(0)}}, C_{3_{1(0)}}$ and $D_{3_{1(0)}}$ to the eigenvalue $1$ is
	\begin{displaymath}
	v=\frac{1}{\sqrt{3}}\left(\begin{matrix}
	                        1\\1\\1
	                        \end{matrix}
				\right).
	\end{displaymath}
We extend $v$ to an orthonormal basis of $\mathbb{C}^{3}$ using the other eigenvectors of $A_{3_{1(0)}}$ ($\Rightarrow$ $S$ is unitary).
	\begin{displaymath}
	S=\frac{1}{\sqrt{3}}\left(\begin{matrix}
			 1 & 1 & 1 \\
			 1 & \omega & \omega^{2} \\
			 1 & \omega^{2} & \omega
	     		    \end{matrix}\right),
	\end{displaymath}
	\begin{displaymath}
	S^{-1}A_{3_{1(0)}}S=\left(\begin{matrix}
			 1 & 0 & 0 \\
			 0 & \omega & 0 \\
			 0 & 0 & \omega^{2}
	     		    \end{matrix}\right),\quad
	S^{-1}B_{3_{1(0)}}S=\left(\begin{matrix}
			 1 & 0 & 0 \\
			 0 & 0 & \omega \\
			 0 & \omega^{2} & 0
	     		    \end{matrix}\right).
	\end{displaymath}
Therefore we have found
	\begin{displaymath}
	\textbf{\underline{2}}_{1}: A\mapsto \left(\begin{matrix}
	                                      \omega & 0\\ 0 &\omega^{2}      
	                                      \end{matrix}\right),\enspace
				B\mapsto \left(\begin{matrix}
	                                      0 & \omega\\ \omega^{2} & 0      
	                                      \end{matrix}\right),\enspace
				C,D\mapsto \left(\begin{matrix}
	                                      1 & 0\\ 0 &1      
	                                      \end{matrix}\right).
	\end{displaymath}
Since only the sign of the generator $B_{3_{2}}$ is different to the representation $\textbf{\underline{3}}_{1}$, $S$ reduces $\textbf{\underline{3}}_{2}$ too, and one gets
	 \begin{displaymath}
	\textbf{\underline{2}}_{2}: A\mapsto \left(\begin{matrix}
	                                      \omega & 0\\ 0 &\omega^{2}      
	                                      \end{matrix}\right),\enspace
				B\mapsto \left(\begin{matrix}
	                                      0 & -\omega\\ -\omega^{2} & 0      
	                                      \end{matrix}\right),\enspace
				C,D\mapsto \left(\begin{matrix}
	                                      1 & 0\\ 0 &1      
	                                      \end{matrix}\right).
	\end{displaymath}
The basis vectors of the invariant subspaces $V_{\textbf{\underline{1}}_{1}}$ and $V_{\textbf{\underline{2}}_{1}}$ are given by
	\begin{displaymath}
	\begin{split}
	& u_{\textbf{\underline{1}}_{r}}=S_{j1}u_{\textbf{\underline{3}}_{r}}(j),\\
	&
	u_{\textbf{\underline{2}}_{1}}(1)=S_{j2}u_{\textbf{\underline{3}}_{r}}(j),\\
	&
	u_{\textbf{\underline{2}}_{1}}(2)=S_{j3}u_{\textbf{\underline{3}}_{r}}(j).
	\end{split}
	\end{displaymath}
Thus we have found the new Clebsch-Gordan coefficients, which we list in table \ref{Delta6nnCGCc}.

\begin{table}
\begin{center}
\renewcommand{\arraystretch}{1.4}
\begin{tabular}{|l|l|l|}
\firsthline
	$\Delta(6n^{2})$ & $\textbf{\underline{3}}_{1(l)}\otimes\textbf{\underline{3}}_{1(n-l)}$ /$\textbf{\underline{3}}_{1(l)}\otimes\textbf{\underline{3}}_{2(n-l)}$/ $\textbf{\underline{3}}_{2(l)}\otimes\textbf{\underline{3}}_{2(n-l)}$ & CGC\\
\hline
	$\textbf{\underline{1}}_{1}$/$\textbf{\underline{1}}_{2}$/$\textbf{\underline{1}}_{1}$
	&	
	\hspace{6.7mm}$u_{\textbf{\underline{1}}}=\frac{1}{\sqrt{3}}e_{11}+\frac{1}{\sqrt{3}}e_{22}+\frac{1}{\sqrt{3}}e_{33}$ & \textbf{Ia}\\
	
	$\textbf{\underline{2}}_{1}$
	&	
	$u_{\textbf{\underline{2}}_{1}}(1)=\frac{1}{\sqrt{3}}e_{11}+\frac{\omega^{2}}{\sqrt{3}}e_{22}+\frac{\omega}{\sqrt{3}}e_{33}$ & \textbf{IIb}\\

	&	
	$u_{\textbf{\underline{2}}_{1}}(2)=\frac{1}{\sqrt{3}}e_{11}+\frac{\omega}{\sqrt{3}}e_{22}+\frac{\omega^{2}}{\sqrt{3}}e_{33}$ & \\
\lasthline
\end{tabular}
\caption[Clebsch-Gordan coefficients for $\textbf{\underline{3}}_{1(l)}\otimes\textbf{\underline{3}}_{1(n-l)}$, $\textbf{\underline{3}}_{1(l)}\otimes\textbf{\underline{3}}_{2(n-l)}$ and $\textbf{\underline{3}}_{2(l)}\otimes\textbf{\underline{3}}_{2(n-l)}$ of $\Delta(6n^{2})$.]{Clebsch-Gordan coefficients for $\textbf{\underline{3}}_{1(l)}\otimes\textbf{\underline{3}}_{1(n-l)}$, $\textbf{\underline{3}}_{1(l)}\otimes\textbf{\underline{3}}_{2(n-l)}$ and $\textbf{\underline{3}}_{2(l)}\otimes\textbf{\underline{3}}_{2(n-l)}$ of $\Delta(6n^{2})$. The basis vectors of the 6-dimensional invariant subspace are not shown. They are the same as in tables \ref{Delta6nnCGCa} and \ref{Delta6nnCGCb}.}
\label{Delta6nnCGCc}
\end{center}
\end{table}

\paragraph{\underline{(ii) $\textbf{\underline{6}}_{(l,0)}=\textbf{\underline{3}}_{1(l)}\oplus \textbf{\underline{3}}_{2(l)}$}}
Here the reduction is not as easy as in case (i). To solve the problem, we will use the (slightly adapted) algorithm explained in chapter \ref{chapterclebsch} on the Clebsch-Gordan decomposition
	\begin{displaymath}
	\textbf{\underline{6}}_{(l,0)}\otimes \textbf{\underline{1}}_{1}=\textbf{\underline{3}}_{1(l)}\oplus\textbf{\underline{3}}_{2(l)}.
	\end{displaymath}
The Clebsch-Gordan coefficients for
	\begin{displaymath}
	D_{a}\otimes D_{b}=\bigoplus_{\lambda}D^{\lambda}
	\end{displaymath}
fulfil
	\begin{displaymath}
	[D^{\lambda}]_{ir}[(D_{a}^{-1})^{T}]_{js}[(D_{b}^{-1})^{T}]_{kt}C^{\lambda}_{ijk}=C^{\lambda}_{rst}.
	\end{displaymath}
Here we will deal with the special case $[D_{b}]=1$. Setting $C^{\lambda}_{ij}:=C^{\lambda}_{ij1}$ we get
	\begin{displaymath}
	[D^{\lambda}]_{ir}[(D_{a}^{-1})^{T}]_{js}C^{\lambda}_{ij}=C^{\lambda}_{rs}.
	\end{displaymath}
We can interpret this equation as an eigenvalue-equation
	\begin{displaymath}
		NC=C
	\end{displaymath}
with
	\begin{displaymath}
	C=\left(\begin{matrix}
		 C_{11}\\
		 \vdots\\
		 C_{1n_{a}}\\
		 C_{21}\\
		 \vdots\\
		 C_{n_{\lambda}n_{a}} 
	  \end{matrix}\right)
	\quad
	\mbox{and}
	\quad
	N=\left(
	\begin{matrix}
	 [D^{\lambda}]_{11}[(D_{a}^{-1})^{T}]_{11} & ... & [D^{\lambda}]_{n_{\lambda}1}[(D_{a}^{-1})^{T}]_{n_{a}1} \\
	 \vdots &  & \vdots \\
	 [D^{\lambda}]_{1n_{\lambda}}[(D_{a}^{-1})^{T}]_{1n_{a}} & ... & [D^{\lambda}]_{n_{\lambda}n_{\lambda}}[(D_{a}^{-1})^{T}]_{n_{a}n_{a}} 
	\end{matrix}\right).
	\end{displaymath}
Using the known algorithm with the new matrix $N$ and the generators given in \cite{escobar} we can calculate the matrix of \textquotedblleft reduction coefficients\textquotedblright\hspace{1mm} $S$ for $\textbf{\underline{6}}_{(l,0)}$.
	\begin{displaymath}
	S=\frac{1}{\sqrt{2}}\left(\begin{matrix}
 0 & 0 & 1 & 0 & 0 & 1 \\
 1 & 0 & 0 & 1 & 0 & 0 \\
 0 & 1 & 0 & 0 & 1 & 0 \\
 1 & 0 & 0 & -1 & 0 & 0  \\
 0 & 0 & 1 & 0 & 0 & -1 \\
 0 & 1 & 0 & 0 & -1 & 0
\end{matrix}\right)
	\end{displaymath}
As in case (i) we find the basis vectors of the new invariant subspaces by
	\begin{displaymath}
	\begin{split}
	& u_{\textbf{\underline{3}}_{1(l)}}(j)=S_{kj}u_{\textbf{\underline{6}}_{(l,0)}}(k)\quad j=1,2,3,\\
	& u_{\textbf{\underline{3}}_{2(l)}}(j)=S_{k\hspace{0.4mm} j+3}u_{\textbf{\underline{6}}_{(l,0)}}(k)\quad j=1,2,3.
	\end{split}
	\end{displaymath}
In which tensor products can $\textbf{\underline{6}}_{(l,0)}$ occur? We already know that
	\begin{displaymath}
	\textbf{\underline{3}}_{i(l)}\otimes \textbf{\underline{3}}_{j(l')}=\textbf{\underline{3}}_{\kappa(i,j)(l+l')}\oplus\textbf{\underline{6}}_{(-l,l-l')}.
	\end{displaymath}
Therefore the only possibility to get $\textbf{\underline{6}}_{(l,0)}$ is
	\begin{equation}\label{6l0reduction}
	\textbf{\underline{3}}_{i(-l)}\otimes \textbf{\underline{3}}_{j(-l)}=\textbf{\underline{3}}_{\kappa(i,j)(-2l)}\oplus\textbf{\underline{6}}_{(l,0)}=\textbf{\underline{3}}_{\kappa(i,j)(-2l)}\oplus\textbf{\underline{3}}_{1(l)}\oplus\textbf{\underline{3}}_{2(l)}.
	\end{equation}
Since $l=0$ is not allowed (then $\textbf{\underline{3}}_{i(l)}$ would be reducible), all representations on the right hand side of equation (\ref{6l0reduction}) are irreducible, except $l=\pm \frac{n}{2}$, which is of course only possible, if $n\in 2\mathbb{N}\backslash\{0\}$. If $l=\pm \frac{n}{2}$ we have
	\begin{equation}\label{6l0reductionb}
	\textbf{\underline{3}}_{i(\pm \frac{n}{2})}\otimes \textbf{\underline{3}}_{j(\pm \frac{n}{2})}=\textbf{\underline{1}}_{\kappa(i,j)}\oplus \textbf{\underline{2}}_{1}\oplus\textbf{\underline{3}}_{1(\mp \frac{n}{2})}\oplus\textbf{\underline{3}}_{2(\mp \frac{n}{2})}.
	\end{equation}
The coefficients of the Clebsch-Gordan decomposition (\ref{6l0reduction}) are listed in tables \ref{Delta6nnCGCd} and \ref{Delta6nnCGCe}. The Clebsch-Gordan coefficients for
	\begin{displaymath}
	\textbf{\underline{3}}_{1(-l)}\otimes \textbf{\underline{3}}_{2(-l)}=\textbf{\underline{3}}_{2(-2l)}\oplus \textbf{\underline{3}}_{1(l)}\oplus\textbf{\underline{3}}_{2(l)}
	\end{displaymath}
are shown in table \ref{Delta6nnCGCe}. These coefficients are of the same type as those listed in table \ref{Delta6nnCGCd}, only the roles of $\textbf{\underline{3}}_{1(l)}$ and $\textbf{\underline{3}}_{2(l)}$ are interchanged.
\medskip
\\
The Clebsch-Gordan coefficients for the decomposition (\ref{6l0reductionb}) can be obtained from tables \ref{Delta6nnCGCd} and \ref{Delta6nnCGCe} by replacing the basis vectors of the invariant subspace $V_{\textbf{\underline{3}}_{\kappa(i,j)(-2l)}}$ with the basis vectors given in table \ref{Delta6nnCGCc} (decomposition of $\textbf{\underline{3}}_{\kappa(i,j)(0)}$) .

\begin{table}
\begin{center}
\renewcommand{\arraystretch}{1.4}
\begin{tabular}{|l|l|l|}
\firsthline
	$\Delta(6n^{2})$ & $\textbf{\underline{3}}_{1(-l)}\otimes\textbf{\underline{3}}_{1(-l)}/\textbf{\underline{3}}_{2(-l)}\otimes\textbf{\underline{3}}_{2(-l)}$ & CGC\\
\hline
	$\textbf{\underline{3}}_{1(-2l)}$
	&	
	$u^{\textbf{\underline{3}}_{1(-l)}\otimes\textbf{\underline{3}}_{1(-l)}/\textbf{\underline{3}}_{2(-l)}\otimes\textbf{\underline{3}}_{2(-l)}}_{\textbf{\underline{3}}_{1(-2l)}}(1)=e_{11}$ & \textbf{IIId}\\

	&	
	$u^{\textbf{\underline{3}}_{1(-l)}\otimes\textbf{\underline{3}}_{1(-l)}/\textbf{\underline{3}}_{2(-l)}\otimes\textbf{\underline{3}}_{2(-l)}}_{\textbf{\underline{3}}_{1(-2l)}}(2)=e_{22}$ & \\

	&	
	$u^{\textbf{\underline{3}}_{1(-l)}\otimes\textbf{\underline{3}}_{1(-l)}/\textbf{\underline{3}}_{2(-l)}\otimes\textbf{\underline{3}}_{2(-l)}}_{\textbf{\underline{3}}_{1(-2l)}}(3)=e_{33}$ & \\

	$\textbf{\underline{3}}_{1(l)}$
	&	
	$u^{\textbf{\underline{3}}_{1(-l)}\otimes\textbf{\underline{3}}_{1(-l)}/\textbf{\underline{3}}_{2(-l)}\otimes\textbf{\underline{3}}_{2(-l)}}_{\textbf{\underline{3}}_{1(l)}}(1)=\frac{1}{\sqrt{2}}e_{23}+\frac{1}{\sqrt{2}}e_{32}$  & \textbf{IIIa$_{(1,1)}$}\\

	&	
	$u^{\textbf{\underline{3}}_{1(-l)}\otimes\textbf{\underline{3}}_{1(-l)}/\textbf{\underline{3}}_{2(-l)}\otimes\textbf{\underline{3}}_{2(-l)}}_{\textbf{\underline{3}}_{1(l)}}(2)=\frac{1}{\sqrt{2}}e_{13}+\frac{1}{\sqrt{2}}e_{31}$ & \\

	&	
	$u^{\textbf{\underline{3}}_{1(-l)}\otimes\textbf{\underline{3}}_{1(-l)}/\textbf{\underline{3}}_{2(-l)}\otimes\textbf{\underline{3}}_{2(-l)}}_{\textbf{\underline{3}}_{1(l)}}(3)=\frac{1}{\sqrt{2}}e_{12}+\frac{1}{\sqrt{2}}e_{21}$ & \\

	$\textbf{\underline{3}}_{2(l)}$
	&	
	$u^{\textbf{\underline{3}}_{1(-l)}\otimes\textbf{\underline{3}}_{1(-l)}/\textbf{\underline{3}}_{2(-l)}\otimes\textbf{\underline{3}}_{2(-l)}}_{\textbf{\underline{3}}_{2(l)}}(1)=\frac{1}{\sqrt{2}}e_{23}-\frac{1}{\sqrt{2}}e_{32}$  & \textbf{IIIa$_{(1,-1)}$}\\

	&	
	$u^{\textbf{\underline{3}}_{1(-l)}\otimes\textbf{\underline{3}}_{1(-l)}/\textbf{\underline{3}}_{2(-l)}\otimes\textbf{\underline{3}}_{2(-l)}}_{\textbf{\underline{3}}_{2(l)}}(2)=-\frac{1}{\sqrt{2}}e_{13}+\frac{1}{\sqrt{2}}e_{31}$ & \\

	&	
	$u^{\textbf{\underline{3}}_{1(-l)}\otimes\textbf{\underline{3}}_{1(-l)}/\textbf{\underline{3}}_{2(-l)}\otimes\textbf{\underline{3}}_{2(-l)}}_{\textbf{\underline{3}}_{2(l)}}(3)=\frac{1}{\sqrt{2}}e_{12}-\frac{1}{\sqrt{2}}e_{21}$ & \\
\lasthline
\end{tabular}
\caption{Clebsch-Gordan coefficients for $\textbf{\underline{3}}_{1(-l)}\otimes\textbf{\underline{3}}_{1(-l)}$ and $\textbf{\underline{3}}_{2(-l)}\otimes\textbf{\underline{3}}_{2(-l)}$ of $\Delta(6n^{2})$.}
\label{Delta6nnCGCd}
\end{center}
\end{table}

\begin{table}
\begin{center}
\renewcommand{\arraystretch}{1.4}
\begin{tabular}{|l|l|l|}
\firsthline
	$\Delta(6n^{2})$ & $\textbf{\underline{3}}_{1(-l)}\otimes\textbf{\underline{3}}_{2(-l)}$ & CGC\\
\hline
	$\textbf{\underline{3}}_{2(-2l)}$
	&	
	$u^{\textbf{\underline{3}}_{1(-l)}\otimes\textbf{\underline{3}}_{2(-l)}}_{\textbf{\underline{3}}_{2(-2l)}}(1)=e_{11}$  & \textbf{IIId}\\

	&	
	$u^{\textbf{\underline{3}}_{1(-l)}\otimes\textbf{\underline{3}}_{2(-l)}}_{\textbf{\underline{3}}_{2(-2l)}}(2)=e_{22}$ & \\

	&	
	$u^{\textbf{\underline{3}}_{1(-l)}\otimes\textbf{\underline{3}}_{2(-l)}}_{\textbf{\underline{3}}_{2(-2l)}}(3)=e_{33}$ & \\

	$\textbf{\underline{3}}_{1(l)}$
	&	
	$u^{\textbf{\underline{3}}_{1(-l)}\otimes\textbf{\underline{3}}_{2(-l)}}_{\textbf{\underline{3}}_{1(l)}}(1)=\frac{1}{\sqrt{2}}e_{23}-\frac{1}{\sqrt{2}}e_{32}$ & \textbf{IIIa$_{(1,-1)}$}\\

	&	
	$u^{\textbf{\underline{3}}_{1(-l)}\otimes\textbf{\underline{3}}_{2(-l)}}_{\textbf{\underline{3}}_{1(l)}}(2)=-\frac{1}{\sqrt{2}}e_{13}+\frac{1}{\sqrt{2}}e_{31}$ & \\

	&	
	$u^{\textbf{\underline{3}}_{1(-l)}\otimes\textbf{\underline{3}}_{2(-l)}}_{\textbf{\underline{3}}_{1(l)}}(3)=\frac{1}{\sqrt{2}}e_{12}-\frac{1}{\sqrt{2}}e_{21}$ & \\

	$\textbf{\underline{3}}_{2(l)}$
	&	
	$u^{\textbf{\underline{3}}_{1(-l)}\otimes\textbf{\underline{3}}_{2(-l)}}_{\textbf{\underline{3}}_{2(l)}}(1)=\frac{1}{\sqrt{2}}e_{23}+\frac{1}{\sqrt{2}}e_{32}$  & \textbf{IIIa$_{(1,1)}$}\\

	&	
	$u^{\textbf{\underline{3}}_{1(-l)}\otimes\textbf{\underline{3}}_{2(-l)}}_{\textbf{\underline{3}}_{2(l)}}(2)=\frac{1}{\sqrt{2}}e_{13}+\frac{1}{\sqrt{2}}e_{31}$ & \\

	&	
	$u^{\textbf{\underline{3}}_{1(-l)}\otimes\textbf{\underline{3}}_{2(-l)}}_{\textbf{\underline{3}}_{2(l)}}(3)=\frac{1}{\sqrt{2}}e_{12}+\frac{1}{\sqrt{2}}e_{21}$ & \\
\lasthline
\end{tabular}
\caption{Clebsch-Gordan coefficients for $\textbf{\underline{3}}_{1(-l)}\otimes\textbf{\underline{3}}_{2(-l)}$ of $\Delta(6n^{2})$.}
\label{Delta6nnCGCe}
\end{center}
\end{table}

\paragraph{\underline{(iii) $\textbf{\underline{6}}_{(\frac{n}{3},\frac{n}{3})}$, $\textbf{\underline{6}}_{(\frac{2n}{3},\frac{2n}{3})}$}}
According to \cite{escobar} $\textbf{\underline{6}}_{(\frac{n}{3},\frac{n}{3})}$ and $\textbf{\underline{6}}_{(\frac{2n}{3},\frac{2n}{3})}$ of $\Delta(6n^{2})$ ($n\in 3\mathbb{N}\backslash\{0\}$) are reducible to
	\begin{displaymath}
	\textbf{\underline{2}}_{2}\oplus \textbf{\underline{2}}_{3}\oplus \textbf{\underline{2}}_{4}.
	\end{displaymath}
This automatically means that $\textbf{\underline{6}}_{(\frac{n}{3},\frac{n}{3})}$ and $\textbf{\underline{6}}_{(\frac{2n}{3},\frac{2n}{3})}$ are equivalent (because they are direct sums of the same representations). Therefore we only need to consider one of these two representations. We choose $\textbf{\underline{6}}_{(\frac{n}{3},\frac{n}{3})}$. \cite{escobar} give the following generators for $\textbf{\underline{2}}_{2}$, $\textbf{\underline{2}}_{3}$ and $\textbf{\underline{2}}_{4}$:
	\begin{displaymath}
	\begin{split}
	& \textbf{\underline{2}}_{2}:
A\mapsto
\left(\begin{matrix}
 \omega & 0 \\
 0 & \omega^{2}
\end{matrix}\right),\enspace
B\mapsto
\left(\begin{matrix}
 0 & 1 \\
 1 & 0
\end{matrix}\right),\enspace
C,D\mapsto
\left(\begin{matrix}
 \omega^{2} & 0 \\
 0 & \omega
\end{matrix}\right),\\
& \textbf{\underline{2}}_{3}:
A\mapsto
\left(\begin{matrix}
 \omega & 0 \\
 0 & \omega^{2}
\end{matrix}\right),\enspace
B\mapsto
\left(\begin{matrix}
 0 & 1 \\
 1 & 0
\end{matrix}\right),\enspace
C,D\mapsto
\left(\begin{matrix}
 \omega & 0 \\
 0 & \omega^{2}
\end{matrix}\right),\\
& \textbf{\underline{2}}_{4}:
A\mapsto
\left(\begin{matrix}
 1 & 0 \\
 0 & 1
\end{matrix}\right),\enspace
B\mapsto
\left(\begin{matrix}
 0 & 1 \\
 1 & 0
\end{matrix}\right),\enspace
C,D\mapsto
\left(\begin{matrix}
 \omega & 0 \\
 0 & \omega^{2}
\end{matrix}\right).
	\end{split}
	\end{displaymath}
Using the adapted algorithm that we have already used to reduce $\textbf{\underline{6}}_{(l,0)}$ we find the following matrix of \textquotedblleft reduction coefficients\textquotedblright
	\begin{displaymath}
	S=\frac{1}{\sqrt{3}}
	\left(\begin{matrix}
 0 & \omega^{2} & \omega & 0 & 1 & 0 \\
 0 & \omega & \omega^{2} & 0 & 1 & 0 \\
 0 & 1 & 1 & 0 & 1 & 0 \\
 \omega^{2} & 0 & 0 & \omega & 0 & 1 \\
 \omega & 0 & 0 & \omega^{2} & 0 & 1 \\
 1 & 0 & 0 & 1 & 0 & 1
	\end{matrix}\right).
	\end{displaymath}
The general structure of tensor products of 3-dimensional irreducible representations of $\Delta(6n^{2})$ is described by
	\begin{displaymath}
	\textbf{\underline{3}}_{i(l)}\otimes \textbf{\underline{3}}_{j(l')}=\textbf{\underline{3}}_{\kappa(i,j)(l+l')}\oplus\textbf{\underline{6}}_{(-l,l-l')}.
	\end{displaymath}
$(-l,l-l')=(\frac{n}{3},\frac{n}{3})$ $\Rightarrow l=-\frac{n}{3}, l'=-\frac{2n}{3}$. Since $l$ and $l'$ are exponents of $\eta=e^{\frac{2\pi i}{n}}$, we can replace them by $l+\alpha n, l'+\beta n$, where $\alpha,\beta\in \mathbb{Z}$. Thus we can set $l=\frac{2n}{3}, l'=\frac{n}{3}$. The tensor products we are interested in are
	\begin{displaymath}
	\textbf{\underline{3}}_{i(\frac{2n}{3})}\otimes \textbf{\underline{3}}_{j(\frac{n}{3})}=\textbf{\underline{3}}_{\kappa(i,j)(0)}\oplus\textbf{\underline{6}}_{(\frac{n}{3},\frac{n}{3})}=\textbf{\underline{1}}_{\kappa(i,j)}\oplus \textbf{\underline{2}}_{1}\oplus \textbf{\underline{2}}_{2}\oplus \textbf{\underline{2}}_{3}\oplus \textbf{\underline{2}}_{4}.
	\end{displaymath}
The Clebsch-Gordan coefficients for $\textbf{\underline{1}}_{\kappa(i,j)}$ and $\textbf{\underline{2}}_{1}$ can be found in table \ref{Delta6nnCGCc}. The remaining coefficients are listed in tables \ref{Delta6nnCGCf} and \ref{Delta6nnCGCg} (they can be calculated using the matrix $S$ as in cases (i) and (ii).) Comparing tables \ref{Delta6nnCGCf} and \ref{Delta6nnCGCg} we see that the Clebsch-Gordan coefficients only differ in sign, thus they are equivalent up to basis transformations.
\bigskip
\\
As in the section on $\Delta(3n^{2})$ one can test the Clebsch-Gordan coefficients for $\Delta(6n^{2})$ by using them to reduce the tensor products. In this way one finds that they are all correct.

\begin{table}
\begin{center}
\renewcommand{\arraystretch}{1.8}
\begin{tabular}{|l|l|l|}
\firsthline
	$\Delta(6n^{2})$ & $\textbf{\underline{3}}_{1(\frac{2n}{3})}\otimes\textbf{\underline{3}}_{1(\frac{n}{3})}$ & CGC\\
\hline
	$\textbf{\underline{2}}_{2}$
	& $u^{\textbf{\underline{3}}_{1(\frac{2n}{3})}\otimes\textbf{\underline{3}}_{1(\frac{n}{3})}}_{\textbf{\underline{2}}_{2}}(1)=\frac{1}{\sqrt{3}}e_{13}+\frac{\omega}{\sqrt{3}}e_{21}+\frac{\omega^{2}}{\sqrt{3}}e_{32}$ & \textbf{IIc}\\

	& $u^{\textbf{\underline{3}}_{1(\frac{2n}{3})}\otimes\textbf{\underline{3}}_{1(\frac{n}{3})}}_{\textbf{\underline{2}}_{2}}(2)=\frac{\omega^{2}}{\sqrt{3}}e_{12}+\frac{\omega}{\sqrt{3}}e_{23}+\frac{1}{\sqrt{3}}e_{31}$ & \\

	$\textbf{\underline{2}}_{3}$
	& $u^{\textbf{\underline{3}}_{1(\frac{2n}{3})}\otimes\textbf{\underline{3}}_{1(\frac{n}{3})}}_{\textbf{\underline{2}}_{3}}(1)=\frac{\omega}{\sqrt{3}}e_{12}+\frac{\omega^{2}}{\sqrt{3}}e_{23}+\frac{1}{\sqrt{3}}e_{31}$ & \textbf{IId} \\

	& $u^{\textbf{\underline{3}}_{1(\frac{2n}{3})}\otimes\textbf{\underline{3}}_{1(\frac{n}{3})}}_{\textbf{\underline{2}}_{3}}(2)=\frac{1}{\sqrt{3}}e_{13}+\frac{\omega^{2}}{\sqrt{3}}e_{21}+\frac{\omega}{\sqrt{3}}e_{32}$ & \\

	$\textbf{\underline{2}}_{4}$
	& $u^{\textbf{\underline{3}}_{1(\frac{2n}{3})}\otimes\textbf{\underline{3}}_{1(\frac{n}{3})}}_{\textbf{\underline{2}}_{4}}(1)=\frac{1}{\sqrt{3}}e_{12}+\frac{1}{\sqrt{3}}e_{23}+\frac{1}{\sqrt{3}}e_{31}$  & \textbf{IIe}\\

	& $u^{\textbf{\underline{3}}_{1(\frac{2n}{3})}\otimes\textbf{\underline{3}}_{1(\frac{n}{3})}}_{\textbf{\underline{2}}_{4}}(2)=\frac{1}{\sqrt{3}}e_{13}+\frac{1}{\sqrt{3}}e_{21}+\frac{1}{\sqrt{3}}e_{32}$ & \\
\lasthline
\end{tabular}
\caption[Clebsch-Gordan coefficients for $\textbf{\underline{3}}_{1(\frac{2n}{3})}\otimes\textbf{\underline{3}}_{1(\frac{n}{3})}$ of $\Delta(6n^{2})$ ($n\in 3\mathbb{N}\backslash\{0\}$).]{Clebsch-Gordan coefficients for $\textbf{\underline{3}}_{1(\frac{2n}{3})}\otimes\textbf{\underline{3}}_{1(\frac{n}{3})}$ of $\Delta(6n^{2})$ ($n\in 3\mathbb{N}\backslash\{0\}$). The remaining basis vectors can be found in table \ref{Delta6nnCGCc}.}
\label{Delta6nnCGCf}
\end{center}
\end{table}

\begin{table}
\begin{center}
\renewcommand{\arraystretch}{1.8}
\begin{tabular}{|l|l|l|}
\firsthline
	$\Delta(6n^{2})$ & $\textbf{\underline{3}}_{1(\frac{2n}{3})}\otimes\textbf{\underline{3}}_{2(\frac{n}{3})}$ & CGC\\
\hline
	$\textbf{\underline{2}}_{2}$
	& $u^{\textbf{\underline{3}}_{1(\frac{2n}{3})}\otimes\textbf{\underline{3}}_{2(\frac{n}{3})}}_{\textbf{\underline{2}}_{2}}(1)=-\frac{1}{\sqrt{3}}e_{13}-\frac{\omega}{\sqrt{3}}e_{21}-\frac{\omega^{2}}{\sqrt{3}}e_{32}$  & \textbf{IIc$^{(\ast)}$}\\

	& $u^{\textbf{\underline{3}}_{1(\frac{2n}{3})}\otimes\textbf{\underline{3}}_{2(\frac{n}{3})}}_{\textbf{\underline{2}}_{2}}(2)=\frac{\omega^{2}}{\sqrt{3}}e_{12}+\frac{\omega}{\sqrt{3}}e_{23}+\frac{1}{\sqrt{3}}e_{31}$ & \\

	$\textbf{\underline{2}}_{3}$
	& $u^{\textbf{\underline{3}}_{1(\frac{2n}{3})}\otimes\textbf{\underline{3}}_{2(\frac{n}{3})}}_{\textbf{\underline{2}}_{3}}(1)=\frac{\omega}{\sqrt{3}}e_{12}+\frac{\omega^{2}}{\sqrt{3}}e_{21}+\frac{1}{\sqrt{3}}e_{31}$ & \textbf{IId$^{(\ast)}$}\\

	& $u^{\textbf{\underline{3}}_{1(\frac{2n}{3})}\otimes\textbf{\underline{3}}_{2(\frac{n}{3})}}_{\textbf{\underline{2}}_{3}}(2)=-\frac{1}{\sqrt{3}}e_{13}-\frac{\omega^{2}}{\sqrt{3}}e_{21}-\frac{\omega}{\sqrt{3}}e_{32}$ & \\

	$\textbf{\underline{2}}_{4}$
	& $u^{\textbf{\underline{3}}_{1(\frac{2n}{3})}\otimes\textbf{\underline{3}}_{2(\frac{n}{3})}}_{\textbf{\underline{2}}_{4}}(1)=\frac{1}{\sqrt{3}}e_{12}+\frac{1}{\sqrt{3}}e_{23}+\frac{1}{\sqrt{3}}e_{31}$ & \textbf{IIe$^{(\ast)}$}\\

	& $u^{\textbf{\underline{3}}_{1(\frac{2n}{3})}\otimes\textbf{\underline{3}}_{2(\frac{n}{3})}}_{\textbf{\underline{2}}_{4}}(2)=-\frac{1}{\sqrt{3}}e_{13}-\frac{1}{\sqrt{3}}e_{21}-\frac{1}{\sqrt{3}}e_{32}$ & \\
\lasthline
\end{tabular}
\caption[Clebsch-Gordan coefficients for $\textbf{\underline{3}}_{1(\frac{2n}{3})}\otimes\textbf{\underline{3}}_{2(\frac{n}{3})}$ of $\Delta(6n^{2})$ ($n\in 3\mathbb{N}\backslash\{0\}$).]{Clebsch-Gordan coefficients for $\textbf{\underline{3}}_{1(\frac{2n}{3})}\otimes\textbf{\underline{3}}_{2(\frac{n}{3})}$ of $\Delta(6n^{2})$ ($n\in 3\mathbb{N}\backslash\{0\}$). The remaining basis vectors can be found in table \ref{Delta6nnCGCc}. ($(\ast)$... up to irrelevant phase factors)}
\label{Delta6nnCGCg}
\end{center}
\end{table}

\subsection{The new $SU(3)$-subgroups found by the FFK/BLW collaboration}

The history of the new finite $SU(3)$-subgroups obtained by the FFK/BLW collaboration begins with the study of so-called $Z$-metacyclic groups investigated by Bovier et al. in \cite{BLW2}.

\begin{define}
The \textit{cyclic group} $\mathbb{Z}_{m}$, $m\in\mathbb{N}\backslash\{0\}$, is defined as
	\begin{displaymath}
	\mathbb{Z}_{m}=\{0,1,2,...,m-1\}
	\end{displaymath}
together with the composition
	\begin{displaymath}
	a\circ b:=(a+b) \hspace{0.5mm}\mathrm{mod}\hspace{0.5mm}m,
	\end{displaymath}
where $+$ denotes the addition of natural numbers. $\mathbb{Z}_{m}$ is an Abelian group isomorphic to the group $Z_{m}$.
\end{define}
\hspace{0mm}\\
In \cite{BLW2} BLW have analysed groups of the form
	\begin{displaymath}
	\mathbb{Z}_{m}\rtimes\mathbb{Z}_{n},
	\end{displaymath}
where $\rtimes$ denotes the so-called \textit{semidirect product} of groups. For our considerations it is not necessary to further explain the structure of these semidirect products. For more information on the $Z$-metacyclic groups we refer the reader to \cite{BLW2}. The endresult of the investigations of FFK/BLW \cite{fairbairn,BLW1,BLW2,fairbairn2} is that some $Z$-metacyclic groups are \textquotedblleft new\textquotedblright\enspace finite $SU(3)$-subgroups, namely
	\begin{displaymath}
	T_{m}:=\mathbb{Z}_{m}\rtimes\mathbb{Z}_{3},
	\end{displaymath}
where $m$ must contain at least one prime factor of the form $3k+1$, $k\in \mathbb{N}\backslash\{0\}$. The smallest group fulfilling this condition is the group $T_{7}$, which has already been applied to lepton physics \cite{luhn3,hagedorn2}. If $m$ contains a factor $q$ consisting of powers of primes which are neither $3$ nor of the form $3k+1$, $k\in \mathbb{N}\backslash\{0\}$, then according to \cite{fairbairn2} $\mathbb{Z}_{m}\rtimes\mathbb{Z}_{3}$ can be written as a direct product
	\begin{displaymath}
	(\mathbb{Z}_{3^{r}p}\rtimes\mathbb{Z}_{3})\times \mathbb{Z}_{q}
	\end{displaymath}
with $m=3^{r}pq$, where $p$ is a product of powers of primes of the form $3k+1$, $k\in \mathbb{N}\backslash \{0\}$. Fairbairn and Fulton \cite{fairbairn2} now concentrate on the group $\mathbb{Z}_{3^{r}p}\rtimes\mathbb{Z}_{3}$, because $\mathbb{Z}_{q}$ is a trivial finite subgroup of $SU(3)$, thus $(\mathbb{Z}_{3^{r}p}\rtimes\mathbb{Z}_{3})\times \mathbb{Z}_{q}$ is a finite subgroup of $SU(3)$ if and only if $\mathbb{Z}_{3^{r}p}\rtimes\mathbb{Z}_{3}$ is a finite subgroup of $SU(3)$. According to \cite{fairbairn2} $\mathbb{Z}_{3^{r}p}\rtimes\mathbb{Z}_{3}$ is a finite subgroup of $SU(3)$ if $r=0,1$. Fairbairn and Fulton \cite{fairbairn2} give generators for both cases:
	\begin{displaymath}
	A=\left(\begin{matrix}
		 0 & 1 & 0 \\
		 0 & 0 & 1 \\
		 1 & 0 & 0
	       \end{matrix}\right),\quad
	B=\left(\begin{matrix}
		 e^{\frac{2\pi i}{p}} & 0 & 0 \\
		 0 & e^{\frac{2\pi i a}{p}} & 0 \\
		 0 & 0 & e^{\frac{2\pi i a^2}{p}}
	       \end{matrix}\right),
	\end{displaymath}
	\begin{itemize}
	 \item $r=0:$ $1+a+a^2=0\hspace{0.5mm}\mathrm{mod}\hspace{0.5mm}p$,
	 \item $r=1:$ $1+a+a^2=0\hspace{0.5mm}\mathrm{mod}\hspace{0.5mm}3p$.
	\end{itemize}
Let $\rho=e^{\frac{2\pi i}{p}}$, then the generator $B$ can be written as
	\begin{displaymath}
	B=\left(\begin{matrix}
		 \rho & 0 & 0 \\
		 0 & \rho^a & 0 \\
		 0 & 0 & \rho^{-1-a}
	       \end{matrix}\right).
	\end{displaymath}
In this form $A$ and $B$ look similar to the generators of the representation $\textbf{\underline{3}}_{(k,l)}$ of $\Delta(3n^2)$:
	\begin{displaymath}
	\textbf{\underline{3}}_{(k,l)} : a\mapsto
			\left(\begin{matrix}
			 0 & 1 & 0 \\
			 0 & 0 & 1 \\
			 1 & 0 & 0
			\end{matrix}\right),\enspace
			c\mapsto
			\left(\begin{matrix}
			 \eta^{l} & 0 & 0 \\
			 0 & \eta^{k} & 0 \\
			 0 & 0 & \eta^{-k-l}
			\end{matrix}\right);\quad \eta=e^{\frac{2\pi i}{n}}.
	\end{displaymath}
(In subsection \ref{d3nnsubsect} we mentioned a third generator $d$, but we later found that $d$ can be expressed through $a$ and $c$.) Comparing $A,B$ with $a,c$ we find that if $\rho=\eta$
	\begin{displaymath}
	A=\textbf{\underline{3}}_{(k,l)}(a),\quad B=\textbf{\underline{3}}_{(a,1)}(c).
	\end{displaymath}
Thus
	\begin{displaymath}
	\mathbb{Z}_{3^{r}p}\rtimes\mathbb{Z}_{3}\simeq \textbf{\underline{3}}_{(a,1)}(\Delta(3p^2)).
	\end{displaymath}
We conclude that faithful 3-dimensional representations of $\mathbb{Z}_{3^{r}p}\rtimes\mathbb{Z}_{3}$ can be interpreted as (not necessarily faithful) 3-dimensional representations of $\Delta(3p^2)$. This automatically includes that all possible Clebsch-Gordan coefficients for $\textbf{\underline{3}}\otimes \textbf{\underline{3}}$-tensor products of $\mathbb{Z}_{3^{r}p}\rtimes\mathbb{Z}_{3}$ are included in the list of Clebsch-Gordan coefficients for $\Delta(3p^2)$.
\medskip
\\
Let us as an example consider the group $T_7=\mathbb{Z}_7\rtimes \mathbb{Z}_3$, which is of the form $\mathbb{Z}_{3^r p}\rtimes \mathbb{Z}_3$ with $r=0,\enspace p=7$.
\\
Since $r=0$ we have to solve the equation
	\begin{displaymath}
	1+a+a^2=0\hspace{0.5mm}\mathrm{mod}\hspace{0.5mm}7
	\end{displaymath}
to determine $a=2$. Thus $T_7$ is isomorphic to
	\begin{displaymath}
	\textbf{\underline{3}}_{(a,1)}(\Delta(3p^2))=\textbf{\underline{3}}_{(2,1)}(\Delta(3\cdot 7^2))=\textbf{\underline{3}}_{(2,1)}(\Delta(147)).
	\end{displaymath}
\hspace{0mm}
\\
Next we have to consider groups of the form $(\mathbb{Z}_{3^{r}p}\rtimes\mathbb{Z}_{3})\times \mathbb{Z}_{q}$. For this purpose let us consider a more general group $G\times A$, where $G$ is an arbitrary finite group and $A$ is a finite Abelian group.

\begin{theorem}\label{producttheorem}
Let $G$ be a finite group and $A$ a finite Abelian group. Then: $G\times A$ is isomorphic to a finite subgroup of $SU(3)$ $\Rightarrow $ $A$ is isomorphic to $Z_{3}$ or the trivial group $\{e\}$.
\end{theorem}

\begin{proof}
By definition ($\rightarrow$ definition \ref{DA27}) $G\times A$ is given by
	\begin{displaymath}
	G\times A=\{(g,a)\vert g\in G,a\in A; (g,a)\circ(g',a')=(g\circ g',a\circ a')\}.
	\end{displaymath}
$G\times A$ is isomorphic to a finite subgroup of $SU(3)$ if and only if there exists an isomorphism $\phi: G\times A\rightarrow \phi(G\times A)\subset SU(3)$. From
	\begin{displaymath}
	(g,a)\circ(e,a')=(g\circ e,a\circ a')=( e\circ g,a'\circ a)=(e,a')\circ (g,a)
	\end{displaymath}
it follows that $(e,a')$ commutes with all elements of $G\times A$ $\Rightarrow$ $\phi((e,a'))$ commutes with all elements of $\phi(G\times A)$. Since $\{\phi((e,a))\vert a\in A\}$ is a group whose elements commute with all elements of $\phi(G\times A)$ it must be a subgroup of the center of $\phi(G\times A)$. Per assumption $\phi(G\times A)$ is a finite subgroup of $SU(3)$, thus using corollary \ref{CSU310} we find that $\{\phi((e,a))\vert a\in A\}$ is either $\{\mathbbm{1}_{3}\}$ or $\{\mathbbm{1}_{3},\omega \mathbbm{1}_{3},\omega^2 \mathbbm{1}_{3}\}$, where $\omega=e^{\frac{2\pi i}{3}}$. Since $\phi$ is an isomorphism we find that the Abelian group $A$ must be isomorphic to $\{e\}$ or $Z_{3}$.
\end{proof}
\hspace{0mm}
\\
As a corollary we find that $(\mathbb{Z}_{3^{r}p}\rtimes\mathbb{Z}_{3})\times \mathbb{Z}_{q}$ can be a finite subgroup of $SU(3)$ only if $q=1$ or $q=3$, but $q=1$ corresponds to $\mathbb{Z}_{3^{r}p}\rtimes\mathbb{Z}_{3}$ and following \cite{fairbairn2} $q=3$ is not allowed, thus $(\mathbb{Z}_{3^{r}p}\rtimes\mathbb{Z}_{3})\times \mathbb{Z}_{q}$ do not form subgroups of $SU(3)$ for $q\neq 1$.
\bigskip
\\
Theorem \ref{producttheorem} gives another new insight into the problem of finite subgroups of $SU(3)$. Let $G$ be a finite subgroup of $SU(3)$ that does not already contain the center $C=\{\mathbbm{1}_{3},\omega \mathbbm{1}_{3},\omega^2 \mathbbm{1}_{3}\}$ of $SU(3)$, then we can construct $G\times C$ and it will be a new finite subgroup of $SU(3)$. This procedure may lead to new finite subgroups of $SU(3)$, but it will not lead to new Clebsch-Gordan coefficients, because of similar arguments as used in the proof of lemma \ref{LSU313} (commutativity of $\omega^n\mathbbm{1}_{3}$ with all elements of the group).
\bigskip
\\
As mentioned in section \ref{nonabeliansu3section} Miller et al. \cite{miller} list two series (C) and (D) of finite subgroups of $SU(3)$ which could contain new groups that have not been found by the FFK/BLW collaboration. In the next two subsections we will concentrate on these two series.

\subsection{The group (C)}
In their listing of finite subgroups of $SU(3)$ Miller et al. \cite{miller} list two series of non-Abelian groups, namely (C) and (D). In this subsection we will investigate the properties of (C).
\medskip
\\
Miller et al. \cite{miller} define (C) as a finite subgroup of $SU(3)$ generated by the two matrices
	\begin{displaymath}
	H=\left(
	\begin{matrix}
 \alpha & 0 & 0 \\
 0 & \beta & 0 \\
 0 & 0 & \gamma
	\end{matrix}
	\right),\quad T=\left(\begin{matrix}
			 0 & 1 & 0 \\
			 0 & 0 & 1 \\
			 1 & 0 & 0
			\end{matrix}\right).
	\end{displaymath}
In subsection \ref{d3nnsubsect} we already found that the series $\Delta(3n^2)$ is a special case of (C). Let us investigate the generators of (C) in more detail. $H$ must be an element of $SU(3)$, thus $\alpha,\beta,\gamma\in U(1)$. Since (C) is a finite group there exists an $n$ s.t. $\alpha^n=\beta^n=\gamma^n=1$. It follows
	\begin{displaymath}
	H=\left(
	\begin{matrix}
 	\eta^{a} & 0 & 0 \\
 	0 & \eta^{b} & 0 \\
 	0 & 0 & \eta^{c}
	\end{matrix}
	\right),
	\end{displaymath}
where $\eta=e^{\frac{2\pi i}{n}}$. From $\mathrm{det}H=1$ it follows that $c=-a-b$, thus
	\begin{displaymath}
	H=\left(
	\begin{matrix}
 	\eta^{a} & 0 & 0 \\
 	0 & \eta^{b} & 0 \\
 	0 & 0 & \eta^{-a-b}
	\end{matrix}
	\right).
	\end{displaymath}
In this form we can see that $H$ and $T$ generate the representation $\textbf{\underline{3}}_{(b,a)}$ of $\Delta(3n^2)$ ($\rightarrow$ subsection \ref{d3nnsubsect}).

\subsection{The group (D)}\label{Dsubsection}

(D) is a series of groups generated by $H,T$ of (C) and the generator
	\begin{displaymath}
	R=\left(
	\begin{matrix}
 	x & 0 & 0 \\
 	0 & 0 & y \\
 	0 & z & 0
	\end{matrix}
	\right).
	\end{displaymath}
From unitarity it follows that $x,y,z\in U(1)$, and $\mathrm{det}R=1$ implies $z=-\frac{1}{xy}$. Obviously $x$ is an eigenvalue of $R$, thus it must be of some finite order $m$ ($R^m=\mathbbm{1}_{3}\Rightarrow x^m=1$) $\Rightarrow x=e^{\frac{2\pi i g}{m}}$, $g\in \mathbb{Z}, m\in \mathbb{N}\backslash\{0\}$. Let us now look at all eigenvalues $\lambda_{i}$ of $R$. We find
	\begin{displaymath}
	\lambda_{1}=x, \quad \lambda_{2}=\frac{i}{\sqrt{x}}, \quad \lambda_{3}=-\frac{i}{\sqrt{x}},
	\end{displaymath}
where $\sqrt{x}$ denotes one of the two squareroots of $x$. Thus the eigenvalues of $R$ do not restrict $y$ (and $z$).
\medskip
\\
We have to analyse $y$ and $z$ further in order to ensure that (D) is a finite group. Analysing some products of the generators $H,T,R$ one finds
	\begin{displaymath}
	R^2=\left(
	\begin{matrix}
 	x^2 & 0 & 0 \\
 	0 & -\frac{1}{x} & 0 \\
 	0 & 0 & -\frac{1}{x}
	\end{matrix}
	\right),\quad (RT^2)^2=\left(
	\begin{matrix}
 	-\frac{1}{y} & 0 & 0 \\
 	0 & y^2 & 0 \\
 	0 & 0 & -\frac{1}{y}
	\end{matrix}
	\right).
	\end{displaymath}
This gives us a restriction on $y$. To ensure that $(RT^2)^2$ is of finite order $y$ must be of the form $y=e^{\frac{2\pi i g'}{m'}}$, $g'\in \mathbb{Z}, m'\in \mathbb{N}\backslash\{0\}$. Let us now compare $x$ and $y$:
	\begin{displaymath}
	x=e^{\frac{2\pi i g}{m}},\quad y=e^{\frac{2\pi i g'}{m'}},\quad g,g'\in \mathbb{Z}; m,m'\in \mathbb{N}\backslash\{0\}
	\end{displaymath}
Since $\frac{g}{m},\frac{g'}{m'}\in \mathbb{Q}$ there exists a common denominator $d\in \mathbb{N}\backslash\{0\}$ s.t.
	\begin{displaymath}
	\exists r,s\in \mathbb{Z}:\quad \frac{g}{m}=\frac{r}{d},\quad \frac{g'}{m'}=\frac{s}{d}.
	\end{displaymath}
Defining $\delta:=e^{\frac{2\pi i}{d}}$ we find $x=\delta^r$, $y=\delta^s$ and
	\begin{displaymath}
	R=\left(
	\begin{matrix}
 	\delta^r & 0 & 0 \\
 	0 & 0 & \delta^s \\
 	0 & -\delta^{-r-s} & 0
	\end{matrix}
	\right),\quad r,s\in\mathbb{Z}.
	\end{displaymath}
$\delta^d=1\Rightarrow r,s\in \{0,1,...,d-1\}$. The same argument holds for the generator $H$ ($\Rightarrow a,b\in \{0,1,...,n-1\}$). Our final result for the generators of the group (D) is
	\begin{displaymath}
	H=\left(
	\begin{matrix}
	 \eta^a & 0 & 0 \\
	 0 & \eta^b & 0 \\
	 0 & 0 & \eta^{-a-b}
	\end{matrix}
	\right),\quad T=\left(\begin{matrix}
			 0 & 1 & 0 \\
			 0 & 0 & 1 \\
			 1 & 0 & 0
			\end{matrix}\right),\quad
	R=\left(
	\begin{matrix}
 	\delta^r & 0 & 0 \\
 	0 & 0 & \delta^s \\
 	0 & -\delta^{-r-s} & 0
	\end{matrix}
	\right),
	\end{displaymath}
where $\eta=e^{\frac{2\pi i}{n}}, \delta=e^{\frac{2\pi i}{d}}$;$\enspace$ $n,d\in\mathbb{N}\backslash\{0\}$; $a,b\in \{0,1,...,n-1\}$; $r,s\in \{0,1,...,d-1\}$. Therefore (D) is a six-parametric series of finite groups
	\begin{displaymath}
	D(n,a,b;d,r,s).
	\end{displaymath}
Note that we have not analysed for which $n,a,b,d,r,s$ the group (D) is a non-Abelian finite subgroup of $SU(3)$. For this purpose we would have to check for which parameter values (D) has a faithful irreducible 3-dimensional representation.
\medskip
\\
Remark: The representation $\textbf{\underline{3}}_{2(l)}$ of $\Delta(6n^2)$ is given by the generators
	\begin{displaymath}
	\left(\begin{matrix}
			 0 & 1 & 0 \\
			 0 & 0 & 1 \\
			 1 & 0 & 0
			\end{matrix}\right), \quad
	\left(\begin{matrix}
			 0 & 0 & -1 \\
			 0 & -1 & 0 \\
			 -1 & 0 & 0
			\end{matrix}\right), \quad
	\left(
	\begin{matrix}
	 1 & 0 & 0 \\
	 0 & \eta^l & 0 \\
	 0 & 0 & \eta^{-l}
	\end{matrix}
	\right), \quad
	\left(
	\begin{matrix}
	 \eta^l & 0 & 0 \\
	 0 & \eta^{-l} & 0 \\
	 0 & 0 & 1
	\end{matrix}
	\right),
	\end{displaymath}
where $\eta=e^{\frac{2\pi i}{n}}$ (see subsection \ref{d6nnsubsection}). Comparing to the group (D) we find
	\begin{displaymath}
	D(n,0,l;2,1,1)=\textbf{\underline{3}}_{2(l)}(\Delta(6n^2)).
	\end{displaymath}
Its rich structure leaves the possibility that the series (D) contains some finite subgroups of $SU(3)$ which have not been analysed yet, but we will end our systematic analysis of the finite subgroups of $SU(3)$ here, let us just give an example for a group that has not been considered yet.
\medskip
\\
Using the computer algebra system \textit{GAP} one can easily search for new groups by \textquotedblleft trial and error\textquotedblright. For example one finds the group $D(2,1,1;3,1,0)$ of order $72$. Using \textit{GAP} one can easily check that $S_{4}$ and $Z_{3}$ are invariant subgroups of $D(2,1,1;3,1,0)$, and that
	\begin{displaymath}
	D(2,1,1;3,1,0)\simeq S_{4}\times Z_{3}, 
	\end{displaymath}
which is a finite subgroup of $SU(3)$ that we had missed up to now.

\section{Summary of the derived Clebsch-Gordan coefficients}\label{summaryCGCsection}
We have now analysed all known non-Abelian finite subgroups of $SU(3)$ that have 3-dimensional irreducible representations (except the series (D) described in subsection \ref{Dsubsection}). In contrast to the infinite number of these subgroups we found that the number of non-equivalent Clebsch-Gordan coefficients is quite small. We will now list all possible Clebsch-Gordan decompositions of tensor products of 3-dimensional irreducible representations of those finite subgroups of $SU(3)$ we have studied, the groups where these decompositions can occur, and the tables where the corresponding Clebsch-Gordan coefficients (up to basis transformations) can be found. For the sake of clarity we will only list the dimensions of the representations. Since we don't have a good method to check whether two sets of Clebsch-Gordan coefficients are equivalent up to basis transformations, the list could contain tensor products for which the Clebsch-Gordan coefficients are equivalent.

\begin{sidewaystable}
\begin{center}
\renewcommand{\arraystretch}{1.4}
\begin{tabular}{|lll|}
\firsthline
	Clebsch-Gordan decomposition of $\textbf{\underline{3}}\otimes\textbf{\underline{3}}$ & Groups & CGCs\\
\hline
	$\textbf{\underline{1}}\oplus \textbf{\underline{1}}\oplus \textbf{\underline{1}}\oplus \textbf{\underline{1}}\oplus \textbf{\underline{1}}\oplus \textbf{\underline{1}}\oplus \textbf{\underline{1}}\oplus \textbf{\underline{1}}\oplus \textbf{\underline{1}}$ & $\Delta(3n^{2})$ & \ref{Delta3nnCGC3}\\
	$\textbf{\underline{1}}\oplus \textbf{\underline{1}}\oplus \textbf{\underline{1}}\oplus \textbf{\underline{3}}\oplus \textbf{\underline{3}}$ & $A_{4},\Sigma(216\phi),\Delta(3n^{2})$ & \ref{A4CGC}\\
	$\textbf{\underline{1}}\oplus \textbf{\underline{2}}\oplus \textbf{\underline{2}}\oplus \textbf{\underline{2}}\oplus \textbf{\underline{2}}$ & $\Delta(6n^{2})$ & \ref{Delta6nnCGCc}+\ref{Delta6nnCGCf}\\
	$\textbf{\underline{1}}\oplus \textbf{\underline{2}}\oplus \textbf{\underline{3}}\oplus \textbf{\underline{3}}$ & $S_{4}$ & \ref{S4CGC}\\
	$\textbf{\underline{1}}\oplus \textbf{\underline{2}}\oplus \textbf{\underline{6}}$ & $\Delta(6n^{2})$ & \ref{Delta6nnCGCa}+\ref{Delta6nnCGCc}\\
	$\textbf{\underline{1}}\oplus \textbf{\underline{3}}\oplus \textbf{\underline{5}}$ & $A_{5}$ & \ref{A5CGCa}\\
	$\textbf{\underline{1}}\oplus \textbf{\underline{4}}\oplus \textbf{\underline{4}}$ & $\Sigma(36\phi)$ & \ref{Sigma36CGCb}\\
	$\textbf{\underline{1}}\oplus \textbf{\underline{8}}$ & $\Sigma(168),\Sigma(72\phi),\Sigma(216\phi),\Sigma(360\phi)$ &\ref{Sigma168CGCb}\\
	$\textbf{\underline{3}}\oplus \textbf{\underline{3}}\oplus \textbf{\underline{3}}$ & $\Sigma(36\phi), \Delta(3n^{2}), \Delta(6n^{2})$ & \ref{Sigma36CGCa}, \ref{Delta3nnCGC2}, \ref{Delta6nnCGCd}\\
	$\textbf{\underline{3}}\oplus \textbf{\underline{6}}$ & $\Sigma(168), \Sigma(72\phi), \Sigma(216\phi),\Sigma(360\phi), \Delta(6n^{2})$ & \ref{Sigma168CGCa}, \ref{Delta6nnCGCa}\\
	$\textbf{\underline{4}}\oplus \textbf{\underline{5}}$ & $A_{5}$ & \ref{A5CGCb}\\
	$\textbf{\underline{9}}$ & $\Sigma(216\phi), \Sigma(360\phi)$ &\\
\lasthline
\end{tabular}
\caption[Possible tensor products and Clebsch-Gordan coefficients of 3-dimensional irreducible representations of (known) finite subgroups of $SU(3)$.]{Possible tensor products and Clebsch-Gordan coefficients (CGCs) of 3-dimensional irreducible representations of (known) finite subgroups of $SU(3)$. We don't list double covers $\tilde{G}$ of finite subgroups $G$ of $SO(3)$, because $\tilde{G}$ and $G$ have the same 3-dimensional irreducible representations.}
\label{SU3tensor}
\end{center}
\end{sidewaystable}
\hspace{0mm}\\
Another very interesting question is, what types of Clebsch-Gordan coefficients are possible for $D_{a}\otimes D_{b}=D\oplus...$ when one fixes the dimension of $D$ ($\mathrm{dim}D_{a}=\mathrm{dim}D_{b}=3$). Instead of $D_{a}\otimes D_{b}=D\oplus...$ we will also write $D_{a}\otimes D_{b}\rightarrow D$. From table \ref{SU3tensor} we find that the possible dimensions of $D$ for the groups we analysed can take all values from 1 to 9, except 7. We also see that there are only few groups with representations that fulfil
	\begin{displaymath}
	\textbf{\underline{3}}\otimes \textbf{\underline{3}}\rightarrow \textbf{\underline{4}}\quad \mbox{or}\quad \textbf{\underline{3}}\otimes \textbf{\underline{3}}\rightarrow \textbf{\underline{5}}.
	\end{displaymath}
Two sets of Clebsch-Gordan coefficients $\{\Gamma_{i}\}_{i}$ and $\{\Gamma_{i}'\}_{i}$ are equivalent if and only if there exist unitary matrices $S_{A},S_{B}$ and $S_{\lambda}$ s.t.
	\begin{equation}\label{basisequ}
	\Gamma_{i}'=S_{A}^{-1}\Gamma_{m}(S_{B}^{-1})^{T}(S_{\lambda})_{mi}.
	\end{equation}
(see definition \ref{DSU317} and proposition \ref{Pclebsch1}). Therefore one can test whether $\{\Gamma_{i}\}_{i}$ and $\{\Gamma_{i}'\}_{i}$ are equivalent by solving equation (\ref{basisequ}). In contrast to the invariance equations (\ref{inveq}), which are linear in the unknown matrices $\Gamma_{i}$, equation (\ref{basisequ}) is nonlinear in the unknown matrices $S_{A}, S_{B}$ and $S_{\lambda}$. Since we don't have a systematic method to solve this nonlinear equation, we will in general not be able to test the different sets of Clebsch-Gordan coefficients on equivalence.
\medskip
\\
Remark: By now we have always talked about Clebsch-Gordan coefficients for the reduction of tensor products to a direct sum of irreducible representations.  Suppose we are interested in \textquotedblleft reduction coefficients\textquotedblright\enspace for a decomposition $D_{a}\otimes D_{b}=D_{c}\oplus...$, where $D_{c}$ is \textit{reducible}. If $D_{c}$ decomposes into irreducible representations $D_{1},...,D_{n}$, the reduction coefficients can be obtained by combining the bases of $V_{D_{1}},...,V_{D_{n}}$ to a basis of $V_{D_{c}}$.
\medskip
\\
In the following we will go through all tables of Clebsch-Gordan coefficients listed in table \ref{SU3tensor} and extract all Clebsch-Gordan coefficients. We will now not give the coefficients by listing the basis vectors $u_{i}=\Gamma_{ijk}e_{j}\otimes e_{k}$ of the invariant subspaces, but by showing the matrices $\Gamma_{i}$ ($(\Gamma_{i})_{jk}:=\Gamma_{ijk}$).
\bigskip
\\
\textbf{1-dimensional representations}
\bigskip
\\
Going through all tensor products listed in table \ref{SU3tensor} that contain 1-dimensional representations we find the following possible Clebsch-Gordan coefficients (ignoring irrelevant phase factors):
	\begin{displaymath}
	\begin{split}
	& \Gamma=\frac{1}{\sqrt{3}}
\left(\begin{matrix}
 1 & 0 & 0 \\
 0 & 1 & 0 \\
 0 & 0 & 1
\end{matrix}\right),\enspace
\Gamma=\frac{1}{\sqrt{3}}
\left(\begin{matrix}
 1 & 0 & 0 \\
 0 & \omega & 0 \\
 0 & 0 & \omega^{2}
\end{matrix}\right),\enspace
\Gamma=\frac{1}{\sqrt{3}}
\left(\begin{matrix}
 1 & 0 & 0 \\
 0 & \omega^{2} & 0 \\
 0 & 0 & \omega
\end{matrix}\right),\\
& \Gamma=\frac{1}{\sqrt{3}}
\left(\begin{matrix}
 0 & 0 & 1 \\
 1 & 0 & 0 \\
 0 & 1 & 0
\end{matrix}\right),\enspace
\Gamma=\frac{1}{\sqrt{3}}
\left(\begin{matrix}
 0 & 0 & 1 \\
 \omega^{2} & 0 & 0 \\
 0 & \omega & 0
\end{matrix}\right),\enspace
\Gamma=\frac{1}{\sqrt{3}}
\left(\begin{matrix}
 0 & 0 & 1 \\
 \omega & 0 & 0 \\
 0 & \omega^{2} & 0
\end{matrix}\right),\\
& \Gamma=\frac{1}{\sqrt{3}}
\left(\begin{matrix}
 0 & 1 & 0 \\
 0 & 0 & 1 \\
 1 & 0 & 0
\end{matrix}\right),\enspace
\Gamma=\frac{1}{\sqrt{3}}
\left(\begin{matrix}
 0 & 1 & 0 \\
 0 & 0 & \omega^{2} \\
 \omega & 0 & 0
\end{matrix}\right),\enspace
\Gamma=\frac{1}{\sqrt{3}}
\left(\begin{matrix}
 0 & 1 & 0 \\
 0 & 0 & \omega \\
 \omega^{2} & 0 & 0
\end{matrix}\right).
	\end{split}
	\end{displaymath}
\textbf{Claim:} All these Clebsch-Gordan coefficients are equivalent.
\medskip
\\
\textit{Proof}. From proposition \ref{Pclebsch1} we know the transformation properties of Clebsch-Gordan coefficients under basis transformations:
	\begin{displaymath}
	\Gamma_{i}\mapsto \Gamma_{i}'=S_{A}^{-1}\Gamma_{m}(S_{B}^{-1})^{T}(S_{\lambda})_{mi},
	\end{displaymath}
where $S_{A}, S_{B}$ and $S_{\lambda}$ are unitary matrices. In the case of 1-dimensional representations we have $S_{\lambda}\in \mathbb{C}$, thus we can absorb it into $(S_{B}^{-1})^{T}$ and get
	\begin{displaymath}
	\Gamma'=S_{A}^{-1}\Gamma (S_{B}^{-1})^{T}.
	\end{displaymath}
All matrices $\Gamma$ listed above are proportional to unitary matrices $\tilde{\Gamma}=\frac{1}{\lambda}\Gamma$ ($\lambda\in \mathbb{R}_{+}$, here: $\lambda=\frac{1}{\sqrt{3}}$), hence $\tilde{\Gamma}$ is diagonalizeable (see theorem \ref{Tclebsch3}).
	\begin{displaymath}
	\Rightarrow \exists A,B\enspace \mbox{(unitary):}\quad A\tilde{\Gamma} B=D,
	\end{displaymath}
where $D$ is a diagonal phase matrix.
	\begin{displaymath}
\Rightarrow A\tilde{\Gamma} B (A\tilde{\Gamma} B)^{\ast}=\mathbbm{1}_{3}	
\Rightarrow \underbrace{A}_{=:S_{A}^{-1}}\Gamma \underbrace{B (A\tilde{\Gamma} B)^{\ast}}_{=:(S_{B}^{-1})^{T}}=\lambda\mathbbm{1}_{3}.
	\end{displaymath}
\qed
\medskip
\\
Therefore the only possible Clebsch-Gordan coefficients for 1-dimensional representations (up to basis transformations) are
	\begin{displaymath}
	\Gamma=\frac{1}{\sqrt{3}}
\left(\begin{matrix}
 1 & 0 & 0 \\
 0 & 1 & 0 \\
 0 & 0 & 1
\end{matrix}\right).
	\end{displaymath}
Due to the simple structure of equation (\ref{basisequ}) in the case of 1-dimensional representations it was possible to prove equivalence. When one goes over to higher dimensional representations equation (\ref{basisequ}) becomes of course more and more complicated. We will not systematically analyse the remaining sets of Clebsch-Gordan coefficients on equivalence. However we can use one simple rule to detect equivalence in at least some cases: \textit{Two sets of Clebsch-Gordan coefficients (for fixed dimension of $D^{\lambda}$ in $D_{a}\otimes D_{b}=D^{\lambda}\oplus...$) which only differ by phase factors or order of the basis vectors of the invariant subspace are equivalent up to basis transformations.}
\bigskip
\\
Following this rule we can easily list all (not obviously equivalent) Clebsch-Gordan coefficients, by just going through all tables listed in table \ref{SU3tensor}. Doing this one finds only one special case which needs further explanation:
	\begin{displaymath}
	A_{4}:\enspace \textbf{\underline{3}}\otimes \textbf{\underline{3}} =\textbf{\underline{1}}\oplus\textbf{\underline{1}}'\oplus\textbf{\underline{1}}''\oplus \textbf{\underline{3}}\oplus\textbf{\underline{3}}.
	\end{displaymath}
This decomposition (which also occurs within $\Delta(3n^{2})$ and $\Sigma(216\phi)$) is special, because \textit{the same} 3-dimensional irreducible representation $\textbf{\underline{3}}$ occurs twice in $\textbf{\underline{3}}\otimes \textbf{\underline{3}}$. Therefore the vectorspace of solutions of the invariance equations is two-dimensional. The two linearly independent sets of Clebsch-Gordan coefficients are
\begin{displaymath}
\Gamma_{1}=\left(
\begin{matrix}
 0 & 0 & 0 \\
 0 & 0 & 1 \\
 0 & 0 & 0
\end{matrix}	
\right),\enspace
\Gamma_{2}=\left(
\begin{matrix}
 0 & 0 & 0 \\
 0 & 0 & 0 \\
 1 & 0 & 0
\end{matrix}	
\right),\enspace
\Gamma_{3}=\left(
\begin{matrix}
 0 & 1 & 0 \\
 0 & 0 & 0 \\
 0 & 0 & 0
\end{matrix}	
\right),
\end{displaymath}
and
\begin{displaymath}
\Gamma_{1}'=\left(
\begin{matrix}
 0 & 0 & 0 \\
 0 & 0 & 0 \\
 0 & 1 & 0
\end{matrix}	
\right),\enspace
\Gamma_{2}'=\left(
\begin{matrix}
 0 & 0 & 1 \\
 0 & 0 & 0 \\
 0 & 0 & 0
\end{matrix}	
\right),\enspace
\Gamma_{3}'=\left(
\begin{matrix}
 0 & 0 & 0 \\
 1 & 0 & 0 \\
 0 & 0 & 0
\end{matrix}	
\right).
\end{displaymath}
The general form of the Clebsch-Gordan coefficients is therefore 2-parametric:
	\begin{displaymath}
	\tilde{\Gamma}_{j}=\frac{\alpha \Gamma_{j}+\beta\Gamma_{j}'}{\sqrt{|\alpha|^{2}+|\beta|^{2}}},\quad \tilde{\Gamma}_{j}'=\frac{\beta^{\ast} \Gamma_{j}-\alpha^{\ast}\Gamma_{j}'}{\sqrt{|\alpha|^{2}+|\beta|^{2}}},
	\end{displaymath}
with two complex parameters $\alpha,\beta$ ($(\alpha,\beta)\in \mathbb{C}^{2}\backslash\{(0,0)\}$). Thus the basis vectors of the invariant subspaces become
	\begin{displaymath}
	\begin{split}
	& u_{\textbf{\underline{3}}}(1)=\frac{1}{\sqrt{|\alpha|^{2}+|\beta|^{2}}}(\alpha e_{23}+\beta e_{32}),\\
	& u_{\textbf{\underline{3}}}(2)=\frac{1}{\sqrt{|\alpha|^{2}+|\beta|^{2}}}(\alpha e_{31}+\beta e_{13}),\\
	& u_{\textbf{\underline{3}}}(3)=\frac{1}{\sqrt{|\alpha|^{2}+|\beta|^{2}}}(\alpha e_{12}+\beta e_{21})
	\end{split}
	\end{displaymath}
for the first invariant subspace, and
	\begin{displaymath}
	\begin{split}
	& u_{\textbf{\underline{3}}}'(1)=\frac{1}{\sqrt{|\alpha|^{2}+|\beta|^{2}}}(\beta^{\ast} e_{23}-\alpha^{\ast} e_{32}),\\
	& u_{\textbf{\underline{3}}}'(2)=\frac{1}{\sqrt{|\alpha|^{2}+|\beta|^{2}}}(\beta^{\ast} e_{31}-\alpha^{\ast} e_{13}),\\
	& u_{\textbf{\underline{3}}}'(3)=\frac{1}{\sqrt{|\alpha|^{2}+|\beta|^{2}}}(\beta^{\ast} e_{12}-\alpha^{\ast} e_{21})
	\end{split}
	\end{displaymath}
for the second invariant subspace. Thus the Clebsch-Gordan coefficients for
	\begin{displaymath}
	A_{4}: \textbf{\underline{3}}\otimes \textbf{\underline{3}}\rightarrow \textbf{\underline{3}}
	\end{displaymath}
have the general form
\begin{displaymath}
\begin{split}
& \Gamma_{1}(\alpha,\beta)=
\frac{1}{\sqrt{|\alpha|^{2}+|\beta|^{2}}}\left(
\begin{matrix}
 0 & 0 & 0 \\
 0 & 0 & \alpha \\
 0 & \beta & 0
\end{matrix}
\right),\\
& \Gamma_{2}(\alpha,\beta)=
\frac{1}{\sqrt{|\alpha|^{2}+|\beta|^{2}}}\left(
\begin{matrix}
 0 & 0 & \beta \\
 0 & 0 & 0 \\
 \alpha & 0 & 0
\end{matrix}
\right),\\
& \Gamma_{3}(\alpha,\beta)=
\frac{1}{\sqrt{|\alpha|^{2}+|\beta|^{2}}}\left(
\begin{matrix}
 0 & \alpha & 0 \\
 \beta & 0 & 0 \\
 0 & 0 & 0
\end{matrix}
\right),
\end{split}
\end{displaymath}
with $(\alpha,\beta)\in \mathbb{C}^{2}\backslash\{(0,0)\}$. Fixing $\alpha,\beta$ the coefficients of the second 3-dimensional invariant subspace are determined by $\Gamma_{j}'(\alpha,\beta)=\Gamma_{j}(\beta^{\ast},-\alpha^{\ast})$.
\medskip
\\
We will call $\Gamma_{j}(\alpha,\beta)$ the Clebsch-Gordan coefficients of type \textbf{IIIa$_{(\alpha,\beta)}$}. When one goes through the tables of Clebsch-Gordan coefficients of this chapter one finds that \textbf{IIIa$_{(1,1)}$} and \textbf{IIIa$_{(1,-1)}$} occur within the study of $S_{4}$. In principle \textbf{IIIa$_{(1,1)}$} and \textbf{IIIa$_{(1,-1)}$} could be treated as independent types of coefficients (because of the fixed $\alpha$ and $\beta$). On the other hand they are included in the more general type \textbf{IIIa$_{(\alpha,\beta)}$}.
\medskip
\\
Going through the lists of Clebsch-Gordan coefficients listed in table \ref{SU3tensor}, one also finds some other types of Clebsch-Gordan coefficients which are special cases of type \textbf{IIIa$_{(\alpha,\beta)}$}. As an example, for $\textbf{\underline{3}}_{(-k,-l)}\otimes \textbf{\underline{3}}_{(k,l)}$ of $\Delta(3n^{2})$ there occur Clebsch-Gordan coefficients of type \textbf{IIIa$_{(1,0)}$} and \textbf{IIIa$_{(0,1)}$}. We will also include these types in the list.
\bigskip
\\
The different types of Clebsch-Gordan coefficients are listed in table \ref{SU3CGClist}.
%Folgendes ist unklar formuliert und nicht unbedingt notwendig - streichen wir.
% Note that we list all Clebsch-Gordan coefficients (ignoring irrelevant phase factors) obtained in this chapter independent on whether they are equivalent or not, because if one would like to construct reduction coefficients for $D_a\otimes D_b=D_c\oplus...$, where $D_c$ is reducible, one needs all Clebsch-Gordan coefficients for the reduction of $D_a\otimes D_b$ not only up to basis transformations, but in the form of basis vectors of $V_{D_a}\otimes V_{D_b}$. (Only if one could show that a set of Clebsch-Gordan coefficients for a whole decomposition is equivalent to another set, one could discard one of the sets. This is (for example) the case for decompositions like $\textbf{\underline{3}}\otimes \textbf{\underline{3}}=\textbf{\underline{3}}_{a}\oplus \textbf{\underline{6}}_{s}$.)
\begin{footnotesize}
\renewcommand{\arraystretch}{1.6}
\begin{longtable}{|ll|}
\hline
Type & Clebsch-Gordan coefficients\\
\hline
\endhead
\hline
\endfoot
\endlastfoot
\textbf{Ia} & $\Gamma=\frac{1}{\sqrt{3}}
\left(
\begin{matrix}
 1 & 0 & 0 \\
 0 & 1 & 0 \\
 0 & 0 & 1
\end{matrix}
\right)
$
\\
\textbf{Ib} & $\Gamma=\frac{1}{\sqrt{3}}
\left(
\begin{matrix}
 \omega^{2} & 0 & 0 \\
 0 & \omega & 0 \\
 0 & 0 & 1
\end{matrix}
\right)
$
\\
\textbf{Ic} & $\Gamma=\frac{1}{\sqrt{3}}
\left(
\begin{matrix}
 \omega & 0 & 0 \\
 0 & \omega^{2} & 0 \\
 0 & 0 & 1
\end{matrix}
\right)
$
\\
\textbf{Id} & $\Gamma=\frac{1}{\sqrt{3}}
\left(\begin{matrix}
 0 & 0 & 1 \\
 1 & 0 & 0 \\
 0 & 1 & 0
\end{matrix}\right)
$
\\
\textbf{Ie} & $\Gamma=\frac{1}{\sqrt{3}}
\left(\begin{matrix}
 0 & 0 & 1 \\
 \omega^{2} & 0 & 0 \\
 0 & \omega & 0
\end{matrix}\right)
$
\\
\textbf{If} & $\Gamma=\frac{1}{\sqrt{3}}
\left(\begin{matrix}
 0 & 0 & 1 \\
 \omega & 0 & 0 \\
 0 & \omega^{2} & 0
\end{matrix}\right)
$
\\
\textbf{Ig} & $\Gamma=\frac{1}{\sqrt{3}}
\left(\begin{matrix}
 0 & 1 & 0 \\
 0 & 0 & 1 \\
 1 & 0 & 0
\end{matrix}\right)
$
\\
\textbf{Ih} & $\Gamma=\frac{1}{\sqrt{3}}
\left(\begin{matrix}
 0 & 1 & 0 \\
 0 & 0 & \omega^{2} \\
 \omega & 0 & 0
\end{matrix}\right)
$
\\
\textbf{Ii} & $\Gamma=\frac{1}{\sqrt{3}}
\left(\begin{matrix}
 0 & 1 & 0 \\
 0 & 0 & \omega \\
 \omega^{2} & 0 & 0
\end{matrix}\right)
$
\\
\textbf{IIa} & $\Gamma_{1}=\frac{1}{\sqrt{2}}
\left(
\begin{matrix}
 0 & 0 & 0 \\
 0 & 1 & 0 \\
 0 & 0 & -1
\end{matrix}
\right)
$,
$\enspace\Gamma_{2}=\frac{1}{\sqrt{6}}
\left(
\begin{matrix}
 -2 & 0 & 0 \\
 0 & 1 & 0 \\
 0 & 0 & 1
\end{matrix}
\right)
$
\\
\textbf{IIb} & $\Gamma_{1}=\frac{1}{\sqrt{3}}
\left(
\begin{matrix}
 1 & 0 & 0 \\
 0 & \omega^{2} & 0 \\
 0 & 0 & \omega
\end{matrix}
\right)
$,$\enspace$
$\Gamma_{2}=\frac{1}{\sqrt{3}}
\left(
\begin{matrix}
 1 & 0 & 0 \\
 0 & \omega & 0 \\
 0 & 0 & \omega^{2}
\end{matrix}
\right)
$
\\
\textbf{IIc} & $\Gamma_{1}=\frac{1}{\sqrt{3}}
\left(
\begin{matrix}
 0 & 0 & 1 \\
 \omega & 0 & 0 \\
 0 & \omega^{2} & 0
\end{matrix}
\right)
$,
$\enspace\Gamma_{2}=\frac{1}{\sqrt{3}}
\left(
\begin{matrix}
 0 & \omega^{2} & 0 \\
 0 & 0 & \omega \\
 1 & 0 & 0
\end{matrix}
\right)
$
\\
\textbf{IId} & $\Gamma_{1}=\frac{1}{\sqrt{3}}
\left(
\begin{matrix}
 0 & \omega & 0 \\
 0 & 0 & \omega^{2} \\
 1 & 0 & 0
\end{matrix}
\right)
$,
$\enspace\Gamma_{2}=\frac{1}{\sqrt{3}}
\left(
\begin{matrix}
 0 & 0 & 1 \\
 \omega^{2} & 0 & 0 \\
 0 & \omega & 0
\end{matrix}
\right)
$
\\
\textbf{IIe} & $\Gamma_{1}=\frac{1}{\sqrt{3}}
\left(
\begin{matrix}
 0 & 1 & 0 \\
 0 & 0 & 1 \\
 1 & 0 & 0
\end{matrix}
\right)
$,
$\enspace\Gamma_{2}=\frac{1}{\sqrt{3}}
\left(
\begin{matrix}
 0 & 0 & 1 \\
 1 & 0 & 0 \\
 0 & 1 & 0
\end{matrix}
\right)
$
\\
\textbf{IIIa$_{(\alpha,\beta)}$} & $\Gamma_{1}=
\frac{1}{n_{\alpha\beta}}
\left(
\begin{matrix}
 0 & 0 & 0 \\
 0 & 0 & \alpha \\
 0 & \beta & 0
\end{matrix}
\right),\enspace
\Gamma_{2}=
\frac{1}{n_{\alpha\beta}}\left(
\begin{matrix}
 0 & 0 & \beta \\
 0 & 0 & 0 \\
 \alpha & 0 & 0
\end{matrix}
\right),\enspace
\Gamma_{3}=
\frac{1}{n_{\alpha\beta}}\left(
\begin{matrix}
 0 & \alpha & 0 \\
 \beta & 0 & 0 \\
 0 & 0 & 0
\end{matrix}
\right)$
\\
\textbf{IIIa$_{(1,1)}$} & $\Gamma_{1}=
\frac{1}{\sqrt{2}}
\left(
\begin{matrix}
 0 & 0 & 0 \\
 0 & 0 & 1 \\
 0 & 1 & 0
\end{matrix}
\right),\enspace
\Gamma_{2}=
\frac{1}{\sqrt{2}}\left(
\begin{matrix}
 0 & 0 & 1 \\
 0 & 0 & 0 \\
 1 & 0 & 0
\end{matrix}
\right),\enspace
\Gamma_{3}=
\frac{1}{\sqrt{2}}\left(
\begin{matrix}
 0 & 1 & 0 \\
 1 & 0 & 0 \\
 0 & 0 & 0
\end{matrix}
\right)$
\\
\textbf{IIIa$_{(1,-1)}$} & $\Gamma_{1}=
\frac{1}{\sqrt{2}}
\left(
\begin{matrix}
 0 & 0 & 0 \\
 0 & 0 & 1 \\
 0 & -1 & 0
\end{matrix}
\right),\enspace
\Gamma_{2}=
\frac{1}{\sqrt{2}}\left(
\begin{matrix}
 0 & 0 & -1 \\
 0 & 0 & 0 \\
 1 & 0 & 0
\end{matrix}
\right),\enspace
\Gamma_{3}=
\frac{1}{\sqrt{2}}\left(
\begin{matrix}
 0 & 1 & 0 \\
 -1 & 0 & 0 \\
 0 & 0 & 0
\end{matrix}
\right)$
\\
\textbf{IIIa$_{(1,0)}$} & $\Gamma_{1}=
\left(
\begin{matrix}
 0 & 0 & 0 \\
 0 & 0 & 1 \\
 0 & 0 & 0
\end{matrix}
\right),\enspace
\Gamma_{2}=
\left(
\begin{matrix}
 0 & 0 & 0 \\
 0 & 0 & 0 \\
 1 & 0 & 0
\end{matrix}
\right),\enspace
\Gamma_{3}=
\left(
\begin{matrix}
 0 & 1 & 0 \\
 0 & 0 & 0 \\
 0 & 0 & 0
\end{matrix}
\right)$
\\
\textbf{IIIa$_{(0,1)}$} & $\Gamma_{1}=
\left(
\begin{matrix}
 0 & 0 & 0 \\
 0 & 0 & 0 \\
 0 & 1 & 0
\end{matrix}
\right),\enspace
\Gamma_{2}=
\left(
\begin{matrix}
 0 & 0 & 1 \\
 0 & 0 & 0 \\
 0 & 0 & 0
\end{matrix}
\right),\enspace
\Gamma_{3}=
\left(
\begin{matrix}
 0 & 0 & 0 \\
 1 & 0 & 0 \\
 0 & 0 & 0
\end{matrix}
\right)$
\\
\textbf{IIIb} & $\Gamma_{1}=
\left(
\begin{matrix}
 -\frac{\tau_{-}}{\sqrt{12}} & 0 & 0 \\
 0 & 0 & \frac{1}{\tau_{-}} \\
 0 & \frac{1}{\tau_{-}} & 0
\end{matrix}
\right),\enspace
\Gamma_{2}=
\left(
\begin{matrix}
 0 & 0 & \frac{1}{\tau_{-}} \\
 0 & -\frac{\tau_{-}}{\sqrt{12}} & 0 \\
 \frac{1}{\tau_{-}} & 0 & 0
\end{matrix}
\right),\enspace
\Gamma_{3}=
\left(
\begin{matrix}
 0 & \frac{1}{\tau_{-}} & 0 \\
 \frac{1}{\tau_{-}} & 0 & 0 \\
 0 & 0 & -\frac{\tau_{-}}{\sqrt{12}}
\end{matrix}
\right)$
\\
\textbf{IIIc} & $\Gamma_{1}=
\left(
\begin{matrix}
 \frac{\tau_{+}}{\sqrt{12}} & 0 & 0 \\
 0 & 0 & \frac{1}{\tau_{+}} \\
 0 & \frac{1}{\tau_{+}} & 0
\end{matrix}
\right),\enspace
\Gamma_{2}=
\left(
\begin{matrix}
 0 & 0 & \frac{1}{\tau_{+}} \\
 0 & \frac{\tau_{+}}{\sqrt{12}} & 0 \\
 \frac{1}{\tau_{+}} & 0 & 0
\end{matrix}
\right),\enspace
\Gamma_{3}=
\left(
\begin{matrix}
 0 & \frac{1}{\tau_{+}} & 0 \\
 \frac{1}{\tau_{+}} & 0 & 0 \\
 0 & 0 & \frac{\tau_{+}}{\sqrt{12}}
\end{matrix}
\right)$
\\
\textbf{IIId} & $\Gamma_{1}=
\left(
\begin{matrix}
 1 & 0 & 0 \\
 0 & 0 & 0 \\
 0 & 0 & 0
\end{matrix}
\right),\enspace
\Gamma_{2}=
\left(
\begin{matrix}
 0 & 0 & 0 \\
 0 & 1 & 0 \\
 0 & 0 & 0
\end{matrix}
\right),\enspace
\Gamma_{3}=
\left(
\begin{matrix}
 0 & 0 & 0 \\
 0 & 0 & 0 \\
 0 & 0 & 1
\end{matrix}
\right)$
\\
\textbf{IVa} & $\Gamma_{1}=\frac{1}{\sqrt{3}}
\left(
\begin{matrix}
 1 & 0 & 0 \\
 0 & 1 & 0 \\
 0 & 0 & 1
\end{matrix}
\right)$,\\ &
$\Gamma_{2}=\left(
\begin{matrix}
 0 & \sqrt{\frac{1}{6} \left(3-\sqrt{5}\right)} & 0 \\
 \frac{1}{2}(-1+\sqrt{5})\sqrt{\frac{3+\sqrt{5}}{9-3\sqrt{5}}} & 0 & 0 \\
 0 & 0 & 0
\end{matrix}
\right)$,\\
&
$\Gamma_{3}=\left(
\begin{matrix}
 0 & 0 & \sqrt{\frac{1}{6} \left(3+\sqrt{5}\right)} \\
 0 & 0 & 0 \\
 \frac{-1+\sqrt{5}}{2\sqrt{3}} & 0 & 0
\end{matrix}
\right)$,
\\ &
$\Gamma_{4}=\left(
\begin{matrix}
 0 & 0 & 0 \\
 0 & 0 & \sqrt{\frac{1}{6} \left(3-\sqrt{5}\right)} \\
 0 & \frac{1}{2}(-1+\sqrt{5})\sqrt{\frac{3+\sqrt{5}}{9-3\sqrt{5}}} & 0
\end{matrix}
\right)$
\\
\textbf{IVb} & $\Gamma_{1}=\frac{1}{\sqrt{3}}
\left(
\begin{matrix}
 0 & \omega^2 & 0 \\
 0 & 0 & 1 \\
 \omega & 0 & 0
\end{matrix}
\right),\enspace
\Gamma_{2}=\frac{1}{\sqrt{3}}\left(
\begin{matrix}
 0 & 0 & \omega^2 \\
 1 & 0 & 0 \\
 0 & \omega & 0
\end{matrix}
\right)$,\\
&
$\Gamma_{3}=\frac{1}{\sqrt{3}}\left(
\begin{matrix}
 0 & \omega^2 & 0 \\
 0 & 0 & \omega \\
 1 & 0 & 0
\end{matrix}
\right),\enspace
\Gamma_{4}=\frac{1}{\sqrt{3}}\left(
\begin{matrix}
 0 & 0 & \omega^2 \\
 \omega & 0 & 0 \\
 0 & 1 & 0
\end{matrix}
\right)$
\\
\textbf{IVc} & $\Gamma_{1}=\frac{1}{\sqrt{3}}
\left(
\begin{matrix}
 0 & 0 & \omega \\
 \omega & 0 & 0 \\
 0 & \omega & 0
\end{matrix}
\right),\enspace
\Gamma_{2}=\frac{1}{\sqrt{3}}\left(
\begin{matrix}
 \omega & 0 & 0 \\
 0 & 1 & 0 \\
 0 & 0 & \omega^2
\end{matrix}
\right)$,\\
&
$\Gamma_{3}=\frac{1}{\sqrt{3}}\left(
\begin{matrix}
 0 & \omega^2 & 0 \\
 0 & 0 & \omega^2 \\
 \omega^2 & 0 & 0
\end{matrix}
\right),\enspace
\Gamma_{4}=\frac{1}{\sqrt{3}}\left(
\begin{matrix}
 \omega & 0 & 0 \\
 0 & \omega^2 & 0 \\
 0 & 0 & 1
\end{matrix}
\right)$
\\
\textbf{Va} & $\Gamma_{1}=
\frac{1}{\sqrt{2}}\left(
\begin{matrix}
 0 & 0 & 0 \\
 0 & 1 & 0 \\
 0 & 0 & -1
\end{matrix}
\right),\enspace
\Gamma_{2}=\frac{1}{\sqrt{6}}\left(
\begin{matrix}
 -2 & 0 & 0 \\
 0 & 1 & 0 \\
 0 & 0 & 1
\end{matrix}
\right),\enspace
\Gamma_{3}=\frac{1}{\sqrt{2}}\left(
\begin{matrix}
 0 & 0 & 0 \\
 0 & 0 & 1 \\
 0 & 1 & 0
\end{matrix}
\right)$,\\
&
$\Gamma_{4}=\frac{1}{\sqrt{2}}\left(
\begin{matrix}
 0 & 0 & 1 \\
 0 & 0 & 0 \\
 1 & 0 & 0
\end{matrix}
\right),\enspace
\Gamma_{5}=\frac{1}{\sqrt{2}}\left(
\begin{matrix}
 0 & 1 & 0 \\
 1 & 0 & 0 \\
 0 & 0 & 0
\end{matrix}
\right)$
\\
\textbf{Vb} & $\Gamma_{1}=
\left(
\begin{matrix}
 -\frac{\sqrt{\frac{5}{3}}}{2} & 0 & 0 \\
 0 & \frac{-3+\sqrt{5}}{4 \sqrt{3}} & 0 \\
 0 & 0 & \frac{3+\sqrt{5}}{4 \sqrt{3}}
\end{matrix}
\right),\enspace
\Gamma_{2}=\left(
\begin{matrix}
 \frac{1}{2} & 0 & 0 \\
 0 & \frac{1}{4} \left(-1-\sqrt{5}\right) & 0 \\
 0 & 0 & \frac{1}{4} \left(-1+\sqrt{5}\right)
\end{matrix}
\right)$,\\
& $\Gamma_{3}=\left(
\begin{matrix}
 0 & 0 & 0 \\
 0 & 0 & \sqrt{\frac{1}{6} \left(3+\sqrt{5}\right)} \\
 0 & -\frac{-1+\sqrt{5}}{2 \sqrt{3}} & 0
\end{matrix}
\right),\enspace
\Gamma_{4}=\left(
\begin{matrix}
 0 & 0 & -\frac{-1+\sqrt{5}}{2 \sqrt{3}} \\
 0 & 0 & 0 \\
 \sqrt{\frac{1}{6} \left(3+\sqrt{5}\right)} & 0 & 0
\end{matrix}
\right)$,\\
& $\Gamma_{5}=\left(
\begin{matrix}
 0 & \sqrt{\frac{1}{6} \left(3+\sqrt{5}\right)} & 0 \\
 -\frac{-1+\sqrt{5}}{2 \sqrt{3}} & 0 & 0 \\
 0 & 0 & 0
\end{matrix}
\right)$
\\
\textbf{VIa} & $\Gamma_{1}=
\left(
\begin{matrix}
 1 & 0 & 0 \\
 0 & 0 & 0 \\
 0 & 0 & 0
\end{matrix}
\right),\enspace
\Gamma_{2}=\left(
\begin{matrix}
 0 & 0 & 0 \\
 0 & 1 & 0 \\
 0 & 0 & 0
\end{matrix}
\right),\enspace
\Gamma_{3}=\left(
\begin{matrix}
 0 & 0 & 0 \\
 0 & 0 & 0 \\
 0 & 0 & 1
\end{matrix}
\right)$,\\
&
$\Gamma_{4}=\frac{1}{\sqrt{2}}\left(
\begin{matrix}
 0 & 0 & 0 \\
 0 & 0 & 1 \\
 0 & 1 & 0
\end{matrix}
\right),\enspace
\Gamma_{5}=\frac{1}{\sqrt{2}}\left(
\begin{matrix}
 0 & 0 & 1 \\
 0 & 0 & 0 \\
 1 & 0 & 0
\end{matrix}
\right),\enspace
\Gamma_{6}=\frac{1}{\sqrt{2}}\left(
\begin{matrix}
 0 & 1 & 0 \\
 1 & 0 & 0 \\
 0 & 0 & 0
\end{matrix}
\right)$
\\
\textbf{VIb} & $\Gamma_{1}=
\left(
\begin{matrix}
 0 & 1 & 0 \\
 0 & 0 & 0 \\
 0 & 0 & 0
\end{matrix}
\right),\enspace
\Gamma_{2}=\left(
\begin{matrix}
 0 & 0 & 0 \\
 0 & 0 & 1 \\
 0 & 0 & 0
\end{matrix}
\right),\enspace
\Gamma_{3}=\left(
\begin{matrix}
 0 & 0 & 0 \\
 0 & 0 & 0 \\
 1 & 0 & 0
\end{matrix}
\right)$,\\
&
$\Gamma_{4}=\left(
\begin{matrix}
 0 & 0 & 0 \\
 0 & 0 & 0 \\
 0 & 1 & 0
\end{matrix}
\right),\enspace
\Gamma_{5}=\left(
\begin{matrix}
 0 & 0 & 0 \\
 1 & 0 & 0 \\
 0 & 0 & 0
\end{matrix}
\right),\enspace
\Gamma_{6}=\left(
\begin{matrix}
 0 & 0 & 1 \\
 0 & 0 & 0 \\
 0 & 0 & 0
\end{matrix}
\right)$
\\
\textbf{VIII} & $\Gamma_{1}=
\left(
\begin{matrix}
 0 & 1 & 0 \\
 0 & 0 & 0 \\
 0 & 0 & 0
\end{matrix}
\right),\enspace
\Gamma_{2}=\left(
\begin{matrix}
 0 & 0 & 1 \\
 0 & 0 & 0 \\
 0 & 0 & 0
\end{matrix}
\right),\enspace
\Gamma_{3}=\left(
\begin{matrix}
 0 & 0 & 0 \\
 1 & 0 & 0 \\
 0 & 0 & 0
\end{matrix}
\right)$,\\
&
$\Gamma_{4}=\frac{1}{\sqrt{2}}\left(
\begin{matrix}
 1 & 0 & 0 \\
 0 & -1 & 0 \\
 0 & 0 & 0
\end{matrix}
\right),\enspace
\Gamma_{5}=\left(
\begin{matrix}
 0 & 0 & 0 \\
 0 & 0 & 1 \\
 0 & 0 & 0
\end{matrix}
\right),\enspace
\Gamma_{6}=\left(
\begin{matrix}
 0 & 0 & 0 \\
 0 & 0 & 0 \\
 1 & 0 & 0
\end{matrix}
\right)$,\\
& $
\Gamma_{7}=\left(
\begin{matrix}
 0 & 0 & 0 \\
 0 & 0 & 0 \\
 0 & 1 & 0
\end{matrix}
\right),\enspace
\Gamma_{8}=\frac{1}{\sqrt{6}}\left(
\begin{matrix}
 1 & 0 & 0 \\
 0 & 1 & 0 \\
 0 & 0 & -2
\end{matrix}
\right)
$
\\
\textbf{IX} & $\Gamma_{1}=
\left(
\begin{matrix}
 1 & 0 & 0 \\
 0 & 0 & 0 \\
 0 & 0 & 0
\end{matrix}
\right),\enspace
\Gamma_{2}=\left(
\begin{matrix}
 0 & 1 & 0 \\
 0 & 0 & 0 \\
 0 & 0 & 0
\end{matrix}
\right),\enspace
\Gamma_{3}=\left(
\begin{matrix}
 0 & 0 & 1 \\
 0 & 0 & 0 \\
 0 & 0 & 0
\end{matrix}
\right)$,\\
&
$\Gamma_{4}=\left(
\begin{matrix}
 0 & 0 & 0 \\
 1 & 0 & 0 \\
 0 & 0 & 0
\end{matrix}
\right),\enspace
\Gamma_{5}=\left(
\begin{matrix}
 0 & 0 & 0 \\
 0 & 1 & 0 \\
 0 & 0 & 0
\end{matrix}
\right),\enspace
\Gamma_{6}=\left(
\begin{matrix}
 0 & 0 & 0 \\
 0 & 0 & 1 \\
 0 & 0 & 0
\end{matrix}
\right)$,\\
& $
\Gamma_{7}=\left(
\begin{matrix}
 0 & 0 & 0 \\
 0 & 0 & 0 \\
 1 & 0 & 0
\end{matrix}
\right),\enspace
\Gamma_{8}=\left(
\begin{matrix}
 0 & 0 & 0 \\
 0 & 0 & 0 \\
 0 & 1 & 0
\end{matrix}
\right),\enspace
\Gamma_{9}=\left(
\begin{matrix}
 0 & 0 & 0 \\
 0 & 0 & 0 \\
 0 & 0 & 1
\end{matrix}
\right)
$
\\
\hline
\caption[The possible types of Clebsch-Gordan coefficients for tensor products of the 3-dimensional irreducible representations of the studied finite subgroups of $SU(3)$ in matrix form.]{The possible types of Clebsch-Gordan coefficients for tensor products of the 3-dimensional irreducible representations of the studied finite subgroups of $SU(3)$ in matrix form. Remark: This list maybe contains equivalent sets of Clebsch-Gordan coefficients.
\\
$\omega=e^{\frac{2\pi i}{3}}$; $n_{\alpha\beta}=\sqrt{|\alpha|^{2}+|\beta|^{2}}, \alpha,\beta\in U(1)$; $\tau_{\pm}=\sqrt{2(3\pm\sqrt{3})}$.}
\label{SU3CGClist}
\end{longtable}
\end{footnotesize}
\hspace{0mm}\\
The types of Clebsch-Gordan coefficients listed in table \ref{SU3CGClist} are also indicated in all tables of Clebsch-Gordan coefficients in this chapter. In table \ref{SU3CGClist2} we present a list showing which groups allow tensor products that lead to a given type of Clebsch-Gordan coefficients. The corresponding tensor products can be found by going through the subsections treating the groups, and especially by going through the tables of Clebsch-Gordan coefficients in these subsections.
\medskip
\\
At last in table \ref{SU3tensor2} we list all Clebsch-Gordan coefficients for $3\otimes 3$-tensor products of 3-dimensional irreducible representations of (known) finite subgroups of $SU(3)$, the groups for which the decompositions occur and the tables where the coefficients can be found.
\medskip
\\
\textbf{Remark}: When we treated tensor products of the form $\textbf{\underline{3}}\otimes \textbf{\underline{3}}'$ with $\textbf{\underline{3}}\neq \textbf{\underline{3}}'$ we (most times) did not mention the Clebsch-Gordan coefficients for $\textbf{\underline{3}}'\otimes \textbf{\underline{3}}$. They can be obtained by replacing $e_{ij}\mapsto e_{ji}$, which corresponds to $\Gamma_{i}\mapsto \Gamma_{i}^{T}$ in table \ref{SU3CGClist}.
\bigskip
\\
Looking at table \ref{SU3tensor2} we notice that we have found 17 types of Clebsch-Gordan decompositions (of which some could be equivalent) for $3\otimes 3$-tensor products of 3-dimensional irreducible representations of finite subgroups of $SU(3)$.

\newpage
\begin{footnotesize}
\renewcommand{\arraystretch}{1.6}
\begin{longtable}{|ll|}
\hline
CGC-type & Groups\\
\hline
\endhead
\hline
\endfoot
\endlastfoot
\textbf{Ia} & $A_{4}$, $S_{4}$, $A_{5}$, $\Sigma(168)$, $\Sigma(36\phi)$, $\Sigma(72\phi)$, $\Sigma(216\phi)$, $\Sigma(360\phi)$, $\Delta(3n^{2})$, $\Delta(6n^{2})$ \\
\textbf{Ib} & $A_{4}$, $\Sigma(216\phi)$, $\Delta(3n^{2})$\\
\textbf{Ic} & $A_{4}$, $\Sigma(216\phi)$, $\Delta(3n^{2})$\\
\textbf{Id},\textbf{Ie},...,\textbf{Ii} & $\Delta(6n^{2})$\\
\textbf{IIa} & $S_{4}$\\
\textbf{IIb} & $\Delta(6n^{2})$\\
\textbf{IIc} & $\Delta(6n^{2})$\\
\textbf{IId} & $\Delta(6n^{2})$\\
\textbf{IIe} & $\Delta(6n^{2})$ \\
\textbf{IIIa$_{(\alpha,\beta)}$} & $A_{4}$, $\Sigma(216\phi)$ \\
\textbf{IIIa$_{(1,1)}$} & $A_{4}$, $S_{4}$, $\Sigma(216\phi)$, $\Delta(6n^{2})$ \\
\textbf{IIIa$_{(1,-1)}$} & $A_{4}$, $S_{4}$, $A_{5}$, $\Sigma(168)$, $\Sigma(36\phi)$, $\Sigma(72\phi)$, $\Sigma(216\phi)$, $\Sigma(360\phi)$, $\Delta(3n^{2})$, $\Delta(6n^{2})$ \\
\textbf{IIIa$_{(1,0)}$} & $\Delta(3n^{2})$ \\
\textbf{IIIa$_{(0,1)}$} & $\Delta(3n^{2})$ \\
\textbf{IIIb} & $\Sigma(36\phi)$\\
\textbf{IIIc} & $\Sigma(36\phi)$\\
\textbf{IIId} & $\Delta(3n^{2})$, $\Delta(6n^{2})$\\
\textbf{IVa} & $A_{5}$\\
\textbf{IVb} & $\Sigma(36\phi)$\\
\textbf{IVc} & $\Sigma(36\phi)$ \\
\textbf{Va} & $A_{5}$ \\
\textbf{Vb} & $A_{5}$ \\
\textbf{VIa} & $\Sigma(168)$, $\Sigma(72\phi)$, $\Sigma(216\phi)$, $\Sigma(360\phi)$, $\Delta(6n^{2})$\\
\textbf{VIb} & $\Delta(6n^{2})$\\
\textbf{VIII} & $\Sigma(168)$, $\Sigma(72\phi)$, $\Sigma(216\phi)$, $\Sigma(360\phi)$\\
\textbf{IX} & $\Sigma(216\phi)$, $\Sigma(360\phi)$\\
\hline
\caption[Known finite subgroups of $SU(3)$ that allow a given type of Clebsch-Gordan coefficients.]{Known finite subgroups of $SU(3)$ that allow a given type of Clebsch-Gordan coefficients. We don't list double covers $\tilde{G}$ of finite subgroups $G$ of $SO(3)$. (C) and (D) are not listed either.}
\label{SU3CGClist2}
\end{longtable}
\end{footnotesize}

\begin{sidewaystable}
\begin{center}
\renewcommand{\arraystretch}{1.4}
\begin{tabular}{|lll|}
\firsthline
	Clebsch-Gordan coefficients for $\textbf{\underline{3}}\otimes\textbf{\underline{3}}$ & Groups & CGCs\\
\hline
	\textbf{Ia}+\textbf{Ib}+\textbf{Ic}+\textbf{IIIa}$_{(1,-1)}$+\textbf{IIIa}$_{(1,1)}$ & $A_{4}$, $\Sigma(216\phi)$ & \ref{A4CGC}\\
	\textbf{Ia}+\textbf{IIa}+\textbf{IIIa}$_{(1,-1)}$+\textbf{IIIa}$_{(1,1)}$ & $S_{4}$ & \ref{S4CGC}\\
	\textbf{Ia}+\textbf{IIIa}$_{(1,-1)}$+\textbf{Va} & $A_{5}$ & \ref{A5CGCa}\\
	\textbf{IVa}+\textbf{Vb} & $A_{5}$ & \ref{A5CGCb}\\
	\textbf{IIIa}$_{(1,-1)}$+\textbf{VIa} & $\Sigma(168)$, $\Sigma(72\phi)$, $\Sigma(216\phi)$, $\Sigma(360\phi)$, $\Delta(6n^{2})$ & \ref{Sigma168CGCa}\\
	\textbf{Ia}+\textbf{VIII} & $\Sigma(168)$, $\Sigma(72\phi)$, $\Sigma(216\phi)$, $\Sigma(360\phi)$ & \ref{Sigma168CGCb}\\
	\textbf{IIIa}$_{(1,-1)}$+\textbf{IIIb}+\textbf{IIIc} & $\Sigma(36\phi)$ & \ref{Sigma36CGCa}\\
	\textbf{Ia}+\textbf{IVb}+\textbf{IVc} & $\Sigma(36\phi)$ & \ref{Sigma36CGCb}\\
	\textbf{IX} & $\Sigma(216\phi)$, $\Sigma(360\phi)$ & \\
	\textbf{Ia}+\textbf{Ib}+\textbf{Ic}+\textbf{IIIa}$_{(1,0)}$+\textbf{IIIa}$_{(0,1)}$ & $\Delta(3n^{2})$ & \ref{Delta3nnCGC1}\\
	\textbf{IIId}+\textbf{IIIa}$_{(1,0)}$+\textbf{IIIa}$_{(0,1)}$ & $\Delta(3n^{2})$ & \ref{Delta3nnCGC2}\\
	\textbf{Ia}+\textbf{Ib}+\textbf{Ic}+\textbf{Id}+\textbf{Ie}+\textbf{If}+\textbf{Ig}+\textbf{Ih}+\textbf{Ii} & $\Delta(3n^{2})$ & \ref{Delta3nnCGC3}\\
	\textbf{IIId}+\textbf{VIb} & $\Delta(6n^{2})$ & \ref{Delta6nnCGCb}\\
	\textbf{Ia}+\textbf{IIb}+\textbf{VIb} & $\Delta(6n^{2})$ & \ref{Delta6nnCGCb}+\ref{Delta6nnCGCc}\\
	\textbf{IIId}+\textbf{IIIa}$_{(1,1)}$+\textbf{IIIa}$_{(1,-1)}$ & $\Delta(6n^{2})$ & \ref{Delta6nnCGCd}, \ref{Delta6nnCGCe}\\
	\textbf{Ia}+\textbf{IIb}+\textbf{IIIa}$_{(1,1)}$+\textbf{IIIa}$_{(1,-1)}$ & $\Delta(6n^{2})$ & \ref{Delta6nnCGCc}+\ref{Delta6nnCGCd}, \ref{Delta6nnCGCc}+\ref{Delta6nnCGCe}\\
	\textbf{Ia}+\textbf{IIb}+\textbf{IIc}+\textbf{IId}+\textbf{IIe} & $\Delta(6n^{2})$ & \ref{Delta6nnCGCc}+\ref{Delta6nnCGCf}, \ref{Delta6nnCGCc}+\ref{Delta6nnCGCg}\\
\lasthline
\end{tabular}
\caption[Clebsch-Gordan coefficients for $3\otimes 3$-tensor products of 3-dimensional irreducible representations of (known) finite subgroups of $SU(3)$.]{Clebsch-Gordan coefficients for $3\otimes 3$-tensor products of 3-dimensional irreducible representations of (known) finite subgroups of $SU(3)$. We don't list double covers $\tilde{G}$ of finite subgroups $G$ of $SO(3)$, because $\tilde{G}$ and $G$ have the same 3-dimensional irreducible representations. (C) and (D) are not listed either.}
\label{SU3tensor2}
\end{center}
\end{sidewaystable}

%% file: leptonsector/leptonsector.tex
\chapter{The application of finite family symmetry groups to the lepton sector}\label{LeptonSector}
After the systematic analysis of the finite subgroups of $SU(3)$ we want to investigate the applications of finite family symmetry groups to the lepton sector.

\section{$G$-invariant Higgs potentials}
So far we have only considered $G$-invariant Yukawa-couplings, but of course we want the whole Lagrangian to be $G$-invariant. The remaining part, which we had not analysed up to now, is the Lagrangian of the Higgs fields.
\\
At first we have to generalize the Lagrangian for one standard model Higgs doublet to several standard model Higgs doublets. This has to be done keeping $SU(2)_{I}\times U(1)_{Y}$-invariance in mind. The easiest way to ensure gauge invariance is to construct a Lagrangian consisting of building blocks of the structures
	\begin{displaymath}
	(D_{\mu}\phi_{i})^{\dagger}(D^{\mu}\phi_{j})\quad \mathrm{and}\quad\phi_{i}^{\dagger}\phi_{j}.
	\end{displaymath}
Constructing the Lagrangian from these building blocks Lorentz-invariance is ensured too. We make the following ansatz for the Lagrangian of the Higgs fields:
	\begin{equation}\label{LHiggs}
	\mathcal{L}_{\mathrm{Higgs}}=A_{ij}(D_{\mu}\phi_{i})^{\dagger}(D^{\mu}\phi_{j})-V(\phi),
	\end{equation}
where
	\begin{equation}
	V(\phi)=\mu^2 B_{ij}(\phi_{i}^{\dagger}\phi_{j})+\lambda C_{ijkl}(\phi_{i}^{\dagger}\phi_{j})(\phi_{k}^{\dagger}\phi_{l}).
	\end{equation}
$G$-invariance of $\mathcal{L}_{\mathrm{Higgs}}$ can be ensured by appropriate choice of the coefficients $A_{ij}, B_{ij}$ and $C_{ijkl}$.
\medskip
\\
The condition for $G$-invariance is
	\begin{displaymath}
	\mathcal{L}_{\mathrm{Higgs}}(\phi)=\mathcal{L}_{\mathrm{Higgs}}(D_{\phi}\phi),
	\end{displaymath}
where $D_{\phi}$ is a representation of $G$. As in chapter \ref{yukawa} we will write $[D_{\phi}]$ for the matrix representation of $D_{\phi}$, but for the sake of simplicity we will write $\phi$ instead of $[\phi]$ for the matrix representation of the vector $\phi$.

\begin{prop}\label{Plepton1}
The Higgs Lagrangian (\ref{LHiggs}) is invariant under $\phi\mapsto D_{\phi}\phi$ if
	\begin{displaymath}
	\begin{split}
	&[D_{\phi}]^{\ast}_{ia}[D_{\phi}]_{jb}A_{ij}=A_{ab},\\
	&[D_{\phi}]^{\ast}_{ia}[D_{\phi}]_{jb}B_{ij}=B_{ab},\\
	&[D_{\phi}]^{\ast}_{ia}[D_{\phi}]_{jb}[D_{\phi}]^{\ast}_{kc}[D_{\phi}]_{ld}C_{ijkl}=C_{abcd}.
	\end{split}
	\end{displaymath}
\end{prop}

\begin{proof}
The calculation is identical for all three statements. For $B_{ij}$ we find:
	\begin{displaymath}
	B_{ij}(\phi_{i}^{\dagger}\phi_{j})\mapsto \underbrace{B_{ij}[D_{\phi}]^{\ast}_{ia}[D_{\phi}]_{jb}}_{B_{ab}}(\phi_{a}^{\dagger}\phi_{b})=B_{ab}(\phi_{a}^{\dagger}\phi_{b}).
	\end{displaymath}
\end{proof}

\begin{prop}\label{Plepton2}
Let $A_{ij}, B_{ij}, C_{ijkl}$ be the coefficients of a $G$-invariant Higgs Lagrangian in a given basis defined by the choice of the matrix representation $[D_{\phi}]$ of the operator $D_{\phi}$. Under the basis transformation
	\begin{displaymath}
	[D_{\phi}]\mapsto S^{-1}[D_{\phi}]S
	\end{displaymath}
the coefficients $A_{ij}, B_{ij}, C_{ijkl}$ transform as
	\begin{equation}
	\begin{split}
	&A_{ij}\mapsto S^{\ast}_{ai}S_{bj}A_{ab},\\
	&B_{ij}\mapsto S^{\ast}_{ai}S_{bj}B_{ab},\\
	&C_{ijkl}\mapsto S^{\ast}_{ai}S_{bj}S^{\ast}_{ck}S_{dl}C_{abcd}.
	\end{split}
	\end{equation}
\end{prop}

\begin{proof}
As in the proof of proposition \ref{Plepton1} the calculation is identical for all three statements, thus we will only investigate the transformation property of $B_{ij}$. $G$-invariance is given if
	\begin{displaymath}
	[D_{\phi}]^{\ast}_{ia}[D_{\phi}]_{jb}B_{ij}=B_{ab}.
	\end{displaymath}
In the new basis this condition becomes
	\begin{displaymath}
	\begin{split}
	& (S^{-1})_{ip}^{\ast}[D_{\phi}]^{\ast}_{pq}S^{\ast}_{qa}(S^{-1})_{jr}[D_{\phi}]_{rs}S_{sb}B_{ij}'=B_{ab}'\\
	&
	[D_{\phi}]^{\ast}_{pq}\{(S^{-1})_{ip}^{\ast}B_{ij}'(S^{-1})_{jr}\}[D_{\phi}]_{rs}S^{\ast}_{qa}S_{sb}=B_{ab}'/\cdot (S^{-1})_{at}^{\ast}(S^{-1})_{bu}\\
	&
	[D_{\phi}]^{\ast}_{pt}\{\underbrace{(S^{-1})_{ip}^{\ast}B_{ij}'(S^{-1})_{jr}}_{B_{pr}}\}[D_{\phi}]_{ru}=\underbrace{(S^{-1})_{at}^{\ast}B_{ab}'(S^{-1})_{bu}}_{B_{tu}}.
	\end{split}
	\end{displaymath}
Thus
	\begin{displaymath}
	(S^{-1})_{at}^{\ast}B_{ab}'(S^{-1})_{bu}=B_{tu}\Rightarrow B_{ij}'=S^{\ast}_{ai}S_{bj}B_{ab}.
	\end{displaymath}
\end{proof}

\begin{lemma}\label{Llepton3}
Under the basis transformation
	\begin{displaymath}
	[D_{\phi}]\mapsto S^{-1}[D_{\phi}]S
	\end{displaymath}
the Higgs potential $V(\phi)$ transforms as
	\begin{displaymath}
	V(\phi)\mapsto V'(\phi)=V(S\phi).
	\end{displaymath}
\end{lemma}

\begin{proof}
Again the calculation is identical for all parts of the Higgs potential, thus we will only investigate one part of it. From proposition \ref{Plepton2} we know that
	\begin{displaymath}
	B_{ij}(\phi_{i}^{\dagger}\phi_{j})\mapsto S^{\ast}_{ai}S_{bj}B_{ab}(\phi_{i}^{\dagger}\phi_{j})=B_{ab}(S\phi)_{a}^{\dagger}(S\phi)_{b}.
	\end{displaymath}
\end{proof}

\section{$G$-invariant Lagrangians}\label{GinvLagr}
We now want to consider a general $G$-invariant Lagrangian and investigate its properties.
\medskip
\\
Let $a,b$ transform as
	\begin{displaymath}
	a\mapsto D_{a}a,\quad b\mapsto D_{b}b,
	\end{displaymath}
then from corollary \ref{CY10} we know that $D_{\phi}$ is determined by
	\begin{displaymath}
	(D_{a}^{-1})^{\dagger}\otimes (D_{b}^{-1})^{T}=\bigoplus_{\sigma}D^{\sigma},
	\end{displaymath}
in the way that $D_{\phi}$ can take all values $D^{\sigma}$ occurring in $(D_{a}^{-1})^{\dagger}\otimes (D_{b}^{-1})^{T}$. Of course one can combine some of the $D^{\sigma}$ to a reducible representation
	\begin{displaymath}
	D_{\phi}=D^1\oplus...\oplus D^m,
	\end{displaymath}
where $D^1,...,D^m$ are irreducible representations occurring in $D_a\otimes D_b$. This corresponds to including $m$ independent $G$-invariant Yukawa-couplings in the Lagrangian. The most general $G$-invariant Yukawa-coupling is then given by
	\begin{equation}\label{GinvYukcoupl}
	\mathcal{L}_{\mathrm{Yukawa}}=-\sum_{\sigma=1}^m g_{\sigma}\Gamma_{ijk}^{\sigma}\bar{a}_{j}\phi_{i}^{(\sigma)}b_{k}+\mathrm{H.c.},
	\end{equation}
where $g_{\sigma}$ are arbitrary coupling constants, $\Gamma^{\sigma}_{ijk}$ are the Clebsch-Gordan coefficients for $(D_{a}^{-1})^{\dagger}\otimes (D_{b}^{-1})^{T}\rightarrow D^{\sigma}$ and $\phi^{(\sigma)}$ transforms as
	\begin{displaymath}
	\phi^{(\sigma)}\mapsto D^{\sigma}\phi^{(\sigma)}
	\end{displaymath}
under $G$.
\\
In addition one will get a Higgs Lagrangian of the form:
	\begin{equation}\label{LHiggs-generalized}
	\mathcal{L}_{\mathrm{Higgs}}=\sum_{\sigma=1}^m (A_{ij}^{\sigma}(D_{\mu}\phi_{i}^{(\sigma)})^{\dagger}(D^{\mu}\phi_{j}^{(\sigma)}))-V(\phi^{(1)},...,\phi^{(m)}),
	\end{equation}
where
	\begin{equation}\label{generalized-Higgs-potential}
	\begin{split}
	V(\phi^{(1)},...,\phi^{(m)}) & =\sum_{\rho,\sigma=1}^{m}\mu_{\rho\sigma}^2 B_{ij}^{\rho\sigma}(\phi_{i}^{(\rho)\dagger}\phi_{j}^{(\sigma)})+\\
	& +\sum_{\rho,\sigma,\tau,\omega=1}^{m}\lambda_{\rho\sigma\tau\omega}C^{\rho\sigma\tau\omega}_{ijkl}(\phi_{i}^{(\rho)\dagger}\phi_{j}^{(\sigma)})(\phi_{k}^{(\tau)\dagger}\phi_{l}^{(\omega)}).
	\end{split}
	\end{equation}
Please do not be confused by the nomenclature $D_{\mu}$ for the covariant derivative and $D^{\sigma}$ for a representation of $G$. In the Lagrangian only the covariant derivative occurs.
\medskip
\\
Remark: In this chapter we don't use Einstein's summation convention on Greek indices, except it are Lorentz indices.
\medskip
\\
Let us investigate the properties of the coefficients $A^{\sigma}_{ij}$, $B^{\rho\sigma}_{ij}$ and $C^{\rho\sigma\tau\omega}_{ijkl}$. The claim for $G$-invariance of the Higgs Lagrangian (\ref{LHiggs-generalized}) leads to the conditions
	\begin{displaymath}
	\begin{split}
	& [D^{\sigma}]_{ia}^{\ast}[D^{\sigma}]_{jb}A^{\sigma}_{ij}=A_{ab}^{\sigma},\\
	& [D^{\rho}]_{ia}^{\ast}[D^{\sigma}]_{jb}B^{\rho\sigma}_{ij}=B_{ab}^{\rho\sigma},\\
	&
	[D^{\rho}]_{ia}^{\ast}[D^{\sigma}]_{jb}[D^{\tau}]_{kc}^{\ast}[D^{\omega}]_{ld}C^{\rho\sigma\tau\omega}_{ijkl}=C_{abcd}^{\rho\sigma\tau\omega}.
	\end{split}
	\end{displaymath}
(The proof works completely analogous to the proof of proposition \ref{Plepton1}.) The transformation properties under basis transformations
	\begin{displaymath}
	[D^{\sigma}]\mapsto S_{\sigma}^{-1}[D^{\sigma}]S_{\sigma}\quad\quad(\sigma=1,...,m)
	\end{displaymath}
are completely analogous to proposition \ref{Plepton2}:
	\begin{displaymath}
	\begin{split}
	& A^{\sigma}_{ij}\mapsto (S_{\sigma})^{\ast}_{ai}(S_{\sigma})_{bj}A^{\sigma}_{ab},\\
	& B_{ij}^{\rho\sigma}\mapsto (S_{\rho})^{\ast}_{ai}(S_{\sigma})_{bj}B_{ab}^{\rho\sigma},\\
	& C_{ijkl}^{\rho\sigma\tau\omega}\mapsto (S_{\rho})^{\ast}_{ai}(S_{\sigma})_{bj}(S_{\tau})^{\ast}_{ck}(S_{\omega})_{dl}C_{abcd}^{\rho\sigma\tau\omega}.
	\end{split}
	\end{displaymath}
(Again the proof works completely analogous to the proof of proposition \ref{Plepton2}.)

\begin{lemma}\label{LLepton2}
Under a basis transformation
	\begin{displaymath}
	[D^{\sigma}]\mapsto S_{\sigma}^{-1}[D^{\sigma}]S_{\sigma}\quad\quad(\sigma=1,...,m)
	\end{displaymath}
the $G$-invariant Higgs potential (\ref{generalized-Higgs-potential}) transforms as
	\begin{displaymath}
	V(\phi^{(1)},...,\phi^{(m)})\mapsto V'(\phi^{(1)},...,\phi^{(m)})=V(S_1 \phi^{(1)},...,S_m \phi^{(m)}).
	\end{displaymath}
\end{lemma}

\begin{proof}
The proof works in the same way as the proof of lemma \ref{Llepton3}.
\end{proof}

\begin{prop}\label{PLepton4}
Let
	\begin{displaymath}
	\mathcal{M}:=\left(\sum_{\sigma=1}^{m}\frac{g_{\sigma}v_{j}^{(\sigma)}}{\sqrt{2}}\Gamma^{\sigma}_{j}\right)^{\dagger}
	\end{displaymath}
be the mass matrix corresponding to the Yukawa-coupling (\ref{GinvYukcoupl}), where we have defined $(\Gamma_{i}^{\sigma})_{jk}:=\Gamma^{\sigma}_{ijk}$ and $v_{j}^{(\sigma)}$ is defined via
	\begin{displaymath}
	\langle 0\vert\phi_{j}^{(\sigma)}(x)\vert0\rangle=\frac{1}{\sqrt{2}}\left(
	\begin{matrix}
	0\\
	v_{j}^{(\sigma)}
	\end{matrix}
	\right).
	\end{displaymath}
Then under the basis transformation
	\begin{displaymath}
	\begin{split}
	& [D_{a}]\mapsto S_{a}^{-1}[D_{a}]S_{a},\\
	& [D_{b}]\mapsto S_{b}^{-1}[D_{b}]S_{b},\\
	& [D^{\sigma}]\mapsto S_{\sigma}^{-1}[D^{\sigma}]S_{\sigma}\quad\quad(\sigma=1,...,m)
	\end{split}
	\end{displaymath}
$\mathcal{M}$ transforms as
	\begin{displaymath}
	\mathcal{M}\mapsto S_{b}^{\dagger}\mathcal{M}S_{a}.
	\end{displaymath}
\end{prop}

\begin{proof}
$\Gamma_{ijk}^{\sigma}$ are Clebsch-Gordan coefficients for
	\begin{displaymath}
	(D_{a}^{-1})^{\dagger}\otimes (D_{b}^{-1})^T=D^{\sigma}\oplus...\enspace.
	\end{displaymath}
$[D^{-1}_{a}]^{\dagger}$ and $[D_{b}^{-1}]^T$ transform as
	\begin{displaymath}
	\begin{split}
	& [D_{a}^{-1}]^{\dagger}\mapsto S_{a}^{\dagger}[D_{a}^{-1}]^{\dagger}(S_{a}^{-1})^{\dagger},\\
	& [D_{b}^{-1}]^{T}\mapsto S_{b}^{T}[D_{b}^{-1}]^{T}(S_{b}^{-1})^{T}
	\end{split}
	\end{displaymath}
and thus, using proposition \ref{Pclebsch1}, $\Gamma_{i}^{\sigma}$ transforms as
	\begin{displaymath}
	\Gamma_{i}^{\sigma}\mapsto S_{a}^{\dagger}\Gamma_{m}^{\sigma}S_{b}(S_{\sigma})_{mi}.
	\end{displaymath}
We also have to take into account that the Higgs potential changes under
	\begin{displaymath}
	[D^{\sigma}]\mapsto S_{\sigma}^{-1}[D^{\sigma}] S_{\sigma},
	\end{displaymath}
and thus $v_{j}^{(\sigma)}$ are not invariant under basis change. $\langle0\vert\phi^{(\sigma)}\vert0\rangle$ are defined by the condition:
	\begin{quote}
	$V(\phi^{(1)},...,\phi^{(m)})$ \textit{has a minimum at} $(\langle0\vert\phi^{(1)}\vert0\rangle,...,\langle0\vert\phi^{(m)}\vert0\rangle)$\textit{, i.s.}
		\begin{displaymath}
		V(\langle0\vert\phi^{(1)}\vert0\rangle,...,\langle0\vert\phi^{(m)}\vert0\rangle)=\mathrm{min}.
		\end{displaymath}
	\end{quote}	
In the new basis this equation becomes ($\rightarrow$ lemma \ref{LLepton2})
	\begin{displaymath}
	V(S_{1}\langle0\vert\phi^{(1)}\vert0\rangle',...,S_m\langle0\vert\phi^{(m)}\vert0\rangle')=\mathrm{min}.
	\end{displaymath}
From this we find
	\begin{displaymath}
	\langle0\vert\phi^{(\sigma)}\vert0\rangle\mapsto\langle0\vert\phi^{(\sigma)}\vert0\rangle'= S_{\sigma}^{-1}\langle0\vert\phi^{(\sigma)}\vert0\rangle\Rightarrow v_{j}^{(\sigma)}\mapsto (S_{\sigma}^{-1})_{jk}v_{k}^{(\sigma)}.
	\end{displaymath}
The transformation property of $\mathcal{M}$ is therefore given by
	\begin{displaymath}
	\mathcal{M}^{\dagger}\mapsto \sum_{\sigma=1}^{m}\frac{g_{\sigma}(S_{\sigma}^{-1})_{jk}v_{k}^{(\sigma)}}{\sqrt{2}}S_{a}^{\dagger}\Gamma_{m}^{\sigma}S_{b}(S_{\sigma})_{mj}=S_{a}^{\dagger}\mathcal{M}^{\dagger}S_{b}\Rightarrow \mathcal{M}\mapsto S_{b}^{\dagger}\mathcal{M}S_{a}.
	\end{displaymath}
\end{proof}
\hspace{0mm}\\
Before we turn to the lepton sector, let us further investigate the coefficients $A_{ij}^{\sigma}$ and $B_{ij}^{\rho\sigma}$ of the Higgs Lagrangian (\ref{LHiggs-generalized}).

\begin{prop}\label{Plepton5}
Let $\mathcal{L}_{\mathrm{Higgs}}$ be a $G$-invariant Lagrangian of the form (\ref{LHiggs-generalized}), then
	\begin{displaymath}
	\begin{split}
	& A_{ij}^{\sigma}=A^{\sigma}\delta_{ij},\\
	& B^{\rho\sigma}_{ij}=\beta^{\rho\sigma}\delta_{ij}\quad\mbox{if }D^{\rho}\sim D^{\sigma},\\
	& B^{\rho\sigma}_{ij}=0\quad\mbox{if }D^{\rho}\not\sim D^{\sigma}.
	\end{split}
	\end{displaymath}
($\beta^{\rho\sigma}$ is a complex number.)
\end{prop}

\begin{proof}
We will only investigate the coefficients $B^{\rho\sigma}_{ij}$, the analysis of $A_{ij}^{\sigma}$ is included in the analysis of $B^{\rho\sigma}_{ij}$ as the special case $\rho=\sigma$.
\medskip
\\
$G$-invariance of $\mathcal{L}_{\mathrm{Higgs}}$ requires
	\begin{displaymath}
	[D^{\rho}]_{ia}^{\ast}[D^{\sigma}]_{jb}B^{\rho\sigma}_{ij}=B_{ab}^{\rho\sigma},
	\end{displaymath}
which is in matrix form
	\begin{displaymath}
	[D^{\rho}]^{\dagger}B^{\rho\sigma}[D^{\sigma}]=B^{\rho\sigma}.
	\end{displaymath}
This can be rewritten as
	\begin{displaymath}
	([D^{\rho}]^{-1})^{\dagger}B^{\rho\sigma}=B^{\rho\sigma}[D^{\sigma}].
	\end{displaymath}
Since $D^{\rho},D^{\sigma}$ are representations of a finite group, they are equivalent to unitary representations, thus $((D^{\rho})^{-1})^{\dagger}\sim D^{\rho}$, and we find
	\begin{displaymath}
	[D^{\rho}]B^{\rho\sigma}=B^{\rho\sigma}[D^{\sigma}].
	\end{displaymath}
Using Schur's lemma ($\rightarrow$ lemma \ref{LA50}) and the fact that $D^{\rho}$ and $D^{\sigma}$ are irreducible, we notice that there are the following two possibilities:
	\begin{enumerate}
	 \item $D^{\rho}\not\sim D^{\sigma}\Rightarrow B^{\rho\sigma}=0.$
	 \item $D^{\rho}\sim D^{\sigma}\Rightarrow$ $B^{\rho\sigma}=\beta^{\rho\sigma}\mathbbm{1}$, with a complex number $\beta^{\rho\sigma}$.
	\end{enumerate}
\end{proof}
\hspace{0mm}
\\
We now want to concentrate on the lepton sector. From proposition \ref{Plepton5} we know that $A^{\sigma}_{ij}=A^{\sigma}\delta_{ij}$. To keep analogy to the Higgs Lagrangian of the standard model we set $A^{\sigma}=1\enspace\forall \sigma\in\{1,...,m\}$.
\bigskip
\bigskip
\\
The total $G$-invariant Lagrangian of our model becomes
	\begin{equation}\label{totalL}
	\begin{split}
	\mathcal{L} & =-\frac{1}{2}\mathrm{Tr}(W_{\lambda\rho}W^{\lambda\rho})-\frac{1}{4}\tilde{B}_{\lambda\rho}\tilde{B}^{\lambda\rho}+\\
	&\phantom{=}+\sum_{\tau=1}^{m}(D_{\mu}\phi_{i}^{(\tau)})^{\dagger}(D^{\mu}\phi_{i}^{(\tau)})-V(\phi^{(1)},...,\phi^{(m)})+\\
	&\phantom{=}+\bar{l}'i\gamma^{\mu}D_{\mu}l'+\bar{\nu}i\gamma^{\mu}D_{\mu}\nu+\\
	&\phantom{=}-(\sum_{\sigma}g_{\sigma}\Gamma_{ijk}^{(cl)\sigma}\bar{D'}_{jL}\phi_{i}^{(\sigma)}l'_{kR}+\mathrm{H.c.})+\\
	&\phantom{=}-(\sum_{\rho}g_{\rho}\Gamma_{ijk}^{(\nu)\rho}\bar{D'}_{jL}\tilde{\phi_{i}}^{(\rho)}\nu_{kR}+\mathrm{H.c.}),
	\end{split}
	\end{equation}
where
	\begin{itemize}
	 \item[] $\tilde{B}_{\lambda\rho}$ is the field strength tensor of the $U(1)_{Y}$-gauge field,
	 \item[] $\mathnormal{l}_{L,R}'=\left(
		\begin{matrix}
		 e_{L,R}' \\ \mu_{L,R}' \\ \tau_{L,R}'
		\end{matrix}
		\right)$, $\nu_{L,R}=\left(
		\begin{matrix}
		 \nu_{e L,R} \\ \nu_{\mu L,R} \\ \nu_{\tau L,R}
		\end{matrix}
		\right)$, $D'_{L}=\left(
		\begin{matrix}
		 D_{e L}' \\ D_{\mu L}' \\ D_{\tau L}'
		\end{matrix}
		\right)$, $D'_{\alpha L}=\left(
		\begin{matrix}
		\nu_{\alpha L} \\ \alpha'_{L}
		\end{matrix}
		\right)$,
	\item[] $\nu=\nu_L+\nu_R$,\quad $l'=l'_L+l'_R$,\quad $\tilde{\phi}=i\tau^{2}\phi^{\ast}$, $\tau^2=\left(
	\begin{matrix}
		0 & -i \\ i & 0
		\end{matrix}
	\right)$.
	\end{itemize}
$m$ is the total number of Higgs doublets of the model. $\rho$ and $\sigma$ label the Higgs doublets that interact with the right-handed neutrinofields and the right-handed charged lepton fields, respectively. The index $\tau$ runs over all Higgs doublets.
\\
As always, the representations under which the $\phi^{(\tau)}$ transform are determined by the claim for $G$-invariance of the two Yukawa-couplings
	\begin{displaymath}
	\begin{split}
	& \mathcal{L}_{\mathrm{Yukawa}}^{(cl)}=-(\sum_{\sigma}g_{\sigma}\Gamma_{ijk}^{(cl)\sigma}\bar{D'}_{jL}\phi_{i}^{(\sigma)}l'_{kR}+\mathrm{H.c.}),\\ & \mathcal{L}_{\mathrm{Yukawa}}^{(\nu)}=-(\sum_{\rho}g_{\rho}\Gamma_{ijk}^{(\nu)\rho}\bar{D'}_{jL}\tilde{\phi_{i}}^{(\rho)}\nu_{kR}+\mathrm{H.c.}).
	\end{split}
	\end{displaymath}
Remark: Due to the extended Higgs sector the formulae for the vector boson masses are different to those of the standard model.
\medskip
\\
We will now investigate $G$-invariance under
	\begin{equation}\label{Gtrafo}
	\begin{split}
	& l_{L,R}'\mapsto [D_{L,R}^{cl}]l_{L,R}',\\
	& \nu_{L,R}\mapsto [D_{L,R}^{\nu}]\nu_{L,R},\\
	& \phi^{(\sigma)}\mapsto [D^{\sigma}]\phi^{(\sigma)},\\
	& \tilde{\phi}^{(\rho)}\mapsto [D^{\rho}]\tilde{\phi}^{(\rho)}.
	\end{split}
	\end{equation}
Since $\nu_{\alpha L}$ and $\alpha'_{L}$ are members of an $SU(2)_{I}$-doublet they must transform equally, thus
	\begin{displaymath}
	[D^{cl}_{L}]=[D^{\nu}_{L}]\quad (\Leftrightarrow D^{cl}_{L}=D^{\nu}_{L}).
	\end{displaymath}
In order to ensure that $\bar{l}'i\gamma^{\mu}D_{\mu}l'+\bar{\nu}i\gamma^{\mu}D_{\mu}\nu$ and $\sum_{\tau=1}^{m}(D_{\mu}\phi_{i}^{(\tau)})^{\dagger}(D^{\mu}\phi_{i}^{(\tau)})$ are invariant under the transformation (\ref{Gtrafo}) we have to assume that $[D_{L,R}^{cl}]$, $[D_{L,R}^{\nu}]$ and $[D^{\tau}]$ \textit{are unitary}. $G$-invariance of the Yukawa-couplings can be achieved by appropriate choice of $\Gamma^{(cl)\sigma}_{ijk}$ and $\Gamma^{(\nu)\rho}_{ijk}$:
	\begin{itemize}
	 \item $\Gamma^{(cl)\sigma}_{ijk}$ are Clebsch-Gordan coefficients for
		\begin{displaymath}
		((D^{cl}_{L})^{-1})^{\dagger}\otimes ((D^{cl}_{R})^{-1})^T=D^{\sigma}\oplus...
		\end{displaymath}
	 \item $\Gamma^{(\nu)\rho}_{ijk}$ are Clebsch-Gordan coefficients for
		\begin{displaymath}
		((D^{cl}_{L})^{-1})^{\dagger}\otimes ((D^{\nu}_{R})^{-1})^T=D^{\rho}\oplus...
		\end{displaymath}
	\end{itemize}
We already found that only unitary matrix representations $[D_{L,R}^{cl}]$ and $[D_{L,R}^{\nu}]$ are allowed, thus we can also write
	\begin{displaymath}
	D^{cl}_{L}\otimes D^{cl\ast}_{R}=D^{\sigma}\oplus...,\quad\quad D^{cl}_{L}\otimes D^{\nu\ast}_{R}=D^{\rho}\oplus...\enspace.
	\end{displaymath}

\section{Lepton masses and mixing}
We now know how to construct the fermionic parts of $G$-invariant Lagrangians of the type (\ref{totalL}). The main objects we are interested in are the mass matrices. From the Lagrangian (\ref{totalL}) we find (using the mass matrix definition of subsection \ref{GIMsubsection})
	\begin{displaymath}
	(\mathcal{M}^{cl})^{\dagger}=\sum_{\sigma}\frac{g_{\sigma}v_{i}^{(\sigma)}}{\sqrt{2}}\Gamma^{(cl)\sigma}_{i},\quad\quad (\mathcal{M}^{\nu})^{\dagger}=\sum_{\rho}\frac{g_{\rho}v_{i}^{(\rho)}}{\sqrt{2}}\Gamma^{(\nu)\rho}_{i},
	\end{displaymath}
where $v_{i}^{(\sigma)}$ and $v_{i}^{(\rho)}$ are defined via the vacuum expectation values
	\begin{displaymath}
	\begin{split}
	& \langle 0\vert\phi_{j}^{(\sigma)}(x)\vert0\rangle=\frac{1}{\sqrt{2}}\left(
	\begin{matrix}
	0\\
	v_{j}^{(\sigma)}
	\end{matrix}
	\right),\\
	& \langle 0\vert\tilde{\phi}_{j}^{(\rho)}(x)\vert0\rangle=\frac{1}{\sqrt{2}}\left(
	\begin{matrix}
	v_{j}^{(\rho)}\\
	0
	\end{matrix}
	\right).
	\end{split}
	\end{displaymath}
We already know that the mass matrices are basis dependent, their transformation properties under basis change were developed in proposition \ref{PLepton4}. Any observable physical quantity must not be basis dependent, and in fact it will turn out that the lepton masses and the mixing matrix are indeed not dependent on the basis.

\begin{theorem}\label{Tlepton5}
The lepton masses and the lepton mixing matrix corresponding to the $G$-invariant Lagrangian (\ref{totalL}) do not depend on the choice of basis.
\end{theorem}

\begin{proof}
The lepton masses are obtained by bidiagonalization of the mass matrices. Looking at the proof of theorem \ref{TY1} we find that the lepton masses are the positive square roots
of the eigenvalues of $\mathcal{M}^{cl}\mathcal{M}^{cl}\hspace{0mm}^{\dagger}$ and $\mathcal{M}^{\nu}\mathcal{M}^{\nu}\hspace{0mm}^{\dagger}$. Under basis change
	\begin{displaymath}
	\begin{split}
	& [D_{L,R}^{cl}]\mapsto (S^{cl}_{L,R})^{-1}[D_{L,R}^{cl}]S^{cl}_{L,R},\\
	& [D_{L,R}^{\nu}]\mapsto (S^{\nu}_{L,R})^{-1}[D_{L,R}^{\nu}]S^{\nu}_{L,R},\\
	& [D^{\sigma}]\mapsto S_{\sigma}^{-1}[D^{\sigma}]S_{\sigma},\\
	& [D^{\rho}]\mapsto S_{\rho}^{-1}[D^{\rho}]S_{\rho}\\
	\end{split}
	\end{displaymath}
(from the unitarity of $[D_{L,R}^{cl}]$ and $[D_{L,R}^{\nu}]$ it follows that $S^{cl}_{L,R}$ and $S^{\nu}_{L,R}$ must be unitary too) $\mathcal{M}^{cl}\mathcal{M}^{cl}\hspace{0mm}^{\dagger}$ and $\mathcal{M}^{\nu}\mathcal{M}^{\nu}\hspace{0mm}^{\dagger}$ transform as ($\rightarrow$ proposition \ref{PLepton4}):
	\begin{displaymath}
	\mathcal{M}^{cl}\mathcal{M}^{cl}\hspace{0mm}^{\dagger}\mapsto (S_{R}^{cl})^{\dagger}\mathcal{M}^{cl}\underbrace{S^{cl}_{L}(S^{cl}_{L})^{\dagger}}_{\mathbbm{1}_{3}}\mathcal{M}^{cl}\hspace{0mm}^{\dagger}S_{R}^{cl}=(S_{R}^{cl})^{\dagger}\mathcal{M}^{cl}\mathcal{M}^{cl}\hspace{0mm}^{\dagger}S_{R}^{cl}
	\end{displaymath}
and in a similar manner
	\begin{displaymath}
	\mathcal{M}^{\nu}\mathcal{M}^{\nu}\hspace{0mm}^{\dagger}\mapsto (S_{R}^{\nu})^{\dagger}\mathcal{M}^{\nu}\underbrace{S^{cl}_{L}(S^{cl}_{L})^{\dagger}}_{\mathbbm{1}_{3}}\mathcal{M}^{\nu}\hspace{0mm}^{\dagger}S_{R}^{\nu}=(S_{R}^{\nu})^{\dagger}\mathcal{M}^{\nu}\mathcal{M}^{\nu}\hspace{0mm}^{\dagger}S_{R}^{\nu}.
	\end{displaymath}
The unitary transformations $S_{R}^{cl}$ and $S_{R}^{\nu}$ do not change the eigenvalues of $\mathcal{M}^{cl}\mathcal{M}^{cl}\hspace{0mm}^{\dagger}$ and $\mathcal{M}^{\nu}\mathcal{M}^{\nu}\hspace{0mm}^{\dagger}$, thus the lepton masses are basis independent.
\\
Now for the mixing matrix. Bidiagonalization of $\mathcal{M}^{cl}$ and $\mathcal{M}^{\nu}$ leads to
	\begin{displaymath}
	\mathcal{M}^{cl}=V_{cl}\hat{\mathcal{M}}^{cl}U_{cl}^{\dagger},\quad \mathcal{M}^{\nu}=V_{\nu}\hat{\mathcal{M}}^{\nu}U_{\nu}^{\dagger}
	\end{displaymath}
In the new basis these equations become
	\begin{displaymath}
	\begin{split}
	& (\mathcal{M}^{cl})'=(S_{R}^{cl})^{\dagger}\mathcal{M}^{cl}S^{cl}_{L}=\underbrace{(S_{R}^{cl})^{\dagger}V_{cl}}_{V_{cl}'}\hat{\mathcal{M}}^{cl}\underbrace{U_{cl}^{\dagger}S^{cl}_{L}}_{U_{cl}'\hspace{0mm}^{\dagger}},\\
	&
	(\mathcal{M}^{\nu})'=(S_{R}^{\nu})^{\dagger}\mathcal{M}^{\nu}S^{cl}_{L}=\underbrace{(S_{R}^{\nu})^{\dagger}V_{\nu}}_{V_{\nu}'}\hat{\mathcal{M}}^{\nu}\underbrace{U_{\nu}^{\dagger}S^{cl}_{L}}_{U_{\nu}'\hspace{0mm}^{\dagger}}.
	\end{split}
	\end{displaymath}
Thus under basis transformation
	\begin{displaymath}
	U_{cl}\mapsto (S^{cl}_{L})^{\dagger}U_{cl},\quad U_{\nu}\mapsto (S^{cl}_{L})^{\dagger}U_{\nu},
	\end{displaymath}
and the lepton mixing matrix
	\begin{displaymath}
	U:=U_{cl}^{\dagger}U_{\nu}
	\end{displaymath}
remains invariant under basis transformation.
\end{proof}
\hspace{0mm}\\
Theorem \ref{Tlepton5} justifies our treatment of the Clebsch-Gordan coefficients in chapter \ref{SU3chapter} - we only need them up to basis transformations.

\section{Model building with finite family symmetry groups}
A model realizing finite family symmetry groups in the lepton sector can be constructed in the following manner:
	\begin{enumerate}
	 \item Choose a finite group $G$ and three-dimensional irreducible representations $D^{(cl)}_{L,R}$ and $D^{(\nu)}_{R}$ ($D^{(\nu)}_{L}=D^{(cl)}_{L}$). (The restriction on irreducible representations is not necessary in general, we just restrict on 3-dimensional irreducible representations, because these are the representations we had studied in this thesis.)
	\item Construct a $G$-invariant Lagrangian of the form (\ref{totalL}) by appropriate choice of Clebsch-Gordan coefficients as described in section \ref{GinvLagr}. The construction of the coefficients for the Higgs Lagrangian will need a similar study as the study of $G$-invariant Yukawa-couplings in chapter \ref{yukawa}. If one uses the vacuum expectation values as free parameters, one does not need to consider the $G$-invariant Higgs Lagrangian at all.
	\item Construct the mass matrices $\mathcal{M}^{cl}$ and $\mathcal{M}^{\nu}$. They will contain the following free parameters (if the Higgs Lagrangian is ignored as described in 2.):
		\begin{displaymath}
		g_{\sigma}, v_{j}^{(\sigma)}, g_{\rho}, v_{j}^{(\rho)}.
		\end{displaymath}
	\item One could now try to fit the free parameters in order to obtain the desired lepton masses and mixing angles. Due to the high number of parameters this may need usage of a computer. Also the masses and mixing angles as functions of the free parameters would be quite interesting, but bidiagonalization is unlikely to work analytically in the presence of many free parameters. However, there are at least four functions of the free parameters that can be easily calculated by hand, namely
		\begin{displaymath}
		\begin{split}
		& \mathrm{Tr}(\mathcal{M}^{cl}\mathcal{M}^{cl}\hspace{0mm}^{\dagger})=m_{e}^{2}+m_{\mu}^{2}+m_{\tau}^{2},\\
		& \mathrm{Tr}(\mathcal{M}^{\nu}\mathcal{M}^{\nu}\hspace{0mm}^{\dagger})=m_{\nu_{1}}^{2}+m_{\nu_{2}}^{2}+m_{\nu_{3}}^{2},\\
		& \mathrm{det}(\mathcal{M}^{cl}\mathcal{M}^{cl}\hspace{0mm}^{\dagger})=(m_{e}m_{\mu}m_{\tau})^{2},\\
		& \mathrm{det}(\mathcal{M}^{\nu}\mathcal{M}^{\nu}\hspace{0mm}^{\dagger})=(m_{\nu_{1}}m_{\nu_{2}}m_{\nu_{3}})^{2}.
		\end{split}
		\end{displaymath}
	\end{enumerate}

\section{Example for a model involving a finite family symmetry group}
Let us at last give a short example for a model involving a finite family symmetry group. Let us choose the group $\Sigma(36\phi)$ which has been investigated in subsection \ref{subsectionS36phi}. $\Sigma(36\phi)$ has 8 non-equivalent three-dimensional irreducible representations, which can all be constructed from $\textbf{\underline{3}}_{1}$ and $\textbf{\underline{3}}_{1}^{\ast}$ by multiplication with one-dimensional representations. Let for example
	\begin{displaymath}
	D_{L}^{cl}=\textbf{\underline{3}}_{1},
	D_{R}^{cl}=D_{R}^{\nu}=\textbf{\underline{3}}_{1}^{\ast}.
	\end{displaymath}
For the Higgs doublets we have to consider the tensor products
	\begin{displaymath}
	\begin{split}
	& ((D^{cl}_{L})^{-1})^{\dagger}\otimes ((D^{cl}_{R})^{-1})^T=\textbf{\underline{3}}_{1}\otimes\textbf{\underline{3}}_{1}=\textbf{\underline{3}}_{4}\oplus \textbf{\underline{3}}_{5}\oplus\textbf{\underline{3}}_{6},\\
	& ((D^{cl}_{L})^{-1})^{\dagger}\otimes ((D^{\nu}_{R})^{-1})^T=\textbf{\underline{3}}_{1}\otimes\textbf{\underline{3}}_{1}=\textbf{\underline{3}}_{4}\oplus \textbf{\underline{3}}_{5}\oplus\textbf{\underline{3}}_{6}.
	\end{split}
	\end{displaymath}
Now we can for example consider a 6-Higgs doublet model based on $\textbf{\underline{3}}_{5}$ and $\textbf{\underline{3}}_{6}$:
	\begin{displaymath}
	\phi^{(\sigma)}\mapsto [\textbf{\underline{3}}_{5}]\phi^{(\sigma)},\quad \tilde{\phi}^{(\rho)}\mapsto [\textbf{\underline{3}}_{6}]\tilde{\phi}^{(\rho)}.
	\end{displaymath}
Because we only consider one irreducible representation for the Higgs doublets there will be no summation over $\sigma$ and $\rho$ in the Lagrangian. We can now construct the mass matrices from the Clebsch-Gordan coefficients for the decompositions
	\begin{displaymath}
	\begin{split}
	& ((D^{cl}_{L})^{-1})^{\dagger}\otimes ((D^{cl}_{R})^{-1})^T=\textbf{\underline{3}}_{1}\otimes\textbf{\underline{3}}_{1}= \textbf{\underline{3}}_{5}\oplus...\Rightarrow \Gamma^{(cl)}_{i},\\
	& ((D^{cl}_{L})^{-1})^{\dagger}\otimes ((D^{\nu}_{R})^{-1})^T=\textbf{\underline{3}}_{1}\otimes\textbf{\underline{3}}_{1}=\textbf{\underline{3}}_{6}\oplus...\Rightarrow \Gamma^{(\nu)}_{j}.
	\end{split}
	\end{displaymath}
From tables \ref{Sigma36CGCa} and \ref{SU3CGClist} we can read off
	\begin{displaymath}
	\Gamma_{1}^{(cl)}=
	\left(
	\begin{matrix}
	 -\frac{\tau_{-}}{\sqrt{12}} & 0 & 0 \\
	 0 & 0 & \frac{1}{\tau_{-}} \\
	 0 & \frac{1}{\tau_{-}} & 0
	\end{matrix}
	\right),\enspace
	\Gamma_{2}^{(cl)}=
	\left(
	\begin{matrix}
	 0 & 0 & \frac{1}{\tau_{-}} \\
	 0 & -\frac{\tau_{-}}{\sqrt{12}} & 0 \\
	 \frac{1}{\tau_{-}} & 0 & 0
	\end{matrix}
	\right),\enspace
	\Gamma_{3}^{(cl)}=
	\left(
	\begin{matrix}
 	0 & \frac{1}{\tau_{-}} & 0 \\
	 \frac{1}{\tau_{-}} & 0 & 0 \\
	 0 & 0 & -\frac{\tau_{-}}{\sqrt{12}}
	\end{matrix}
	\right),
	\end{displaymath}
	\begin{displaymath}
	\Gamma_{1}^{(\nu)}=
	\left(
	\begin{matrix}
 	\frac{\tau_{+}}{\sqrt{12}} & 0 & 0 \\
 	0 & 0 & \frac{1}{\tau_{+}} \\
	 0 & \frac{1}{\tau_{+}} & 0
	\end{matrix}
	\right),\enspace\Gamma_{2}^{(\nu)}=
	\left(
	\begin{matrix}
	 0 & 0 & \frac{1}{\tau_{+}} \\
	 0 & \frac{\tau_{+}}{\sqrt{12}} & 0 \\
	 \frac{1}{\tau_{+}} & 0 & 0
	\end{matrix}
	\right),\enspace
	\Gamma_{3}^{(\nu)}=
	\left(
	\begin{matrix}
	 0 & \frac{1}{\tau_{+}} & 0 \\
	 \frac{1}{\tau_{+}} & 0 & 0 \\
	 0 & 0 & \frac{\tau_{+}}{\sqrt{12}}
	\end{matrix}
	\right)
	\end{displaymath}
where $\tau_{\pm}=\sqrt{2(3\pm \sqrt{3})}$. Using these Clebsch-Gordan coefficients we find the mass matrices
	\begin{displaymath}
	\mathcal{M}^{cl}=\frac{g_{\sigma}^{\ast}}{\sqrt{2}}\left(
	\begin{matrix}
 	-\frac{v^{(\sigma)}_{1}\tau_{-}}{\sqrt{12}} & \frac{v^{(\sigma)}_{3}}{\tau_{-}} & \frac{v^{(\sigma)}_{2}}{\tau_{-}} \\
 	\frac{v^{(\sigma)}_{3}}{\tau_{-}} & -\frac{v^{(\sigma)}_{2}\tau_{-}}{\sqrt{12}} & \frac{v^{(\sigma)}_{1}}{\tau_{-}} \\
	 \frac{v^{(\sigma)}_{2}}{\tau_{-}} & \frac{v^{(\sigma)}_{1}}{\tau_{-}} & -\frac{v^{(\sigma)}_{3}\tau_{-}}{\sqrt{12}}
	\end{matrix}
	\right)^{\dagger},
	\end{displaymath}
	\begin{displaymath}
	\mathcal{M}^{\nu}=\frac{g_{\rho}^{\ast}}{\sqrt{2}}\left(
	\begin{matrix}
 	\frac{v^{(\rho)}_{1}\tau_{+}}{\sqrt{12}} & \frac{v^{(\rho)}_{3}}{\tau_{+}} & \frac{v^{(\rho)}_{2}}{\tau_{+}} \\
 	\frac{v^{(\rho)}_{3}}{\tau_{+}} & \frac{v^{(\rho)}_{2}\tau_{+}}{\sqrt{12}} & \frac{v^{(\rho)}_{1}}{\tau_{+}} \\
	 \frac{v^{(\rho)}_{2}}{\tau_{+}} & \frac{v^{(\rho)}_{1}}{\tau_{+}} & \frac{v^{(\rho)}_{3}\tau_{+}}{\sqrt{12}}
	\end{matrix}
	\right)^{\dagger}.
	\end{displaymath}
We find that this simple model has 8 complex parameters
	\begin{displaymath}
	g_{\sigma}, g_{\rho}, v_1^{(\sigma)}, v_2^{(\sigma)}, v_3^{(\sigma)}, v_1^{(\rho)}, v_2^{(\rho)}, v_3^{(\rho)}.
	\end{displaymath}
$g_{\sigma}$ and $g_{\rho}$ can be absorbed into $v_1^{(\sigma)}, v_2^{(\sigma)}, v_3^{(\sigma)}, v_1^{(\rho)}, v_2^{(\rho)}, v_3^{(\rho)}$, thus we have 6 complex parameters, which correspond to 12 real parameters. Since one can absorb two phases into the fermion fields, there are 10 real free parameters.

%% file: conclusions/conclusions.tex
\chapter{Conclusions}
In this thesis we systematically analysed finite family symmetry groups and their application to the lepton sector. Especially we investigated the finite subgroups of $SU(3)$ in detail. We defined and studied $G$-invariant Yukawa-cou\-plings and found a direct relation between $G$-invariant Yukawa-couplings and Clebsch-Gordan coefficients of finite family symmetry groups.
\\
An important result from the study of (finite) family symmetry groups in the lepton sector is the fact that the lepton mass matrices and the lepton mixing matrix are directly related to the Clebsch-Gordan decompositions of tensor products of representations of finite family symmetry groups in a basis-independent way. As a consequence the lepton masses and the mixing matrix become group properties to some extent. Such models may be helpful tools to understand lepton physics beyond the standard model. Our work opens many possibilities for future studies in this field, for example
	\begin{itemize}
	%\item Bidiagonalization of the Dirac mass matrices that are allowed by the analysed finite family symmetry groups. These mass matrices are para\-metrized by arbitrary coupling constants and the vacuum expectation values (VEVs) of the Higgs doublets.
	\item Construction of $G$-invariant  $\phi^{4}$-Lagrangians for the Higgs doublets. This would restrict the vacuum expectation values (VEVs) of the Higgs doublets.
	\item Consideration of models with Majorana neutrinos.
	\item Concentration onto concrete models (seesaw mechanism,...). Are there interesting models that give results in agreement with the experiments without a big amount of fine tuning? These models may be interesting objects of study using computer supported methods.
	\item Extension of the analysis to finite subgroups of $U(3)$.
	\item The theory can also be used on the quark mixing problem. Can one reproduce the CKM matrix using $G$-invariant Yukawa-couplings?
	\item Suppose there exists a model involving a finite family symmetry group, which reproduces well experimental values for the lepton masses and the lepton mixing matrix. What can this finite family symmetry group tell us about underlying, more general theories? Could such theories \textquotedblleft explain\textquotedblright\hspace{1mm} the existence of the three generations of fermions?
	\end{itemize}
I hope that the analysis performed in this work will help to answer these questions.

%% file: group_theory/group_theory.tex
\chapter{Basic definitions and theorems in group theory}\label{groupappendix}
% Definitionen folgen im Wesentlichen Hamermesh.
In this section we will briefly discuss the basic definitions and theorems of the theory of finite groups and their representations. Many of the definitions and theorems presented here follow the classic textbook \cite{hamermesh}.
\\
Other good resources are the books \cite{chen},\cite{joshi} and \cite{sternberg}.

\begin{conv}%Eventuell löschen, falls dies später an anderer Stelle der Diplomarbeit erwähnt wird.
	If not stated otherwise we will always use Einstein's summation convention.
\end{conv}

\section{General definitions and theorems in group theory}\label{A1}

\begin{define}\label{DA1}
	A \textit{group} $(G,\circ)$ is a set $G$ together with a \textit{composition}
		\begin{displaymath}
		\circ: G\times G\longrightarrow G
		\end{displaymath}
	such that
	\begin{enumerate}
 		\item $a\circ(b\circ c)=(a\circ b)\circ c \hspace{10mm}\forall a,b,c\in G.$
		\item $\exists e\in G: \quad e\circ a=a \quad \forall a\in G.$
		\item $\forall a\in G: \quad \exists a^{-1}\in G: \quad a^{-1}\circ a=e. $
	\end{enumerate}
$e$ is called \textit{identity}, $a^{-1}$ is called the \textit{inverse} of $a$.
\end{define}

\begin{prop}\label{PA2}
	Let $(G,\circ)$ be a group, then
	\begin{enumerate}
		\item $a\circ e=a \quad \forall a\in G.$
		\item $a\circ a^{-1}=e \quad \forall a\in G.$
		\item $e$ is unique.
		\item $\forall a\in G: \exists!\hspace{0.5mm}a^{-1}: a^{-1}\circ a=e.$
	\end{enumerate}
\end{prop}

\begin{proof}
	\begin{itemize}
	\item[]
	\item[2.] $a\circ a^{-1}=e\circ (a\circ a^{-1})=((a^{-1})^{-1}\circ a^{-1})\circ a\circ a^{-1}=$\\$=(a^{-1})^{-1}\circ (a^{-1}\circ a)\circ a^{-1}=(a^{-1})^{-1}\circ a^{-1}=e$.
	\item[1.] $a\circ e=a\circ (a^{-1}\circ a)=$[using 2.]$=e\circ a=a$.
	\item[3.] Let $e$ and $e'$ be identity elements: $e=e\circ e'=e'$.
	\item[4.] Let $f$ and $f'$ be inverse elements of $a$: $e=a\circ f=a\circ f'\Longrightarrow a^{-1}\circ a\circ f=a^{-1}\circ a\circ f' \Longrightarrow e\circ f=e\circ f' \Longrightarrow f=f'$.
	\end{itemize}
\end{proof}

\begin{define}\label{DA3}
Two elements $a,b$ of a group $G$ \textit{commute} if $a\circ b=b\circ a$.
\\
If $a\circ b=b\circ a \quad \forall a,b\in G$, $G$ is called \textit{commutative} or \textit{Abelian}.
\\
The number of different elements of a group is called \textit{order} of the group, i.s. %in symbols
$\mathrm{ord}(G)$.
\\
If $\mathrm{ord}(G)\in \mathbb{N}\backslash\lbrace0\rbrace$, $G$ is called \textit{finite}.
\\
$a^{n}:=\underbrace{a\circ...\circ a}_n$,\hspace{0mm} $a^{0}:=e$.
\\
Let $a\in G$: The smallest number $n\in \mathbb{N}\backslash\lbrace0\rbrace$ s.t. $a^{n}=e$, is called \textit{order} of $a$, i.s. $\mathrm{ord}(a)$. If $a^{n}\neq e \quad \forall n\in\mathbb{N}\backslash\lbrace0\rbrace$, $a$ is said to be of order infinity, i.s. $\mathrm{ord}(a)=\infty$.
\end{define}

\begin{conv}
From now on we will often write $G$ instead of $(G,\circ)$ and $ab$ instead of $a\circ b$.
\end{conv}

\begin{define}\label{DA4}
Two groups $G$ and $G'$ are called \textit{isomorphic}, i.s. $G\simeq G'$, if $\exists \phi:G\rightarrow G'$ bijective with:
	\begin{displaymath}
		\phi(ab)=\phi(a)\phi(b) \quad \forall a,b \in G.
	\end{displaymath}
\end{define}

%\begin{define}\label{DA5}
%A group $G$ is called \textit{cyclic}, if all elements can be written as powers of an element $a\in G$.
%\end{define}

\begin{define}\label{DA6}
The \textit{symmetric group} $\mathcal{S}_{n}$ of order $n$ is the group of all \textit{permutations} of $n$ elements. Permutations of $n$ elements are bijective mappings of a set of $n$ different elements onto itself. $\mathrm{ord}(\mathcal{S}_{n})=n!$ . The product of permutations is understood in the sense of composition of functions ($\phi\psi:=\phi\circ\psi$).
\end{define}

\begin{example}
 	\begin{gather*}
	\phi: \lbrace 1,2,3\rbrace\rightarrow\lbrace1,2,3\rbrace\\
	\phi(1)=3, \phi(2)=1, \phi(3)=2
	\end{gather*}
is a permutation of 3 elements.
\end{example}

\begin{conv}
	We will often write permutations as tables, e.g.
	\begin{displaymath}
		\left(
		\begin{array}{ccc}
			...&i&...\\
			...&j&...
		\end{array}
		\right)
	\end{displaymath}
is a permutation that maps $i$ on $j$. The permutation $\phi$ from the example above would be written as
	\begin{displaymath}
		\left(
		\begin{array}{ccc}
			1&2&3\\
			3&1&2
		\end{array}
		\right).
	\end{displaymath}
Permutations can be written as products of \textit{cycles} (\textit{cyclic permutations}):
\medskip
\\
\begin{displaymath}
	\left(
		\begin{array}{cccc}
			i&i+1&...&i+m\\
			i+1&i+2&...&i
		\end{array}
		\right)=
(i \quad i+1 \quad ... \quad i+m).
\end{displaymath}
\end{conv}

\begin{example}
	\begin{displaymath}
	\left(
	\begin{array}{cccccccc}
			1&2&3&4&5&6&7&8\\
			2&3&1&5&4&7&6&8
		\end{array}
		\right)=(1\enspace2\enspace3)(4\enspace5)(6\enspace7)(8)=(1\enspace2\enspace3)(4\enspace5)(6\enspace7).
	\end{displaymath}
(Cycles containing only one element are often suppressed.)
\end{example}
\bigskip
\bigskip
\hspace{0mm}\\
Cycles are invariant under cyclic permutations of their elements, e.g.
\begin{displaymath}
	(1\enspace2\enspace3)=(3\enspace1\enspace2)=(2\enspace3\enspace1),
\end{displaymath}
but \begin{displaymath}
    	(1\enspace2\enspace3)\neq(1\enspace3\enspace2).
    \end{displaymath}

%\begin{define}\label{DA7}
%	A cycle of two elements (2-\textit{cycle}) is called %\textit{transposition}.  
%\end{define}

%\begin{prop}\label{PA8} %Vielleicht später herauslöschen? Zu %ausführlich für kurze Einführung in die Gruppentheorie?
%	Any permutation can be written as a product of %transpositions.
%\end{prop}
%
%\begin{proof}
%	Any permutation can be written as a product of cycles, %and any cycle can be written as a product of transpositions:
%	\begin{displaymath}
%		(k \enspace k+1 \enspace ... \enspace k+n)=(k %\enspace k+n)(k \enspace k+n-1)...(k\enspace k+1).
%	\end{displaymath}
%
%\end{proof}
%Vielleicht noch hinzufügen: Eigenschaften von Transpositionen, Dekrement, etc. (Exzerpt(Hamermesh) Seite 7.)

\begin{define}\label{DA9}
	A subset $H\subset G$ of a group $(G,\circ)$ is called \textit{subgroup}, if $(H,\circ)$ is a group.
	\\
	If $H\neq G$ and $H\neq\lbrace e\rbrace$, $H$ is called \textit{proper subgroup} of $G$. $\lbrace e\rbrace$ and $G$ are the \textit{trivial subgroups} of $G$.
\end{define}
\begin{prop}\label{PA10}
	\textbf{(Criterion for subgroups of finite groups)} A subset $H$ of a finite group $G$ is a subgroup if and only if
	\begin{displaymath}
		h_{1}\circ h_{2}\in H \quad \forall h_{1},h_{2}\in H.
	\end{displaymath}
\end{prop}

\begin{proof}
	\begin{itemize}
		\item[]
		\item[$\Longrightarrow$]: Clear by definition of a group.
		\item[$\Longleftarrow$]: Because $G$ is finite $\exists n\neq 0:h^{n}=e$ for any $h\in H \Rightarrow e\in H$ and $h^{-1}=h^{n-1}\in H \Rightarrow H$ is a group.
	\end{itemize}
\end{proof}

\begin{theorem}\label{TA11} %Hamermesh Seite 17.
	\textbf{(Cayley)} Every group of order $n$ is isomorphic to a subgroup of $\mathcal{S}_{n}$.
\end{theorem}

\begin{proof}
	We construct an isomorphism: Let $b\in G=\lbrace a_{1},...,a_{n}\rbrace$.
	\begin{displaymath}
	\phi: b\mapsto \phi(b)=\left( \begin{matrix}
					 a_{1} & ... & a_{n} \\
					 ba_{1}& ... & ba_{n}
	                              \end{matrix}
				\right).
	\end{displaymath}
$ba_{1},...,ba_{n}$ are the elements of the group $G$ in a new order, so $\phi$ is a map of $G$ to a subset of $\mathcal{S}_{n}$. We will now show that $\phi$ is an isomorphism.

		\begin{displaymath}
		\begin{split}
			\phi(b)\phi(c) & =\left( \begin{matrix}
					 a_{1} & ... & a_{n} \\
					 ba_{1}& ... & ba_{n}
	                              \end{matrix}
				\right)
				\cdot
				\left( \begin{matrix}
					 a_{1} & ... & a_{n} \\
					 ca_{1}& ... & ca_{n}
	                              \end{matrix}
				\right)=\\
			& =\left( \begin{matrix}
					 a_{1} & ... & a_{n} \\
					 ba_{1}& ... & ba_{n}
	                              \end{matrix}
				\right)
				\cdot
				\left( \begin{matrix}
					 ba_{1} & ... & ba_{n} \\
					 bca_{1}& ... & bca_{n}
	                              \end{matrix}
				\right)=\\
			& = \left( \begin{matrix}
					 a_{1} & ... & a_{n} \\
					 (bc)a_{1}& ... & (bc)a_{n}
	                              \end{matrix}
				\right)=\phi(bc).
		\end{split}
		\end{displaymath}
		\item[] $\phi(G)$ is a group, because
		\begin{displaymath}
			\begin{split}
				&\phi(e)\phi(a_{i})=\phi(a_{i}).\\
				&\phi(a^{-1})\phi(a)=\phi(e).\\
			\end{split}
		\end{displaymath}
It follows that $\phi$ is an isomorphism of $G$ on a subgroup of $\mathcal{S}_{n}$.
\end{proof}

\begin{define}\label{DA12}
	Let $G$ be an $n$-dimensional subgroup of $\mathcal{S}_{n}$ with the following property:
	\begin{itemize}
		\item[] $G\ni a\neq e: a(...,i,...)=(...,a(i),...)\neq(...,i,...) \quad \forall i\in\lbrace1,2,...,n\rbrace$. 
	\end{itemize}
($a$ leaves no element unchanged.) Then $G$ is called a \textit{regular permutation group} and the elements of $G$ are called \textit{regular permutations}.
\end{define}

\begin{cor}\label{CA12a}
	Every group of order $n$ is isomorphic to a regular subgroup of $\mathcal{S}_{n}$. 
\end{cor}

\begin{proof}
	In the explicit construction of the isomorphism in the proof of Cayley's theorem (theorem \ref{TA11}) one can see that the permutation group $\phi(G)$ is regular.
\end{proof}

\begin{theorem}\label{TA13}
	\textbf{(Lagrange)} Let $G$ be a finite group of order $n$, and $H\subset G$ a subgroup of $G$ with order $m$. Then $m$ is a divisor of $n$.
\end{theorem}

\begin{proof}
	\begin{itemize}
		\item[]
		\item[] If $H\supset G\Rightarrow H=G\Rightarrow n=m\Rightarrow m$ is a divisor of $n$.
		\item[] If $H\not\supset G$, let $a_{2}$ be an element of $G$ that is not contained in $H$.
		\\
		$H=\lbrace e,h_{2},...,h_{m}\rbrace$.
		\\
		$a_{2}H:=\lbrace a_{2}e,a_{2}h_{2},...,a_{2}h_{m} \rbrace$.
		\item[] All $a_{2}h_{i}$ are different, because $a_{2}h_{i}=a_{2}h_{j}\Rightarrow h_{i}=h_{j}$.
		\\
		$a_{2}h_{i}\not\in H$, because else: $a_{2}h_{i}=h_{j}\in H \Rightarrow a_{2}=h_{i}h_{j}^{-1}\in H\Rightarrow$ contradiction to $a_{2}\not\in H$.
		\item[] We now have won two disjoint subsets of $G$. If $H\cup a_{2}H\neq G\Rightarrow$ take an element $a_{3}\in G, a_{3}\not\in(H\cup a_{2}H)\Rightarrow H,a_{2}H,a_{3}H$ are disjoint subsets of $G$ (by the same arguments as before).
		\item[] Continue till $G=H\cup a_{2}H\cup ...\cup a_{s}H$, $G$ is divided into $s$ disjoint subsets with $m$ elements. $\Rightarrow n=s\cdot m$.
	\end{itemize}
\end{proof}

\begin{define}\label{DA14}
	Let $G=H\cup a_{2}H\cup ...\cup a_{s}H$ as in the proof of theorem \ref{TA13}. It follows $n=s\cdot m$. $s$ is called \textit{index} of the subgroup $H$ under $G$. The sets $a_{i}H$ are called \textit{left cosets} of $H$ in $G$. The cosets are not subgroups of $G$, because $e\not\in a_{i}H$.
\end{define}

%\begin{prop}\label{PA15}
%	A group of prime order is cyclic.
%\end{prop}

%\begin{proof}
%	Let $a\in G\backslash\lbrace e\rbrace$ arbitrary. We construct the \textit{period} of $a$ ($\mathrm{ord}(a)=m$), that is $\lbrace a^{0},a^{1},...,a^{m-1}\rbrace$, which is the smallest subgroup of $G$ that contains $a$.
%	\\
%	From Lagrange's theorem it follows, that $m$ is a divisor of $n=\mathrm{ord}(G)\Rightarrow$ the orders of all $a\in G\backslash\lbrace e\rbrace$ must be divisors of $n$.
%	\\
%	Let now $n$ be prime. $\Rightarrow \mathrm{ord}(a)=1$ or $n \Rightarrow \mathrm{ord}(a)=n$, because $a\neq e$.
%	\\
%	$\Rightarrow G=\lbrace e,a,a^{2},...,a^{n-1}\rbrace$. $\Rightarrow$ $G$ is cyclic.
%\end{proof}

\section{Conjugate classes, invariant subgroups, homomorphisms}

\begin{define}\label{DA16}
	$b\in G$ is \textit{conjugate} to $a\in G$, i.s. $b\sim a$, if $\exists u\in G: uau^{-1}=b$.
\end{define}

\begin{defprop}\label{DA17}
	A relation $\sim$ is called \textit{equivalence relation} if
	\begin{enumerate}
		\item $a\sim a.$
		\item $a\sim b \Leftrightarrow b\sim a.$
		\item $a\sim b, b\sim c \Rightarrow a\sim c.$ 
	\end{enumerate}
The relation \textit{conjugate} $\sim$ is an equivalence relation.
\end{defprop}

\begin{proof}
	\begin{itemize}
		\item[]
		\item[1.] $eae^{-1}=a\Rightarrow a\sim a$.
		\item[2.] $uau^{-1}=b \Rightarrow u^{-1}bu=u^{-1}b(u^{-1})^{-1}=a\Rightarrow \quad a\sim b\Leftrightarrow b\sim a$.
		\item[3.] $uau^{-1}=b$, $vbv^{-1}=c\Rightarrow$ $(vu)a(vu)^{-1}=c\Rightarrow \quad$ $a\sim b,b\sim c\Rightarrow a\sim c$. 
	\end{itemize}
\end{proof}

\begin{define}\label{DA18}
	Given an equivalence relation $\sim$, a set $M$ can be divided into \textit{equivalence classes} $C_{i}$ such that
		\begin{enumerate}
			\item $a_{i}\sim b_{i} \quad \forall a_{i},b_{i}\in C_{i}.$
			\item $a_{i}\not\sim b_{j}$ if $a_{i}\in C_{i},b_{j}\in C_{j}, i\neq j.$
			\item $M=\bigcup_{i}C_{i}$.
		\end{enumerate}
	The equivalence classes to the equivalence relation \textquotedblleft conjugate\textquotedblright\hspace{1mm} are called conjugate classes.
\end{define}

\begin{define}\label{DA19}
	Let $H$ be a subgroup of $G$, $a\in G$.
	\begin{displaymath}
		aHa^{-1}:=\lbrace aha^{-1}\vert h\in H \rbrace
	\end{displaymath}
is a \textit{conjugate subgroup} of $H$ in $G$.
\medskip
\\
If $aHa^{-1}=H \quad \forall a\in G$, $H$ is called an \textit{invariant subgroup} in $G$.
\end{define}

\begin{crit}\label{CA20}
	A subgroup $H$ of $G$ is invariant if and only if it consists of complete conjugate classes of $G$.
\end{crit}

\begin{proof}
	\begin{itemize}
		\item[]
		\item[$\Longrightarrow$]: Let $H$ be an invariant subgroup. $h\in H\Rightarrow$ all conjugates $aha^{-1}$ are elements of $H\Rightarrow H$ contains the conjugate class of $h$.  
		\item[$\Longleftarrow$]: Let $H$ be a subgroup consisting of conjugate classes. Suppose $\exists a\in G: aha^{-1}=f\not\in H$. $f\sim h\Rightarrow f$ is an element of the conjugate class of $h\Rightarrow f\in H \Rightarrow$ contradiction. 
	\end{itemize}
\end{proof}

\begin{define}\label{DA21}
	A group is called \textit{simple} if it contains no proper invariant subgroup.
	\\
	A group is called \textit{semisimple} if none of its proper invariant subgroups is Abelian.
\end{define}

\begin{define}\label{DA22}
	\textit{Coset multiplication}: 
	\begin{displaymath}
		(aH)(bH):=\lbrace (ah_{1})(bh_{2})\vert h_{1},h_{2}\in H\rbrace.
	\end{displaymath}
	The set of all cosets of an invariant subgroup $H\subset G$ together with the composition \textquotedblleft coset multiplication\textquotedblright\hspace{1mm} forms a group, the \textit{factor group} $G/H$. ($e=H, (aH)^{-1}=a^{-1}H$.)
\end{define}

\begin{prop}
Let $G$ be a finite group and let $H\subset G$ be an invariant subgroup of $G$. Then
	\begin{displaymath}
	\mathrm{ord}(G/H)=\frac{\mathrm{ord}(G)}{\mathrm{ord}(H)}.
	\end{displaymath}
\end{prop}

\begin{proof}
Every coset $gH$, $g\in G$ has $\mathrm{ord}(H)$ elements, thus the number of different cosets is $\frac{\mathrm{ord}(G)}{\mathrm{ord}(H)}=\mathrm{ord}(G/H)$.
\end{proof}

\begin{define}\label{DA23}
	A mapping $\phi:G\rightarrow G'$ is called \textit{homomorphism}, if
	\begin{displaymath}
		\phi(ab)=\phi(a)\phi(b) \quad \forall a,b\in G.
	\end{displaymath}
\end{define}

\begin{prop}\label{PA24}
	\textbf{Properties of homomorphisms:} Let $\phi:G\rightarrow G'$ be a homomorphism, then
	\begin{enumerate}
		\item $\phi(e)=e'.$
		\item $\phi(a^{-1})=\phi(a)^{-1}.$
	\end{enumerate}
\end{prop}

\begin{proof}
	\begin{itemize}
		\item[]
		\item[1.]
		$\phi(a)=\phi(ea)=\phi(e)\phi(a)\Rightarrow \phi(e)=e'$.
		\item[2.] $e'=\phi(e)=\phi(aa^{-1})=\phi(a)\phi(a^{-1})\Rightarrow \phi(a^{-1})=\phi(a)^{-1}$.
	\end{itemize}
\end{proof}

%\begin{define}\label{D25}
%	A group $G$ is the \textit{direct product} of its (proper) %subgroups $H_{1},...,H_{n}$, i.s. $G=H_{1}\times H_{2}\times %...\times H_{n}$, if
%	\begin{enumerate}
%		\item The elements of different subgroups commute.
%		\item Every element $g\in G$ can be uniquely %written as
%		\begin{displaymath}
%			g=h_{1}h_{2}\cdots h_{n}, \quad h_{j}\in %H_{j}.
%		\end{displaymath}
%	\end{enumerate}
%The proper subgroups $H_{1},...,H_{n}$ are called \textit{direct %factors} of $G$.
%\end{define}
%
%\begin{prop}\label{PA26}
%	\begin{enumerate}
%		\item $H_{i}\cap H_{j}=\lbrace e\rbrace$ for %$i\neq j$.
%		\item The direct factors of $G$ are invariant %subgroups of $G$.
%	\end{enumerate}
%\end{prop}
%
%\begin{proof}
%	\begin{enumerate}
%		\item Let $i\neq j$. Suppose $\exists a\neq e: %a\in H_{i}\cap H_{j}:$ w.l.o.g. $a\in H_{1}\cap H_{2}$. Then %$g=e\cdot a\cdot h_{3}\cdot h_{4}\cdots h_{n}$ and $g=a\cdot %e\cdot h_{3}\cdot h_{4}\cdots h_{n}$ are two different %decompositions of $g\Rightarrow$ contradiction to definition %\ref{DA25} (2).
%		\item Let $g=h_{1}\cdots h_{n}$. Let $h_{i}'\in %H_{i}$ ($H_{i}$ are the factors of $G$.) $\Rightarrow$ All %conjugates of $h_{i}'$ have the form $g %h_{i}'g^{-1}=(h_{1}h_{2}\cdots h_{n})h_{i}'(h_{1}h_{2}\cdots %h_{n})^{-1}$. ($g\in G$ arbitrary)
%		\\
%		Because elements of different subgroups $H_{i}$ %commute it remains $gh_{i}'g^{-1}=h_{i}h_{i}'h_{i}^{-1}\in %H_{i}\Rightarrow$ the complete conjugation class of $h_{i}'$ is %contained in $H_{i}\Rightarrow H_{i}$ invariant.
%	\end{enumerate}
%\end{proof}

\begin{define}\label{DA27}
	Let $G,G'$ be groups. The \textit{direct product} $G\times G'$ is defined by
	\begin{displaymath}
		G\times G':=(\lbrace (a,a')\vert a\in G, a'\in G'\rbrace,(a,a')\circ(b,b')=(a\circ b,a'\circ b')).
	\end{displaymath}
Remark: $\mathrm{ord}(G\times G')=\mathrm{ord}(G)\cdot\mathrm{ord}(G')$.
\end{define}

\section{Representations}

\begin{define}\label{DA28}
	Let $D(G)$ be a group of operators on a vectorspace $V$. A homomorphism
	\begin{displaymath}
		D: G\rightarrow D(G)
	\end{displaymath}
is called \textit{representation} of the group $G$ on the vectorspace $V$, and $V$ is called \textit{representation space}. $\mathrm{dim}V$ is called \textit{dimension} of the representation.
\\
If $D(G)$ is a group of linear operators on $V$, $D$ is called a \textit{linear representation}.
\\
If the linear operators in $D(G)$ are represented as matrices, $D$ is called a \textit{matrix representation} of $G$. 
\end{define}

\begin{conv}
In this work we will only consider linear representations. Therefore by \textquotedblleft representation\textquotedblright\hspace{1mm} we will always mean \textquotedblleft linear representation\textquotedblright.
\\
We will only consider finite dimensional representations, so all usual operations of linear algebra ($\mathrm{det}$, $\mathrm{Tr}$,...) are well defined.
\end{conv}

\begin{prop}\label{PA29}
	\textbf{Properties of representations:}
	\begin{enumerate}
		\item $D(ab)=D(a)D(b) \quad \forall a,b\in G$.
		\item $D(a^{-1})=D(a)^{-1}$.
		\item $D(e)=id_{V}$.
	\end{enumerate}
\end{prop}

\begin{proof} By definition. (Homomorphism. See proposition \ref{PA24}.)
\end{proof}

\begin{cor}\label{CA30}
	Let $D$ be a representation of a group $G$, then \begin{displaymath}
	\mathrm{det}(D(a))\neq 0 \quad \forall a\in G.
	\end{displaymath}
\end{cor}

\begin{proof}
$D(a)\in D(G)\Rightarrow\exists D(a)^{-1}=D(a^{-1})\in D(G)\Rightarrow$ $D(a)$ invertible.
\end{proof}

\begin{define}\label{DA31}
	A representation $D:G\rightarrow D(G)$ is called \textit{faithful}, if $D$ is an isomorphism.
\end{define}

\begin{define}\label{DA32}
	A representation $D'$ is \textit{equivalent} to a representation $D$, i.s. $D'\sim D$, if there exists an invertible linear operator $C$ such that
	\begin{displaymath}
		D'(a)=CD(a)C^{-1} \quad \forall a\in G.
	\end{displaymath}
\end{define}

\begin{define}\label{DA41}
	A representation is called \textit{unitary} if $D(a)$ is unitary for all $a\in G$.
\end{define}

\begin{theorem}\label{TA42}%Hamermesh S.92
	Let $G$ be a finite group. Then every representation $D(G)$ is equivalent to a unitary representation.
\end{theorem}

\begin{proof}
	\begin{itemize}
		\item[] Let $x,y\in V$. Define
		\begin{displaymath}
			\langle x,y\rangle:=\frac{1}{\mathrm{ord}(G)}\sum_{a\in G}(D(a)x,D(a)y).
		\end{displaymath}
		where $(\cdot,\cdot)$ is a Hermitian scalar product, thus $\langle\cdot,\cdot\rangle$ is a Hermitian scalar product too. Furthermore $\langle\cdot,\cdot\rangle$ is invariant under $D(G)$, because
		\begin{displaymath}
		\begin{split}
			\langle D(a)x,D(a)y\rangle
			& = \frac{1}{\mathrm{ord}(G)}\sum_{b\in G}(D(b)D(a)x,D(b)D(a)y)=\\
			& = \frac{1}{\mathrm{ord}(G)}\sum_{b\in G}(D(ba)x,D(ba)y)=\\
			& =\frac{1}{\mathrm{ord}(G)}\sum_{ba\in G}(D(ba)x,D(ba)y)=\\
			& = \langle x,y\rangle.
		\end{split}
		\end{displaymath}
	
		\item[] Let $\lbrace u_{i}\rbrace_{i}$ and $\lbrace v_{j}\rbrace_{j}$ be orthonormal bases with respect to $(\cdot,\cdot)$ and $\langle\cdot ,\cdot\rangle$. Then there exists a linear operator $T$ such that $v_{i}=Tu_{i}$.\\
		$\Rightarrow \langle Tx,Ty\rangle=x_{i}^{\ast}y_{j}\langle Tu_{i},Tu_{j}\rangle=x_{i}^{\ast}y_{j}\langle v_{i},v_{j}\rangle=x_{i}^{\ast}y_{j}\delta_{ij}=x_{i}^{\ast}y_{j}(u_{i},u_{j})=(x,y)$.

		\item[] Define $D'(a):=T^{-1}D(a)T$. It follows:
		\begin{displaymath}
		\begin{split}
			(D'(a)x,D'(a)y) & = (T^{-1}D(a)Tx,T^{-1}D(a)Ty)=\\
					& = [(x,y)=\langle Tx,Ty\rangle]=\\
					& = \langle D(a)Tx,D(a)Ty\rangle=\\
					& = \langle Tx,Ty\rangle=\\
					& = (x,y).
		\end{split}
		\end{displaymath}
		
		\item[] So $D'$ is unitary with respect to the (arbitrary) Hermitian scalar product $(\cdot,\cdot)$. 
	\end{itemize}
\end{proof}

\begin{define}\label{DA43}%Sternberg S.49
	Let $D$ be a representation of a group $G$ on a vectorspace $V$. A subspace $W\subset V$ is called \textit{invariant} if $D(a)W\subset W \enspace\forall a\in G$.\\
	Because $D(a)$ is invertible it follows that $D(a)W=W$.\\
	$\lbrace 0\rbrace$ and $V$ are the trivial invariant subspaces of $V$.
\end{define}

\begin{define}\label{DA44}
	A representation $D(G)$ is called \textit{irreducible}, if there is no nontrivial invariant subspace of $V$.
\end{define}

\begin{define}\label{DA45}
	Let $V_{1}$ and $V_{2}$ be vectorspaces. The \textit{direct sum} of $V_{1}$ and $V_{2}$, i.s. $V_{1}\oplus V_{2}$, is defined by
	\begin{displaymath}
	\begin{split}
		V_{1}\oplus V_{2} := & \lbrace (v_{1},v_{2})\vert v_{1}\in V_{1},v_{2}\in V_{2};\\ & (v_{1},v_{2})+(w_{1},w_{2})=(v_{1}+w_{1},v_{2}+w_{2}),\\
		& \lambda(v_{1},v_{2})=(\lambda v_{1},\lambda v_{2})\rbrace.
	\end{split}
	\end{displaymath}
	Instead of the pair $(v_{1},v_{2})$ we will write $v_{1}\oplus v_{2}$.
	\medskip
	\\
	Let $D_{1}$ and $D_{2}$ be representations of a group $G$ on the vectorspaces $V_{1}$ and $V_{2}$. The \textit{direct sum} of the representations $D_{1}$ and $D_{2}$, i.s. $D_{1}\oplus D_{2}$, on the vectorspace $V_{1}\oplus V_{2}$ is defined by
	\begin{displaymath}
		(D_{1}\oplus D_{2})(a)(v_{1}\oplus v_{2}):=D_{1}(a)v_{1}\oplus D_{2}(a)v_{2} \quad \forall a\in G.		
	\end{displaymath}
	If $[D]$ is the matrix representation of $D$ in a certain basis, then the matrix representation of $D:=D_{1}\oplus D_{2}$ is given by
	\begin{displaymath}
		[D(a)]=\left( \begin{matrix}
				[D_{1}(a)]  & \textbf{0} \\
				\textbf{0} & [D_{2}(a)]
		             \end{matrix} \right),
	\end{displaymath}
	where $\textbf{0}$ are null matrices (with appropriate dimensions).
\end{define}

\begin{define}\label{DA46}
	A representation $D(G)$ is called \textit{completely reducible}, if it can be written as a direct sum of irreducible representations.
	\begin{displaymath}
		D=D_{1}\oplus \cdots \oplus D_{n} \quad\mbox{on}\quad V=V_{1}\oplus\cdots\oplus V_{n}.
	\end{displaymath}
\end{define}

\begin{theorem}\label{TA47}
	Every finite dimensional unitary representation of a finite group is completely reducible.
\end{theorem}

\begin{proof}
	Let $D$ be a finite dimensional unitary representation of a finite group $G$.
	
	\begin{itemize}
		\item If $D$ is irreducible, it is obviously completely reducible.
		\item $D$ is not irreducible. $\Rightarrow$ There exists a nontrivial invariant subspace $W\subset V$:
		\begin{displaymath}
			D(a)W=W \quad \forall a\in G.
		\end{displaymath}
		$W^{\perp}:=\lbrace x\in V\vert (x,w)=0 \enspace\forall w\in W\rbrace$, where $(\cdot,\cdot)$ is the Hermitian scalar product.
		\\
		Claim: $W^{\perp}$ is an invariant subspace. This is true, because:
		\\
		Let $v\in W^{\perp},w\in W$.
		\begin{displaymath}
			0=(v,w)=(D(a)v,\underbrace{D(a)w}_{\in W})\Rightarrow D(a)v\in W^{\perp}.
		\end{displaymath}
		$\Rightarrow W^{\perp}$ is an invariant subspace. $\Rightarrow V=W\oplus W^{\perp}$ is a decomposition of $V$ into invariant subspaces.
		\medskip
		\\
		We now take for $W$ the smallest nontrivial invariant subspace of $V$. (This means that $W$ has the smallest possible positive dimension.)
		\\
		$\Rightarrow$ The restriction of $D$ on $W$, i.s. $D\vert_{W}$, is irreducible.

		\item If $D\vert_{W^{\perp}}$ is irreducible, $D\vert_{V}$ is completely reducible.

		\item If $D\vert_{W^{\perp}}$ is not irreducible $\Rightarrow$ Performing the same procedure as before $D\vert_{V}$ can be decomposed into a direct sum of irreducible representations:
		\begin{displaymath}
			V=V_{1}\oplus\cdots\oplus V_{n} \quad\quad D\vert_{V}=D\vert_{V_{1}}\oplus\cdots\oplus D\vert_{V_{n}}.
		\end{displaymath}
	\end{itemize}
$\Rightarrow D$ is completely reducible.
\end{proof}

\begin{cor}\label{CA48}
	Every finite dimensional representation of a finite group is completely reducible.
\end{cor}

\begin{proof}
	As seen in the proof of theorem \ref{TA42} any representation $D$ is unitary w.r.t.
	\begin{displaymath}
		\langle x,y\rangle:=\frac{1}{\mathrm{ord}(G)}\sum_{a\in G}(D(a)x,D(a)y).
	\end{displaymath}
	$\Rightarrow$ By the same arguments as in the proof of theorem \ref{TA47} $D$ is completely reducible.
\end{proof}

\begin{sublemma}\label{SA49}
	Let $V$ be a finite dimensional vectorspace over $\mathbb{C}$. Then every linear operator on $V$ has at least one eigenvalue.
\end{sublemma}

\begin{proof}
	The eigenvalues $\lambda$ of an operator $A$ are uniquely characterized by 
		\begin{displaymath}
			\mathrm{det}(A-\lambda id)=0.
		\end{displaymath}
$\mathrm{det}(A-\lambda id)$ is a polynomial of order $\mathrm{dim}V$ in $\lambda$. From the fundamental theorem of algebra we know that there exists at least one zero of the polynomial, therefore it exists at least one eigenvalue $\lambda$.
\end{proof}

\begin{lemma}\label{LA50}%Sternberg p.55 Wesentlich besserer Beweis als in Hamermesh oder Joshi
	\textbf{(Schur)} Let $D_{1}$ and $D_{2}$ be finite dimensional irreducible representations of a finite group $G$ on the vectorspaces $V_{1}$ and $V_{2}$.
	\\
	Let $A: V_{1}\rightarrow V_{2}$ be a linear operator s.t.
	\begin{displaymath}
		D_{2}(a) A = A D_{1}(a) \quad\forall a\in G.
	\end{displaymath}
	Then:

	\begin{enumerate}
		\item If $D_{1}\not\sim D_{2} \Rightarrow A=0$.
		\item If $D_{1}=D_{2} \Rightarrow \exists \lambda\in \mathbb{C}: A=\lambda id$.
	\end{enumerate}
\end{lemma}

\begin{proof}
	\begin{itemize}
		\item[]
		\item[1.] $\mathrm{ker}(A):=\lbrace v\in V_{1}\vert A(v_{1})=0\rbrace$ (\textit{kernel of $A$}).
		\smallskip
		\\
		Let $v\in \mathrm{ker}(A)$: $A D_{1}(a)v=D_{2}(a) A v=D_{2}(a)0=0$.
		\smallskip
		\\
		$\Rightarrow D_{1}(a)\mathrm{ker}(A)\subset \mathrm{ker}(A)\Rightarrow \mathrm{ker}(A)$ is an invariant subspace.
		\smallskip
		\\
		Because $D_{1}$ is irreducible, it follows that $\mathrm{ker}(A)=\lbrace 0\rbrace$, or $\mathrm{ker}(A)=V_{1}$.
		\bigskip
		\\
		First case: $\mathrm{ker}(A)=V_{1}\Rightarrow A=0$.
		\smallskip
		\\
		Second case: $\mathrm{ker}(A)=\lbrace 0\rbrace \Rightarrow D_{2}(a) A(V_{1})\neq \lbrace 0\rbrace$.
		\smallskip
		\\
		$\Rightarrow A(V_{1})$ is an invariant subspace of $V_{2}\Rightarrow A(V_{1})=V_{2}$, because $D_{2}$ is irreducible.
		\smallskip
		\\
		$\Rightarrow A(V_{1})=V_{2}$ and $\mathrm{ker}(A)=\lbrace 0\rbrace\Rightarrow A$ is an isomorphism.
		\\
		$\Rightarrow AD_{1}(a)A^{-1}=D_{2}(a)\Rightarrow D_{1}\sim D_{2}$.

		\item[2.] $D_{1}=D_{2}\Rightarrow V_{1}=V_{2}\Rightarrow$
		\begin{displaymath}
			D_{2}(a)(A-\lambda id)=(A-\lambda id)D_{1}(a)
		\end{displaymath}
		By sublemma \ref{SA49} every linear operator has at least one eigenvalue ($\in \mathbb{C}$).
		\\
		$\Rightarrow$ Let $\lambda$ be this eigenvalue $\Rightarrow (A-\lambda id)$ is singular.
		\smallskip
		\\
		$\Rightarrow \mathrm{ker}(A-\lambda id)\neq\lbrace 0\rbrace\Rightarrow \mathrm{ker}(A-\lambda id)=V_{1}$ (see proof of 1.)
		\smallskip
		\\
		$\Rightarrow A-\lambda id = 0 \Rightarrow A=\lambda id$.
	\end{itemize}
\end{proof}
\hspace{0mm}\\
This very elegant and compact proof follows \cite{sternberg}(p.55).

\section{Orthogonality relations}

\begin{lemma}\label{TA51}
	Let $G$ be a finite group with order $g=\mathrm{ord}(G)$, and let $D(G)$ be an $n$-dimensional irreducible representation of $G$. Then
	\begin{displaymath}
		\sum_{a\in G}D_{ij}(a)D_{kl}(a^{-1})=\frac{g}{n}\delta_{jk}\delta_{il}.
	\end{displaymath}
	($D_{ij}$ means the $ij$-component of the matrix representation of $D$ w.r.t. a chosen basis.)
\end{lemma}

\begin{proof}
We define
\begin{displaymath}
	A:=\sum_{a\in G}D(a)BD(a^{-1})
\end{displaymath}
with an arbitrary linear operator $B$ on $V$.
\medskip
\\
Claim: $D(a)A=AD(a)$. This is true, because:
\begin{displaymath}
		\begin{split}
			D(a)A & = \sum_{b\in G}D(a)D(b)BD(b^{-1})=\\
			& = \sum_{b\in G} D(ab)BD(b^{-1})D(a^{-1})D(a)=\\
			& = \sum_{b\in G} D(ab)BD((ab)^{-1})D(a)=\\
			& = \sum_{ab\in G} D(ab)BD((ab)^{-1})D(a)=AD(a).
		\end{split}
		\end{displaymath}
Using Schur's lemma 2 (lemma \ref{LA50} 2.) it follows: $A=\lambda id$ for a $\lambda\in\mathbb{C}$.
\bigskip
\\
We now set $B_{ij}=\delta_{il}\delta_{jm}$ for fixed $l,m$. It follows
\begin{displaymath}
	\begin{split}
	A_{ij} & =\sum_{a\in G} D_{ik}(a)B_{kp}D_{pj}(a^{-1})=\\
	& = \sum_{a\in G} D_{ik}(a)\delta_{kl}\delta_{pm}D_{pj}(a^{-1})=\\
	& = \sum_{a\in G}D_{il}(a)D_{mj}(a^{-1})=\lambda(l,m) \delta_{ij}.
	\end{split}
\end{displaymath}
$\lambda(l,m)$ can be determined by setting $i=j$:
\begin{displaymath}
	\begin{split}
		A_{ii} & = \lambda(l,m)\delta_{ii}=n\lambda(l,m)=\\
		& =\sum_{a\in G}D_{il}(a)D_{mi}(a^{-1})=\\
		& = \sum_{a\in G}D_{mi}(a^{-1})D_{il}(a)=\\
		& = \sum_{a\in G}D_{ml}(e)=g\delta_{ml}.
	\end{split}
\end{displaymath}
$$\Rightarrow \lambda(l,m)=\frac{g}{n}\delta_{lm}\Rightarrow \sum_{a\in G}D_{il}(a)D_{mj}(a^{-1})=\frac{g}{n}\delta_{lm} \delta_{ij}.$$
\end{proof}

\begin{lemma}\label{LA52}
	Let $D^{(1)}$ and $D^{(2)}$ be two non-equivalent irreducible representations of $G$ ($\mathrm{dim}V_{1}=n_{1}$, $\mathrm{dim}V_{2}=n_{2}$). Then
	\begin{displaymath}
			\sum_{a\in G}D^{(2)}_{il}(a)D^{(1)}_{mj}(a^{-1})=0 \quad\forall i,j,l,m.
	\end{displaymath}
\end{lemma}

\begin{proof}
We define
\begin{displaymath}
	A:=\sum_{a\in G}D^{(2)}(a)BD^{(1)}(a^{-1})
\end{displaymath}
with an arbitrary linear operator $B:V_{1}\rightarrow V_{2}$. Completely analogous to the proof of lemma \ref{TA51} we have
\begin{displaymath}
	D^{(2)}(a)A=AD^{(1)}(a) \quad\forall a\in G.
\end{displaymath}
Using Schur's lemma 1 (lemma \ref{LA50} 1.) it follows $A=0$. By choosing $B_{ij}=\delta_{il}\delta_{jm}$ for fixed $l,m$ as in the proof of lemma \ref{TA51} we obtain:
\begin{displaymath}
	\sum_{a\in G}D^{(2)}_{il}(a)D^{(1)}_{mj}(a^{-1})=0 \quad\forall i,j,l,m.
\end{displaymath}
\end{proof}

\begin{theorem}\label{TA53}%Anschauliche Erklärung zum group space aus Joshi.
	Let $D^{(\alpha)},\enspace \alpha\in\lbrace 1,...,k\rbrace$ be all non-equivalent irreducible representations of a finite group $G$ on vectorspaces $V_{\alpha}$ ($\mathrm{dim}V_{\alpha}=n_{\alpha}$). Then the following \textit{orthogonality relation} holds:
	\begin{displaymath}
		\sum_{a\in G}D^{(\alpha)}_{il}(a)D^{(\beta)}_{mj}(a^{-1})=\frac{g}{n_{\alpha}}\delta_{\alpha\beta}\delta_{ij}\delta_{lm}.
	\end{displaymath}
\end{theorem}

\begin{proof}
	$\lbrace D^{(\alpha)}_{ij}(a)\rbrace_{\alpha,i,j}$ are functions (with the discrete variable $a$), which are elements in a vectorspace (the so called \textit{group space}). 
	%Folgendes ist nicht ganz verständlich und auch nicht notwendig zu erwähnen.
	%\\
	%The set $\lbrace D^{(\alpha)}_{ij}\rbrace_{\alpha,i,j}$ contains $\sum_{\alpha=1}^{k}n_{\alpha}^{2}$ of these functions, but the group space is $g$-dimensional, because every function $D^{(\alpha)}_{ij}$ is uniquely determined by the $g$ numbers $D^{(\alpha)}_{ij}(a)$ ($a\in G$).
	%\\
	%We will later see that $\sum_{\alpha=1}^{k}n_{\alpha}^{2}=g\Rightarrow$ $\lbrace D^{(\alpha)}_{ij}\rbrace_{\alpha,i,j}$ span the group space, but by now we just know that $\sum_{\alpha=1}^{k}n_{\alpha}^{2}\leq g$.
	\smallskip
	\\
	We define a scalar product on the group space:
	\begin{displaymath}
		(D^{(\alpha)}_{ij},D^{(\beta)}_{kl}):=\frac{1}{g}\sum_{a\in G}D^{(\alpha)}_{ij}(a)D^{(\beta)}_{kl}(a^{-1}).
	\end{displaymath}
	From lemma \ref{TA51} and lemma \ref{LA52} we deduce
	\begin{displaymath}
		(D^{(\alpha)}_{ij},D^{(\beta)}_{kl})=0 \quad\mbox{for}\quad \alpha\neq\beta,
	\end{displaymath}
	and
	\begin{displaymath}
		(D^{(\alpha)}_{ij},D^{(\alpha)}_{kl})=\frac{1}{n_{\alpha}}\delta_{jk}\delta_{il}.
	\end{displaymath}
	It follows
	\begin{displaymath}
		(D^{(\alpha)}_{il},D^{(\beta)}_{mj})=\frac{1}{n_{\alpha}}\delta_{\alpha\beta}\delta_{ij}\delta_{lm}
		\Rightarrow
		\sum_{a\in G}D^{(\alpha)}_{il}(a)D^{(\beta)}_{mj}(a^{-1})=\frac{g}{n_{\alpha}}\delta_{\alpha\beta}\delta_{ij}\delta_{lm}.
	\end{displaymath}
\end{proof}
\hspace{0mm}\\
To formulate the orthogonality relations in an elegant way, we will now define a scalar product similar to the product in the proof of theorem \ref{TA53}.

\begin{defprop}\label{DA54}
	Let $D^{(\alpha)}$ and $D^{(\beta)}$ be representations of a finite group $G$. We define a scalar product on the group space by
	\begin{displaymath}
		(D^{(\alpha)}_{ij},D^{(\beta)}_{kl}):=\frac{1}{\mathrm{ord}(G)}\sum_{a\in G}D^{(\alpha)}_{ij}(a)D^{(\beta)}_{kl}(a^{-1}).
	\end{displaymath}
	\begin{enumerate}
		\item Let $D$ be an $n$-dimensional irreducible representation, then
		\begin{displaymath}
			(D_{ij},D_{kl})=\frac{1}{n}\delta_{jk}\delta_{il}.
		\end{displaymath}
		\item Let $D^{(1)}$ and $D^{(2)}$ be two non-equivalent irreducible representations of $G$, then
		\begin{displaymath}
			(D^{(1)}_{ij},D^{(2)}_{kl})=0 \quad \forall i,j,k,l.
		\end{displaymath}
		\item Let $D^{(\alpha)},\enspace \alpha\in\lbrace 1,...,k\rbrace$ be all non-equivalent irreducible representations of a finite group $G$ on vectorspaces $V_{\alpha}$ ($\mathrm{dim}V_{\alpha}=n_{\alpha}$), then
		\begin{displaymath}
			(D^{(\alpha)}_{ij},D^{(\beta)}_{kl})=\frac{1}{n_{\alpha}}\delta_{\alpha\beta}\delta_{jk}\delta_{il}.
		\end{displaymath}
	\end{enumerate}
\end{defprop}

\begin{proof}
	This proposition directly follows from lemma \ref{TA51}, lemma \ref{LA52} and theorem \ref{TA53}.
\end{proof}

\section{Characters and their orthogonality relations}

\begin{defprop}\label{DA55}
	Let $A,B$ be different matrix representations of a linear operator $\phi$. Then $\mathrm{Tr}(A)=\mathrm{Tr}(B)=:\mathrm{Tr}(\phi)$.
\end{defprop}

\begin{proof}
Because $A,B$ are matrix representations of the same linear operator, there exists an invertible matrix $C$ such that $A=CBC^{-1}.\Rightarrow \mathrm{Tr}(A)=\mathrm{Tr}(CBC^{-1})=\mathrm{Tr}(C^{-1}CB)=\mathrm{Tr}(B).$
\end{proof}

\begin{define}\label{DA56}
	Let $D$ be a representation of $G\ni a$. $\chi_{D}(a):=\mathrm{Tr}(D(a))$ is called \textit{character} of $a$ in the representation $D$.
\end{define}

\begin{define}\label{DA57}
	Hermitian scalar product for characters:
	\begin{displaymath}
		(\chi_{D_{1}},\chi_{D_{2}}):=\frac{1}{\mathrm{ord}(G)}\sum_{a\in G}\chi_{D_{1}}(a^{-1})\chi_{D_{2}}(a).
	\end{displaymath}
\end{define}

\begin{prop}\label{PA58}
	\textbf{Properties of characters:}
	\begin{enumerate}
		\item If $a$ and $b$ are conjugate elements of $G$, then
		\begin{displaymath}
			\chi_{D}(a)=\chi_{D}(b).
		\end{displaymath}
		\item If $D$ and $D'$ are equivalent, then
		\begin{displaymath}
			\chi_{D}(a)=\chi_{D'}(a) \quad\forall a\in G.
		\end{displaymath}
		\item $\chi_{D}(a^{-1})=\chi_{D}^{\ast}(a).$
		\item Let $D_{1}$ and $D_{2}$ be representations of $G$, then
		\begin{displaymath}
			\chi_{D_{1}\oplus D_{2}}(a)=\chi_{D_{1}}(a)+\chi_{D_{}2}(a).
		\end{displaymath}
		\item $D_{1}\not\sim D_{2}\Rightarrow (\chi_{D_{1}},\chi_{D_{2}})=0.$
		\item $D$ irreducible $\Leftrightarrow (\chi_{D},\chi_{D})=1.$ 
	\end{enumerate}
\end{prop}

\begin{proof}
	\begin{itemize}
		\item[]
		\item[1.] Because of the invariance of the trace. (See definition and proposition \ref{DA55}.)
		\item[2.] See 1.
		\item[3.] Every representation is equivalent to a unitary representation, and equivalent representations have the same characters. Thus we can assume that $D$ is unitary for all calculations involving only characters.
		\smallskip
		\\
		$\Rightarrow D(a)^{\dagger}=D(a^{-1})$ $\Rightarrow \mathrm{Tr}D(a)^{\dagger}=\mathrm{Tr}D(a^{-1})$ $\Rightarrow \chi_{D}^{\ast}(a)=\chi_{D}(a^{-1})$.
		\item[4.] This is easy to see in the matrix representation of $D_{1}\oplus D_{2}$.
		\item[5.] $D_{1}\not\sim D_{2}\Rightarrow $ From definition and proposition \ref{DA54} 2. follows:
		\smallskip
		\\
		$(D^{(1)}_{ij},D^{(2)}_{kl})=0 \quad \forall i,j,k,l.$ $\Rightarrow (D^{(1)}_{ii},D^{(2)}_{kk})=0\Rightarrow (\chi_{D_{1}},\chi_{D_{2}})=0.$
		\item[6.] 
		\begin{itemize}
		\item[]$\Rightarrow$: From definition and proposition \ref{DA54} 3. follows:
		\begin{displaymath}
			(D^{(\alpha)}_{ij},D^{(\alpha)}_{kl})=\frac{1}{n_{\alpha}}\delta_{jk}\delta_{il}\Rightarrow (\chi^{(\alpha)},\chi^{(\alpha)})=\frac{1}{n_{\alpha}}\delta_{ik}\delta_{ik}=\frac{n_{\alpha}}{n_{\alpha}}=1.
		\end{displaymath}
		\item[]$\Leftarrow$: Suppose $(\chi_{D},\chi_{D})=1$ with $D$ not irreducible. $\Rightarrow D$ can be decomposed into irreducible representations $D^{(\alpha)}$:
		\begin{displaymath}
			D(a)=\bigoplus_{\alpha=1}^{k}D^{(\alpha)}(a)
		\end{displaymath}
		with $k\geq 2$. It follows:
		\begin{displaymath}
		\begin{split}
			(\chi_{D},\chi_{D}) & =\sum_{\alpha,\beta}(\chi^{(\alpha)},\chi^{(\beta)})=\\
			& =\sum_{\alpha=1}^k \underbrace{(\chi^{(\alpha)},\chi^{(\alpha)})}_{1}+\underbrace{\sum_{\alpha\neq\beta}(\chi^{(\alpha)},\chi^{(\beta)})}_{0}=k\geq 2.
		\end{split}
		\end{displaymath}
		$\Rightarrow$ contradiction to $(\chi_{D},\chi_{D})=1$.
		\end{itemize}
		
	\end{itemize}
\end{proof}

\begin{theorem}\label{TA59}
	\textbf{Orthogonality of characters:} Let $D^{(\alpha)},\enspace \alpha\in\lbrace 1,...,k\rbrace$ be all non-equivalent irreducible representations of a finite group $G$ on vectorspaces $V_{\alpha}$ ($\mathrm{dim}V_{\alpha}=n_{\alpha}$), then
	\begin{displaymath}
		(\chi^{(\alpha)},\chi^{(\beta)})=\delta_{\alpha\beta}.
	\end{displaymath}
\end{theorem}

\begin{proof}
	From definition and proposition \ref{DA54}:
	\begin{displaymath}
		(D^{(\alpha)}_{ij},D^{(\beta)}_{kl})=\frac{1}{n_{\alpha}}\delta_{\alpha\beta}\delta_{jk}\delta_{il}.
	\end{displaymath}
	Forming the trace leads to:
	\begin{displaymath}
		(\chi^{(\alpha)},\chi^{(\beta)})=(D^{(\alpha)}_{ii},D^{(\beta)}_{kk})=\frac{1}{n_{\alpha}}\delta_{\alpha\beta}\delta_{jk}\delta_{jk}=\delta_{\alpha\beta}.
	\end{displaymath}
\end{proof}

\begin{prop}\label{PA59a}
Using the orthogonality relation for characters one can easily calculate how often an irreducible representation is contained in a reducible representation.
\smallskip
\\
Let $D$ be an arbitrary representation. $D$ can be decomposed into a direct sum of non-equivalent irreducible representations.
	\begin{displaymath}
		D(a)=\bigoplus_{\alpha}b_{\alpha}D^{(\alpha)}(a) \quad\quad \chi_{D}(a)=b_{\alpha}\chi^{(\alpha)}(a)
	\end{displaymath}
	with $b_{\alpha}\in \mathbb{N}$. $\alpha$ labels the non-equivalent irreducible representations.
	\smallskip
	\\
	Using the orthogonality relation for characters (theorem \ref{TA59}) one gets
	\begin{displaymath}
		b_{\alpha}=(\chi^{(\alpha)},\chi_{D}).
	\end{displaymath}
	\qed
\end{prop}

\begin{lemma}\label{LA60}
	Let $D^{(\alpha)},\enspace \alpha\in\lbrace 1,...,k\rbrace$ be all non-equivalent irreducible representations of a finite group $G$, and let $C_{i}$ ($i\in \lbrace 1,...,c\rbrace$) be the conjugate classes of $G$. Then
	\begin{displaymath}
		k\leq c.
	\end{displaymath}
\end{lemma}

\begin{proof}
	Let $c_{i}$ be the number of elements in $C_{i}$.
	\smallskip
	\\
	\begin{displaymath}
		\sum_{a\in C_{i}}\chi^{(\alpha)}(a^{-1})\chi^{(\beta)}(a)=c_{i}\chi^{(\alpha)}(a^{-1})\chi^{(\beta)}(a)=c_{i}\chi^{(\alpha)\ast}(a)\chi^{(\beta)}(a).
	\end{displaymath}
	\begin{displaymath}
		(\chi^{(\alpha)},\chi^{(\beta)})=\frac{1}{\mathrm{ord}(G)}\sum_{i=1}^{c}c_{i}\chi^{(\alpha)\ast}_{i}\chi^{(\beta)}_{i}=\delta_{\alpha\beta},
	\end{displaymath}
with $\chi^{(\alpha)}_{i}=\chi^{(\alpha)}(a_{i})$, where $a_{i}\in C_{i}$ arbitrary.
$(\sqrt{\frac{c_{i}}{\mathrm{ord}(G)}}\chi^{(\alpha)}_{i})_{i}$ are elements in a $c$-dimensional vectorspace, and vectors with different $\alpha$ are orthogonal and unequal $0$.
\smallskip
\\
$\Rightarrow$ There are $k$ orthogonal vectors in a $c$-dimensional vectorspace. 
\begin{displaymath}
	\Rightarrow k\leq c.
\end{displaymath}
\end{proof}

\section{General theorems}

\begin{define}\label{DA63}
	Let $G=\lbrace a_{1},...,a_{g}\rbrace$ be a group of order $g$. The \textit{regular representation} $R(G)$ on $V=\mathbb{R}^{g}$ is defined by
	\begin{displaymath}
		a_{i}a_{j}=R_{jk}(a_{i})a_{k}.
	\end{displaymath}
	By corollary \ref{CA12a} $a_{i}$ can be interpreted as a regular permutation, and therefore $R(G)$ is called regular representation.
\end{define}

\begin{prop}\label{PA64}
	\begin{displaymath}
		\chi_{R}(a)=\left\{
			\begin{array}{ll}
				0 & \mbox{for}\enspace a\neq e. \\ 
				g & \mbox{for}\enspace a=e.
		        \end{array}
			\right.
	\end{displaymath}
\end{prop}

\begin{proof}
	\begin{displaymath}
		a_{i}a_{j}=R_{jk}(a_{i})a_{k}=a_{p_{j}}\Rightarrow R_{jk}(a_{i})=\delta_{p_{j}k}.
	\end{displaymath}
We form the trace:
	\begin{displaymath}
		\mathrm{Tr}(R(a_{i}))=R_{jj}(a_{i})=\delta_{p_{j}j}.
	\end{displaymath}
Because the permutation $a_{j}\mapsto a_{i}a_{j}=a_{p_{j}}$ is regular we find
\begin{displaymath}
	\delta_{p_{j}j}=\left\{
			\begin{array}{ll}
				0 & \mbox{for}\enspace a_{i}\neq e. \\ 
				g & \mbox{for}\enspace a_{i}=e.
		        \end{array}
			\right.
\end{displaymath}
\end{proof}

\begin{theorem}\label{TA65}
	Let $D^{(\alpha)}$ be the non-equivalent irreducible representations of a group $G$ with order $g$ on vectorspaces $V_{\alpha}$  ($\mathrm{dim}V_{\alpha}=n_{\alpha}$). Then
	\begin{displaymath}
		\sum_{\alpha}n_{\alpha}^{2}=g.
	\end{displaymath}
\end{theorem}

\begin{proof}
	Let $C_{i}$ ($i\in \lbrace 1,...,c\rbrace$) be the conjugate classes of $G$.
	\\
	We decompose the regular representation into non-equivalent irreducible representations $\Rightarrow$
	\begin{displaymath}
		\chi_{R}(a_{i})=\sum_{\alpha}b_{\alpha}\chi^{(\alpha)}(a_{i}),
	\end{displaymath}
	where $a_{i}\in C_{i}$ arbitrary.
	\medskip
	\\
	Consider now the conjugate class containing $e$:
	\begin{displaymath}
		\chi_{R}(e)=g \quad\quad\quad \chi^{(\alpha)}(e)=\mathrm{Tr}D^{(\alpha)}(e)=\mathrm{dim}V_{\alpha}=n_{\alpha}.
	\end{displaymath}
	\begin{displaymath}
		\chi_{R}(e)=\sum_{\alpha}b_{\alpha}\chi^{(\alpha)}(e)\Rightarrow g=\sum_{\alpha}b_{\alpha}n_{\alpha}.
	\end{displaymath}
	By proposition \ref{PA59a}
	\begin{displaymath}
		b_{\alpha}=(\chi^{(\alpha)},\chi_{R})=\frac{1}{g}\sum_{a\in G}\chi^{(\alpha)}(a^{-1})\chi_{R}(a)=\frac{1}{g}\chi^{(\alpha)}(e)\chi_{R}(e)=\chi^{(\alpha)}(e)=n_{\alpha}.
	\end{displaymath}
It follows
	\begin{displaymath}
		g=\sum_{\alpha}b_{\alpha}n_{\alpha}=\sum_{\alpha}n_{\alpha}^{2}.
	\end{displaymath}
\end{proof}

\begin{cor}\label{CA66}
Let $D^{(\alpha)}$ be the non-equivalent irreducible representations of a group $G$ on vectorspaces $V_{\alpha}$  ($\mathrm{dim}V_{\alpha}=n_{\alpha}$), and let $R(G)$ be the regular representation of $G$. Then
	\begin{displaymath}
		R=\bigoplus_{\alpha}n_{\alpha}D^{(\alpha)}.
	\end{displaymath}
	\qed
\end{cor}

\begin{define}\label{DA67}
	Let $G$ be a finite group. The \textit{group algebra} $\mathcal{A}(G)$ is the set
	\begin{displaymath}
		\mathcal{A}(G):=\lbrace A=\sum_{b\in G}A_{b}b \vert A_{b}\in\mathbb{C}\rbrace.
	\end{displaymath}
	together with the compositions
	\begin{itemize}
		\item[] Addition: 
		\begin{displaymath}
			A+B=\sum_{b\in G}A_{b}b+\sum_{b\in G}B_{b}b:=\sum_{b\in G}(A_{b}+B_{b})b.
		\end{displaymath}
		\item[] Scalar multiplication:
		\begin{displaymath}
			\lambda A=\lambda\sum_{b\in G}A_{b}b:=\sum_{b\in G}(\lambda A_{b})b.
		\end{displaymath}
		\item[] Multiplication: 
		\begin{displaymath}
			AB=(\sum_{b\in G}A_{b}b)(\sum_{d\in G}B_{d}d):=\sum_{b,d\in G}(A_{b}B_{d})bd.
		\end{displaymath}
	\end{itemize}
With addition and scalar multiplication $\mathcal{A}(G)$ becomes a vectorspace over $\mathbb{C}$. Including the multiplication the vectorspace $\mathcal{A}(G)$ becomes an \textit{algebra} over $\mathbb{C}$.
\end{define}

\begin{define}\label{DA68}
	Let $C_{i}=\lbrace a^{i}_{1},...,a^{i}_{c_{i}}\rbrace$ ($i\in \lbrace 1,...,c\rbrace$, $\sum_{i}c_{i}=\mathrm{ord}(G)$) be the conjugate classes of a finite group $G$. For every conjugate class $C_{i}$ we define a corresponding element of the group algebra.
	\begin{displaymath}
		\mathcal{C}_{i}:=\sum_{j=1}^{c_{i}}a^{i}_{j}
	\end{displaymath}
\end{define}

\begin{defprop}\label{PA69}%Name "structure constants" wird nur in Chen erwähnt.
	\begin{displaymath}
		\mathcal{C}_{i}\mathcal{C}_{j}=c_{ijk}\mathcal{C}_{k}
	\end{displaymath}
with $c_{ijk}\in \mathbb{N}$. $c_{ijk}$ are called \textit{structure constants} of the finite group $G$.
\end{defprop}

\begin{proof}
	The proposition follows if $\mathcal{C}_{i}\mathcal{C}_{j}$ is a sum of all elements of complete conjugate classes. This is the case if
	\begin{displaymath}
		a(\mathcal{C}_{i}\mathcal{C}_{j})a^{-1}=\mathcal{C}_{i}\mathcal{C}_{j} \quad\forall a\in G.
	\end{displaymath}
	This immediately follows from the definition of a conjugate class:
	\begin{displaymath}
		a\mathcal{C}_{i}a^{-1}=\mathcal{C}_{i}\Rightarrow a(\mathcal{C}_{i}\mathcal{C}_{j})a^{-1}=(a\mathcal{C}_{i}a^{-1})(a\mathcal{C}_{j}a^{-1})=\mathcal{C}_{i}\mathcal{C}_{j}.
	\end{displaymath}
\end{proof}

\begin{sublemma}\label{SA70}
	Let $D$ be an $n$-dimensional irreducible representation of a finite group $G$ with conjugate classes $C_{i}=\lbrace a^{i}_{1},...,a^{i}_{c_{i}}\rbrace$ ($i\in \lbrace 1,...,c\rbrace$, $\sum_{i}c_{i}=\mathrm{ord}(G)$), and let $\chi_{i}=\chi_{D}(a^{i}_{j})$. Then
	\begin{displaymath}
		c_{i}\chi_{i}c_{j}\chi_{j}=n\sum_{l}c_{ijl}c_{l}\chi_{l}.
	\end{displaymath}
In this sublemma (and the proof of the sublemma) we do \underline{not} use Einstein's summation convention!
\end{sublemma}

\begin{proof}
	We define the linear operators $D_{i}$:
	\begin{displaymath}
		D_{i}:=\sum_{a\in C_{i}}D(a).
	\end{displaymath}
	Analogous to $\mathcal{C}_{i}$ the operators $D_{i}$ are invariant under $D_{i}\mapsto D(a)D_{i}D(a)^{-1}$.
	\begin{displaymath}
		\Rightarrow [D_{i},D(a)]=0 \quad \forall a\in G.
	\end{displaymath}
	Using Schur's lemma (lemma \ref{LA50}) we get
	\begin{displaymath}
		D_{i}=\lambda_{i}id.
	\end{displaymath}
	We calculate the trace:
	\begin{displaymath}
		\mathrm{Tr}D_{i}=\sum_{a\in C_{i}}\chi_{D}(a_{i})=\lambda_{i}\mathrm{Tr}(id)=\lambda_{i}n.
	\end{displaymath}
	\begin{displaymath}
		\sum_{a\in C_{i}}\chi_{D}(a_{i})=c_{i}\chi_{i}\Rightarrow c_{i}\chi_{i}=n\lambda_{i}\Rightarrow \lambda_{i}=\frac{c_{i}\chi_{i}}{n}.
	\end{displaymath}
	\begin{displaymath}
		\mathcal{C}_{i}\mathcal{C}_{j}=\sum_{l}c_{ijl}\mathcal{C}_{l}\Rightarrow D_{i}D_{j}=\sum_{l}c_{ijl}D_{l}\Rightarrow \lambda_{i}\lambda_{j}=\sum_{l}c_{ijl}\lambda_{l}.
	\end{displaymath}
Using $\lambda_{i}=\frac{c_{i}\chi_{i}}{n}$ we get
	\begin{displaymath}
		\frac{c_{i}\chi_{i}}{n}\cdot \frac{c_{j}\chi_{j}}{n}=\sum_{l}c_{ijl}\frac{c_{l}\chi_{l}}{n}\Rightarrow c_{i}\chi_{i}c_{j}\chi_{j}=n\sum_{l}c_{ijl}c_{l}\chi_{l}.
	\end{displaymath}
\end{proof}

\begin{sublemma}\label{SA71}
	Let $C_{i}$ be a conjugate class, then
	\begin{displaymath}
		C_{i}^{-1}:=\lbrace a^{-1}\vert a\in C_{i}\rbrace
	\end{displaymath}
	is also a conjugate class.
\end{sublemma}

\begin{proof}
	Let $a,b\in C_{i}$, $a\sim b\Rightarrow$ $\exists u\in G: a=ubu^{-1}.$
	\smallskip
	\\
	$\Rightarrow a^{-1}=u^{-1}b^{-1}u=:vb^{-1}v^{-1}.$
	\smallskip
	\\
	$\Rightarrow \exists v\in G: a^{-1}=vb^{-1}v^{-1}\Rightarrow a^{-1}\sim b^{-1}\quad \forall a,b\in C_{i}$.
	\medskip
	\\
	$\Rightarrow C_{i}^{-1}$ is a conjugate class.
\end{proof}

\begin{sublemma}\label{SA72}
	Let $C_{i}=\lbrace a^{i}_{1},...,a^{i}_{c_{i}}\rbrace$ be the conjugate classes of a finite group $G$, and let $C_{j}^{-1}=:C_{j'}$. Furthermore let $C_{1}:=\lbrace e\rbrace$. Then the following equality holds for the structure constants:
	\begin{displaymath}
		c_{ij1}=c_{i}\delta_{ij'} \quad(\Leftrightarrow c_{ij'1}=c_{i}\delta_{ij}).
	\end{displaymath}
In this sublemma (and its proof) we do \underline{not} use Einstein's summation convention!
\end{sublemma}

\begin{proof}
	\begin{displaymath}
		C_{i}C_{i}^{-1}=\lbrace \underbrace{e,e,...,e}_{c_{i}},... \rbrace\Rightarrow \mathcal{C}_{i}\mathcal{C}_{i}^{-1}=c_{i}e+...=\mathcal{C}_{i}\mathcal{C}_{i'}=c_{ii'1}\mathcal{C}_{1}+...
	\end{displaymath}
	\begin{displaymath}
		\Rightarrow c_{ii'1}=c_{i}.
	\end{displaymath}
$C_{i}C_{j}^{-1}=C_{i}C_{j'}$  does not contain $e$ if $i\neq j$.
	\begin{displaymath}
		\Rightarrow c_{ij'1}=0 \enspace\mbox{if}\enspace i\neq j.
	\end{displaymath}
	\begin{displaymath}
		\Rightarrow c_{ij'1}=c_{i}\delta_{ij}.
	\end{displaymath}
Clearly $j''=j\Rightarrow$ $c_{ij''1}=c_{i}\delta_{ij'}$.
	\begin{displaymath}
		\Rightarrow c_{ij1}=c_{i}\delta_{ij'}.
	\end{displaymath}
\end{proof}

\begin{lemma}\label{LA73}
	Let $D^{(\alpha)},\enspace \alpha\in\lbrace 1,...,k\rbrace$ be all non-equivalent irreducible representations of a group $G$ with order $g$, and let $C_{i}$ ($i\in \lbrace 1,...,c\rbrace$) be the conjugate classes of $G$. Then
	\begin{displaymath}
		k\geq c.
	\end{displaymath}
\end{lemma}

\begin{proof}
	In this proof we do \underline{not} use Einstein's summation convention!
	\\
	Consider $D^{(\alpha)}$ with characters $\chi^{(\alpha)}_{i}$, where $i$ denotes the conjugate class. From sublemma \ref{SA70} we know
	\begin{displaymath}
		c_{i}\chi^{(\alpha)}_{i}c_{j}\chi^{(\alpha)}_{j}=n_{\alpha}\sum_{l}c_{ijl}c_{l}\chi^{(\alpha)}_{l}.
	\end{displaymath}
	We sum over $\alpha$ and get
	\begin{displaymath}
		c_{i}c_{j}\sum_{\alpha=1}^{k}\chi^{(\alpha)}_{i}\chi^{(\alpha)}_{j}=\sum_{l}c_{ijl}c_{l}\sum_{\alpha=1}^{k}n_{\alpha}\chi^{(\alpha)}_{l}.	
	\end{displaymath}
	From corollary \ref{CA66} $\sum_{\alpha=1}^{k}n_{\alpha}\chi^{(\alpha)}_{l}=\chi^{R}_{l}$, where $\chi^{R}$ is the character of the regular representation. If we denote the conjugate class $\lbrace e\rbrace$ with $C_{1}$, using proposition \ref{PA64} we see $\chi^{R}_{l}=g\delta_{1l}$.
	\\
	It follows
	\begin{displaymath}
		c_{i}c_{j}\sum_{\alpha=1}^{k}\chi^{(\alpha)}_{i}\chi^{(\alpha)}_{j}=gc_{ij1}c_{1}=gc_{ij1}.
	\end{displaymath}
	From sublemma \ref{SA72} we know that $c_{ij1}=c_{i}\delta_{ij'}$, so
	\begin{displaymath}
		c_{i}c_{j}\sum_{\alpha=1}^{k}\chi^{(\alpha)}_{i}\chi^{(\alpha)}_{j}=gc_{i}\delta_{ij'}\Rightarrow \sum_{\alpha=1}^{k}\chi^{(\alpha)}_{i}\chi^{(\alpha)}_{j}=\frac{g}{c_{j}}\delta_{ij'}.
	\end{displaymath}
	We now use proposition \ref{PA58} 3.: $\chi(a^{-1})=\chi^{\ast}(a)$. Therefore $\chi^{(\alpha)}_{j}=\chi^{(\alpha)\ast}_{j'}$.
	\begin{displaymath}
		\Rightarrow \sum_{\alpha=1}^{k}\chi^{(\alpha)}_{i}\chi^{(\alpha)\ast}_{j'}=\frac{g}{c_{j}}\delta_{ij'}\Rightarrow \sum_{\alpha=1}^{k}\chi^{(\alpha)}_{i}\chi^{(\alpha)\ast}_{j}=\frac{g}{c_{j}}\delta_{ij},
	\end{displaymath}
where we have also used $c_{j'}=c_{j}$.
\medskip
\\
This has the form of a Hermitian scalar product, thus we can interpret $\chi_{i}$ as orthogonal vectors in a $k$-dimensional unitary vectorspace.
\medskip
\\
$\Rightarrow$ $c$ orthogonal vectors in a $k$-dimensional vectorspace $\Rightarrow k\geq c$. 
\end{proof}

\begin{theorem}\label{TA74}
	The number of non-equivalent irreducible representations of a finite group is equal to the number of conjugate classes of the group, i.s.
	\begin{displaymath}
		k=c.
	\end{displaymath}
\end{theorem}

\begin{proof}
	From lemma \ref{LA60} we know $k\leq c$, from lemma \ref{LA73} we know $k\geq c$.
	\begin{displaymath}
		\Rightarrow k=c.
	\end{displaymath}
\end{proof}

\section{Tensor products of representations}

\begin{define}\label{DA81}
	Let $V_{1}$ and $V_{2}$ be vectorspaces over $\mathbb{R}(\mathbb{C})$. The \textit{tensor product} of $V_{1}$ and $V_{2}$, i.s. $V_{1}\otimes V_{2}$, is defined by
	\begin{displaymath}
		V_{1}\otimes V_{2}:= \mathrm{Span}(\lbrace v_{1}\otimes v_{2}\vert v_{1}\in V_{1}, v_{2}\in V_{2}\rbrace)
	\end{displaymath}
	with
	\begin{displaymath}
		\begin{split}
			& (\alpha v_{1}+\beta w_{1})\otimes v_{2}=\alpha v_{1}\otimes v_{2}+\beta w_{1}\otimes v_{2},\\
			& v_{1}\otimes(\alpha v_{2}+\beta w_{2})=\alpha v_{1}\otimes v_{2}+\beta v_{1}\otimes w_{2}.
		\end{split}
	\end{displaymath}
If $V_{1}$ and $V_{2}$ are Euclidian (unitary) we furthermore define
	\begin{displaymath}
		\langle v_{1}\otimes v_{2},w_{1}\otimes w_{2}\rangle=\langle v_{1},w_{1}\rangle\cdot \langle v_{2},w_{2}\rangle.
	\end{displaymath}
\end{define}

\begin{define}\label{DA82}
	Let $A_{1}$ and $A_{2}$ be linear operators on $V_{1}$ and $V_{2}$. We define the tensor product of the operators $A_{1}$ and $A_{2}$, i.s. $A_{1}\otimes A_{2}$, by
	\begin{displaymath}
		(A_{1}\otimes A_{2})(a_{ij}v_{i}\otimes v_{j})=a_{ij}(A_{1}v_{i})\otimes(A_{2}v_{j}) \quad\forall v=a_{ij}v_{i}\otimes v_{j}\in V_{1}\otimes V_{2}.
	\end{displaymath}
$\Rightarrow A_{1}\otimes A_{2}$ is a linear operator on $V_{1}\otimes V_{2}$.
\end{define}
\bigskip
\bigskip
\hspace{0mm}\\
Let $A_{1}$ and $A_{2}$ be linear operators on $V_{1}$ and $V_{2}$ ($\mathrm{dim}V_{1}=n_{1}, \mathrm{dim} V_{2}=n_{2}$), and let $\lbrace v_{i}\rbrace_{i}$ be a basis of $V_{1}$ and $\lbrace w_{j}\rbrace_{j}$ be a basis of $V_{2}$.
	\begin{displaymath}
		\begin{split}
			\Rightarrow & A_{1}v_{i}=A_{1ji}v_{j},\\
			& A_{2}w_{k}=A_{2lk}w_{l},
		\end{split}
	\end{displaymath}
where $[A_{1ij}]$ and $[A_{2kl}]$ are the matrix representations of $A_{1}$ and $A_{2}$ w.r.t. the bases $\lbrace v_{i}\rbrace_{i}$ and $\lbrace w_{j}\rbrace_{j}$.
\	\begin{displaymath}
 		(A_{1}\otimes A_{2})(v_{i}\otimes w_{k})=(A_{1}v_{i})\otimes(A_{2}w_{k})=A_{1ji}A_{2lk}(v_{j}\otimes w_{l}).
 	\end{displaymath}
$\Rightarrow [A_{1ij}A_{2kl}]$ is the matrix representation of $A_{1}\otimes A_{2}$ in the basis $\lbrace v_{j}\otimes w_{l}\rbrace_{jl}$.
\\
If we interpret the components $a_{ij}$ of vectors in $V_{1}\otimes V_{2}$ as $n_{1}n_{2}$-column-vectors
	\begin{displaymath}
		v=a_{ij}v_{i}\otimes v_{j}\rightarrow [v]=
			\left(
			% use packages: array
			\begin{array}{c}
				a_{11} \\ 
				a_{12} \\ 
				\vdots \\ 
				a_{1n_{2}} \\ 
				a_{21} \\ 
				\vdots \\ 
				a_{n_{1}n_{2}}            \end{array}
				\right),
	\end{displaymath}
we can interpret the matrix representation of $A_{1}\otimes A_{2}$ as the $(n_{1}n_{2})\cdot(n_{1}n_{2})$-matrix
	\begin{displaymath}
	% use packages: array
	[A_{1}\otimes A_{2}]=
	\left(
	\begin{array}{ccc}
	A_{111}\cdot [A_{2ij}] & ... & A_{11n_{1}}\cdot [A_{2ij}] \\ 
	\vdots & × & \vdots \\ 
	A_{1n_{1}1}\cdot [A_{2ij}] & ... & A_{1n_{1}n_{1}}\cdot [A_{2ij}]
	\end{array}
	\right).
	\end{displaymath}

\begin{define}\label{DA83}
	The matrix $[A_{1}\otimes A_{2}]$ is called \textit{Kronecker product} of the matrices $[A_{1}]$ and $[A_{2}]$.
	\medskip
	\\
	$\Rightarrow$ The Kronecker product is the matrix representation of the tensor product.
\end{define}

\begin{define}\label{DA84}
	Let $D_{1}$ and $D_{2}$ be representations of a group $G$ on vectorspaces $V_{1}$ and $V_{2}$ ($\mathrm{dim}V_{1}=n_{1}$, $\mathrm{dim}V_{2}=n_{2}$). The \textit{tensor product} (Kronecker product) \textit{of the representations} $D_{1}$ and $D_{2}$, i.s. $D_{1}\otimes D_{2}$, is defined as the tensor product of the linear operators $D_{1}$ and $D_{2}$, therefore $D_{1}\otimes D_{2}$ is an $n_{1}n_{2}$-dimensional representation of $G$ on $V_{1}\otimes V_{2}$.
\end{define}

\begin{prop}\label{PA85}
	\begin{displaymath}
		\chi_{D_{1}\otimes D_{2}}=\chi_{D_{1}}\cdot\chi_{D_{2}}.
	\end{displaymath}
\end{prop}

\begin{proof}
	We calculate the character of $D_{1}\otimes D_{2}$ using the matrix representation:
	\begin{displaymath}
		\begin{split}
			\mathrm{Tr}[D_{1}\otimes D_{2}] & = D_{111}\cdot \mathrm{Tr}[D_{2ij}]+...+D_{1n_{1}n_{1}}\cdot \mathrm{Tr}[D_{2ij}] =\\
			& = \mathrm{Tr}[D_{1kl}]\mathrm{Tr}[D_{2ij}]
		\end{split}
	\end{displaymath}
	\begin{displaymath}
		\Rightarrow \mathrm{Tr}(D_{1}\otimes D_{2})=\mathrm{Tr}D_{1}\cdot \mathrm{Tr}D_{2}
		\Rightarrow \chi_{D_{1}\otimes D_{2}}=\chi_{D_{1}}\cdot\chi_{D_{2}}.
	\end{displaymath}
\end{proof}

\begin{prop}\label{PA86}
	Let $D$ be a representation of a group $G$ on a vectorspace $V$ ($\mathrm{dim}V=n>1$), then the product $D\otimes D$ is reducible (even if $D$ is irreducible).
\end{prop}

\begin{proof}
	Let $\lbrace e_{j}\rbrace_{j}$ be a basis of $V$ ($j=1,...,n$). $\Rightarrow \lbrace e_{i}\otimes e_{j}\rbrace_{ij}$ is a basis of $V\otimes V$.
	\begin{displaymath}
		\begin{split}
			& \Rightarrow (D\otimes D)(e_{i}\otimes e_{j})=\underbrace{D_{ki}D_{lj}}_{(D\otimes D)_{ki,lj}}(e_{k}\otimes e_{l}).\\
			& \Rightarrow (D\otimes D)(e_{j}\otimes e_{i})=(D\otimes D)_{kj,li}(e_{k}\otimes e_{l}).
		\end{split}
	\end{displaymath}
Calculating the sum and the difference of these equations we get
	\begin{displaymath}
		(D\otimes D)(\frac{1}{2}(e_{i}\otimes e_{j}\pm e_{j}\otimes e_{i}))=\frac{1}{2}((D\otimes D)_{ki,lj}\pm (D\otimes D)_{kj,li})(\frac{1}{2}(e_{k}\otimes e_{l}\pm e_{l}\otimes e_{k})),
	\end{displaymath}
where we have used that $A_{ij}B_{ij}=A_{ij}B_{(ij)}$ for symmetric $A$.
\medskip
\\
We conclude that
	\begin{displaymath}
		\mathrm{Span}(\frac{1}{2}(e_{k}\otimes e_{l}+e_{l}\otimes e_{k}))=\mathrm{Span}(e_{(k}\otimes e_{l)})
	\end{displaymath}
and
	\begin{displaymath}
		\mathrm{Span}(\frac{1}{2}(e_{k}\otimes e_{l}-e_{l}\otimes e_{k}))=\mathrm{Span}(e_{[k}\otimes e_{l]})
	\end{displaymath}
are invariant subspaces of $V\otimes V$. $\Rightarrow$ $D\otimes D$ is reducible.
\end{proof}

\section{The Clebsch-Gordan-series}\label{ACGC}

\begin{define}\label{DA87}
	Let $D$ be a representation of a finite group $G$. The representation $D^{\ast}$ which is defined by
	\begin{displaymath}
		[D^{\ast}]=[D]^{\ast},
	\end{displaymath}
	is called \textit{complex conjugate representation} of $D$.
	\medskip
	\\
	The representation $\bar{D}$, which is defined by
	\begin{displaymath}
		[\bar{D}]=[D^{-1}]^{T},
	\end{displaymath}
	is called \textit{adjoint representation} of $D$. 
\end{define}
\hspace{0mm}\\
Remark: For finite groups $D^{\ast}$ and $(D^{-1})^{T}$ are equivalent, because $D$ is equivalent to a unitary representation.

\begin{prop}\label{PA88}
	\begin{itemize}
		\item[]
		\item[1.] $\chi_{\bar{D}}(a)=\chi_{D}(a^{-1}).$
		\item[2.] $\chi_{D^{\ast}}=\chi_{D}^{\ast}.$
	\end{itemize}
\end{prop}

\begin{proof}
	\begin{itemize}
		\item[]
		\item[1.] $\chi_{\bar{D}}(a)=\mathrm{Tr}([D(a)^{-1}]^{T})=\mathrm{Tr}([D(a)^{-1}])=\mathrm{Tr}([D(a^{-1})])=\chi_{D}(a^{-1}).$
		\item[2.]$\chi_{D^{\ast}}=\mathrm{Tr}(D^{\ast})=\mathrm{Tr}([D]^{\ast})=\mathrm{Tr}(D)^{\ast}=\chi_{D}^{\ast}.$
	\end{itemize}
\end{proof}

\begin{define}\label{DA89}
	Let $D^{(\alpha)}$ be the non-equivalent irreducible representations of a finite group $G$. In general the tensor product
		\begin{displaymath}
			D^{(\alpha)}\otimes D^{(\beta)}
		\end{displaymath}
	is not irreducible, thus
	\begin{displaymath}
		D^{(\alpha)}\otimes D^{(\beta)}=\bigoplus_{\gamma}a_{\gamma}D^{(\gamma)}.	
	\end{displaymath}
	This sum is called \textit{Clebsch-Gordan-series}.
\end{define}

\begin{define}\label{DA90}
	Let $D^{(\mu)}(G)$ and $D^{(\nu)}(G)$ be irreducible representations of a finite group $G$ on vectorspaces $V^{(\mu)}$ and $V^{(\nu)}$ ($\mathrm{dim}V^{(\mu)}=n_{\mu}$, $\mathrm{dim}V^{(\nu)}=n_{\nu}$).
	\smallskip
	\\
	Let $\lbrace v_{j}\rbrace_{j=1,...,n_{\mu}}$ be a basis of $V^{(\mu)}$ and $\lbrace w_{k}\rbrace_{k=1,...,n_{\nu}}$ a basis of $V^{(\nu)}$.
	\smallskip
	\\
	Consider the Clebsch-Gordan decomposition
	\begin{displaymath}
		D^{(\mu)}\otimes D^{(\nu)}=\bigoplus_{\lambda}a_{\lambda}D^{(\lambda)},
	\end{displaymath}
where $D^{(\lambda)}$ are the non-equivalent irreducible representations of $G$.
\smallskip
\\
$\lbrace v_{j}\otimes w_{k}\rbrace_{j,k}$ is a basis of $V^{(\mu)}\otimes V^{(\nu)}$. Let $\lbrace u^{(\lambda)}_{i}\rbrace_{i=1,...,n_{\lambda}}$ be a basis of $V^{(\lambda)}$, then there exist complex numbers $C^{\lambda}_{ijl}$, such that
	\begin{displaymath}
		u^{(\lambda)}_{i}=C^{\lambda}_{ijl}(v_{j}\otimes w_{l}).
	\end{displaymath}
$C^{\lambda}_{ijl}$ are called \textit{Clebsch-Gordan coefficients}. Clearly they are basis-dependent, since they form the matrix representation of a basis-change transformation from the basis of $V^{(\mu)}\otimes V^{(\nu)}$ to bases of the invariant vectorspaces $V^{(\lambda)}$.
\end{define}

\begin{prop}\label{PA91}
Let $D=[D^{(\mu)}\otimes D^{(\nu)}]$ be the matrix representation of $D^{(\mu)}\otimes D^{(\nu)}$ in the basis $\lbrace v_{i}\otimes w_{j}\rbrace_{ij}$, and let $M$ be the matrix of basis transformation to the basis $\lbrace u^{(\lambda)}_{i}\rbrace_{\lambda,i}$ of $V^{(\mu)}\otimes V^{(\nu)}$ (the elements of $M$ are the Clebsch-Gordan coefficients), then
	\begin{displaymath}
		M^{-1}DM=\left(
		\begin{matrix}
		 [D^{(1)}] & \textbf{0} & \textbf{0} & \textbf{0} \\
		 \textbf{0} & [D^{(2)}] & \textbf{0} & \textbf{0} \\
		 \textbf{0} & \textbf{0} & \ddots & \textbf{0} \\
		 \textbf{0} & \textbf{0} & \textbf{0} & [D^{(p)}]
		\end{matrix}
		\right),
	\end{displaymath}
where $[D^{(j)}]$ are matrix representations of the irreducible representations $D^{(\lambda)}$ of $G$ ($j\in \lbrace 1,...,k \rbrace$), $p=\sum_{\lambda}a_{\lambda}$ , and $\textbf{0}$ are appropriate null matrices.
\end{prop}

\begin{proof}
	This is clear, since $D^{(\mu)}\otimes D^{(\nu)}=\bigoplus_{\lambda}a_{\lambda}D^{(\lambda)}$, and $\lbrace u^{(\lambda)}_{i}\rbrace_{i}$ is a basis of the invariant subspace $V^{(\lambda)}$.
\end{proof}

\begin{example}
	$D^{(\mu)}\otimes D^{(\nu)}=2D^{(1)}\oplus D^{(2)}$
	\begin{displaymath}
		\Rightarrow M^{-1}DM=\left(
	\begin{matrix}
	 [D^{(1)}] & \textbf{0} & \textbf{0} \\
	 \textbf{0} & [D^{(1)}] & \textbf{0} \\
	 \textbf{0} & \textbf{0} & [D^{(2)}]
	\end{matrix}
	\right).
	\end{displaymath}
\end{example}

%% file: lie/lie.tex
\chapter{The Lie groups $SU(2)$, $SO(3)$ and $SL(2,\mathbb{C})$}\label{lieappendix}

\section{Definition and properties of $SU(2)$ and $SO(3)$}

\begin{define}\label{DSU2SU31}
	\begin{displaymath}
	\begin{split}
	& SU(2):=\{A\in \mathrm{Mat}(2,\mathbb{C})\vert A^{\dagger}=A^{-1}, \mathrm{det}A=1\}
	\\ &
	SO(3):=\{A\in \mathrm{Mat}(3,\mathbb{R})\vert A^{T}=A^{-1}, \mathrm{det}A=1\}
	\end{split}
	\end{displaymath}
where $\mathrm{Mat}(n,\mathbb{F})$ denotes the set of $n\times n$-matrices over the field $\mathbb{F}$.
\end{define}

\begin{prop}\label{PSU2SU32}
	\textbf{Properties of $SU(2)$ and $SO(3)$.} Let $\alpha\in \mathbb{R}$, $\vec{n}\in \mathbb{R}^{3}, \|\vec{n}\|=1$.

	\begin{enumerate}
		\item Every $U\in SU(2)$ can be written as
		\begin{equation}
		U(\alpha,\vec{n})=e^{-i\alpha\vec{n}\cdot\frac{\vec{\sigma}}{2}}=\mathrm{cos}\frac{\alpha}{2}\hspace{0.5mm}\mathbbm{1}_{2}-i\mathrm{sin}\frac{\alpha}{2}\hspace{0.5mm}\vec{n}\cdot\vec{\sigma}.
		\end{equation}
		\item Every $R\in SO(3)$ can be written as
		\begin{equation}
		R(\alpha,\vec{n})=e^{-i\alpha\vec{n}\cdot\vec{T}}=\mathrm{cos}\alpha\hspace{0.5mm} \mathbbm{1}_{3}+(1-\mathrm{cos}\alpha)\vec{n}\vec{n}^{T}-i\mathrm{sin}\alpha \hspace{0.5mm}\vec{n}\cdot\vec{T}
		\end{equation}
		with $(T_{j})_{kl}=\frac{1}{i}\varepsilon_{jkl}$.
		\end{enumerate}
\end{prop}

\begin{proof}
	The statements of this proposition can be easily proved using the relations between matrix Lie groups and their Lie algebras.
	\begin{enumerate}
		\item $\{-i\frac{\sigma^{j}}{2}\}_{j=1,2,3}$ is a basis of the Lie algebra
		\begin{displaymath}
		su(2)=\{A\in \mathrm{Mat}(2,\mathbb{C})\vert A^{\dagger}=-A, \mathrm{Tr}A=0\},
		\end{displaymath}
		thus every $U\in SU(2)$ can be written as $U(\alpha,\vec{n})=e^{-i\alpha\vec{n}\cdot\frac{\vec{\sigma}}{2}}$.
		\medskip
		\\
		Using $(\vec{n}\cdot\vec{\sigma})^{2k}=\mathbbm{1}_{2}$ we find
		\begin{displaymath}
		\begin{split}
			&e^{-i\alpha\vec{n}\cdot\frac{\vec{\sigma}}{2}}=
			\sum_{n=0}^{\infty}\frac{1}{n!}(-\frac{i\alpha}{2}\vec{n}\cdot\vec{\sigma})^{n}=
			\\&
			=\sum_{k=0}^{\infty}\frac{(-1)^{k}}{(2k)!}\left(\frac{\alpha}{2}\right)^{2k}\underbrace{(\vec{n}\cdot\vec{\sigma})^{2k}}_{\mathbbm{1}_{2}}-i\sum_{k=0}^{\infty}\frac{(-1)^{k}}{(2k+1)!}\left(\frac{\alpha}{2}\right)^{2k+1}\underbrace{(\vec{n}\cdot\vec{\sigma})^{2k+1}}_{\vec{n}\cdot\vec{\sigma}}=
			\\&=
			\mathrm{cos}\left(\frac{\alpha}{2}\right)\mathbbm{1}_{2}-i\mathrm{sin}\left(\frac{\alpha}{2}\right)\vec{n}\cdot\vec{\sigma}.
		\end{split}
		\end{displaymath}
		\item $\{iT^{j}\}_{j=1,2,3}$ is a basis of the Lie algebra
		\begin{displaymath}
		so(3)=\{A\in \mathrm{Mat}(3,\mathbb{R})\vert A^{T}=-A, \mathrm{Tr}A=0\},
		\end{displaymath}
		thus every $R\in SO(3)$ can be written as $R(\alpha,\vec{n})=e^{-i\alpha\vec{n}\cdot\vec{T}}$.
		\medskip
		\\
		$(\vec{n}\cdot\vec{T})$ fulfils:
		\begin{displaymath}
		\begin{split}
			&(\vec{n}\cdot\vec{T})^{0}=\mathbbm{1}_{3},\\
			&(\vec{n}\cdot\vec{T})^{2k}=\mathbbm{1}_{3}-\vec{n}\vec{n}^{T}\quad (k>0),\\
			&(\vec{n}\cdot\vec{T})^{2k+1}=\vec{n}\cdot\vec{T}\quad (k\geq0).
		\end{split}
		\end{displaymath}
		Performing the same series expansion as in 1. we find
		\begin{displaymath}
		e^{-i\alpha\vec{n}\cdot\vec{T}}=\mathrm{cos}\alpha\hspace{0.5mm} \mathbbm{1}_{3}+(1-\mathrm{cos}\alpha)\vec{n}\vec{n}^{T}-i\mathrm{sin}\alpha \hspace{0.5mm}\vec{n}\cdot\vec{T}.
		\end{displaymath}
	\end{enumerate}
\end{proof}
\hspace{0mm}\\
\textbf{The Lie group homomorphism between $SU(2)$ and $SO(3)$}

\begin{lemma}\label{LSU2SU33}
	Let $g,h$ be Lie algebras of the Lie groups $G$ and $H$, and let $\gamma:g\rightarrow h$ be a Lie algebra homomorphism. If $G$ is simply connected, then there exists a unique Lie group homomorphism $\phi:G\rightarrow H$ such that $d\phi=\gamma$. 
\end{lemma}
\hspace{0mm}\\
The proof of this lemma needs good knowledge about Lie groups and Lie algebras. It can be found in \cite{warner} (p.101).

\begin{theorem}\label{TSU2SU34}
	\begin{gather*}
	\phi:SU(2)\rightarrow SO(3)
	\\
	U(\alpha,\vec{n})\mapsto R(\alpha,\vec{n})
	\end{gather*}
is a group homomorphism.
\end{theorem}

\begin{proof}
	Because $SU(2)$ is simply connected we only need to show that
	\begin{gather*}
	d\phi: su(2)\rightarrow so(3)
	\\
	-i\frac{\sigma^{j}}{2}\mapsto -iT^{j}
	\end{gather*}
	is a Lie algebra homomorphism. Since $d\phi$ is linear, and $-i\frac{\sigma^{j}}{2}$ and $-iT^{j}$ both fulfil the commutation relations (Lie bracket)
	\begin{displaymath}
	[-iT^{j},-iT^{k}]=\varepsilon_{jkl}(-iT^{l}),
	\end{displaymath}
	$d\phi$ is a Lie algebra homomorphism.
\end{proof}
\hspace{0mm}\\
One usually calls $\phi$ a \textquotedblleft 2 to 1\textquotedblright\hspace{1mm} homomorphism. This comes from
	\begin{displaymath}
	\phi(-U)=\phi(U),
	\end{displaymath}
which is true, because $-\mathbbm{1}_{2}=U(2\pi,\vec{n})$, and $R(\alpha+2\pi,\vec{n})=R(\alpha,\vec{n})$.

\section{The group $SL(2,\mathbb{C})$}

\begin{define}\label{DSU2SU35}
	\begin{displaymath}
	SL(2,\mathbb{C}):=\{ A\in \mathrm{Mat}(2,\mathbb{C})\vert \mathrm{det}A=1\}.
	\end{displaymath}
We furthermore define
	\begin{displaymath}
	\sigma^{0}=\mathbbm{1}_{2},\quad (\sigma^{\mu})=(\mathbbm{1},\vec{\sigma}),
	\end{displaymath}
where $\sigma^{i}$ ($i=1,2,3$) are the Pauli matrices.
	\begin{displaymath}
	\sigma_{\mu}:=\eta_{\mu\nu}\sigma^{\nu} \Rightarrow (\sigma_{\mu})=(\mathbbm{1},-\vec{\sigma}).
	\end{displaymath}
\end{define}

\begin{theorem}\label{TSU2SU36}
There exists a \textquotedblleft 2 to 1\textquotedblright-homomorphism
	\begin{displaymath}
	\phi: SL(2,\mathbb{C})\rightarrow \mathcal{L}_{+}^{\uparrow},
	\end{displaymath}
where $\mathcal{L}_{+}^{\uparrow}$ is the group of \textit{proper orthochronous Lorentz-transformations}.
\end{theorem}

\begin{proof}
Let $x\in \mathbb{R}^{4}:\quad$ $X:=x^{\mu}\sigma_{\mu}=
\left(\begin{matrix}
 x^{0}-x^{3} & -x^{1}+ix^{2} \\
 -x^{1}-ix^{2} & x^{0}+x^{3}
\end{matrix}\right)=X^{\dagger}.
$
\\
Consider now the mapping
	\begin{displaymath}
	X\mapsto AXA^{\dagger}=X'=x'\hspace{0mm}^{\mu}\sigma_{\mu},\quad A\in SL(2,\mathbb{C}).
	\end{displaymath}
	\begin{displaymath}
	\Rightarrow \mathrm{det}X=(x^{0})^{2}-\vec{x}^{2}=\mathrm{det}X'=(x'\hspace{0mm}^0)^{2}-\vec{x}'\hspace{0mm}^{2}\Rightarrow x'\hspace{0mm}^{\mu}=(L_{A})^{\mu}_{\phantom{\mu}\nu}x^{\nu}.
	\end{displaymath}
$X'$ is again Hermitian, so $(x'\hspace{0mm}^\mu)\in \mathbb{R}^{4}$. Because the mapping leaves the 4-distance $\eta_{\mu\nu}x^{\mu}x^{\nu}=(x^{0})^{2}-\vec{x}^{2}$ invariant, one can interpret $A$ acting on 
\begin{center}
$\{X\in \mathrm{Mat}(2,\mathbb{C})\vert X^{\dagger}=X\}$\end{center}
as a Lorentz transformation $L_{A}$ acting on $\mathbb{R}^{4}$. (Here we have of course interpreted $\mathbb{R}^{4}$ as the 4-dimensional Minkowski-space.)
\medskip
\\
To show that $\phi: A\mapsto L_{A}$ is a group homomorphism we have to show that $AB\mapsto L_{A}L_{B}$:
	\begin{displaymath}
	X\mapsto ABX(AB)^{\dagger}=A\underbrace{BXB^{\dagger}}_{(L_{B})^{\mu}_{\phantom{\mu}\nu}x^{\nu}\sigma_{\mu}}A^{\dagger}=(L_{A})^{\mu}_{\phantom{\mu}\rho}(L_{B})^{\rho}_{\phantom{\rho}\nu}x^{\nu}\sigma_{\mu}=(L_{A}L_{B})^{\mu}_{\phantom{\mu}\nu}x^{\nu}\sigma_{\mu}.
	\end{displaymath}
Since $\mathbbm{1}_{2}\in SL(2,\mathbb{C})$ is mapped onto $\mathbbm{1}_{4}\in \mathcal{L}_{+}^{\uparrow}$, we can also conclude $A^{-1}\mapsto L_{A}^{-1}=L_{A^{-1}}$.
\\
Currently we only know that there exists a homomorphism onto some subgroup of the full Lorentz-group. In fact $\phi$ is a homomorphism onto $\mathcal{L}_{+}^{\uparrow}$ (see \cite{sexl} (p.237) ).
\medskip
\\
In analogy to the group homomorphism $SU(2)\rightarrow SO(3)$,  $\phi:SL(2,\mathbb{C})\rightarrow \mathcal{L}_{+}^{\uparrow}$ is \textquotedblleft 2 to 1\textquotedblright\hspace{1mm} too (see \cite{sexl} (p.237)).
\end{proof}
\hspace{0mm}\\
In nonrelativistic quantum mechanics the transition from the rotation group $SO(3)$ to its universal covering group $SU(2)$ leads to a description of particles with half-integer spin. Here we have the same situation. \textit{The natural laws are not only Lorentz-invariant, they are $SL(2,\mathbb{C})$-invariant}, which leads to the description of half integer spin as a relativistic quantum effect.
\bigskip
\\
We will later need the following important property of $SL(2,\mathbb{C})$:

\begin{prop}\label{PSU3SU2}
Let $\varepsilon:=\left(
\begin{matrix}
 0 & 1 \\
 -1 & 0
\end{matrix}
\right)
$, and let $A\in SL(2,\mathbb{C})$, then
	\begin{displaymath}
	\varepsilon A \varepsilon^{-1}=(A^{-1})^{T}.
	\end{displaymath}
\end{prop}

\begin{proof}
Let $A=\left(
\begin{matrix}
 a & b \\
 c & d
\end{matrix}
\right)$, then
	\begin{displaymath}
	\varepsilon A \varepsilon^{-1}A^{T}=\left(
\begin{matrix}
 ad-bc & 0 \\
 0 & ad-bc
\end{matrix}
\right)=(\mathrm{det}A)\mathbbm{1}_{2}.
	\end{displaymath}
For $A\in SL(2,\mathbb{C})$ we find $\varepsilon A \varepsilon^{-1}A^{T}=\mathbbm{1}_{2}$.
\end{proof}

%% file: majorana/majorana.tex
\chapter{Dirac and Majorana particles}\label{appendixdiracmajorana}

\section{The Majorana equation}

The \textit{Majorana equation} is an $SL(2,\mathbb{C})$-invariant field equation for massive neutral spin-$\frac{1}{2}$-particles. It reads:
	\begin{equation}\label{majoranaequation1}
	i\sigma^{\mu}\partial_{\mu}\bar{\chi}-m\chi=0,
	\end{equation}
where $\chi$ is a so-called 2-spinor ($\chi\in \mathbb{C}^{2}$) and $\bar{\chi}=\varepsilon\chi^{\ast}$. $\varepsilon$ and $\sigma^{\mu}$ are defined as in appendix \ref{lieappendix}.
\medskip
\\
By complex conjugation and multiplication of equation (\ref{majoranaequation1}) with $\varepsilon$ from the left we find
	\begin{displaymath}
	-i\underbrace{\varepsilon \sigma^{\mu\ast}\varepsilon}_{-\bar{\sigma}^{\mu}}\partial_{\mu}\chi-m\bar{\chi}=0,
	\end{displaymath}
where we have defined $(\bar{\sigma}^{\mu}):=(\mathbbm{1}_{2},-\vec{\sigma})$. So we have obtained an equivalent form of the Majorana equation:
	\begin{equation}\label{majoranaequation2}
	i\bar{\sigma}^{\mu}\partial_{\mu}\chi-m\bar{\chi}=0.
	\end{equation}

\begin{theorem}\label{Tmajor1}
The Majorana equations (\ref{majoranaequation1}) and (\ref{majoranaequation2}) are $SL(2,\mathbb{C})$-invariant.
\end{theorem}

\begin{proof}
For the proof we will consider the more general equation 
	\begin{displaymath}
	i\sigma^{\mu}\partial_{\mu}\bar{\varphi}-m\chi=0
	\end{displaymath}
with independent spinors $\varphi$ and $\chi$. If this equation is $SL(2,\mathbb{C})$-invariant, then the Majorana equations (\ref{majoranaequation1}) and (\ref{majoranaequation2}) are $SL(2,\mathbb{C})$-invariant too.
\bigskip
\\
$y:=L_{A}^{-1}x$.
\medskip
\\
$\overline{A\varphi(y)}=\varepsilon (A\varphi(y))^{\ast}=\varepsilon A^{\ast}\varphi^{\ast}(y)=\varepsilon A^{\ast}\varepsilon^{-1}\bar{\varphi}(y)$.
\smallskip
\\
From proposition \ref{PSU3SU2} we know that $\varepsilon A^{\ast}\varepsilon^{-1}=(A^{-1})^{\dagger}$, therefore
	\begin{displaymath}
	\begin{split}
	i\sigma^{\mu}\partial_{\mu}\overline{A\varphi(y)}-mA\chi(y) & = i\sigma^{\mu}\partial_{\mu} (A^{-1})^{\dagger}\bar{\varphi}(y)-mA\chi(y)=\\
	& = i\sigma^{\mu} (A^{-1})^{\dagger} \underbrace{\frac{\partial y^{\nu}}{\partial x^{\mu}}}_{(L_{A}^{-1})^{\nu}_{\phantom{\nu}\mu}}\frac{\partial \bar{\varphi}(y)}{\partial y^{\nu}}-mA\chi(y)=\\
	& = AA^{-1}i\sigma^{\mu}(A^{-1})^{\dagger}(L_{A}^{-1})^{\nu}_{\phantom{\nu}\mu}\frac{\partial \bar{\varphi}(y)}{\partial y^{\nu}}-mA\chi(y).
	\end{split}
	\end{displaymath}
We now use the results of appendix \ref{lieappendix}, where we had found:
	\begin{displaymath}
	Ax^{\mu}\sigma_{\mu}A^{\dagger}=X'=x'\hspace{0mm}^{\mu}\sigma_{\mu}=(L_{A})^{\mu}_{\phantom{\mu}\nu}x^{\nu}\sigma_{\mu}\Rightarrow A\sigma_{\mu}A^{\dagger}=(L_{A})^{\rho}_{\phantom{\rho}\mu}\sigma_{\rho}.
	\end{displaymath}
	\begin{displaymath}
	\Rightarrow Ax_{\mu}\sigma^{\mu}A^{\dagger}=X'=x'_{\mu}\sigma^{\mu}=(L_{A}^{-1})^{\nu}_{\phantom{\nu}\mu}x_{\nu}\sigma^{\mu}\Rightarrow A\sigma^{\mu}A^{\dagger}=(L_{A}^{-1})^{\mu}_{\phantom{\mu}\rho}\sigma^{\rho}.
	\end{displaymath}
We conclude that $A^{-1}\sigma^{\mu}(A^{-1})^{\dagger}=(L_{A})^{\mu}_{\phantom{\mu}\rho}\sigma^{\rho}$. Using this relation we find
	\begin{displaymath}
	i\sigma^{\mu}\partial_{\mu}\overline{A\varphi(y)}-mA\chi(y)= A(\underbrace{i\sigma^{\rho}\frac{\partial}{\partial y^{\rho}}\bar{\varphi}(y)-m\chi(y)}_{0})=0.
	\end{displaymath}
\end{proof}

\section{The Dirac equation}
In the previous section we showed that the Majorana equation is $SL(2,\mathbb{C})$-invariant, but in the proof of theorem \ref{Tmajor1} we already saw that $SL(2,\mathbb{C})$-invariance also holds for a more general equation including two independent spinors $\bar{\varphi}$ and $\chi$:
	\begin{displaymath}
	i\sigma^{\mu}\partial_{\mu}\bar{\varphi}-m\chi=0.
	\end{displaymath}
Combining this equation with another Majorana equation
	\begin{displaymath}
	i\bar{\sigma}^{\mu}\partial_{\mu}\chi-m\bar{\varphi}=0
	\end{displaymath}
we find
	\begin{displaymath}
	\left[i\left(\begin{matrix}
	 \textbf{0} & \sigma^{\mu} \\
	 \bar{\sigma}^{\mu} & \textbf{0}
	\end{matrix}\right)\partial_{\mu}-m\right] \left(\begin{matrix}
	 \chi \\
	 \bar{\varphi}
	\end{matrix}\right)=0.
	\end{displaymath}
	\begin{displaymath}
	\gamma^{\mu}:=\left(\begin{matrix}
	 \textbf{0} & \sigma^{\mu} \\
	 \bar{\sigma}^{\mu} & \textbf{0}
	\end{matrix}\right)
	\end{displaymath}
are the Dirac gamma matrices in the \textit{Weyl basis} or \textit{chiral representation}.
	\begin{displaymath}
	\psi:=\left(\begin{matrix}
		\chi\\ \bar{\varphi}
	\end{matrix}\right)
	\end{displaymath}
is called \textit{Dirac spinor} or 4-spinor. Using these new definitions we find the \textit{Dirac equation}
	\begin{equation}\label{Diracequation}
	(i\gamma^{\mu}\partial_{\mu}-m)\psi=0.
	\end{equation} 
\underline{Remark:} The Dirac equation is Lorentz-invariant ($SL(2,\mathbb{C})$-invariant) for any set of four matrices $\gamma^{\mu}$ that fulfil the anticommutation relation
	\begin{displaymath}
	\{\gamma^{\mu},\gamma^{\nu}\}=2\eta^{\mu\nu}\mathbbm{1}_{4}.
	\end{displaymath}
(For a justification of this statement see for example \cite{PeskinSchroeder} p.40f.) Therefore there are many alternatives to the Weyl basis. Note that one gets the Majorana equations \ref{majoranaequation1} and \ref{majoranaequation2} from the Dirac equation \ref{Diracequation} by setting $\varphi=\chi$ only in the Weyl basis.

\section{The charge conjugation matrix}
The standard procedure for deriving the Dirac equation with an external electromagnetic field is the so called \textit{minimal coupling} via the introduction of the covariant derivative
	\begin{equation}
		\partial_{\mu}\mapsto \partial_{\mu}+ieA_{\mu},
	\end{equation} 
where $A_{\mu}$ is the electromagnetic four-potential. The new Dirac equation
	\begin{equation}\label{diraceq}
		(i\gamma^{\mu}\partial_{\mu}-e\gamma^{\mu}A_{\mu}-m)\psi=0
	\end{equation}
is the fermion field equation of an $U(1)$-invariant theory - namely QED. The aim of this section is to solve the following problem:
	\begin{quote}
		Given a solution $\psi(x)$ of the Dirac equation (\ref{diraceq}), can we construct a solution $\psi^{c}(x)$ of the \textit{charge conjugate} equation
	\end{quote}
			\begin{equation}\label{conjeq}
				(i\gamma^{\mu}\partial_{\mu}+e\gamma^{\mu}A_{\mu}-m)\psi^{c}=0?
			\end{equation}
We adjoin equation (\ref{diraceq}).
	\begin{displaymath}
		\psi^{\dagger}(-i\gamma^{\mu\dagger}\partial_{\mu}-e\gamma^{\mu\dagger}A_{\mu}-m)=0.
	\end{displaymath}
Note that $\partial_{\mu}$ is now a left acting operator. We insert $\gamma^{0}\gamma^{0}=\mathbbm{1}_{4}$ after $\psi^{\dagger}$ and multiply the equation by $\gamma^{0}$ from the right.
	\begin{displaymath}
		\bar{\psi}(-i(\underbrace{\gamma^{0}\gamma^{\mu\dagger}\gamma^{0}}_{\gamma^{\mu}})\partial_{\mu}-e(\underbrace{\gamma^{0}\gamma^{\mu\dagger}\gamma^{0}}_{\gamma^{\mu}})A_{\mu}-m)=0.
	\end{displaymath}
Transposing the equation we get
	\begin{equation}\label{MEQ1}
		(-i(\gamma^{\mu})^{T}\partial_{\mu}-eA_{\mu}(\gamma^{\mu})^{T}-m)\bar{\psi}^{T}=0.
	\end{equation}
We now define the \textit{charge conjugation matrix} $C$ via
	\begin{equation}\label{cmatrix}
		C(\gamma^{\mu})^{T}C^{-1}=-\gamma^{\mu}.
	\end{equation}
$C$ is not uniquely defined, but we only need to know that there exists a matrix that fulfils (\ref{cmatrix}).

\begin{theorem}\label{TB0}
Let $\gamma^{\mu}$ and $\gamma^{\mu}\hspace{0mm}'$ be two sets of gamma matrices fulfilling
	\begin{displaymath}
	\{\gamma^{\mu},\gamma^{\nu}\}=\{\gamma^{\mu}\hspace{0mm}',\gamma^{\nu}\hspace{0mm}' \}=2\eta^{\mu\nu}\mathbbm{1}_{4},
	\end{displaymath}
then there exists a non-singular matrix $C$ s.t.
	\begin{displaymath}
	C\gamma^{\mu}\hspace{0mm}'C^{-1}=\gamma^{\mu}.
	\end{displaymath}
\end{theorem}
\hspace{0mm}\\
The proof of this theorem can be found in \cite{messiah2}. Using this theorem, and the fact that $-(\gamma^{\mu})^{T}=:\gamma^{\mu}\hspace{0mm}'$ fulfils the needed anticommutation relations, we see that the charge conjugation matrix exists.

\begin{prop}\label{PB1}
One possible choice for the charge conjugation matrix is
		\begin{equation}\label{CCmatrix}
			C=i\alpha\gamma^{2}\gamma^{0},
		\end{equation}
where $\alpha\in U(1)$. Using the Weyl basis and setting $\alpha=1$ we find
		\begin{displaymath}
		C=\left(\begin{matrix}
				i\sigma^{2} & \textbf{0} \\
				\textbf{0} & -i\sigma^{2}
		                \end{matrix}\right).
		\end{displaymath}
In this case $C$ is real. Furthermore with this choice of $C$
	\begin{enumerate}
		\item $C^{T}=-C,$
		\item $C^{\dagger}=C^{-1},$
		\item $C=C^{\ast}=-C^{T}=-C^{-1}=-C^{\dagger},$
		\item $C^{-1}\gamma^{5}C=(\gamma^{5})^{T}$.
	\end{enumerate}
\end{prop}

\begin{proof}
This is true, because
\bigskip
\\
$\gamma^{2}\gamma^{0}\underbrace{(\gamma^{0})^{T}}_{\gamma^{0}}(\gamma^{2}\gamma^{0})^{-1}=-\gamma^{0}\gamma^{2}\gamma^{0}(\gamma^{2}\gamma^{0})^{-1}=-\gamma^{0}$,
\medskip
\\
$\gamma^{2}\gamma^{0}\underbrace{(\gamma^{1})^{T}}_{-\gamma^{1}}(\gamma^{2}\gamma^{0})^{-1}=-\gamma^{2}\gamma^{0}\gamma^{1}(\gamma^{2}\gamma^{0})^{-1}=-\gamma^{1}\gamma^{2}\gamma^{0}(\gamma^{2}\gamma^{0})^{-1}=-\gamma^{1}$,
\medskip
\\
$\gamma^{2}\gamma^{0}\underbrace{(\gamma^{2})^{T}}_{\gamma^{2}}(\gamma^{2}\gamma^{0})^{-1}=-\gamma^{2}\gamma^{2}\gamma^{0}(\gamma^{2}\gamma^{0})^{-1}=-\gamma^{2}$,
\medskip
\\
$\gamma^{2}\gamma^{0}\underbrace{(\gamma^{3})^{T}}_{-\gamma^{3}}(\gamma^{2}\gamma^{0})^{-1}=-\gamma^{2}\gamma^{0}\gamma^{3}(\gamma^{2}\gamma^{0})^{-1}=-\gamma^{3}\gamma^{2}\gamma^{0}(\gamma^{2}\gamma^{0})^{-1}=-\gamma^{3}$.
\medskip
\\
Replacing $\gamma^{2}\gamma^{0}$ by $i\gamma^{2}\gamma^{0}$, we find a real $C$.
\\
Now for the properties of $C$: With our choice of $C$, the properties 1.,2.,3. are obvious. In fact 1. is true in general, and 2. is true for every representation of the $\gamma$-matrices, where $\gamma^{0}$ is Hermitian and $\gamma^{i}$ are anti-Hermitian  \cite{grimus1}. 3. follows from 1. and 2.
\medskip
\\
4. is true in general, because
	\begin{displaymath}
	\begin{split}
		C^{-1}\gamma^{5}C & =
		-iC^{-1}\gamma^{0}\gamma^{1}\gamma^{2}\gamma^{3}C=-iC^{-1}\gamma^{0}CC^{-1}\gamma^{1}CC^{-1}\gamma^{2}CC^{-1}\gamma^{3}C=
		\\&=
		-i(\gamma^{0})^{T}(\gamma^{1})^{T}(\gamma^{2})^{T}(\gamma^{3})^{T}=-i(\gamma^{3}\gamma^{2}\gamma^{1}\gamma^{0})^{T}=
		\\&=
		-i(-1)^{3+2+1}(\gamma^{0}\gamma^{1}\gamma^{2}\gamma^{3})^{T}=(-i\gamma^{0}\gamma^{1}\gamma^{2}\gamma^{3})^{T}=(\gamma^{5})^{T}.
	\end{split}
	\end{displaymath}
\end{proof}
\hspace{0mm}\\
We multiply (\ref{MEQ1}) by $C$ from the left and insert $C^{-1}C$ between the bracket and $\bar{\psi}^{T}$.
	\begin{displaymath}
		(-i\underbrace{C(\gamma^{\mu})^{T}C^{-1}}_{-\gamma^{\mu}}\partial_{\mu}-eA_{\mu}\underbrace{C(\gamma^{\mu})^{T}C^{-1}}_{-\gamma^{\mu}}-m)C\bar{\psi}^{T}=0.
	\end{displaymath}
Thus
	\begin{displaymath}
		(i\gamma^{\mu}\partial_{\mu}+e\gamma^{\mu}A_{\mu}-m)C\bar{\psi}^{T}=0,
	\end{displaymath}
so we have found the solution of the charge conjugate Dirac equation (\ref{conjeq})
	\begin{displaymath}
		\psi^{c}=C\bar{\psi}^{T}=C(\gamma^{0})^{T}\psi^{\ast}.
	\end{displaymath}

\section{Chiral fermion fields}\label{Chiralfields}

\begin{define}\label{DB2}
We define the following operators
	\begin{displaymath}
		P_{L}:=\frac{\mathbbm{1}_{4}-\gamma^{5}}{2} \quad\quad P_{R}:=\frac{\mathbbm{1}_{4}+\gamma^{5}}{2}.
	\end{displaymath}
A 4-spinor $\psi$ is called \textit{left-handed}, if
	\begin{displaymath}
		P_{L}\psi=\psi.
	\end{displaymath}
A 4-spinor $\psi$ is called \textit{right-handed}, if
	\begin{displaymath}
		P_{R}\psi=\psi.
	\end{displaymath}
Thus left- and right-handed spinors are eigenvectors of the operator $\gamma^{5}$ (which is also called \textit{chirality}) to the eigenvalues $+1$ and $-1$, therefore they are called \textit{chiral spinors}.
\medskip
\\
Let $\psi$ be a 4-spinor, then we define
	\begin{displaymath}
		\psi_{L}:=P_{L}\psi \quad\quad \psi_{R}:=P_{R}\psi.
	\end{displaymath}
\end{define}

\begin{prop}\label{PB3} $P_{L}$ and $P_{R}$ are have the following properties:
	\begin{enumerate}
		\item $P_{L}P_{R}=P_{R}P_{L}=0$,
		\item $P_{L}P_{L}=P_{L}$, $P_{R}P_{R}=P_{R}$,
		\item $P_{L}+P_{R}=\mathbbm{1}_{4}$, $P_{L}-P_{R}=-\gamma^{5}$.
	\end{enumerate}
\end{prop}

\begin{proof}
	This is easy to see in a straight forward calculation using $(\gamma^{5})^{2}=\mathbbm{1}_{4}$.
\end{proof}

\begin{prop}\label{PB4}
	Let $\psi$ be a chiral spinor, then $\psi^{c}$ has the opposite chirality.
\end{prop}

\begin{proof}
We only consider $\psi_{R}$, the calculation for $\psi_{L}$ is completely analogous.
	\begin{displaymath}
	\begin{split}
		P_{L}(\psi_{R})^{c}
		&=
		\frac{1}{2}(\mathbbm{1}_{4}-\gamma^{5})C\overline{\frac{1}{2}(\mathbbm{1}_{4}+\gamma^{5})\psi}^{T}=\frac{1}{2}(\mathbbm{1}_{4}-\gamma^{5})C(\gamma^{0})^{T}\frac{1}{2}(\mathbbm{1}_{4}+\gamma^{5\ast})\psi^{\ast}=
		\\&=
		\frac{1}{2}C\underbrace{C^{-1}(\mathbbm{1}_{4}-\gamma^{5})C}_{\mathbbm{1}_{4}-(\gamma^{5})^{T}=\mathbbm{1}_{4}-\gamma^{5}}\gamma^{0}\frac{1}{2}(\mathbbm{1}_{4}+\gamma^{5})\psi^{\ast}=
		\\&=
		C\frac{1}{2}(\mathbbm{1}_{4}-\gamma^{5})\gamma^{0}\frac{1}{2}(\mathbbm{1}_{4}+\gamma^{5})\psi^{\ast}=C\gamma^{0}\frac{1}{2}(\mathbbm{1}_{4}+\gamma^{5})\frac{1}{2}(\mathbbm{1}_{4}+\gamma^{5})\psi^{\ast}=
		\\&=
		C(\gamma^{0})^{T}P_{R}^{2}\psi^{\ast}=C\bar{\psi}_{R}^{T}=(\psi_{R})^{c}.
	\end{split}
	\end{displaymath}
$\Rightarrow$ $(\psi_{R})^{c}$ is left-handed. Note that we have used $(\gamma^0)^T=\gamma^0$.
\end{proof}

\section{Dirac and Majorana particles}
Massive spin-$\frac{1}{2}$ fermions are described by spinor fields that obey the Dirac-equation (\ref{Diracequation}). It turns out that in the case of neutral massive spin-$\frac{1}{2}$ particles (like neutrinos) there is an alternative to the usually used \textit{Dirac particle}, namely the concept of the \textit{Majorana particle}.

\subsection{Dirac particles}
A Dirac-particle is described by a 4-spinor field $\psi$ with independent chiral components $\psi_{L}$ and $\psi_{R}$:
	\begin{displaymath}
		\psi=\psi_{L}+\psi_{R}.
	\end{displaymath}
The Lagrangian of a free Dirac particle reads
	\begin{displaymath}
		\mathcal{L}_{\mathrm{Dirac}}^{(\mathrm{free})}=\bar{\psi}(i\gamma^{\mu}\partial_{\mu}-m)\psi,
	\end{displaymath}
thus the \textit{mass term}\footnote{A mass term is defined as a Lorentz-invariant, bilinear in the spinor fields.} is
	\begin{displaymath}
		\mathcal{L}_{\mathrm{Dirac}}^{(\mathrm{mass})}=-m\bar{\psi}\psi=-m(\bar{\psi}_{R}\psi_{L}+\bar{\psi}_{L}\psi_{R}).
	\end{displaymath}
Variation of the Lagrangian with respect to $\bar{\psi}$ gives the equation of motion, which is the Dirac equation.

\subsection{Majorana particles}\label{majoranapsubsection}
Majorana spinors are built from a single chiral spinor $\psi_{L}$ by
	\begin{displaymath}
	\psi=\psi_{L}+\psi_{L}^{c},
	\end{displaymath}
where by $\psi_{L}^{c}$ we always mean $(\psi_{L})^c$. Thus Majorana spinors fulfil the \textit{Majorana condition}
	\begin{displaymath}
	\psi^{c}=\psi.
	\end{displaymath}
A possible mass term for Majorana particles is given by
	\begin{displaymath}
	\mathcal{L}_{\mathrm{Majorana}}^{(\mathrm{mass})}=-\frac{1}{2}m(\overline{\psi^{c}_{L}}\psi_{L}+\bar{\psi}_{L}\psi_{L}^{c})=\frac{1}{2}m(\psi_{L}^{T}C^{-1}\psi_{L}-\bar{\psi}_{L}C\bar{\psi}_{L}^{T})=-\frac{1}{2}m\bar{\psi}\psi.
	\end{displaymath}
(see e.g. \cite{grimus1}). Therefore
	\begin{displaymath}
		\mathcal{L}_{\mathrm{Majorana}}^{(\mathrm{free})}=\bar{\psi}_{L}(i\gamma^{\mu}\partial_{\mu})\psi_{L}-\frac{1}{2}m(\overline{\psi^{c}_{L}}\psi_{L}+\bar{\psi}_{L}\psi_{L}^{c}).
	\end{displaymath}

\begin{prop}\label{PB5}
	The Euler-Lagrange equations given by $\mathcal{L}_{\mathrm{Majorana}}^{(\mathrm{free})}$ lead to the Dirac equation for the Majorana spinor $\psi$.
\end{prop}

\begin{proof}
In the case of Dirac particles the independent fields are $\psi_{L}, \psi_{R}, \bar{\psi}_{L}$ and $\bar{\psi}_{R}$. Here we work with Majorana particles, whose independent fields are just $\psi_{L}$ and $\bar{\psi}_{L}$.
\medskip
\\
For the sake of clarity we introduce the abbreviation $\mathcal{L}_{\mathrm{Majorana}}^{(\mathrm{free})}=\mathcal{L}$.
\\
$\overline{\psi_{L}^{c}}$ depends only on $\psi_{L}$, because
	\begin{displaymath}
		\overline{\psi_{L}^{c}}=(C\bar{\psi}_{L}^{T})^{\dagger}\gamma^{0}=\bar{\psi}_{L}^{\ast}C^{\dagger}\gamma^{0}=(\psi_{L}^{\dagger}\gamma^{0})^{\ast}C^{\dagger}\gamma^{0}=\psi_{L}^{T}(\gamma^{0})^{\ast}C^{\dagger}\gamma^{0}\Rightarrow \frac{\partial \overline{\psi_{L}^{c}}}{\partial\bar{\psi}_{L}}=0.
		\end{displaymath}
We first variate with respect to $\bar{\psi}_{L}$:
\medskip
\\
$\frac{\partial \mathcal{L}}{\partial \bar{\psi}_{L}}=i\gamma^{\mu}\partial_{\mu}\psi_{L}-\frac{1}{2}m\frac{\partial}{\partial \bar{\psi}_{L}}(\bar{\psi}_{L}C\bar{\psi}_{L}^{T})$
Using the anticommutation property of fermion fields and the antisymmetry of $C$ we get
\medskip
\\
$\frac{\partial}{\partial \bar{\psi}_{Li}}(\bar{\psi}_{Lk}C_{kl}\bar{\psi}_{Ll})=C_{kl}(\delta_{ik}\bar{\psi}_{Ll}-\delta_{il}\bar{\psi}_{Lk})=2C_{ik}\bar{\psi}_{Lk}$
$\Rightarrow \frac{\partial}{\partial \bar{\psi}_{L}}(\bar{\psi}_{L}C\bar{\psi}_{L}^{T})=2C\bar{\psi}_{L}^{T}$
\medskip
\\
$\frac{\partial \mathcal{L}}{\partial(\partial_{\mu}\bar{\psi}_{L})}=0$
\medskip
\\
Thus the first set of Euler-Lagrange-equations reads
	\begin{equation}\label{majorana1}
		i\gamma^{\mu}\partial_{\mu}\psi_{L}-mC\bar{\psi}_{L}^{T}=0.
	\end{equation}
Variation with respect to $\psi_{L}$ gives (in a similar manner)
	\begin{displaymath}
		\partial_{\mu}\bar{\psi}_{L}i\gamma^{\mu}-m\psi_{L}^{T}C^{-1}=0.
	\end{displaymath}
We transpose this equation and insert $C^{-1}C$ after $(\gamma^{\mu})^{T}$.
	\begin{displaymath}
		i(\gamma^{\mu})^{T}C^{-1}\partial_{\mu}\psi_{L}^{c}-m(C^{-1})^{T}\psi_{L}=0
	\end{displaymath}
Multiplying by $-C$ using $C(\gamma^{\mu})^{T}C^{-1}=-\gamma^{\mu}$ and $C^{T}=-C$ one gets
	\begin{equation}\label{majorana2}
		i\gamma^{\mu}\partial_{\mu}\psi_{L}^{c}-m\psi_{L}=0.
	\end{equation}
Adding equations (\ref{majorana1}) and (\ref{majorana2}) one gets the Dirac equation
	\begin{displaymath}
		(i\gamma^{\mu}\partial_{\mu}-m)(\psi_{L}+\psi_{L}^{c})=0.
	\end{displaymath}
\end{proof}
\hspace{0mm}\\
Let us at last consider the connection between the concept of the Majorana particle and the Majorana equation. As already mentioned the Dirac equation is basis independent and so is the charge conjugation matrix $C$ given in the form of equation (\ref{CCmatrix}). In contrast to the Dirac equation the Majorana equations (\ref{majoranaequation1}) and (\ref{majoranaequation2}) explicitly correspond to the Weyl basis. Let us therefore consider the Majorana condition in the Weyl basis:
\medskip
\\
The Majorana condition was $\psi^{c}=\psi$. In the Weyl basis $C$ is given by
	\begin{displaymath}
	C=i\alpha\gamma^{2}\gamma^{0}=\alpha
\left(\begin{matrix}
 \varepsilon & \textbf{0} \\
 \textbf{0} & -\varepsilon
\end{matrix}\right).
	\end{displaymath}
The charge conjugate of the Dirac spinor $\psi=
\left(\begin{matrix}
 \chi \\ \bar{\varphi}
\end{matrix}\right)
$ in the Weyl basis is
	\begin{displaymath}
	\psi^{c}=C\bar{\psi}^{T}=C(\gamma^{0})^{T}\psi^{\ast}=\alpha\left(\begin{matrix}
	 \varepsilon & \textbf{0} \\ \textbf{0} & -\varepsilon
	\end{matrix}\right)
	\left(\begin{matrix}
	 \textbf{0} & \mathbbm{1}_{2} \\ \mathbbm{1}_{2} & \textbf{0}
	\end{matrix}\right)
	\left(\begin{matrix}
	 \chi^{\ast} \\ \varepsilon\varphi\end{matrix}\right)=-\alpha \left(\begin{matrix}
	 \varphi \\ \bar{\chi}
	\end{matrix}\right).
	\end{displaymath}
Thus the Majorana condition becomes
	\begin{displaymath}
	\left(\begin{matrix}
 \chi \\ \bar{\varphi}
\end{matrix}\right)=-\alpha \left(\begin{matrix}
	 \varphi \\ \bar{\chi}
	\end{matrix}\right).
	\end{displaymath}
Choosing $\alpha=-1$ we find
	\begin{displaymath}
	\psi^{c}=\psi\Leftrightarrow \psi=\left(\begin{matrix}
	 \chi \\ \bar{\chi}
	\end{matrix}\right),
	\end{displaymath}
and inserting this spinor into the Dirac equation (in the Weyl basis) directly leads to the Majorana equations (\ref{majoranaequation1}) and (\ref{majoranaequation2}).

%% file: programexamples/examples.tex
\chapter{Examples for group theoretical calculations using a computer algebra system}\label{appendixD}

\fancyhf{}
\fancyhead[RO,LE]{ \thepage}
\fancyhead[RE]{\small \rmfamily \nouppercase{Appendix D. Examples for group theoretical calculations}}
\fancyhead[LO]{\small \rmfamily \nouppercase{\rightmark}}
\fancyfoot{} %no footline, especially no page numbering on the lower side of the page

\section{Example for a calculation of Clebsch-Gordan coefficients}
\begin{small}
In this \textit{ Mathematica} 6 notebook we will calculate the Clebsch-Gordan coefficients for \(3_1\)\(\otimes \)\(3_1\)=\(3_4\)\(\oplus \)\(3_5\)\(\oplus
\)\(3_6\) of the group \(\Sigma \)(36$\phi $). For the sake of clarity we only calculate the coefficients corresponding to \(3_5\). The calculation
of the remaining coefficients works completely analogous.
\medskip
\\
\noindent\(\pmb{\text{Remove}[\text{$\texttt{"}$Global$\grave{ }$*$\texttt{"}$}]}\)
\medskip
\\
\noindent\(\pmb{t=\text{SessionTime}[\hspace{1mm}];}\)
\medskip
\\
\noindent\(\pmb{\omega =\text{Exp}[2 \pi  I / 3];}\)
\medskip
\\
\noindent\(\pmb{\text{rep1}=\texttt{"}3_1\texttt{"};}\\
\pmb{\text{rep2}=\texttt{"}3_1\texttt{"};}\\
\pmb{\text{rep3}=\texttt{"}3_4\texttt{"};}\\
\pmb{\text{rep4}=\texttt{"}3_5\texttt{"};}\\
\pmb{\text{rep5}=\texttt{"}3_6\texttt{"};}\)
\medskip
\\
At first we enter the generators of the needed 3-dimensional representations of the finite group \(\Sigma (36\phi )\).
\medskip
\\
\noindent\(\pmb{\text{A12}=1;}\\
\pmb{\text{B12}=1;}\\
\pmb{\text{C12}=-1;}\\
\pmb{\text{A13}=1;}\\
\pmb{\text{B13}=1;}\\
\pmb{\text{C13}=I;}\)
\medskip
\\
\noindent\(\pmb{\text{A34}=\left(
\begin{array}{ccc}
 1 & 0 & 0 \\
 0 & \omega  & 0 \\
 0 & 0 & \omega ^2
\end{array}
\right);}\\
\pmb{\text{B34}=\left(
\begin{array}{ccc}
 0 & 1 & 0 \\
 0 & 0 & 1 \\
 1 & 0 & 0
\end{array}
\right);}\\
\pmb{\text{C34}=\frac{1}{\omega -\omega ^2}\left(
\begin{array}{ccc}
 1 & 1 & 1 \\
 1 & \omega  & \omega ^2 \\
 1 & \omega ^2 & \omega 
\end{array}
\right);}
\medskip
\\
\pmb{\text{A33}=\text{Conjugate}[\text{A34}];}\\
\pmb{\text{B33}=\text{Conjugate}[\text{B34}];}\\
\pmb{\text{C33}=\text{Conjugate}[\text{C34}];}\\
\pmb{\text{A31}=\text{A33}*\text{A12};}\\
\pmb{\text{B31}=\text{B33}*\text{B12};}\\
\pmb{\text{C31}=\text{C33}*\text{C12};}\\
\pmb{\text{A35}=\text{A34}*\text{A13};}\\
\pmb{\text{B35}=\text{B34}*\text{B13};}\\
\pmb{\text{C35}=\text{C34}*\text{C13};}\)
\medskip
\\
\noindent\(\pmb{\text{AA1}=\text{A31};}\\
\pmb{\text{BB1}=\text{B31};}\\
\pmb{\text{CC1}=\text{C31};}\\
\pmb{\text{AA2}=\text{A31};}\\
\pmb{\text{BB2}=\text{B31};}\\
\pmb{\text{CC2}=\text{C31};}\\
\pmb{\text{AA3}=\text{A35};}\\
\pmb{\text{BB3}=\text{B35};}\\
\pmb{\text{CC3}=\text{C35};}\)
\medskip
\\
We now define the matrices needed to construct the matrix NN.\\
At first we tranform the generators AA to a basis where they are diagonal.
\medskip
\\
\noindent\(\pmb{\text{UUA}=\text{Transpose}[\text{Eigensystem}[\text{Transpose}[\text{Inverse}[\text{AA1}]]][[2]]];}\\
\pmb{A=\text{Simplify}[\text{Inverse}[\text{UUA}].\text{Transpose}[\text{Inverse}[\text{AA1}]].\text{UUA}];}\\
\pmb{\text{UUB}=\text{Transpose}[\text{Eigensystem}[\text{Transpose}[\text{Inverse}[\text{AA2}]]][[2]]];}\\
\pmb{B=\text{Simplify}[\text{Inverse}[\text{UUB}].\text{Transpose}[\text{Inverse}[\text{AA2}]].\text{UUB}];}\\
\pmb{\text{UU$\Phi $}=\text{Transpose}[\text{Eigensystem}[\text{AA3}][[2]]];}\\
\pmb{\Phi =\text{Simplify}[\text{Inverse}[\text{UU$\Phi $}].\text{AA3}.\text{UU$\Phi $}];}\)
\medskip
\\
Construction of NN: We define the columns of NN.
\medskip
\\
\noindent\(\pmb{\text{Col}[\text{i$\_$},\text{j$\_$},\text{k$\_$}]\text{:=}\left(
\begin{array}{c}
 \Phi [[i,1]]A[[j,1]]B[[k,1]] \\
 \Phi [[i,1]]A[[j,1]]B[[k,2]] \\
 \Phi [[i,1]]A[[j,1]]B[[k,3]] \\
 \Phi [[i,1]]A[[j,2]]B[[k,1]] \\
 \Phi [[i,1]]A[[j,2]]B[[k,2]] \\
 \Phi [[i,1]]A[[j,2]]B[[k,3]] \\
 \Phi [[i,1]]A[[j,3]]B[[k,1]] \\
 \Phi [[i,1]]A[[j,3]]B[[k,2]] \\
 \Phi [[i,1]]A[[j,3]]B[[k,3]] \\
 \Phi [[i,2]]A[[j,1]]B[[k,1]] \\
 \Phi [[i,2]]A[[j,1]]B[[k,2]] \\
 \Phi [[i,2]]A[[j,1]]B[[k,3]] \\
 \Phi [[i,2]]A[[j,2]]B[[k,1]] \\
 \Phi [[i,2]]A[[j,2]]B[[k,2]] \\
 \Phi [[i,2]]A[[j,2]]B[[k,3]] \\
 \Phi [[i,2]]A[[j,3]]B[[k,1]] \\
 \Phi [[i,2]]A[[j,3]]B[[k,2]] \\
 \Phi [[i,2]]A[[j,3]]B[[k,3]] \\
 \Phi [[i,3]]A[[j,1]]B[[k,1]] \\
 \Phi [[i,3]]A[[j,1]]B[[k,2]] \\
 \Phi [[i,3]]A[[j,1]]B[[k,3]] \\
 \Phi [[i,3]]A[[j,2]]B[[k,1]] \\
 \Phi [[i,3]]A[[j,2]]B[[k,2]] \\
 \Phi [[i,3]]A[[j,2]]B[[k,3]] \\
 \Phi [[i,3]]A[[j,3]]B[[k,1]] \\
 \Phi [[i,3]]A[[j,3]]B[[k,2]] \\
 \Phi [[i,3]]A[[j,3]]B[[k,3]]
\end{array}
\right);}\)
\medskip
\\
\noindent\(\pmb{\text{NN}=\text{Join}[\text{Col}[1,1,1],\text{Col}[1,1,2],\text{Col}[1,1,3],\text{Col}[1,2,1],}\\
\pmb{\text{Col}[1,2,2],\text{Col}[1,2,3],\text{Col}[1,3,1],\text{Col}[1,3,2],\text{Col}[1,3,3],}\\
\pmb{\text{Col}[2,1,1],\text{Col}[2,1,2],\text{Col}[2,1,3],\text{Col}[2,2,1],\text{Col}[2,2,2],}\\
\pmb{\text{Col}[2,2,3],\text{Col}[2,3,1],\text{Col}[2,3,2],\text{Col}[2,3,3],\text{Col}[3,1,1],}\\
\pmb{\text{Col}[3,1,2],\text{Col}[3,1,3],\text{Col}[3,2,1],\text{Col}[3,2,2],\text{Col}[3,2,3],}\\
\pmb{\text{Col}[3,3,1],\text{Col}[3,3,2],\text{Col}[3,3,3],2];}\)
\medskip
\\
We will now calculate the eigenvectors of NN to the eigenvalue 1.
\medskip
\\
\noindent\(\pmb{\text{MMA}=\text{NN}-\text{IdentityMatrix}[27];}\)
\medskip
\\
The command \(NullSpace\)[MM] gives a set of linearly independent vectors v, such that MMv\(=0\).
\medskip
\\
\noindent\(\pmb{\text{EVA}=\text{NullSpace}[\text{MMA}];}\)
\medskip
\\
The rows of the matrix EVA are linearly independent eigenvectors of NN to the eigenvalue 1 in the new basis, we now have to transform them
back.
\medskip
\\
The number of independent eigenvectors is nA.
\medskip
\\
\noindent\(\pmb{\text{nA}=\text{Dimensions}[\text{EVA}][[1]];}\)
\medskip
\\
Back - transformation of the obtained eigenvectors to the standard basis :
\medskip
\\
\noindent\(\pmb{\text{YY}[\text{n$\_$},\text{k$\_$}]\text{:=}\text{EVA}[[n,k]];}\)
\medskip
\\
\noindent\(\pmb{\text{$\Gamma $Y1}[\text{n$\_$}]\text{:=}\left(
\begin{array}{ccc}
 \text{YY}[n,1] & \text{YY}[n,2] & \text{YY}[n,3] \\
 \text{YY}[n,4] & \text{YY}[n,5] & \text{YY}[n,6] \\
 \text{YY}[n,7] & \text{YY}[n,8] & \text{YY}[n,9]
\end{array}
\right);}\\
\pmb{\text{$\Gamma $Y2}[\text{n$\_$}]\text{:=}\left(
\begin{array}{ccc}
 \text{YY}[n,10] & \text{YY}[n,11] & \text{YY}[n,12] \\
 \text{YY}[n,13] & \text{YY}[n,14] & \text{YY}[n,15] \\
 \text{YY}[n,16] & \text{YY}[n,17] & \text{YY}[n,18]
\end{array}
\right);}\\
\pmb{\text{$\Gamma $Y3}[\text{n$\_$}]\text{:=}\left(
\begin{array}{ccc}
 \text{YY}[n,19] & \text{YY}[n,20] & \text{YY}[n,21] \\
 \text{YY}[n,22] & \text{YY}[n,23] & \text{YY}[n,24] \\
 \text{YY}[n,25] & \text{YY}[n,26] & \text{YY}[n,27]
\end{array}
\right);}\)
\medskip
\\
\noindent\(\pmb{\text{EV}[\text{n$\_$},\text{k$\_$}]\text{:=}\text{Simplify}[\text{Transpose}[\text{Inverse}[\text{UUA}]].\text{$\Gamma $Y1}[n].\text{Inverse}[\text{UUB}]*(\text{Inverse}[\text{UU$\Phi
$}][[1,k]])+}\\
\pmb{\text{Transpose}[\text{Inverse}[\text{UUA}]].\text{$\Gamma $Y2}[n].\text{Inverse}[\text{UUB}]*(\text{Inverse}[\text{UU$\Phi $}][[2,k]])}+
\\
\pmb{\text{Transpose}[\text{Inverse}[\text{UUA}]].\text{$\Gamma
$Y3}[n].\text{Inverse}[\text{UUB}]*(\text{Inverse}[\text{UU$\Phi $}][[3,k]])];}\)
\medskip
\\
\noindent\(\pmb{\text{Ev}[\text{n$\_$}]\text{:=}\text{Flatten}[\text{Transpose}[(\text{EV}[n,1][[1,1]] ,  \text{EV}[n,1][[1,2]] ,  \text{EV}[n,1][[1,3]] ,  \text{EV}[n,1][[2,1]],}\\
 \pmb{\text{EV}[n,1][[2,2]] ,  \text{EV}[n,1][[2,3]] , 
\text{EV}[n,1][[3,1]] ,  \text{EV}[n,1][[3,2]] ,  \text{EV}[n,1][[3,3]] ,}\\
\pmb{\text{EV}[n,2][[1,1]] ,  \text{EV}[n,2][[1,2]] ,  \text{EV}[n,2][[1,3]] ,  \text{EV}[n,2][[2,1]]
,  \text{EV}[n,2][[2,2]] ,}\\
\pmb{\text{EV}[n,2][[2,3]] ,  \text{EV}[n,2][[3,1]] ,  \text{EV}[n,2][[3,2]] ,  \text{EV}[n,2][[3,3]] ,  \text{EV}[n,3][[1,1]] ,}\\
\pmb{\text{EV}[n,3][[1,2]] ,  \text{EV}[n,3][[1,3]] ,  \text{EV}[n,3][[2,1]] ,  \text{EV}[n,3][[2,2]] ,  \text{EV}[n,3][[2,3]] ,}\\
\pmb{\text{EV}[n,3][[3,1]] ,  \text{EV}[n,3][[3,2]]
,  \text{EV}[n,3][[3,3]])]}\)
\medskip
\\
Now we will test if the vectors Ev[n] really are elements of the kernel of MMA.
\medskip
\\
\noindent\(\pmb{A=\text{Transpose}[\text{Inverse}[\text{AA1}]];}\\
\pmb{B=\text{Transpose}[\text{Inverse}[\text{AA2}]];}\\
\pmb{\Phi =\text{AA3};}\)
\medskip
\\
\noindent\(\pmb{\text{NN}=\text{Join}[\text{Col}[1,1,1],\text{Col}[1,1,2],\text{Col}[1,1,3],\text{Col}[1,2,1],}\\
\pmb{\text{Col}[1,2,2],\text{Col}[1,2,3],\text{Col}[1,3,1],\text{Col}[1,3,2],\text{Col}[1,3,3],}\\
\pmb{\text{Col}[2,1,1],\text{Col}[2,1,2],\text{Col}[2,1,3],\text{Col}[2,2,1],\text{Col}[2,2,2],}\\
\pmb{\text{Col}[2,2,3],\text{Col}[2,3,1],\text{Col}[2,3,2],\text{Col}[2,3,3],\text{Col}[3,1,1],}\\
\pmb{\text{Col}[3,1,2],\text{Col}[3,1,3],\text{Col}[3,2,1],\text{Col}[3,2,2],\text{Col}[3,2,3],}\\
\pmb{\text{Col}[3,3,1],\text{Col}[3,3,2],\text{Col}[3,3,3],2];}\)
\medskip
\\
\noindent\(\pmb{\text{MMA}=\text{NN}-\text{IdentityMatrix}[27];}\)
\medskip
\\
\noindent\(\pmb{n=1;}\\
\pmb{\text{While}[n\leq \text{nA},\text{Print}[\text{Simplify}[\text{MMA}.\text{Ev}[n]]];n\text{++}]}\)
\medskip
\\
\noindent\(\{0,0,0,0,0,0,0,0,0,0,0,0,0,0,0,0,0,0,0,0,0,0,0,0,0,0,0\}\)

\noindent\(\{0,0,0,0,0,0,0,0,0,0,0,0,0,0,0,0,0,0,0,0,0,0,0,0,0,0,0\}\)

\noindent\(\{0,0,0,0,0,0,0,0,0,0,0,0,0,0,0,0,0,0,0,0,0,0,0,0,0,0,0\}\)

\noindent\(\{0,0,0,0,0,0,0,0,0,0,0,0,0,0,0,0,0,0,0,0,0,0,0,0,0,0,0\}\)

\noindent\(\{0,0,0,0,0,0,0,0,0,0,0,0,0,0,0,0,0,0,0,0,0,0,0,0,0,0,0\}\)

\noindent\(\{0,0,0,0,0,0,0,0,0,0,0,0,0,0,0,0,0,0,0,0,0,0,0,0,0,0,0\}\)

\noindent\(\{0,0,0,0,0,0,0,0,0,0,0,0,0,0,0,0,0,0,0,0,0,0,0,0,0,0,0\}\)

\noindent\(\{0,0,0,0,0,0,0,0,0,0,0,0,0,0,0,0,0,0,0,0,0,0,0,0,0,0,0\}\)

\noindent\(\{0,0,0,0,0,0,0,0,0,0,0,0,0,0,0,0,0,0,0,0,0,0,0,0,0,0,0\}\)
\medskip
\\
\noindent\(\pmb{n=1;}\\
\pmb{\text{While}[n\leq \text{nA},\text{EvA}[n]=\text{Ev}[n];n\text{++}];}\)
\medskip
\\
Thus we have found the kernel of NN[A] in the standard basis.
\medskip
\\
\noindent\(\pmb{A=\text{Transpose}[\text{Inverse}[\text{BB1}]];}\\
\pmb{B=\text{Transpose}[\text{Inverse}[\text{BB2}]];}\\
\pmb{\Phi =\text{BB3};}\)
\medskip
\\
\noindent\(\pmb{\text{NN}=\text{Join}[\text{Col}[1,1,1],\text{Col}[1,1,2],\text{Col}[1,1,3],\text{Col}[1,2,1],}\\
\pmb{\text{Col}[1,2,2],\text{Col}[1,2,3],\text{Col}[1,3,1],\text{Col}[1,3,2],\text{Col}[1,3,3],}\\
\pmb{\text{Col}[2,1,1],\text{Col}[2,1,2],\text{Col}[2,1,3],\text{Col}[2,2,1],\text{Col}[2,2,2],}\\
\pmb{\text{Col}[2,2,3],\text{Col}[2,3,1],\text{Col}[2,3,2],\text{Col}[2,3,3],\text{Col}[3,1,1],}\\
\pmb{\text{Col}[3,1,2],\text{Col}[3,1,3],\text{Col}[3,2,1],\text{Col}[3,2,2],\text{Col}[3,2,3],}\\
\pmb{\text{Col}[3,3,1],\text{Col}[3,3,2],\text{Col}[3,3,3],2];}\)
\medskip
\\
\noindent\(\pmb{\text{MMB}=\text{NN}-\text{IdentityMatrix}[27];}\)
\medskip
\\
\noindent\(\pmb{A=\text{Transpose}[\text{Inverse}[\text{CC1}]];}\\
\pmb{B=\text{Transpose}[\text{Inverse}[\text{CC2}]];}\\
\pmb{\Phi =\text{CC3};}\)
\medskip
\\
\noindent\(\pmb{\text{NN}=\text{Join}[\text{Col}[1,1,1],\text{Col}[1,1,2],\text{Col}[1,1,3],\text{Col}[1,2,1],}\\
\pmb{\text{Col}[1,2,2],\text{Col}[1,2,3],\text{Col}[1,3,1],\text{Col}[1,3,2],\text{Col}[1,3,3],}\\
\pmb{\text{Col}[2,1,1],\text{Col}[2,1,2],\text{Col}[2,1,3],\text{Col}[2,2,1],\text{Col}[2,2,2],}\\
\pmb{\text{Col}[2,2,3],\text{Col}[2,3,1],\text{Col}[2,3,2],\text{Col}[2,3,3],\text{Col}[3,1,1],}\\
\pmb{\text{Col}[3,1,2],\text{Col}[3,1,3],\text{Col}[3,2,1],\text{Col}[3,2,2],\text{Col}[3,2,3],}\\
\pmb{\text{Col}[3,3,1],\text{Col}[3,3,2],\text{Col}[3,3,3],2];}\)
\medskip
\\
\noindent\(\pmb{\text{MMC}=\text{NN}-\text{IdentityMatrix}[27];}\)
\medskip
\\
\noindent\(\pmb{n=1;}\\
\pmb{\text{While}[n\leq \text{nA},\text{XA}[n]=\text{Simplify}[\text{MMB}.\text{EvA}[n]];n\text{++}];}\\
\pmb{X=\{\text{XA}[1]\};}\\
\pmb{\text{Do}[X=\text{Join}[X,\{\text{XA}[j]\}],\{j,2,\text{nA}\}];}\\
\pmb{X=\text{Transpose}[X];}\\
\pmb{Y=\text{NullSpace}[X];}\\
\pmb{Y=\text{Simplify}[Y];}\)
\medskip
\\
\noindent\(\pmb{Y\text{//}\text{MatrixForm}}\)
\medskip
\\
\noindent\(\left(
\begin{array}{ccccccccc}
 1 & 0 & 0 & 0 & 1 & 0 & 0 & 0 & 1 \\
 0 & 0 & 1 & 1 & 0 & 0 & 0 & 1 & 0 \\
 0 & 1 & 0 & 0 & 0 & 1 & 1 & 0 & 0
\end{array}
\right)\)
\medskip
\\
\noindent\(\pmb{\text{YY1}=\text{Flatten}[\text{Simplify}[\text{Sum}[Y[[1,j]] \text{EvA}[j],\{j,1,\text{nA}\}]]];}\\
\pmb{\text{YY2}=\text{Flatten}[\text{Simplify}[\text{Sum}[Y[[2,j]] \text{EvA}[j],\{j,1,\text{nA}\}]]];}\\
\pmb{\text{YY3}=\text{Flatten}[\text{Simplify}[\text{Sum}[Y[[3,j]] \text{EvA}[j],\{j,1,\text{nA}\}]]];}\)
\medskip
\\
We have now obtained a basis $\{$YY1, YY2, YY3$\}$ of the vectorspace of common eigenvectors of NN[A] and NN[B] to the eigenvalue 1. We will now
intersect with the eigenspace of NN[C] to the eigenvalue 1.
\medskip
\\
\noindent\(\pmb{X=\text{Transpose}[\{\text{Simplify}[\text{MMC}.\text{YY1}]\}];}\\
\pmb{X=\text{Join}[X,\text{Transpose}[\{\text{Simplify}[\text{MMC}.\text{YY2}]\}],2];}\\
\pmb{X=\text{Join}[X,\text{Transpose}[\{\text{Simplify}[\text{MMC}.\text{YY3}]\}],2];}\)
\medskip
\\
\noindent\(\pmb{Y=\text{Flatten}[\text{NullSpace}[X]];}\\
\pmb{\text{MatrixForm}[Y]\text{//}\text{FullSimplify}}\)
\medskip
\\
\noindent\(\left(
\begin{array}{c}
 1-\sqrt{3} \\
 1 \\
 1
\end{array}
\right)\)
\medskip
\\
\noindent\(\pmb{\text{YY}=\text{Simplify}[\text{YY1}*Y[[1]]+\text{YY2}*Y[[2]]+\text{YY3}*Y[[3]]];}\)
\medskip
\\
\noindent\(\pmb{\text{YY}\text{//}\text{FullSimplify}}\)
\medskip
\\
\noindent\(\left\{1-\sqrt{3},0,0,0,0,1,0,1,0,0,0,1,0,1-\sqrt{3},0,1,0,0,0,1,0,1,0,0,0,0,1-\sqrt{3}\right\}\)
\medskip
\\
Thus we have found one common eigenvector YY of NN[A], NN[B] and NN[C] to the eigenvalue 1.
\medskip
\\
The Clebsch - Gordan - Coefficients are :
\medskip
\\
\noindent\(\pmb{\text{$\Gamma $1}=\left(
\begin{array}{ccc}
 \text{YY}[[1]] & \text{YY}[[2]] & \text{YY}[[3]] \\
 \text{YY}[[4]] & \text{YY}[[5]] & \text{YY}[[6]] \\
 \text{YY}[[7]] & \text{YY}[[8]] & \text{YY}[[9]]
\end{array}
\right);}\\
\pmb{\text{$\Gamma $2}=\left(
\begin{array}{ccc}
 \text{YY}[[10]] & \text{YY}[[11]] & \text{YY}[[12]] \\
 \text{YY}[[13]] & \text{YY}[[14]] & \text{YY}[[15]] \\
 \text{YY}[[16]] & \text{YY}[[17]] & \text{YY}[[18]]
\end{array}
\right);}\\
\pmb{\text{$\Gamma $3}=\left(
\begin{array}{ccc}
 \text{YY}[[19]] & \text{YY}[[20]] & \text{YY}[[21]] \\
 \text{YY}[[22]] & \text{YY}[[23]] & \text{YY}[[24]] \\
 \text{YY}[[25]] & \text{YY}[[26]] & \text{YY}[[27]]
\end{array}
\right);}\)
\medskip
\\
Test if the found coefficients fulfil the invariance equations. Zero vectors correspond to right solutions of the invariance equations.
\medskip
\\
\noindent\(\pmb{\text{FullSimplify}[\text{MMA}.\text{YY}]}\\
\pmb{\text{FullSimplify}[\text{MMB}.\text{YY}]}\\
\pmb{\text{FullSimplify}[\text{MMC}.\text{YY}]}\)
\medskip
\\
\noindent\(\{0,0,0,0,0,0,0,0,0,0,0,0,0,0,0,0,0,0,0,0,0,0,0,0,0,0,0\}\)

\noindent\(\{0,0,0,0,0,0,0,0,0,0,0,0,0,0,0,0,0,0,0,0,0,0,0,0,0,0,0\}\)

\noindent\(\{0,0,0,0,0,0,0,0,0,0,0,0,0,0,0,0,0,0,0,0,0,0,0,0,0,0,0\}\)
\medskip
\\
Normalization of the CGC' s :
\medskip
\\
\noindent\(\pmb{\text{$\Gamma $1n}=}\\
\pmb{\text{Simplify}[}\\
\pmb{\text{$\Gamma $1}/}\\
\pmb{\text{Sqrt}[(\text{Abs}[\text{$\Gamma $1}[[1,1]]]{}^{\wedge}2+\text{Abs}[\text{$\Gamma $1}[[1,2]]]{}^{\wedge}2+\text{Abs}[\text{$\Gamma $1}[[1,3]]]{}^{\wedge}2+}\\
\pmb{\text{Abs}[\text{$\Gamma $1}[[2,1]]]{}^{\wedge}2+\text{Abs}[\text{$\Gamma $1}[[2,2]]]{}^{\wedge}2+\text{Abs}[\text{$\Gamma $1}[[2,3]]]{}^{\wedge}2+}\\
\pmb{\text{Abs}[\text{$\Gamma $1}[[3,1]]]{}^{\wedge}2+\text{Abs}[\text{$\Gamma $1}[[3,2]]]{}^{\wedge}2+\text{Abs}[\text{$\Gamma $1}[[3,3]]]{}^{\wedge}2)]];}\)
\medskip
\\
\noindent\(\pmb{\text{$\Gamma $2n}=}\\
\pmb{\text{Simplify}[}\\
\pmb{\text{$\Gamma $2}/}\\
\pmb{\text{Sqrt}[(\text{Abs}[\text{$\Gamma $2}[[1,1]]]{}^{\wedge}2+\text{Abs}[\text{$\Gamma $2}[[1,2]]]{}^{\wedge}2+\text{Abs}[\text{$\Gamma $2}[[1,3]]]{}^{\wedge}2+}\\
\pmb{\text{Abs}[\text{$\Gamma $2}[[2,1]]]{}^{\wedge}2+\text{Abs}[\text{$\Gamma $2}[[2,2]]]{}^{\wedge}2+\text{Abs}[\text{$\Gamma $2}[[2,3]]]{}^{\wedge}2+}\\
\pmb{\text{Abs}[\text{$\Gamma $2}[[3,1]]]{}^{\wedge}2+\text{Abs}[\text{$\Gamma $2}[[3,2]]]{}^{\wedge}2+\text{Abs}[\text{$\Gamma $2}[[3,3]]]{}^{\wedge}2)]];}\)
\medskip
\\
\noindent\(\pmb{\text{$\Gamma $3n}=}\\
\pmb{\text{Simplify}[}\\
\pmb{\text{$\Gamma $3}/}\\
\pmb{\text{Sqrt}[(\text{Abs}[\text{$\Gamma $3}[[1,1]]]{}^{\wedge}2+\text{Abs}[\text{$\Gamma $3}[[1,2]]]{}^{\wedge}2+\text{Abs}[\text{$\Gamma $3}[[1,3]]]{}^{\wedge}2+}\\
\pmb{\text{Abs}[\text{$\Gamma $3}[[2,1]]]{}^{\wedge}2+\text{Abs}[\text{$\Gamma $3}[[2,2]]]{}^{\wedge}2+\text{Abs}[\text{$\Gamma $3}[[2,3]]]{}^{\wedge}2+}\\
\pmb{\text{Abs}[\text{$\Gamma $3}[[3,1]]]{}^{\wedge}2+\text{Abs}[\text{$\Gamma $3}[[3,2]]]{}^{\wedge}2+\text{Abs}[\text{$\Gamma $3}[[3,3]]]{}^{\wedge}2)]];}\)
\medskip
\\
We will now use the command FullSimplify to simplify the Clebsch - Gordan coefficients. To avoid results of the type Root[\textit{equation}], we
exclude this output format explicitly.
\medskip
\\
\noindent\(\pmb{i=1;}\\
\pmb{\text{While}[i\leq 3,}\\
\pmb{j=1;}\\
\pmb{\text{While}[j\leq 3,}\\
\pmb{\text{If}[\text{StringCount}[\text{ToString}[\text{FullSimplify}[\text{Abs}[\text{$\Gamma $1n}[[i,j]]]],\text{StandardForm}],\text{$\texttt{"}$Root$\texttt{"}$}]==0,}\\
\pmb{\text{$\Gamma
$1n}[[i,j]]=\text{FullSimplify}[\text{Abs}[\text{$\Gamma $1n}[[i,j]]]]*}
\pmb{\text{Exp}[I* \text{FullSimplify}[\text{Arg}[\text{$\Gamma $1n}[[i,j]]]]];}
\hspace{2mm}\pmb{j\text{++}]}\\
\pmb{i\text{++}]}\)
\medskip
\\
\noindent\(\pmb{i=1;}\\
\pmb{\text{While}[i\leq 3,}\\
\pmb{j=1;}\\
\pmb{\text{While}[j\leq 3,}\\
\pmb{\text{If}[\text{StringCount}[\text{ToString}[\text{FullSimplify}\text{Abs}[[\text{$\Gamma $2n}[[i,j]]]],\text{StandardForm}],\text{$\texttt{"}$Root$\texttt{"}$}]==0,}\\
\pmb{\text{$\Gamma
$2n}[[i,j]]=\text{FullSimplify}[\text{Abs}[\text{$\Gamma $2n}[[i,j]]]]*}
\pmb{\text{Exp}[I* \text{FullSimplify}[\text{Arg}[\text{$\Gamma $2n}[[i,j]]]]];}
\hspace{2mm}\pmb{j\text{++}]}\\
\pmb{i\text{++}]}\)
\medskip
\\
\noindent\(\pmb{i=1;}\\
\pmb{\text{While}[i\leq 3,}\\
\pmb{j=1;}\\
\pmb{\text{While}[j\leq 3,}\\
\pmb{\text{If}[\text{StringCount}[\text{ToString}[\text{FullSimplify}[\text{Abs}[\text{$\Gamma $3n}[[i,j]]]],\text{StandardForm}],\text{$\texttt{"}$Root$\texttt{"}$}]==0,}\\
\pmb{\text{$\Gamma
$3n}[[i,j]]=\text{FullSimplify}[\text{Abs}[\text{$\Gamma $3n}[[i,j]]]]*}
\pmb{\text{Exp}[I* \text{FullSimplify}[\text{Arg}[\text{$\Gamma $3n}[[i,j]]]]];}
\hspace{2mm}\pmb{j\text{++}]}\\
\pmb{i\text{++}]}\)
\medskip
\\
We will now calculate normalized basis vectors of the invariant subspace corresponding to the representation \(3_5\).
\medskip
\\
\noindent\(\pmb{\text{ee}=\text{Flatten}\left[\left(
\begin{array}{ccccccccc}
 e_{11} & e_{12} & e_{13} & e_{21} & e_{22} & e_{23} & e_{31} & e_{32} & e_{33}
\end{array}
\right)\right];}\)
\medskip
\\
\noindent\(\pmb{\text{out1}=\text{ToString}[\text{Subscript}[\text{$\texttt{"}$u$\texttt{"}$},\text{rep4}]{}^{\wedge}(\text{rep1}<>\texttt{"}\otimes
\texttt{"}<>\text{rep2}),\text{StandardForm}];}\\
\pmb{i=1;}\\
\pmb{q=0;}\\
\pmb{\text{out1}=\text{out1}<>\text{$\texttt{"}$(1)$\texttt{"}$}<>\text{$\texttt{"}$=$\texttt{"}$};}\\
\pmb{\text{While}[i\leq 9,}\\
\pmb{\text{If}[\text{Flatten}[\text{$\Gamma $1n}][[i]]\neq 0,}\\
\pmb{\text{If}[q==1,\text{out1}=\text{out1}<>\text{$\texttt{"}$ $\texttt{"}$}<>\text{ToString}[\text{Style}[\text{$\texttt{"}$+$\texttt{"}$},\text{Bold}],\text{StandardForm}]<>\text{$\texttt{"}$
$\texttt{"}$}]];}\\
\pmb{\text{If}[\text{Flatten}[\text{$\Gamma $1n}][[i]]\neq 0,}\\
\pmb{\text{out1}=\text{out1}<>\text{ToString}[\text{Flatten}[\text{$\Gamma $1n}][[i]],\text{StandardForm}]<>\text{ToString}[\text{Style}[\text{ToString}[\text{ee}[[i]],}
\\
\pmb{\text{StandardForm}],\text{Bold}],\text{StandardForm}];q=1];}\\
\pmb{i\text{++}]}\)
\medskip
\\
\noindent\(\pmb{\text{out2}=\text{ToString}[\text{Subscript}[\text{$\texttt{"}$u$\texttt{"}$},\text{rep4}]{}^{\wedge}(\text{rep1}<>\texttt{"}\otimes
\texttt{"}<>\text{rep2}),\text{StandardForm}];}\\
\pmb{i=1;}\\
\pmb{q=0;}\\
\pmb{\text{out2}=\text{out2}<>\text{$\texttt{"}$(2)$\texttt{"}$}<>\text{$\texttt{"}$=$\texttt{"}$};}\\
\pmb{\text{While}[i\leq 9,}\\
\pmb{\text{If}[\text{Flatten}[\text{$\Gamma $2n}][[i]]\neq 0,}\\
\pmb{\text{If}[q==1,\text{out2}=\text{out2}<>\text{$\texttt{"}$ $\texttt{"}$}<>\text{ToString}[\text{Style}[\text{$\texttt{"}$+$\texttt{"}$},\text{Bold}],\text{StandardForm}]<>\text{$\texttt{"}$
$\texttt{"}$}]];}\\
\pmb{\text{If}[\text{Flatten}[\text{$\Gamma $2n}][[i]]\neq 0,}\\
\pmb{\text{out2}=\text{out2}<>\text{ToString}[\text{Flatten}[\text{$\Gamma $2n}][[i]],\text{StandardForm}]<>\text{ToString}[\text{Style}[\text{ToString}[\text{ee}[[i]],}\\
\pmb{\text{StandardForm}],\text{Bold}],\text{StandardForm}];q=1];}\\
\pmb{i\text{++}]}\)
\medskip
\\
\noindent\(\pmb{\text{out3}=\text{ToString}[\text{Subscript}[\text{$\texttt{"}$u$\texttt{"}$},\text{rep4}]{}^{\wedge}(\text{rep1}<>\texttt{"}\otimes
\texttt{"}<>\text{rep2}),\text{StandardForm}];}\\
\pmb{i=1;}\\
\pmb{q=0;}\\
\pmb{\text{out3}=\text{out3}<>\text{$\texttt{"}$(3)$\texttt{"}$}<>\text{$\texttt{"}$=$\texttt{"}$};}\\
\pmb{\text{While}[i\leq 9,}\\
\pmb{\text{If}[\text{Flatten}[\text{$\Gamma $3n}][[i]]\neq 0,}\\
\pmb{\text{If}[q==1,\text{out3}=\text{out3}<>\text{$\texttt{"}$ $\texttt{"}$}<>\text{ToString}[\text{Style}[\text{$\texttt{"}$+$\texttt{"}$},\text{Bold}],\text{StandardForm}]<>\text{$\texttt{"}$
$\texttt{"}$}]];}\\
\pmb{\text{If}[\text{Flatten}[\text{$\Gamma $3n}][[i]]\neq 0,}\\
\pmb{\text{out3}=\text{out3}<>\text{ToString}[\text{Flatten}[\text{$\Gamma $3n}][[i]],\text{StandardForm}]<>\text{ToString}[\text{Style}[\text{ToString}[\text{ee}[[i]],}\\
\pmb{\text{StandardForm}],\text{Bold}],\text{StandardForm}];q=1];}\\
\pmb{i\text{++}]}\)
\medskip
\\
\noindent\(\pmb{\text{table}=
\begin{array}{|c|}
\hline
 \text{out1} \\
 \text{out2} \\
 \text{out3} \\
\hline
\end{array}
\hspace{0.5mm};}\)
\medskip
\\
\noindent\(\pmb{\text{tab}=\text{TableForm}[\text{table},\text{TableSpacing}\to \{5,5\},}\\
\pmb{\text{TableHeadings}\to }
\pmb{\{\{\text{rep4},\texttt{"}\texttt{"},\texttt{"}\texttt{"}\},}\\
\pmb{\{\text{rep1}<>\texttt{"}\otimes \texttt{"}<>\text{rep2}<>\text{$\texttt{"}$=$\texttt{"}$}<>\text{rep3}<>\texttt{"}\oplus \texttt{"}<>\text{rep4}<>\texttt{"}\oplus
\texttt{"}<>\text{rep5}\}\}]}\)
\medskip
\\
\noindent\(
\renewcommand{\arraystretch}{1.8}
\begin{array}{|l|l|}
\hline
  & 3_1\otimes 3_1=3_4\oplus 3_5\oplus 3_6 \\
\hline
 3_5 & u_{3_5}^{3_1\otimes 3_1}\text{(1)=}-\sqrt{\frac{1}{6} \left(3-\sqrt{3}\right)}e_{11} + \frac{1}{\sqrt{6-2 \sqrt{3}}}e_{23} + \frac{1}{\sqrt{6-2
\sqrt{3}}}e_{32} \\
  & u_{3_5}^{3_1\otimes 3_1}\text{(2)=}\frac{1}{\sqrt{6-2 \sqrt{3}}}e_{13} + \sqrt{\frac{1}{6} \left(3-\sqrt{3}\right)}e_{22} + \frac{1}{\sqrt{6-2
\sqrt{3}}}e_{31} \\
  & u_{3_5}^{3_1\otimes 3_1}\text{(3)=}\frac{1}{\sqrt{6-2 \sqrt{3}}}e_{12} + \frac{1}{\sqrt{6-2 \sqrt{3}}}e_{21} + \sqrt{\frac{1}{6} \left(3-\sqrt{3}\right)}e_{33}\\
\hline
\end{array}
\)
\bigskip
\\
The number SessionTime[\hspace{1mm}] - t is the time the calculation needed in seconds.
\medskip
\\
\noindent\(\pmb{\text{SessionTime}[\hspace{1mm}]-t}\)
\medskip
\\
\noindent\(1.9062500\)
\end{small}

\section{Example for a calculation of tensor products}
In this \textit{Mathematica 6} notebook we will calculate the Kronecker products for $\Sigma(216\phi)$.
\medskip
\\
\begin{small}
\noindent\(\pmb{\text{Remove}[\text{$\texttt{"}$Global$\grave{ }$*$\texttt{"}$}];}\)
\medskip
\\
\noindent\(\pmb{\omega =\text{Exp}[2 \pi  I / 3];}\\
\pmb{a=-\omega ^2;}\\
\pmb{b=2 \omega ^2;}\\
\pmb{c=-3 \omega ^2;}\\
\pmb{d=-\text{Exp}[10 \pi  I / 9];}\\
\pmb{e=-\text{Exp}[4 \pi  I / 9];}\\
\pmb{f=\text{Exp}[4 \pi  I / 9]+\text{Exp}[10 \pi  I / 9];}\\
\pmb{g=6\omega ^2;}\\
\pmb{h=9\omega ^2;}\\
\pmb{i=2 \text{Exp}[4 \pi  I / 9]+\text{Exp}[10 \pi  I / 9];}\\
\pmb{j=-\text{Exp}[4 \pi  I / 9]-2\text{Exp}[10 \pi  I / 9];}\\
\pmb{k=-\text{Exp}[4 \pi  I / 9]+\text{Exp}[10 \pi  I / 9];}\)
\medskip
\\
Order of $\Sigma $ (216$\phi $) :
\medskip
\\
\noindent\(\pmb{\text{gg}=216*3;}\)
\medskip
\\
The matrix CT represents the character table of $\Sigma $ (216 $\phi $). (We cannot display the matrix here, because it is too large.)\\
Rows : Characters for the different classes { }for an irreducible representation.\\
Columns : Characters for a class in different non - equivalent irreducible representations.
\medskip
\\
The number of conjugation classes is
\medskip
\\
\noindent\(\pmb{\text{nc}=24;}\)
\medskip
\\
We also need the number of elements contained in the classes :
\medskip
\\
\begin{tiny}
\noindent\(\pmb{n=\left(
\begin{array}{cccccccccccccccccccccccc}
 1 & 9 & 9 & 24 & 72 & 36 & 36 & 36 & 72 & 36 & 36 & 36 & 9 & 54 & 54 & 54 & 1 & 1 & 12 & 12 & 12 & 12 & 12 & 12
\end{array}
\right);}\)
\end{tiny}
\medskip
\\
Given the character table one can easily compute the number of times bb an irreducible representation is contained in a tensor product of two irreducible representations.
\medskip
\\
\noindent\(\pmb{\text{bb}[\text{x$\_$},\text{y$\_$},\text{z$\_$}]\text{:=}}\\
\pmb{1/\text{gg} \text{Sum}[n[[1]][[k]]*\text{Conjugate}[\text{CT}[[x]][[k]]]*\text{CT}[[y]][[k]]*\text{CT}[[z]][[k]],}\\
\pmb{\{k,1,\text{nc}\} ]}\)
\medskip
\\
The irreducible representations of $\Sigma $ (216 $\phi $) are
\medskip
\\
\begin{tiny}
\noindent\(\pmb{\text{irreps}=}\\
\pmb{\text{Flatten}\left[\left(
\begin{array}{cccccccccccccccccccccccc}
 1_1 & 1_2 & 1_3 & 2_1 & 2_2 & 2_3 & 3_1 & 3_2 & 3_3 & 3_4 & 3_5 & 3_6 & 3_7 & 6_1 & 6_2 & 6_3 & 6_4 & 6_5 & 6_6 & 8_1 & 8_2 & 8_3 & 9_1 & 9_2
\end{array}
\right)\right];}\)
\end{tiny}
\medskip
\\
\noindent\(\pmb{\text{ii}=7;}\\
\pmb{\text{While}[\text{ii}\leq  13,}\\
\pmb{\text{jj}=\text{ii};}\\
\pmb{\text{While}[\text{jj}\leq 13,}\\
\pmb{\text{tt}=\text{Table}[\text{FullSimplify}[\text{bb}[\text{kk},\text{ii},\text{jj}]],\{\text{kk},1,\text{nc}\}];}\\
\pmb{\text{Print}[\text{irreps}[[\text{ii}]],\texttt{"}\otimes \texttt{"},\text{irreps}[[\text{jj}]],\text{$\texttt{"}$=$\texttt{"}$},\text{tt}.\text{irreps}];}\\
\pmb{\text{jj}\text{++}];}\\
\pmb{\text{ii}\text{++}];}\)
\bigskip
\\
\noindent\(3_1\otimes 3_1=1_1+1_2+1_3+2\enspace 3_1\)

\noindent\(3_1\otimes 3_2=9_1\)

\noindent\(3_1\otimes 3_3=9_1\)

\noindent\(3_1\otimes 3_4=9_1\)

\noindent\(3_1\otimes 3_5=9_2\)

\noindent\(3_1\otimes 3_6=9_2\)

\noindent\(3_1\otimes 3_7=9_2\)

\noindent\(3_2\otimes 3_2=3_6+6_5\)

\noindent\(3_2\otimes 3_3=3_7+6_6\)

\noindent\(3_2\otimes 3_4=3_5+6_4\)

\noindent\(3_2\otimes 3_5=1_3+8_3\)

\noindent\(3_2\otimes 3_6=1_1+8_1\)

\noindent\(3_2\otimes 3_7=1_2+8_2\)

\noindent\(3_3\otimes 3_3=3_5+6_4\)

\noindent\(3_3\otimes 3_4=3_6+6_5\)

\noindent\(3_3\otimes 3_5=1_1+8_1\)

\noindent\(3_3\otimes 3_6=1_2+8_2\)

\noindent\(3_3\otimes 3_7=1_3+8_3\)

\noindent\(3_4\otimes 3_4=3_7+6_6\)

\noindent\(3_4\otimes 3_5=1_2+8_2\)

\noindent\(3_4\otimes 3_6=1_3+8_3\)

\noindent\(3_4\otimes 3_7=1_1+8_1\)

\noindent\(3_5\otimes 3_5=3_3+6_1\)

\noindent\(3_5\otimes 3_6=3_4+6_2\)

\noindent\(3_5\otimes 3_7=3_2+6_3\)

\noindent\(3_6\otimes 3_6=3_2+6_3\)

\noindent\(3_6\otimes 3_7=3_3+6_1\)

\noindent\(3_7\otimes 3_7=3_4+6_2\)
\end{small}

\section{Example for a calculation of the elements of a group from its generators}
\begin{small}
In this \textit{Mathematica 6} notebook we will construct the group $\Sigma $ (216$\phi $) from its generators.
\medskip
\\
\noindent\(\pmb{\text{Remove}[\text{$\texttt{"}$Global$\grave{ }$*$\texttt{"}$}]}\)

\noindent\(\pmb{\omega =\text{Exp}[2 \pi  i / 3];}\)

\noindent\(\pmb{\epsilon =\text{Exp}[4 \pi  i / 9];}\)
\medskip
\\
The generators of $\Sigma $ (216$\phi $) are :
\medskip
\\
\noindent\(\pmb{\text{A1}=\left(
\begin{array}{ccc}
 1 & 0 & 0 \\
 0 & \omega  & 0 \\
 0 & 0 & \omega ^2
\end{array}
\right);}\\
\pmb{\text{A2}=\left(
\begin{array}{ccc}
 0 & 1 & 0 \\
 0 & 0 & 1 \\
 1 & 0 & 0
\end{array}
\right);}\\
\pmb{\text{A3}=\frac{1}{\omega -\omega ^2}\left(
\begin{array}{ccc}
 1 & 1 & 1 \\
 1 & \omega  & \omega ^2 \\
 1 & \omega ^2 & \omega 
\end{array}
\right);}\\
\pmb{\text{A4}=\epsilon \left(
\begin{array}{ccc}
 1 & 0 & 0 \\
 0 & 1 & 0 \\
 0 & 0 & \omega 
\end{array}
\right);}\)
\medskip
\\
Using the following program one finds, that A3 and A4 alone generate $\Sigma $ (216 $\phi $). Thus in this example we just use A3 and A4 as generators.
\medskip
\\
At first we will multiply each generator with each generator and look for new group elements. Then we take the generators and the new group elements
and multiply each with each, ... . This is done until there occur no new group elements, then the group is complete.
\medskip
\\
\noindent\(\pmb{L = \{N[\text{A3}],N[\text{A4}]\};}\\
\pmb{\text{L2}= \{\text{A3},\text{A4}\};}\\
\pmb{t = \text{SessionTime}[];}\\
\pmb{\text{While}[0 < 1,}\\
\pmb{\quad i = 1;}\\
\pmb{\quad j = 1;}\\
\pmb{\quad d = \text{Dimensions}[L][[1]];}\\
\pmb{\quad m=0;}\\
\pmb{\quad\quad \text{While}[i \leq  d,}\\
\pmb{\quad\quad\quad \text{While}[j \leq  d,}\\
\pmb{\quad\quad\quad\quad X = L[[i]].L[[j]];}\\
\pmb{\quad\quad\quad\quad k = 1;}\\
\pmb{\quad\quad\quad\quad p = 0;}\\
\pmb{\quad\quad\quad\quad m\text{++};}\\
\pmb{\quad\quad\quad\quad \text{While}[k \leq  \text{Dimensions}[L][[1]],r = 1;q=1;}\\
\pmb{\quad\quad\quad\quad\quad \text{While}[r \leq  3, s = 1;}\\
\pmb{\quad\quad\quad\quad\quad\quad \text{If}[q==0,\text{Break}[]];}\\
\pmb{\quad\quad\quad\quad\quad\quad \text{While}[s \leq  3,}\\
\pmb{\quad\quad\quad\quad\quad\quad\quad \text{If}[\text{Abs}[\text{Chop}[X[[r]][[s]]] - \text{Chop}[L[[k]][[r]][[s]]]] > 0.000001, }\\
\pmb{\quad\quad\quad\quad\quad\quad\quad q=0;\text{Break}[]];}\\
\pmb{\quad\quad\quad\quad\quad\quad s\text{++}];}\\
\pmb{\quad\quad\quad\quad\quad r\text{++}];}\\
\pmb{\quad\quad\quad\quad\quad \text{If}[q==1, p=1;\text{Break}[]];}\\
\pmb{\quad\quad\quad\quad k\text{++}];}\\
\pmb{\quad\quad\quad \text{If}[p ==0, L = \text{Join}[L, \{X\}];\text{L2}=\text{Join}[\text{L2},\{\text{Simplify}[\text{L2}[[i]].\text{L2}[[j]]]\}]];}\\
\pmb{\quad\quad\quad j\text{++}];}\\
\pmb{\quad\quad j = 1;}\\
\pmb{\quad\quad i\text{++}];}\\
\pmb{\quad \text{Print}[\text{$\texttt{"}$Step finished: $\texttt{"}$},\{\text{Dimensions}[L][[1]], \text{SessionTime}[] - t,m\}];}\\
\pmb{\quad \text{   }\text{If}[d - \text{Dimensions}[L][[1]] \text{==} 0,}\\
\pmb{\quad \text{Print}[\text{$\texttt{"}$Finished, total time: $\texttt{"}$},\text{SessionTime}[] - t];}\\
\pmb{\quad \text{Export}[\text{$\texttt{"}$Sigma216.nb$\texttt{"}$},\text{L2}]; \text{Beep}[]; \text{Break}[]]}\\
\pmb{]}\)

\noindent\(\text{Step finished: }\{6,0.0156250,4\}\)

\noindent\(\text{Step finished: }\{29,0.1562500,36\}\)

\noindent\(\text{Step finished: }\{242,4.6250000,841\}\)

\noindent\(\text{Step finished: }\{648,1081.5156250,58564\}\)

\noindent\(\text{Step finished: }\{648,9542.8281250,419904\}\)

\noindent\(\text{Finished, total time: }9542.8281250\)
\bigskip
\\
Notation : $\{$a, b, c$\}$\\
a ... number of elements that have been found yet,\\
b ... current session time in seconds,\\
c ... current number of performed matrix multiplications.
\bigskip
\\
Thus we see, that $\Sigma(216\phi)$ has 648 elements.
\end{small}

%% file: symbols/symbols.tex
%\chapter*{Mathematical Symbols\small{\rmfamily{\markboth{Mathematical Symbols}{Mathematical Symbols}}}}
\newpage

\addcontentsline{toc}{chapter}{Mathematical symbols, Abbreviations}

\fancyhf{}
\fancyhead[RO,LE]{ \thepage}
\fancyhead[RE]{\small \rmfamily \nouppercase{Mathematical symbols, Abbreviations}}
\fancyhead[LO]{\small \rmfamily \nouppercase{Mathematical symbols, Abbreviations}}
\fancyfoot{} %no footline, especially no page numbering on the lower side of the page

\begin{flushleft}
\textbf{Mathematical symbols}
\end{flushleft}

\begin{center}
% use packages: array
\begin{tabular}{ll}
$A_{(i_{1}...i_{n})}$ & $\quad$ $\frac{1}{n!}\sum_{\pi\in\mathcal{S}_{n}}A_{\pi(i_{i})...\pi(i_{n})}$ (total symmetrization)
\medskip
\\
$A_{[i_{1}...i_{n}]}$ & $\quad$ $\frac{1}{n!}\sum_{\pi\in\mathcal{S}_{n}}sign(\pi)A_{\pi(i_{i})...\pi(i_{n})}$ (total antisymmetrization)
\medskip
\\

$\mathrm{Span}(M)$ & $\quad$ Linear span of the set $M$
\medskip
\\

$[A]$ or $[A_{ij}]$ & $\quad$ Matrix representation of the linear operator $A$
\medskip
\\

$\lbrace e_{j}\rbrace_{j}$ with $j\in M$ & $\quad$ $\lbrace e_{a},e_{b},...\rbrace$ where $M=\lbrace a,b,...\rbrace\subset \mathbb{N}\backslash \lbrace 0\rbrace$
\medskip
\\

$\mathrm{Tr}(A)$ & $\quad$ Trace of the linear operator $A$
\medskip
\\

$A^{T}$ & $\quad$ Transposed of the matrix $A$
\medskip
\\

$V\simeq W$ & $\quad$ The vectorspaces $V$ and $W$ are isomorphic.
\medskip
\\

$a^{\ast}$ & $\quad$ Complex conjugate of $a$
\medskip
\\

$A^{\dagger}$ & $\quad$ Hermitian conjugate of $A$: $A^{\dagger}=A^{T\ast}$ for matrices.

\\
& $\quad$ For operators $A^{\dagger}$ is the adjoint operator.
\medskip
\\

H.c. & $\quad$ Hermitian conjugate of an expression
\medskip
\\

$\mathbbm{1}_{n}$ & $\quad$ The $n\times n$-identity matrix
\medskip
\\

$\mathrm{gen}(G)$ & $\quad$ A set of generators of a group $G$
\medskip
\\

$\mathbb{N}$ & $\quad$ The natural numbers $\mathbb{N}=\{0,1,2,...\}$
\medskip
\\

$\mathbb{Z}$ & $\quad$ The integer numbers $\mathbb{Z}=\{...,-2,-1,0,1,2,...\}$
\medskip
\\

$\mathbb{R}$ & $\quad$ The real numbers
\medskip
\\

$\mathbb{C}$ & $\quad$ The complex numbers
\medskip
\\

$\mathbb{Q}$ & $\quad$ The rational numbers
\medskip
\\

$\oplus$ & $\quad$ Direct sum
\medskip
\\

$\otimes$ & $\quad$ Tensor product
\medskip
\\

$G\rtimes H$ & $\quad$ The semidirect product of the two groups $G$ and $H$
\medskip
\\

$e$ or $E$ & $\quad$ Identity element of a group
\medskip
\\

$G/H$ & $\quad$ Factor group
\medskip
\\

$\circ$ & $\quad$ Composition
\medskip
\\

$a\sim b$ & $\quad$ $a$ is conjugate to $b$.
\medskip
\\

$\mathrm{ker}$ & $\quad$ Kernel of a homomorphism (or a linear map)
\medskip
\\
$\mathrm{im}$ & $\quad$ Image of a map
\end{tabular}
\end{center}

\newpage
\begin{flushleft}
\textbf{Mathematical symbols (continued)}
\end{flushleft}

\begin{center}
% use packages: array
\begin{tabular}{ll}
$\mathrm{Mat}(n,\mathbb{F})$ & $\quad$ The set of invertible $n\times n$-matrices over the field $\mathbb{F}$.
\medskip
\\

$SU(n)$ & $\quad$ The set of unitary $n\times n$-matrices of determinant 1.
\medskip
\\

$SO(n)$ & $\quad$ The set of orthogonal $n\times n$-matrices of determinant 1.
\medskip
\\

$SL(n,\mathbb{F})$ & $\quad$ The set of $n\times n$-matrices of determinant 1 over the field $\mathbb{F}$.
\medskip
\\
\end{tabular}
\end{center}

\bigskip
\bigskip

\begin{flushleft}
\textbf{Abbreviations}
\end{flushleft}

\begin{center}
% use packages: array
\begin{tabular}{ll}
i.s. & in symbols \\
p. & page\\
s.t. & such that \\ 
w.l.o.g. & without loss of generality \\ 
w.r.t. & with respect to
\end{tabular}
\end{center} 
 

%% file: conventions/conventions.tex
\newpage

\addcontentsline{toc}{chapter}{Conventions}

\fancyhf{}
\fancyhead[RO,LE]{ \thepage}
\fancyhead[RE]{\small \rmfamily \nouppercase{Conventions}}
\fancyhead[LO]{\small \rmfamily \nouppercase{Conventions}}
\fancyfoot{} %no footline, especially no page numbering on the lower side of the page

\begin{flushleft}
\large{\textbf{Conventions}}
\end{flushleft}

\begin{flushleft}
\large{\textbf{Indices and summation}}
\end{flushleft}
\medskip
If not stated otherwise we will always use Einstein's summation convention, which means that it is always summed over equal indices.
\medskip
\\
\textbf{Lorentz-indices}: \textit{Greek} letters, $\mu=0,1,2,3$.
\medskip
\\
\textbf{Spatial indices}: \textit{Latin} letters, $i=1,2,3$.
\bigskip
\begin{flushleft}
\large{\textbf{Minkowski-metric}}
\end{flushleft}
We will use the following sign convention for the Minkowski-metric in Cartesian coordinates:
	\begin{displaymath}
		(\eta_{\mu\nu})=\left(
			\begin{matrix}
			 1 & 0 & 0 & 0 \\
			 0 & -1 & 0 & 0 \\
			 0 & 0 & -1 & 0 \\
			 0 & 0 & 0 & -1
			\end{matrix}
			\right).
	\end{displaymath}
\bigskip
\begin{flushleft}
\large{\textbf{su(2)}}
\end{flushleft}
	\begin{displaymath}
	su(2)=\{A\in \mathrm{Mat}(2,\mathbb{C})\vert A^{\dagger}=-A, \mathrm{Tr}A=0\}.
	\end{displaymath}
We use the Pauli matrices
	\begin{displaymath}
		\sigma^{1}=\left(
\begin{matrix}
 0 & 1 \\
 1 & 0
\end{matrix}
\right), \sigma^{2}=\left(
\begin{matrix}
 0 & -i \\
 i & 0
\end{matrix}
\right), \sigma^{3}=\left(
\begin{matrix}
 1 & 0 \\
 0 & -1
\end{matrix}
\right)
	\end{displaymath}
as a basis of $isu(2)=\{iA\vert A\in su(2)\}$. They satisfy the commutation relations
	\begin{displaymath}
		[\frac{\sigma^{i}}{2},\frac{\sigma^{j}}{2}]=i\varepsilon_{ijk}\frac{\sigma^{k}}{2},
	\end{displaymath}
which makes them to generators of $SU(2)$. If we work with $SU(2)_{I}$ we will use the symbol $\tau^{i}$ instead of $\sigma^{i}$.

\newpage
\begin{flushleft}
\large{\textbf{Dirac-$\gamma$-matrices}}
\end{flushleft}
If not stated otherwise we will always use the \textit{chiral representation} of the Dirac-$\gamma$-matrices $\gamma^{\mu}$, which is also called \textit{Weyl basis}.
	\begin{displaymath}
		\gamma^{0}=\left(\begin{matrix}\textbf{0}_{2} & \mathbbm{1}_{2}\\
				 \mathbbm{1}_{2} & \textbf{0}_{2}
				\end{matrix}\right),
				\gamma^{i}=\left(\begin{matrix}
				 \textbf{0}_{2} & \sigma^{i}\\
				 -\sigma^{i} & \textbf{0}_{2}
				\end{matrix}\right).
	\end{displaymath}
	\begin{displaymath}
		\gamma^{5}=i\varepsilon_{\mu\nu\sigma\lambda}\gamma^{\mu}\gamma^{\nu}\gamma^{\sigma}\gamma^{\lambda}=i\gamma^{0}\gamma^{1}\gamma^{2}\gamma^{3}=\left(\begin{matrix}-\mathbbm{1}_{2} & \textbf{0}_{2}\\
				 \textbf{0}_{2} & \mathbbm{1}_{2}
				\end{matrix}\right).
	\end{displaymath}
Here $\mathbbm{1}_{2}$ is the $2\times2$-identity matrix, and $\textbf{0}_{2}$ is the $2\times 2$-null matrix.
\bigskip
\\
\textbf{Properties of the $\gamma$-matrices in the chiral representation}
\smallskip
\\
Of course the $\gamma$-matrices fulfil the required anticommutation relations
	\begin{displaymath}
		\{\gamma^{\mu},\gamma^{\nu}\}=2\eta^{\mu\nu} \mathbbm{1}_{4}.
	\end{displaymath}
Furthermore they have the following properties:
	\begin{itemize}
		\item $\gamma^{0\ast}=(\gamma^{0})^{T}=\gamma^{0\dagger}=\gamma^{0}$,
		\item $\gamma^{i\dagger}=-\gamma^{i}$, $\gamma^{5\dagger}=\gamma^{5}$,
		\item $(\gamma^{0})^{2}=\mathbbm{1}_{4}$, $(\gamma^{i})^{2}=-\mathbbm{1}_{4}$, $(\gamma^{5})^{2}=\mathbbm{1}_{4}$,
		\item $\gamma^{0}\gamma^{\mu\dagger}\gamma^{0}=\gamma^{\mu}$.
	\end{itemize}

%% file: zusammenfassung/zusammenfassung.tex
\chapter*{Zusammenfassung}
W\"ahrend die Eichbosonen $\gamma, W^{\pm}$ und $Z$ physikalische Manifestationen des ma\-thematischen Formalismus der Eichtheorien sind, gibt es bis heute keine zufrie\-denstellende Erkl\"arung f\"ur den Fermiongehalt des Standardmodells der Teil\-chenphysik. Die Suche nach einer mathematischen Theorie welche die (zur Zeit bekannten) 3 Generationen von Fermionen erkl\"art, blieb im 20. Jahrhundert erfolglos. Es besteht jedoch Hoffnung, dass am Beginn des 21. Jahrhunderts zumindest ein kleiner Lichtschimmer auf das im Dunkel liegende Problem geworfen werden kann. 
\\
Die fesselndste experimentelle Tatsache ist, dass, w\"ahrend die Massenverh\"altnis\-se der Mitglieder einer der Fermion-Familien (geladene Leptonen, Neutrinos, up-Quarks, down-Quarks) sehr hohe Werte annehmen k\"onnen, wie zum Bei\-spiel 
	\begin{displaymath}
	\frac{m_{\tau}}{m_{e}}\sim 3500,
	\end{displaymath}
die Kopplungen an die Eichbosonen innerhalb einer Fermionfamilie gleich sind. Das ist der Grund weshalb die massiveren Mitglieder der Familien oft als \\
\textquotedblleft schwe\-re Verwandte\textquotedblright\hspace{1mm} der leichtes\-ten Familienmitglieder bezeichnet werden.
\\
Bis zum heutigen Tage gibt es keine zufriedenstellende Erkl\"arung f\"ur die Exis\-tenz der drei Generationen von Fermionen. Der einzige Punkt im Standardmodell, an dem ein Unterschied zwischen den Fermionengenerationen in der Lagrangedichte auftritt, ist in der Yukawa-Kopplung zu finden, \"uber welche die Fermionmassenterme mit Hilfe der spontanen Symmetriebrechung erzeugt werden. An diesem Punkt tritt eine Verbindung zu anderen messbaren Gr\"o\ss en auf - zu den Mischungsmatrizen $U_{CKM}$ und $U_{PMNS}$.
\\
Die Leptonmischungsmatrix $U_{PMNS}$ scheint beinahe die Werte der Harrison-Perkins-Scott Mischungsmatrix
	\begin{displaymath}
	U_{\mathrm{HPS}}=\left(\begin{matrix}
		 \sqrt{\frac{2}{3}} & \frac{1}{\sqrt{3}} & 0 \\
		 -\frac{1}{\sqrt{6}} & \frac{1}{\sqrt{3}} & \frac{1}{\sqrt{2}} \\
		 \frac{1}{\sqrt{6}} & -\frac{1}{\sqrt{3}} & \frac{1}{\sqrt{2}}
	\end{matrix}\right)
	\end{displaymath}
anzunehmen. Die sch\"one und interessante Gestalt dieser Matrix hat einige Physiker dazu veranlasst, die Idee einer diskreten Symmetrie in der Lagrange\-dichte des Leptonsektors zu entwickeln, welche gerade diese Mischungsmatrix verlangt. In dieser Arbeit werden Lagrangedichten eines um drei rechtsh\"andi\-ge Neutrinos erweiterten Standardmodells untersucht und die Konsequenzen einer Invarianz unter diskreten Symmetrien erl\"autert. Der Hauptteil der Arbeit ist der systematischen Untersuchung von endlichen Untergruppen der $SU(3)$ gewidmet. Ich hoffe, dass meine Analyse dieser Gruppen in Zukunft als Grundlage zur Konstruktion und Untersuchung neuer Modelle dienen kann.
\medskip
\\
Die zuk\"unftige Entdeckung einer geeigneten endlichen Gruppe, die sowohl Leptonmassen als auch Mischungswinkel beschreiben kann, wird das Mys\-te\-ri\-um der drei Generationen von Fermionen nicht l\"osen k\"onnen, aber viel\-leicht kann eine geeignete Symmetrie einen Hinweis auf eine allgemeinere Theorie liefern, welche das heutige Standardmodell in einem bestimmten Limes enth\"alt.

%% file: cv/cv.tex
\section*{Curriculum vitae}
Ich wurde am 22. M\"arz 1985 in Wiener Neustadt (Nieder\"osterreich) als Sohn von Emmerich und Maria Gabriele Ludl geboren.
\bigskip
\\
\textbf{1991-1995} Besuch der Volksschule Baumkirchnerring-West in Wiener Neustadt.\hspace{10mm}
\\
\textbf{1995-2003} Besuch des Bundesrealgymnasiums Gr\"ohrm\"uhlgasse mit naturwissenschaftlichem Schwerpunkt in Wiener Neustadt. Wahlpflichtf\"acher: Phy\-sik, Chemie.
\bigskip
\\
\textbf{1999-2003} Teilnahme an den Landeswettbewerben des Bundeslandes Nieder\-\"osterreich der \"Osterreichischen Chemie-Olympia\-de.
\bigskip
\\
\textbf{2000-2003} Teilnahme an den Bundeswettbewerben der \"Osterreichischen Che\-mie-Olympia\-de.
\bigskip
\\
\textbf{5.7.-14.7. 2002} Silbermedaille bei der 34. internationalen Che\-mie-Olympia\-de in Groningen (Niederlande).
\bigskip
\\
\textbf{Schuljahr 2002/2003} Verfassen einer Fachbereichsarbeit im Fach Chemie zum Thema \textquotedblleft\textit{Die Chemie des Stickstoffs und die Rolle seiner Verbindungen f\"ur die Chemie der Atmosph\"are}\textquotedblright\hspace{0.5mm} bei Mag. Jana Ederova.
\bigskip
\\
\textbf{5.6.2003} Reifepr\"ufung mit ausgezeichnetem Erfolg bestanden.
\bigskip
\\
\textbf{5.7.-14.7. 2003} Silbermedaille bei der 35. internationalen Chemie-Olympia\-de in Athen (Griechenland).
\bigskip
\\
\textbf{1.9.2003-30.4.2004} Pr\"asenzdienst.
\bigskip
\\
\textbf{Ab Mai 2004} Studium an der Universit\"at Wien.
\bigskip
\\
\textbf{Seit Oktober 2004} Studium der Physik an der Universit\"at Wien.
\bigskip
\\
\textbf{7.10.2005} 1. Diplompr\"ufung mit Auszeichnung bestanden.
\bigskip
\\
\textbf{11.7.2007} 2. Diplompr\"ufung mit Auszeichnung bestanden.
\bigskip
\\
\textbf{August 2008 - Juni 2009} Diplomarbeit bei Ao. Univ.-Prof. Dr. Walter Grimus.